\documentclass[conference]{IEEEtran}
\IEEEoverridecommandlockouts
\def\BibTeX{{\rm B\kern-.05em{\sc i\kern-.025em b}\kern-.08em
    T\kern-.1667em\lower.7ex\hbox{E}\kern-.125emX}}

\usepackage[show]{chato-notes}
\usepackage[full]{switch-version} % to switch between a 'full' (arXiv) and 'review' version

\newcommand{\revision}[1]{\textcolor{black}{#1}}

\newcommand{\arXivtitle}{Query the model: precomputations for efficient inference with Bayesian Networks}

\newcommand{\reviewtitle}{Workload-aware Materialization for Efficient Variable Elimination on Bayesian Networks}

\newcommand{\hide}[1]{\iffalse{#1}\fi}
\usepackage{algorithm}
\usepackage{graphicx}
\usepackage{cite}
\usepackage{amsmath,amssymb,amsfonts,amsthm}
\usepackage{graphicx}
\usepackage{textcomp}
\usepackage{xcolor}

\usepackage{balance}  % for  \balance command ON LAST PAGE  (only there!)
\usepackage{xspace}
\usepackage{booktabs}
\usepackage{multirow}
\usepackage{algorithmicx}
\usepackage[noend]{algpseudocode}
\usepackage{tablefootnote}
\usepackage{xcolor}
\usepackage{url}

\newcommand{\CA}[1]{\textcolor{black}{#1}}

\newcommand{\eat}[1]{} % to hide text
\newtheorem{theorem}{Theorem}
\newtheorem{definition}{Definition}
\newtheorem{problem}{Problem}
\newtheorem{lemma}{Lemma}

\newcommand{\np}{\ensuremath{\mathbf{NP}}}
\newcommand{\nphard}{{{\np}-hard}}

\newcommand{\preorderTraversal}{{\tt pot}\xspace}
\newcommand{\DAG}{{\sc\large dag}\xspace}
\newcommand{\separated}{{\ensuremath{m}-separated}\xspace}
\newcommand{\isseparated}[3]{\ensuremath{\mathrm{sep}\!\left({#1},{#2},{#3}\right)}\xspace}

\newcommand{\bayesianNet}{\ensuremath{\mathcal{N}}\xspace}
\newcommand{\otherBayesianNet}{\ensuremath{\bayesianNet^{\prime}}\xspace}
\newcommand{\yetAnotherBayesianNet}{\ensuremath{\bayesianNet^{\prime\prime}}\xspace}
\newcommand{\shrunkBayesianNet}{\ensuremath{\mathcal{S}}\xspace}
\newcommand{\resultBN}{\ensuremath{\bayesianNet_{s}}\xspace}
\newcommand{\moralGraph}{\ensuremath{\mathcal{M}}\xspace}
\newcommand{\bn}{Bayesian network\xspace}

\newcommand{\bigtable}{\ensuremath{{H}}\xspace}
\newcommand{\ordergraph}{\ensuremath{\mathcal{H}}\xspace}

\newcommand{\domainsize}{\ensuremath{\mathfrak{\nu}}\xspace}
\newcommand{\eliminationOrder}{\ensuremath{\sigma}\xspace}
\newcommand{\query}{\ensuremath{q}\xspace}
\newcommand{\queriesset}{\ensuremath{\mathcal{Q}}\xspace}
\newcommand{\qvarfree}{\ensuremath{X_{\query}}\xspace}
\newcommand{\qvarbound}{\ensuremath{Y_{\query}}\xspace}
\newcommand{\qvarvals}{\ensuremath{\mathbf{y}_{\query}}\xspace}
\newcommand{\qvarmiss}{\ensuremath{Z_{\query}}\xspace}
\newcommand{\qsize}{\ensuremath{r_{\query}}\xspace}

\newcommand{\allvars}{\ensuremath{X}\xspace}
\newcommand{\somevars}{\ensuremath{U}\xspace}

\newcommand{\vars}[1]{\ensuremath{\allvars_{#1}}\xspace}
\newcommand{\redundantmarginal}{\ensuremath{R_m}\xspace}
\newcommand{\redundantconditional}{\ensuremath{R_c}\xspace}
\newcommand{\redundantvars}{\ensuremath{R}\xspace}

\newcommand{\lowqvarval}{\ensuremath{\mathbf{y}_{\ell}}\xspace}
\newcommand{\highqvarval}{\ensuremath{\mathbf{y}_{u}}\xspace}

\newcommand{\vara}{\ensuremath{a}\xspace}

\newcommand{\varb}{\ensuremath{b}\xspace}
\newcommand{\varbvalue}{\ensuremath{{b_0}}\xspace}
\newcommand{\varc}{\ensuremath{c}\xspace}

\newcommand{\varxi}{\ensuremath{\xi}\xspace}
\newcommand{\varz}{\ensuremath{z}\xspace}

\newcommand{\varomega}{\ensuremath{\omega}\xspace}

\newcommand{\treebudgetallocation}{\ensuremath{\partialbudget_\nodeu^*}\xspace}
\newcommand{\tree}{\ensuremath{T}\xspace}
\newcommand{\otherTree}{\ensuremath{T'}\xspace}
\newcommand{\coreTree}{\ensuremath{T_{in}}\xspace}
\newcommand{\otherCoreTree}{\ensuremath{T'_{in}}\xspace}
\newcommand{\nodes}{\ensuremath{V}\xspace}
\newcommand{\edges}{\ensuremath{E}\xspace}
\newcommand{\nonodes}{\ensuremath{n}\xspace}
\newcommand{\height}{\ensuremath{h}\xspace}
\newcommand{\nodeu}{\ensuremath{u}\xspace}
\newcommand{\nodew}{\ensuremath{w}\xspace}
\newcommand{\nodev}{\ensuremath{v}\xspace}
\newcommand{\nodex}{\ensuremath{x}\xspace}

\newcommand{\rout}{\ensuremath{r}\xspace}
\newcommand{\subtree}[1]{\ensuremath{{\tree}_{#1}}\xspace}
\newcommand{\ancestors}[1]{\ensuremath{A_{#1}}} 
\newcommand{\extendedancestors}[1]{\ensuremath{\bar{A}_{#1}}}
\newcommand{\lsa}[2]{\ensuremath{{a}_{#1}^{#2}}} %cigdem: lowest solution ancestor
\newcommand{\hsd}[2]{\ensuremath{\bar{D}_{#1}^{#2}}} %cigdem: set of highest solution descendants
\newcommand{\rchild}{\ensuremath{r}\xspace}
\newcommand{\lchild}{\ensuremath{\ell}\xspace}

\newcommand{\totalweight}[1]{\ensuremath{{S}_{#1}}\xspace}
\newcommand{\weight}[1]{\ensuremath{{s}_{#1}}\xspace}
\newcommand{\rightsub}[1]{\ensuremath{{r}({#1})}\xspace}
\newcommand{\leftsub}[1]{\ensuremath{{\ell}({#1})}\xspace}

\newcommand{\solution}{\ensuremath{R}\xspace}

\newcommand{\budget}{\ensuremath{k}\xspace}
\newcommand{\spacebudget}{\ensuremath{K}\xspace}
\newcommand{\partialbudget}{\ensuremath{\kappa}\xspace}
\newcommand{\partialbudgetl}{\ensuremath{\kappa_\ell}\xspace}
\newcommand{\partialbudgetr}{\ensuremath{\kappa_r}\xspace}

\newcommand{\DPtable}[3]{\ensuremath{F({#1},{#2},{#3})}\xspace}
\newcommand{\DPtableInc}[3]{\ensuremath{F^+({#1},{#2},{#3})}\xspace}
\newcommand{\DPtableExc}[3]{\ensuremath{F^-({#1},{#2},{#3})}\xspace}

\newcommand{\nocap}{\ensuremath{\epsilon}\xspace}
\newcommand{\bigO}[1]{\ensuremath{\mathcal{O}\!\left({#1}\right)}\xspace}
\newcommand{\probability}{\ensuremath{\mathrm{Pr}}\xspace}
\newcommand{\prob}[1]{\ensuremath{\probability\!\left({#1}\right)}\xspace} % cigdem: version that uses paranthesis instead of brackets
\newcommand{\traindistr}{\ensuremath{\mathrm{\mathcal{P}}}\xspace}
\newcommand{\testdistr}{\ensuremath{\mathrm{\mathcal{P}'}}\xspace}
\newcommand{\netprob}[2]{\ensuremath{\mathrm{Pr}_{#1}\!\left({#2}\right)}\xspace}
\newcommand{\factor}[2]{\ensuremath{\psi_{#1}\!\left({#2}\right)}\xspace}
\newcommand{\isuseful}[3]{\ensuremath{\delta_{#1}\!\left({#2};{#3}\right)}\xspace}
\newcommand{\utility}[1]{\ensuremath{b\!\left({#1}\right)}\xspace}
\newcommand{\cost}[1]{\ensuremath{c\!\left({#1}\right)}\xspace}
\newcommand{\benefit}[1]{\ensuremath{B\!\left({#1}\right)}\xspace}
\newcommand{\benefitfn}{\ensuremath{B}\xspace}
\newcommand{\partialbenefit}[2]{\ensuremath{B_{#1}\!\left({#2}\right)}\xspace}

\newcommand{\mgbenefit}[2]{\ensuremath{B\!\left({#1}\mid {#2}\right)}\xspace}
\newcommand{\posreals}{\mathbb{R}_{\ge 0}}
\newcommand{\expuseful}[2]{\ensuremath{{\mathrm{E}}\!\left[{\isuseful{\query}{{#1}}{{#2}}}\right]}\xspace}

\newcommand{\optimalbenefit}[2]{\ensuremath{B_{#1}\!\left({#2}\right)}\xspace}
\newcommand{\globalbenefit}[1]{\ensuremath{G\!\left({#1}\right)}\xspace}  
\newcommand{\globallyoptimalbenefit}[1]{\ensuremath{{OPT}_{#1}\!}\xspace}

\newcommand{\shrink}[1]{\ensuremath{\mathrm{shrink}\!\left({#1}\right)}\xspace}
\newcommand{\deque}[1]{\ensuremath{\mathrm{deque}\!\left({#1}\right)}\xspace}
\newcommand{\enque}[1]{\ensuremath{\mathrm{enque}\!\left({#1}\right)}\xspace}
\newcommand{\lattice}{\ensuremath{\mathcal{L}}\xspace}
\newcommand{\fullLattice}{\ensuremath{\mathcal{L}^{+}}\xspace}
\newcommand{\latticeSize}{\ensuremath{{\ell}}\xspace}
\newcommand{\children}[1]{\ensuremath{\mathrm{children}\!\left({#1}\right)}\xspace}
\newcommand{\mapping}{\ensuremath{M}\xspace}
\newcommand{\bnProb}{\ensuremath{\pi}\xspace}
\newcommand{\fullBNProb}{\ensuremath{\rho}\xspace}
\newcommand{\queue}{\ensuremath{Q}\xspace}

\newcommand{\nodepath}[2]{\ensuremath{\mathit{path}({#1}, {#2})}}

\DeclareMathOperator*{\assign}{:=}

\newcommand{\survey}{\ensuremath{\text{\tt survey}}\xspace}
\newcommand{\age}{\ensuremath{\text{\tt A}}\xspace}
\newcommand{\sex}{\ensuremath{\text{\tt S}}\xspace}
\newcommand{\education}{\ensuremath{\text{\tt E}}\xspace}
\newcommand{\occupation}{\ensuremath{\text{\tt O}}\xspace}
\newcommand{\residence}{\ensuremath{\text{\tt R}}\xspace}
\newcommand{\transport}{\ensuremath{\text{\tt T}}\xspace}

\newcommand{\cpcs}{\ensuremath{\text{\sc\large cpcs}}\xspace}

\newcommand{\diabetes}{\textsc{Diabetes}\xspace}
\newcommand{\bnpathfinder}{\textsc{Pathfinder}\xspace}
\newcommand{\andes}{\textsc{Andes}\xspace}
\newcommand{\link}{\textsc{Link}\xspace}
\newcommand{\muninb}{\textsc{Munin}\xspace} %% big munin
\newcommand{\munins}{\textsc{Munin}{\small\#1}\xspace} %% small munin
\newcommand{\muninm}{\textsc{Munin}{\small\#2}\xspace} %% medium munin
\newcommand{\mildew}{\textsc{Mildew}\xspace}
\newcommand{\tpch}{\textsc{TpcH}\xspace}
\newcommand{\tpchsmall}{\textsc{\tpch}{\small\#1}\xspace}
\newcommand{\tpchmedium}{\textsc{\tpch}{\small\#2}\xspace}
\newcommand{\tpchverylarge}{\textsc{\tpch}{\small\#3}\xspace}
\newcommand{\tpchlarge}{\textsc{\tpch}{\small\#4}\xspace}

\newcommand{\qtm}[1]{{VE-{#1}}\xspace}
\newcommand{\jt}{\textsc{JT}\xspace}
\newcommand{\kanagal}{\textsc{IND}\xspace}

\newcommand{\wmf}{\textsc{wmf}\xspace}
\newcommand{\mf}{\textsc{mf}\xspace}
\newcommand{\mw}{\textsc{mw}\xspace}
\newcommand{\mn}{\textsc{mn}\xspace}

\newcommand{\spara}[1]{\smallskip\noindent{\bf{#1}}}
\newcommand{\mpara}[1]{\medskip\noindent{\bf{#1}}}

\newcommand{\squishlist}{\begin{list}{$\bullet$}
  {\setlength{\itemsep}{0pt}
    \setlength{\parsep}{3pt}
    \setlength{\topsep}{3pt}
    \setlength{\partopsep}{0pt}
    \setlength{\leftmargin}{1.5em}
    \setlength{\labelwidth}{1em}
    \setlength{\labelsep}{0.5em}}}
\newcommand{\squishend}{
\end{list}}

\newcommand{\varelim}{\textsc{Variable Elimination}\xspace}

\makeatletter
\renewcommand{\cleardoublepage}{\clearpage \ifodd\c@page\else
	\hbox{}\newpage\if@twocolumn\hbox{}\newpage\fi\fi}
\makeatother

\begin{document}

%Subtitles are not captured in Xplore 
%\ReviewOnly{\title{[\reviewtitlenobreak]{\reviewtitle}}}
\ReviewOnly{\title{\reviewtitle}}
\FullOnly{\title{\arXivtitle}}

\author{\IEEEauthorblockN{Cigdem Aslay}
\IEEEauthorblockA{\textit{Aarhus University} \\
Aarhus, Denmark \\
cigdem@cs.au.dk}
\and
\IEEEauthorblockN{Martino Ciaperoni}
\IEEEauthorblockA{\textit{Aalto University} \\
Espoo, Finland \\
martino.ciaperoni@aalto.fi}
\and
\IEEEauthorblockN{Aristides Gionis}
\IEEEauthorblockA{
% \textit{\smallskip KTH Royal Institute of Technology}\\
\textit{KTH Royal Institute of Technology}\\
Stockholm, Sweden \\
argioni@kth.se}
\and
\IEEEauthorblockN{Michael Mathioudakis}
\IEEEauthorblockA{\textit{University of Helsinki} \\
Helsinki, Finland \\
michael.mathioudakis@helsinki.fi}
}

%\hide {
%	\author{
%	% The command \alignauthor (no curly braces needed) should
%	% precede each author name, affiliation/snail-mail address and
%	% e-mail address. Additionally, tag each line of
%	% affiliation/address with \affaddr, and tag the
%	% e-mail address with \email.
%	%
%	% 1st. author
%	\alignauthor
%	Cigdem Aslay\\% \titlenote{Note}\\
%	       \affaddr{\smallskip Aalto University}\\
%	       \email{\sf\large cigdem.aslay@aalto.fi}
%	% 2nd. author
%	\alignauthor
%	Aristides Gionis\\
%	       \affaddr{\smallskip Aalto University}\\
%	       \email{\sf\large aristides.gionis@aalto.fi}
%	% 3rd. author
%	\alignauthor 
%	Michael Mathioudakis\\
%	       \affaddr{\smallskip University of Helsinki}\\
%	       \email{\sf\large michael.mathioudakis@helsinki.fi}
%	% \and  % use '\and' if you need 'another row' of author names
%	}
%}

\maketitle

\ReviewOnly{
	\begin{abstract}
	Bayesian networks are general, well-studied probabilistic models that capture dependencies among a set of variables. 
	{\em Variable Elimination} is a fundamental algorithm for probabilistic inference over Bayesian networks.
	In this paper, we propose a novel materialization method, which can lead to significant efficiency gains when processing inference queries using the Variable Elimination algorithm.
	In particular, we address the problem of choosing a set of intermediate results to precompute and materialize, so as to maximize the expected efficiency gain over a given query workload.
	For the problem we consider, we provide an optimal polynomial-time algorithm and discuss alternative methods.
	We validate our technique using real-world Bayesian networks.
	Our experimental results confirm that a modest amount of materialization can lead to significant improvements in the running time of queries, with an average gain of $70\%$, and reaching up to a gain of $99\%$, for a uniform workload of queries.
	Moreover, in comparison with existing junction tree methods that also rely on materialization, our approach achieves competitive efficiency during inference using significantly lighter materialization.
	\end{abstract}
}
\FullOnly{
	\begin{abstract}
	Variable Elimination is a fundamental algorithm for probabilistic inference over Bayesian networks.
	In this paper, we propose a novel materialization method for Variable Elimination, which can lead to significant efficiency gains when answering inference queries.
	We evaluate our technique using real-world Bayesian networks.
	Our results show that a modest amount of materialization can lead to significant improvements in the running time of queries.
	Furthermore, in comparison with junction tree methods that also rely on materialization, our approach achieves comparable efficiency during inference using significantly lighter materialization.
	\end{abstract}
}

\section{Introduction}
\label{sec:introduction}

Research in \emph{machine learning} has led to powerful methods for building probabilistic models 
which are used to carry out general predictive tasks~\cite{bishop2013model}.
% As a result, machine learning is currently employed in a wide range of fields, 
% enabling us to automate tasks that until recently were seen as particularly challenging. 
% %
% Examples include image and speech recognition~\cite{NIPS2018_imagerecognition, NIPS2018_speechrecognition}, 
% % natural language processing~\cite{NIPS2018_NLP}, 
% and machine translation~\cite{NIPS2018_Translation}.
%
In a typical setting, a model is trained off\-line from available data and it is subsequently employed online for processing probabilistic inference queries. %, i.e., to compute probabilities from the model.
For example, a joint-probability model trained off\-line over the attributes of a relational database can be employed by the DBMS at query time to produce selectivity estimates for each query, for the purpose of query optimization~\cite{getoor2001selectivity}.
For efficiency in such a setting, it is important to consider computational costs not only for 
training the model, but also for performing inference at query time.
In this paper, we focus on the efficiency of inference in Bayesian networks.

Bayesian networks are probabilistic models 
that capture the joint distribution of a set of variables through local 
dependencies between small groups of variables~\cite{pearl2014probabilistic}. 
They have an intuitive representation in terms of a directed acyclic graph (each directed edge between two variables represents a dependency of the child on the parent variable) and probabilities %based on the model 
can be evaluated with compact sum-of-products computations. 
Being intuitive and modular makes Bayesian networks suitable for general settings where one wishes to represent the joint distribution of a large number of variables, e.g., the column attributes of a large relational table as the one shown in Figure~\ref{figure:introduction}. 
Nevertheless, inference in Bayesian networks is \nphard\
and one cannot preclude the possibility that the evaluation of some queries becomes expensive for cases of interest.

\begin{figure*}
\begin{center}
\includegraphics[width=0.9\textwidth]{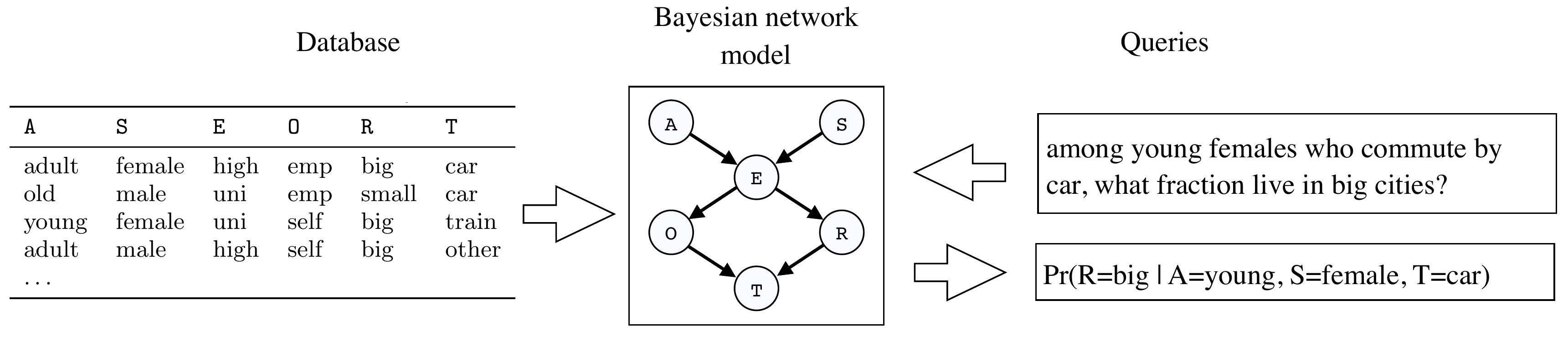}
\end{center}
\caption{\label{figure:introduction}
The Bayesian network learned from the \survey dataset~\protect\cite{scutari2014bayesian}.
The database consists of a single table with variables 
\age~(age)
% :\ indicating the age bracket of a person as {\em young}, {\em adult}, or {\em old}
, \sex~(sex)
% :\ which can be {\em female} or {\em male}
, \education~(education level)
% :\ indicating whether a person has finished {\em high-school} or {\em university}
, \occupation~(occupation)
% :\ indicating whether a person is {\em employed} or {\em self-employed}
, \residence~(size of residence city)
% :\ which can be {\em small} or {\em big}
, and \transport~(means of transportation)
% :\ which takes values {\em train}, {\em car}, or {\em other}
.
The corresponding Bayesian network model can be used to answer probabilistic queries over the data, such as the query shown above.
}
\end{figure*}

How can we mitigate this risk?
Our approach is to consider the anticipated workload of probabilistic queries, and materialize the results of selected computations, so as to enable fast answers to costly inference queries.
Moreover, we choose to work with a specific inference algorithm for Bayesian networks, namely variable elimination~\cite{zhang1994simple, zhang1996exploiting}.
While there exist several inference algorithms for Bayesian networks (e.g., junction trees, arithmetic circuits, and others~\cite{jensen1996introduction, 
	lauritzen1988local, 
	% dechter1999bucket, 
	darwiche2003differential, 
	% chavira2005compiling, 
	chavira2007compiling, 
	% kanagal2009indexing, 
	poon2011sum}), variable elimination is arguably the most straightforward algorithm for inference with Bayesian networks and is used as sub-routine in other algorithms (e.g., junction trees~\cite{jensen1996introduction}).
Since, to the best of our knowledge, this is the first work on {\it workload-aware materialization} for inference in Bayesian networks, we opt to use variable elimination as the inference algorithm, and demonstrate results to motivate the same approach for other inference algorithms as future work.

% How can we mitigate this risk?
Our key observations are the following: 
first, the execution of variable elimination involves intermediate results, in the form of relational tables, some of which might be very costly to compute every time a query requires them;  
second, the same intermediate tables can be used for the evaluation of many different queries.
Therefore, we set to pre\-compute and materialize the intermediate tables that bring the largest computational benefit, i.e., those that are involved in the evaluation of many expensive queries.
The problem formulation is general enough to accommodate arbitrary query workloads and takes as input a budget constraint on the number or space requirements of the materialized tables (Section~\ref{sec:setting}).

Our contributions are the following:
 \squishlist
 \item 
An exact poly\-no\-mial-time algorithm to choose a materialization for the variable elimination algorithm under cardinality constraints; the choice is optimal for a fixed ordering of variables. (Section~\ref{sec:algorithms})
% (Section~\ref{sec:dynamic-programming}); 
 \item 
A greedy approximation algorithm. (Section~\ref{sec:algorithms})
% (Section~\ref{sec:greedy});
 \item 
A pseudo-polynomial algorithm for the problem under a budget constraint on the space used for materialization and an extension to deal with redundant variables. (Section \ref{sec:extensions})
% (Section~\ref{sec:space-budget});
 \item 
%($iv$) a further-optimized computational scheme that avoids query-specific redundant computations 
%(Section~\ref{sec:redundant});
% \item
Experiments over real data, showing that modest materialization can significantly improve the running time  of variable elimination, with an average gain of $70\%$, and reaching up to a gain of $99\%$, over a uniform workload of queries.
In addition, for large datasets and query sizes, our approach achieves competitive inference speeds using much lighter materialization compared with junction tree  methods, which also rely on materialization. (Section~\ref{sec:experiments})
% (Section~\ref{sec:experiments}).
 \squishend
%
% The rest of the paper is organized as follows. 
% Section~\ref{section:related} reviews the related work to provide context for our paper.
% Section~\ref{sec:setting} reviews the required background for inference in Bayesian networks with variable elimination.
% Section~\ref{sec:algorithms} presents our algorithms for the materialization problem, 
% in particular an exact dynamic-programming algorithm
% and a greedy algorithm with approximation guarantees.
% Section \ref{sec:extensions} discusses practical extensions, 
% such as a pseudo-poly\-nomial algorithm for the case of space budget constraints, 
% and dealing with redundant variables.
% Section~\ref{sec:experiments} presents our experimental evaluation, 
% and Section~\ref{section:conclusion} a short conclusion.
\ReviewOnly{
Proofs of lemmas and theorems, omitted due to page limit, can be found in an extended version of the paper\cite{aslay2020query}.
%\footnote{\CA{Available online at \url{https://tinyurl.com/y55yuoat}.}}}
}

\section{Related work}
\label{section:related}

% \spara{Bayesian network inference.}
For exact inference on Bayesian networks, i.e., exact computation of marginal and conditional probabilities, the conceptually simplest algorithm is \emph{\varelim}~\cite{zhang1994simple, zhang1996exploiting}. 
Another common algorithm for exact inference is the \emph{junction tree} algorithm~\cite{jensen1996introduction, lauritzen1988local}.
 %, which is based on ``message passing'' among variables in the triangulated moral graph formed from the input Bayesian network.
%
	%For the latter algorithm, let us provide here a simplified description with the parts that are essential for our discussion and refer to~\cite{lauritzen1988local} for details.
	%
	Obtaining the junction tree of a \bn requires building the triangulated moral graph of the \bn,
	extracting the maximal cliques of that graph, and assigning each clique to one node of the junction tree. 
	The junction tree is then ``calibrated" via a message-passing procedure that precomputes and materializes potentials that correspond to joint probability distributions of the variables belonging to each of its %clique- and separator-
	nodes. % and edges 
	This way, a query that involves variables in the same %clique or separator 
	node can be answered directly by marginalizing those variables from the potential of that node. 
	For queries whose query variables reside in more than one node, 
	referred to as ``out-of-clique'' queries, 
	additional message passing is performed between the nodes containing the query variables. 
	To improve the efficiency of handling out-of-clique queries, Kanagal and Deshpande~\cite{kanagal2009indexing} proposed to materialize additional probability distributions arising from a hierarchical partitioning of the calibrated junction tree. 
	We compare experimentally our proposal with the aforementioned junction tree algorithms \cite{lauritzen1988local,kanagal2009indexing}.

Additionally, there have been several data structures proposed for efficient exact inference with graphical models in the knowledge-compilation literature~\cite{dechter1999bucket, darwiche2003differential, chavira2005compiling, chavira2007compiling, kanagal2009indexing, poon2011sum}. 
Darwiche~\cite{darwiche2003differential} showed that every Bayesian network can be interpreted as an exponentially-sized multi-linear function whose evaluation and differentiation solves the exact inference problem. Accordingly, several techniques were proposed to efficiently factor these multi-linear functions into more compact representations, such as arithmetic circuits and sum-product networks, by taking advantage of possible sparsity in conditional probabilities and context-specific independence for a given query~\cite{chavira2005compiling, chavira2007compiling, poon2011sum}. We note that optimizing for a given workload is orthogonal to these efforts: our cost-benefit frame\-work could be extended to account for other concise representations of the conditional-probability tables used by inference algorithms. As this paper is the first work to address budget-constrained and workload-aware materialization for query evaluation on Bayesian networks, we opt to work with the variable-elimination algorithm~\cite{zhang1994simple} over tabular-factor representations due to its conceptual simplicity. 

% \hide{
% \spara{Machine learning for approximate query processing}. 
\revision{
Model-based inference finds application in many settings, including tasks related to data management.
For example, models like Bayesian networks can naturally be used for selectivity estimation, as explained by Getoor et al.~\cite{getoor2001selectivity} and Tzoumas et al.~\cite{tzoumas2013adapting}, 
and today there is renewed interest in approximate query processing (AQP) based on probabilistic approaches~\cite{Chaudhuri:2017:AQP, Kraska:2017:AQP, ma2019dbest, Mozafari:2017:AQE, hilprecht2019deepdb}.
For tasks such as selectivity estimation and AQP, where inference queries are repeatedly executed over a learned model, it is important to consider the efficiency of online inference and the role that careful materialization can play to increase efficiency.
Our work can be viewed as a step in this direction.
Note, however, that we focus on the efficiency of \emph{\varelim} for Bayesian network inference 
and not on any specific application such as AQP.}%
\section{Setting and problem statement}
\label{sec:setting}

In this section we formally define the materialization problem we study. 
We first provide an overview of Bayesian-network inference
and \varelim algorithm, 
which is central to our contribution.

A Bayesian network \bayesianNet is a directed acyclic graph (\DAG), 
where nodes represent variables
and edges represent dependencies among variables.
Each node % in the network 
is associated with 
% a conditional probability distribution %(not shown in Figure~\ref{figure:introduction})
a table quantifying the probability that the node takes a particular value, 
% on each combination of values of its parent variables.
conditioned on the values of its parent nodes.
For instance, if a node associated with a \domainsize-ary variable \vara 
has~$\ell$ parents, all of which are \domainsize-ary variables, 
the associated probability distribution for \vara is 
a table with $\domainsize^{\ell+1}$ entries.

A key property of Bayesian networks is that each variable is conditionally independent 
of all its non-descendants given the value of its parents. 
Due to this property, the joint probability of all the variables can be factorized into marginal and conditional probabilities associated with the nodes of the network.
\eat{A key property of Bayesian networks is that, conditional on the  values of its parents, a variable is independent of other variables.
This leads to simple formulas for the evaluation of
marginal and conditional probabilities.}
Using the network of Figure~\ref{figure:introduction} as our running example,
the joint probability of all the variables is given by
\begin{eqnarray}
\Pr(\age,\sex,\education,\occupation,\residence,\transport) 
& = & 
\Pr(\transport \mid \occupation,\residence)
\Pr(\occupation \mid \education)
\Pr(\residence \mid \education) \nonumber \nonumber \\
&   & 
\Pr(\education \mid \age,\sex)
\Pr(\age)
\Pr(\sex).
\label{equation:survey}
\end{eqnarray}
% Note that 
Each factor on the right-hand side of Equation~(\ref{equation:survey})
is part of the specification of the Bayesian network and represents the marginal and conditional probabilities of its variables. 
%probability of each variable conditional on its parents. 
%% Notice that 
%Variables with no parents 
%correspond to factors with no conditional part.

\iffalse
For example, in Figure~\ref{figure:cpcs}
we display the \cpcs network with 422 nodes from the work of Pradhan et al.~\cite{pradhan1994knowledge}.
\begin{figure}[t]
\begin{center}
\includegraphics[width=0.8\columnwidth]{img/CPCS}
\caption{\label{figure:cpcs}
The \cpcs  Bayesian network~\protect\cite{pradhan1994knowledge}.}
\end{center}
\end{figure}
\fi

In what follows, we assume we are given a Bayesian network \bayesianNet that is discrete, i.e., all variables are categorical;
numerical variables can be discretized in categorical intervals. 
%(as the age variable \age in the \survey dataset). 
% Notice that, 
% In many cases, 
% access to the probability tables of a Bayesian network
% can replace access to the original data; 
% indeed, we can answer queries via the Bayesian-network model
% rather than through direct processing of potentially huge volumes of data.
% As discussed in the introduction, 
% this approach is not only more efficient, but it can also lead to 
% more accurate estimates as it avoids over-fitting.

% \subsection{The task}

\spara{Querying the Bayesian network.}
We consider the task of answering
probabilistic queries over a Bayesian network \bayesianNet.
%For instance, for the model shown in Figure~\ref{figure:introduction},
%example queries are: 
%\emph{``what is the probability that a person is a university-graduate female, 
%lives in a small city, and is self employed?''}
%or 
%\emph{``for each possible means of transport, 
%what is the probability that a person is young and uses the particular means of transport?''}
Specifically, we consider queries of the form 
%We consider queries of the form 
\begin{equation}
\label{equation:query} 
\query = \probability(\qvarfree,\qvarbound=\qvarvals),
\end{equation} 
where $\qvarfree\subseteq\allvars$ is a set of free variables
and $\qvarbound\subseteq\allvars$ is a set of bound variables
with corresponding values~\qvarvals.
Notice that free variables \qvarfree indicate that 
the query requests a probability 
for {\it each} of their possible values.
In the example of Figure~\ref{figure:introduction},
$\prob{\sex=\text{female},\education=\text{uni}}$ and
$\prob{\transport,\age=\text{young}}$
are instances of such queries.

We denote by $\qvarmiss = \allvars \setminus (\qvarfree\cup\qvarbound)$
the set of variables that do not appear in the query \query. 
The variables in the set \qvarmiss are those that need to be {\it summed out}
in order to compute the query~\query.
Specifically, query \query is computed via the summation
\begin{equation}
\label{equation:summation}
\prob{\qvarfree, \qvarbound=\qvarvals} = 
\sum_{\qvarmiss} \prob{\qvarfree, \qvarbound=\qvarvals, \qvarmiss}.
\end{equation}
The answer to the query $\prob{\qvarfree, \qvarbound=\qvarvals}$
is a table indexed by 
the combinations of values of variables in \qvarfree. Note that conditional probabilities of the form
$\probability(\qvarfree\mid\qvarbound=\qvarvals)$
can be computed from the corresponding joint probabilities 
% with one additional summation, since
by
%\begin{eqnarray*}
%\probability(\qvarfree\mid\qvarbound=\qvarvals)
%& = & \frac{\probability(\qvarfree,\qvarbound=\qvarvals)}{\probability(\qvarbound=\qvarvals)} = \frac{\probability(\qvarfree,\qvarbound=\qvarvals)}{\sum_{\qvarfree}\probability(\qvarfree,\qvarbound=\qvarvals)}. \nonumber
%\end{eqnarray*}
%% cigdem: the above version of this equation spills over to the second column so simplified it 
\begin{eqnarray*}
\probability(\qvarfree\mid\qvarbound=\qvarvals)
& = & \frac{\probability(\qvarfree,\qvarbound=\qvarvals)}{\sum_{\qvarfree}\probability(\qvarfree,\qvarbound=\qvarvals)}. \nonumber
\end{eqnarray*}
In addition, for variables with discrete numerical domain, {\it range-queries} 
of the form $\probability(\qvarfree, \lowqvarval\leq \qvarbound\leq\highqvarval)$
can be computed from the corresponding joint probabilities by
\begin{equation*}
\probability(\qvarfree, \lowqvarval\leq \qvarbound\leq\highqvarval) = 
\sum_{\lowqvarval\leq \qvarvals\leq\highqvarval}{\probability(\qvarfree,\qvarbound=\qvarvals)},
\end{equation*}
where `$\le$' here denotes an element-wise operator over variable vectors.
So without loss of generality, we focus on queries of type~(\ref{equation:query}).

%\spara{Answering queries.}
%The {\em variable-elimination algorithm}, 
%proposed by Zhang et al.~\cite{zhang1994simple}
%to answer queries
%$\query = \probability(\qvarfree,$ $\qvarbound=\qvarvals)$, 
%introduces the concept of {\em elimination tree}.
%% We start by describing these concepts
%% as they are central to our problem formulation. 
%% On a high level, 
%The algorithm 
%computes \query by summing out the variables \qvarmiss 
%that do not appear in \query, 
%according to Equation~(\ref{equation:summation}).
%When we sum out a variable, 
%we say that we {\em eliminate} it. 
%The elimination tree represents the order in which variables are eliminated
%and the intermediate results that are passed along.

\spara{Answering queries.} A query \query can be computed by {\em brute-force} approach as follows. 
First, compute the joint probability of all the variables in the network via a natural join over the tables of the factors of Equation~(\ref{equation:survey}). Let \bigtable be the resulting table, indexed by the combination of values of all variables. Then, select those entries of \bigtable that satisfy the corresponding equality condition $\qvarbound=\qvarvals$. Finally, for each $\vara\in\qvarmiss$, compute a sum over each {\em group} of values of \vara (i.e., ``sum out'' \vara). The table that results from this process is the answer to query \query.

\eat{First, compute into a table~\bigtable the joint probability for each 
combination of values of all variables via a natural join over the tables of the factors of Equation~(\ref{equation:survey}), and 
%Specifically, table \bigtable can be computed via a natural join over the tables of the factors of Equation~(\ref{equation:survey}).
select those entries of \bigtable that satisfy the  
equality condition in $\qvarbound=\qvarvals$. Then, for each $\vara\in\qvarmiss$, 
compute a sum over each {\em group} of values of \vara (i.e., we ``sum out'' \vara).
% as the query neither imposes a condition on its values, not requires it removed.
The table that results from this process is the answer to query \query. } 

\spara{Variable elimination.} 
The \varelim algorithm proposed by Zhang et al.~\cite{zhang1994simple}  %~\cite{zhang1994simple}
improves upon the brute-force approach by 
taking advantage of the factorization of the joint probability, 
hence, avoiding the computation of \bigtable.
% \eat{taking advantage of the fact that it is not necessary to compute \bigtable.}
%
Given a total order~\eliminationOrder\ on variables,
the \varelim algorithm computes the quary \query by summing out the variables in \qvarmiss, 
one at a time, from the intermediate tables obtained via natural join of the \emph{relevant} factors.

For example, 
consider the Bayesian network of Figure~\ref{figure:introduction}, 
the query $\query = \prob{\transport,\age=\text{young}}$, 
and the order
$\eliminationOrder = \langle\age, \sex, \transport, \education, \occupation, \residence\rangle$. 
Note that for this query \query, we have $\qvarfree=\{\transport\}$, $\qvarbound=\{\age\}$, and $\qvarmiss=\{\sex, \education, \occupation, \residence\}$.
The first variable in \eliminationOrder is $\age\in\qvarbound$.
The \varelim algorithm considers only the tables of those factors that include variable \age and select the rows that satisfy the equality condition $(\age=\text{young})$; 
then it performs the natural join over the resulting tables.
This computation corresponds to the following two equations:
\eat{The brute-force algorithm would first compute a natural join over all factors
in Equation~(\ref{equation:survey}) and 
then select those rows that match the condition $(\age=\text{young})$.
An equivalent but more efficient computation is to first
consider only the tables of those factors
that include variable \age and select % from each 
only the rows that satisfy the equality condition $(\age=\text{young})$; then perform the natural join over 
the resulting tables. 
This computation corresponds to the following two equations.}
\begin{align}
\factor{\age}{\sex, \education; \age=\text{young}} =
\Pr(\education \mid \age  = \text{young},\sex)
\Pr(\age  = \text{young}) \nonumber, 
\end{align}
%and
\begin{align}
\Pr(\age = \text{young},\sex,\education,\occupation,\residence,\transport) 
=\, & 
\factor{\age}{\sex, \education; \age=\text{young}}
\Pr(\transport \mid \occupation,\residence) \nonumber \\
& 
\Pr(\occupation \mid \education) 
\Pr(\residence \mid \education)
\Pr(\sex).  \nonumber % \label{eq:afterfirstelimination}
\end{align}
Note that the computation involves only two factors that contain the variable \age, 
resulting in a table $\factor{\age}{\sex, \education; \age=\text{young}}$, 
which is indexed by \sex and \education, and replaces these factors in the factorization of the joint probability.
%  (Eq.~\ref{eq:afterfirstelimination}). \eat{the joint probability formula.} % (Equation~(\ref{equation:survey})).}

\eat{The crucial observation is that the equality condition $(\age=\text{young})$
concerns only two factors of the joint probability formula. % (Equation~(\ref{equation:survey})).
After processing variable \age 
these two factors can be replaced by a factor $\factor{\age}{\sex, \education; \age=\text{young}}$,
a table indexed by \sex and \education
and containing only entries with $(\age=\text{young})$.}

The algorithm then proceeds with the next variables in the elimination order \eliminationOrder, each time considering the current set of relevant factors in the factorization,
%  for natural join operation, 
until there are no more variables left to eliminate. 
The answer to the query is given by the final remaining factor upon all the variables in \eliminationOrder are processed this way.

\begin{figure}
\begin{center}
\includegraphics[width=0.95\columnwidth]{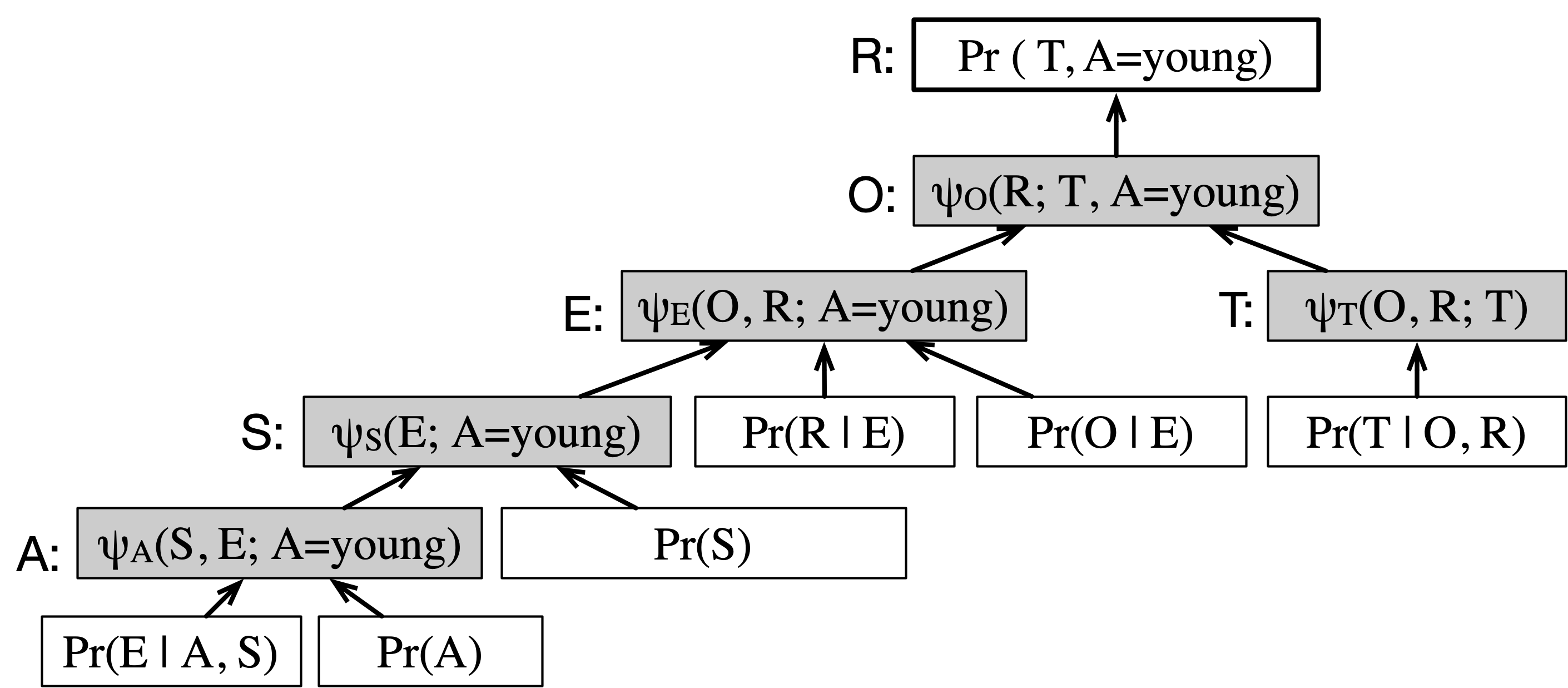}
\end{center}
\caption{The elimination tree \tree for the query $\query = \prob{\transport,\age=\text{young}}$ 
and order of variables
$\eliminationOrder = (\age, \sex, \transport, \education, \occupation, \residence)$, 
for the dataset and Bayesian network shown in Figure~\ref{figure:introduction}.}
\label{fig:elimination_tree}
\end{figure}

\spara{Elimination tree.} The \varelim algorithm gives rise to a graph, like the one shown
in Figure~\ref{fig:elimination_tree} for the example of Figure~\ref{figure:introduction}.
Each node is associated with a factor and there is a directed edge between two factors if one is used for the computation of the other.
In particular, each \emph{leaf node} corresponds a factor, which is a conditional probability table
in the Bayesian network.
In our running example, these are the factors that appear in Equation~(\ref{equation:survey}).
Each \emph{internal node} corresponds to a factor that is computed from 
its children (i.e., the factors of its incoming edges),
and is computed during the execution of the \varelim algorithm.
Moreover, 
% as we saw, 
each internal node corresponds to one variable.
The last factor computed is the answer to the query.

Notice that the graph constructed in this manner 
is either a tree or a forest. 
It is not difficult to see that {\em the elimination graph 
is a tree if and only if the corresponding Bayesian network is a weakly 
connected \DAG}. %, irrespectively of elimination order.
To simplify our discussion, we will focus
on connected Bayesian networks
and deal with an \emph{elimination tree} \tree for each query. 
All results can be directly extended to the case of forests.

We note that the exact form of the factor for each
internal node of \tree depends on the query.
For the elimination tree in Figure~\ref{fig:elimination_tree},
we have factor $\factor{\age}{\sex, \education; \age=\text{young}}$
on the node that corresponds to variable \age.
However, if the query contained variable \age as a free 
variable rather than bound to value $(\age=\text{young})$, then
the same node in \tree would contain a factor $\factor{\age}{\sex, \education; \age}$.
% i.e., a factor that is fully indexed by combinations of values of \sex, \education, and \age. 	
%i.e., a factor that is also indexed by \sex and \education
%but does not contain only entries that satisfy an equality condition for \age.
And if the query did not contain variable \age, then the same node in \tree
would contain a factor $\factor{\age}{\sex, \education}$.
% i.e., a factor that is also indexed by \sex and \education but does not contain \age at all (as it would be summed out).
On the other hand, the structure of the tree, the factors that 
correspond to leaf nodes and the 
variables that index the variables of the factors 
that correspond to internal nodes are query independent.

\FullOnly{
	\spara{Note}. The elimination algorithm we use here differs slightly from the one 
	presented by Zhang et al.~\cite{zhang1994simple}. 
	Specifically, the variable-elimination algorithm of Zhang et al.\ 
	computes the factors associated with the bound variables \qvarbound 
	at a special initialization step, 
	which leads to benefits in practice (even though the running time remains super-polynomial in the worst case). 
	On the other hand, we compute factors in absolute accordance with the elimination order. 
	This allows us to consider the variable-elimination order fixed for all variables independently of the query.
}

\spara{Materialization of factors.}
Materializing % (i.e., precomputing and storing) 
factor tables for
internal nodes of the elimination tree \tree
can speed up the computation of queries that require those factors.
% $\query = \prob{\qvarfree, \qvarbound=\qvarvals}$ 
% as we will avoid evaluating those factors at query time. 
As we saw earlier, factors 
are computed in a sequence of steps, one for each variable.
Each step involves the natural join over other factor tables,
followed by:
(i) either variable summation (to sum-out variables \qvarmiss);
or (ii) row selection (for variables \qvarbound);
or (iii) no operation (for variables \qvarfree).
In what follows, {\em we focus on materializing factors
that involve only variable summation}, the first out of these three types of operations.
Materializing such factors
is often useful for multiple queries \query
and sufficient to make the case for the materialization of factors
that lead to the highest performance gains over a given query workload.
Dealing with the materialization of general factors is a straightforward but somewhat tedious extension which is left for future work.

To formalize our discussion, we introduce some notation.
Given a node $\nodeu\in\nodes$ in an elimination tree \tree,
we write \subtree{\nodeu}
to denote the subtree of \tree that is rooted at node \nodeu. 
We also write \vars{\nodeu} to denote the subset of variables of \allvars
that are associated with the nodes of \subtree{\nodeu}.
Finally, we write $\ancestors{\nodeu}$ to denote 
the set of ancestors of \nodeu in \tree, 
that is, all nodes between \nodeu and the root of the tree~\tree, 
excluding \nodeu. % and including the root.

Computing a factor for a query \query incurs a computational cost.
We distinguish two notions of cost: first, if the children factors of a node \nodeu in the elimination tree~\tree are given as input, computing \nodeu incurs a {\em partial} cost of computing the factor of \nodeu from its children; second, starting from the factors associated with conditional probability tables in the Bayesian network, the {\em total} cost of computing a node includes the partial costs of computing all intermediate factors, from the {\em leaf} nodes to \nodeu.
% --- however having the factor materialized is associated with a utility value, namely the 
% computational effort we save when the materialized factor is useful for a query.
Formally, we have the following definitions.
\begin{definition}[Partial-Cost]
\label{definition:cost}
The partial cost \cost{\nodeu} of a node $\nodeu\in\nodes$ in the  elimination tree $\tree=(\nodes,\edges)$ 
is the computational cost required to compute the corresponding factor 
given the factors of its children nodes.
\end{definition}
\begin{definition}[Total-Cost]
\label{definition:utility}
The total cost of a node $\nodeu\in\nodes$ in the  elimination tree $\tree=(\nodes,\edges)$ 
is the total cost of computing the factor at node \nodeu, i.e., 
\[
\utility{\nodeu} = \sum_{\substack{\nodex\in\subtree{\nodeu}}} \cost{\nodex},
\]
where \cost{\nodex} is the partial cost of node \nodex.
\end{definition}

When we say that we materialize a node $\nodeu\in\nodes$, we mean that we
precompute and store the factor that is the result of summing out
all variables below it on \tree.
% and we say that node \nodeu is useful for a query \query
% only if the corresponding factor is useful.
{\em When is a materialized factor useful for a query \query?}
Intuitively, it is useful if it is one of the factors computed during the evaluation
of \query, in which case we save the total cost of computing it from scratch, 
provided that there is no other materialized factor that could be used in its place, with greater savings in cost.
The following definition of usefulness formalizes this intuition.
\begin{definition}[Usefulness]
\label{definition:usefulness}
Let 
$\query =$
$\prob{\qvarfree, \qvarbound=\qvarvals}$ 
be a query, and 
$\solution\subseteq\nodes$ a set of nodes of the elimination tree \tree
that are materialized.
We say that a node $\nodeu\in\nodes$ is useful 
for the query \query with respect to the set of nodes \solution, if 
{\em (}$i${\em )} $\nodeu\in\solution$; 
{\em (}$ii${\em )} $\vars{\nodeu}\subseteq\qvarmiss$;  and
{\em (}$iii${\em )} there is no other node $\nodev\in\ancestors{\nodeu}$ 
for which conditions {\em (}$i${\em )} and {\em (}$ii${\em )} hold.
\end{definition}
% The three conditions above state that for a node to be useful
% ($i$) the corresponding factor must have been materialized; 
% ($ii$) the missing variables \qvarmiss of the query must be a superset of the 
% variables that are summed out to compute the given factor; and 
% ($iii$) there must be no other 
% nodes higher in the elimination tree for which the previous conditions hold.
% \note[Aris]{Explain why and/or give example.}
To indicate that a node \nodeu is {\em useful} 
for the query \query with respect to a set of nodes \solution with materialized factors,
we use the indicator function $\isuseful{\query}{\nodeu}{\solution}$.
That is, $\isuseful{\query}{\nodeu}{\solution}=1$
if node $\nodeu\in\nodes$ is useful 
for the query \query with respect to the set of nodes \solution, 
and $\isuseful{\query}{\nodeu}{\solution}=0$ otherwise.
When a materialized node is {useful} for a query \query, it saves us the total cost of
computing it from scratch. Considering a query workload, 
where different queries appear with different probabilities,
we define the benefit of a set of materialized nodes \solution as the total cost we save in expectation.
\begin{definition}[Benefit]
\label{definition:benefit}
Consider an elimination tree $\tree=(\nodes,\edges)$, 
a set of nodes $\solution\subseteq\nodes$, and query probabilities $\prob{\query}$
for the set of all possible queries \query.
The benefit \benefit{\solution} of the node set \solution is defined as:
\begin{eqnarray*}
\benefit{\solution} 
& = & \sum_{\query} \prob{\query} \sum_{\nodeu\in\solution} \isuseful{\query}{\nodeu}{\solution} \utility{\nodeu} \\
& = & \sum_{\nodeu\in\solution} \prob{\isuseful{\query}{\nodeu}{\solution}=1} \utility{\nodeu} \\
& = & \sum_{\nodeu\in\solution}  \expuseful{\nodeu}{\solution} \utility{\nodeu}.
\end{eqnarray*}
\end{definition}

\spara{Problem definition.}
We can now define formally the problem we consider:
for a space budget \spacebudget, 
our goal is to select a set of factors to materialize to achieve optimal benefit.

\begin{problem}
Given a Bayesian network \bayesianNet,
an elimination tree $\tree=(\nodes,\edges)$
for answering probability queries over~\bayesianNet, and budget \spacebudget, 
select a set of nodes $\solution\subseteq\nodes$ to materialize, 
whose total size is at most \spacebudget, 
so as to optimize~\benefit{\solution}.
\label{problem:qtm-space}
\end{problem}

For simplicity, % of exposition
we also consider a version of the problem where 
we are given a total budget \budget on the number of nodes
that we can materialize. 
We first present algorithms for Problem~\ref{problem:qtm} in Section~\ref{sec:algorithms},
and discuss how to address the more general Problem~\ref{problem:qtm-space} in Section~\ref{sec:extensions}.

\begin{problem}
Given a Bayesian network \bayesianNet,
an elimination tree $\tree=(\nodes,\edges)$
for answering probability queries over~\bayesianNet, and an integer \budget, 
select at most \budget nodes $\solution\subseteq\nodes$ to materialize
so as to optimize~\benefit{\solution}.
\label{problem:qtm}
\end{problem}

\revision{Note that the factors to materialize are chosen under the assumption that the queries will follow a given probability distribution $\traindistr = \left\{\prob{\query}\right\}$. 
% Once the chosen factors are materialized, 
If the queries that are eventually encountered follow a different distribution \testdistr, 
then query answering may be faster or slower on average than under \traindistr;
this depends on whether \testdistr assigns higher or lower probability to queries that are favorable or not 
for the chosen materialization.
Notice, however, that even when the materialization has been optimized for a different workload, 
the runtime cost of queries can never be worse than the runtime without materialization.
The worst case arises when the chosen materialization 
is useful for no query with non-zero probability in \testdistr.
In such a case the benefit of the materialization is zero.}%

\iffalse
For ease of presentation we show our notation in Table~\ref{table:notation}.
\input{notation-figure.tex}
\fi

\section{Algorithms}
\label{sec:algorithms}

In this section we present our algorithms for Problem~\ref{problem:qtm}:
Section~\ref{sec:dynamic-programming} presents an exact 
polynomial-time dynamic-program\-ming algorithm;
Section~\ref{sec:greedy} presents a greedy algorithm, 
which yields improved time complexity but provides only an approximate solution, 
yet with quality guarantee.
% The two algorithms are subsequently adapted for Problem~\ref{problem:qtm-space}
% in Section~\ref{sec:space-budget}.

\subsection{Dynamic programming}  
\label{sec:dynamic-programming}

We discuss our dynamic-programming algorithm % for Problem~\ref{problem:qtm} 
in three steps.
First, we introduce the notion of {\it partial benefit} that allows us to explore 
partial solutions for the problem.
Second, % using the notion of {\it partial benefit},
we demonstrate the optimal-substructure property of the problem, 
and third, we present the algorithm. 

\mpara{Partial benefit.}
% Recall that 
In Definition~\ref{definition:benefit} we defined the 
(total)~{\em benefit} of a subset of nodes $\solution\subseteq\nodes$ (i.e., a potential solution) 
for the whole elimination tree \tree. 
Here we define the {\em partial benefit} 
of a subset of nodes~\solution for a subtree \subtree{\nodeu}
of a given node \nodeu of the elimination tree~\tree.

\begin{definition}[partial benefit]
\label{definition:partial-benefit}
Consider an elimination tree $\tree=(\nodes,\edges)$, 
a subset of nodes $\solution\subseteq\nodes$, and probabilities $\prob{\query}$
for the set of all possible queries \query. 
The partial benefit \partialbenefit{\nodeu}{\solution} of the node set \solution 
at a given node $\nodeu \in \nodes$ is % defined as:
\begin{align*}
\partialbenefit{\nodeu}{\solution} = \sum_{\nodev\in\solution \cap \subtree{\nodeu}} \expuseful{\nodev}{\solution} \utility{\nodev}.
\end{align*}
\end{definition}

% Our dynamic-programming solution relies on the following lemmas.
The following lemma 
% Lemma~\ref{lemma:lowestAncestor} 
states that, 
given a set of nodes \solution, and a node $\nodeu\in\solution$, 
the probability that \nodeu is useful for a random query with respect to \solution 
depends only on the lowest ancestor of \nodeu in \solution.

\begin{lemma}% [Lowest Solution Ancestor]
\label{lemma:lowestAncestor}
Consider an elimination tree $\tree=(\nodes,\edges)$ and a set $\solution \subseteq \nodes$ of nodes. 
Let $\nodeu, \nodev \in \solution$ such that $\nodev \in \ancestors{\nodeu}$ and 
$\nodepath{\nodeu}{\nodev} \cap \solution = \emptyset$. 
Then we have:
\begin{align*}
\expuseful{\nodeu}{\solution} = \expuseful{\nodeu}{\nodev}, 
\end{align*}
where the expectation is taken over a distribution of queries~\query.
%\begin{align*}
%\prob{\isuseful{\query}{\nodeu}{\solution}=1} = \prob{\isuseful{\query}{\nodeu}{\{\nodeu, \nodev\}}=1}.
%\end{align*}
\end{lemma}
\FullOnly{
  \begin{proof}
  To prove the lemma, we will show that for any query \query, 
  it is 
  $\isuseful{\query}{\nodeu}{\solution}=1$ if and only if  
  $\isuseful{\query}{\nodeu}{\nodev} = 1$. 

  We first show that $\isuseful{\query}{\nodeu}{\solution}=1$ implies  $\isuseful{\query}{\nodeu}{\nodev} = 1$. 
  From Definition~\ref{definition:usefulness}, 
  we have that $\isuseful{\query}{\nodeu}{\solution}=1$ if $\vars{\nodeu}\subseteq\qvarmiss$ and 
  there is no $\nodew \in \ancestors{\nodeu} \cap \solution$ such that $\vars{\nodew}\subseteq\qvarmiss$. 
  Given that $\isuseful{\query}{\nodeu}{\solution}=1$ and $\nodev\in\ancestors{\nodeu} \cap \solution$, 
  it follows that $\vars{\nodev}\not\subseteq\qvarmiss$, hence, $\isuseful{\query}{\nodeu}{\nodev} = 1$. 
  Reversely,  
  we show that $\isuseful{\query}{\nodeu}{\nodev} = 1$ implies $\isuseful{\query}{\nodeu}{\solution}=1$. 
  Notice that $\isuseful{\query}{\nodeu}{\nodev} = 1$ if 
  $\vars{\nodeu}\subseteq\qvarmiss$ and $\vars{\nodev}\not\subseteq\qvarmiss$. 
  This means that, for all $\nodew \in \solution \cap \ancestors{\nodev}$ 
  we have $\vars{\nodew}\not\subseteq\qvarmiss$ 
  since $\vars{\nodev} \subseteq \vars{\nodew}$. 
  Given also that $\nodepath{\nodeu}{\nodev} \cap \solution = \emptyset$, we have that 
  for all $\nodew \in \solution \cap \ancestors{\nodeu}$ it is $\vars{\nodew}\not\subseteq\qvarmiss$, 
  hence, $\isuseful{\query}{\nodeu}{\solution}=1$. 

  Given the one-to-one correspondence between the set of queries in which $\isuseful{\query}{\nodeu}{\solution}=1$ and the set of queries in which $\isuseful{\query}{\nodeu}{\nodev} = 1$, we have
  $\sum_{\query} \prob{\query} \isuseful{\query}{\nodeu}{\solution}  = \sum_{\query} \prob{\query} \isuseful{\query}{\nodeu}{\nodev}$, hence, the result follows.   
  %spacesaving
  %\begin{align*}
  %\prob{\isuseful{\query}{\nodeu}{\solution}=1} &= \sum_{\query} \prob{\query} \isuseful{\query}{\nodeu}{\solution}  \\ 
  %&= \sum_{\query} \prob{\query} \isuseful{\query}{\nodeu}{\nodev} \\ 
  %&= \prob{\isuseful{\query}{\nodeu}{\nodev}=1}.
  %\end{align*}
  \end{proof}
}

Building on Lemma~\ref{lemma:lowestAncestor},
we arrive to Lemma~\ref{lemma:optimal-substructure-prep} below,
which states that
the partial benefit 
\partialbenefit{\nodeu}{\solution} of a
node-set \solution 
at a node \nodeu
depends only on ($i$) the nodes of \subtree{\nodeu} that are included
in \solution, and ($ii$) 
the lowest ancestor \nodev of \nodeu in \solution, and therefore it does not
depend on what other nodes ``above''~\nodev are included in \solution. 

\FullOnly{
  For the proof of Lemma~\ref{lemma:optimal-substructure-prep},
  we introduce some additional notation. 
  Let $\tree=(\nodes,\edges)$ be an elimination tree, 
  \nodeu a node of \tree, and 
  \subtree{\nodeu} the subtree of \tree rooted at \nodeu.
  Let $\solution \subseteq \nodes$ be a set of nodes.
  For each node $\nodew \in \subtree{\nodeu} \cap \solution$, 
  we define 
  $\lsa{\nodew}{\solution}$ 
  to be the lowest ancestor of $\nodew$ that is included in $\solution$.
}

\begin{lemma}% [Solution Partiality Relation]
\label{lemma:optimal-substructure-prep}
Consider an elimination tree $\tree=(\nodes,\edges)$
and a node $\nodeu\in\nodes$. Let $\nodev\in\ancestors{\nodeu}$ be an ancestor of \nodeu. 
Consider two sets of nodes \solution and $\solution'$ for which
% \squishlist
  % \item [{\em (}$i${\em )}] 
  (i) 
  $\nodev\in\solution$ and $\nodev\in\solution'$;
  % \item [{\em (}$ii${\em )}] 
  (ii) $\subtree{\nodeu} \cap \solution = \subtree{\nodeu} \cap \solution'$; and
  % \item [{\em (}$iii${\em )}] 
  (iii) $\nodepath{\nodeu}{\nodev} \cap \solution = \nodepath{\nodeu}{\nodev} \cap \solution' = \emptyset$.
% \squishend
Then, 
% for all $\nodew \in \subtree{\nodeu} \cap \solution$  (equivalently, for all $\nodew \in \subtree{\nodeu} \cap \solution'$),
we have:
$\partialbenefit{\nodeu}{\solution} = \partialbenefit{\nodeu}{\solution'}$.
\end{lemma}

\FullOnly{
  \begin{proof}
  From direct application of Lemma~\ref{lemma:lowestAncestor}, we have
  \[ 
  \expuseful{\nodew}{\solution} = \expuseful{\nodew}{\lsa{\nodew}{\solution}}, 
  \text{  for all } \nodew \in \subtree{\nodeu} \cap \solution,\] 
  and similarly,
  \[ 
  \expuseful{\nodew}{\solution'} = \expuseful{\nodew}{\lsa{\nodew}{\solution'}},
  \text{  for all }\nodew \in \subtree{\nodeu} \cap \solution'.
  \]
  Now, given that $\subtree{\nodeu} \cap \solution = \subtree{\nodeu} \cap \solution'$ and  
    $\nodepath{\nodeu}{\nodev} \cap \solution = \nodepath{\nodeu}{\nodev} \cap \solution' = \emptyset$, 
    we have $\lsa{\nodew}{\solution} = \lsa{\nodew}{\solution'}$, 
    for all $\nodew \in \subtree{\nodeu} \cap \solution$;
    and similarly $\lsa{\nodew}{\solution} = \lsa{\nodew}{\solution'}$, 
    for all $\nodew \in \subtree{\nodeu} \cap \solution'$.
    \iffalse
    (Specifically, 
    if $\left(\nodepath{\nodew}{\nodeu} \cup \{\nodeu\}\right) \cap \solution  = \emptyset$ 
    then $\lsa{\nodew}{\solution} = \nodev$;
    otherwise $\lsa{\nodew}{\solution} \in \subtree{\nodeu} \cap \solution$ and 
    $\nodepath{\nodew}{\lsa{\nodew}{\solution}} \cap \solution = \emptyset$).
    \fi
    It then follows that for all $\nodew \in \subtree{\nodeu} \cap \solution$, 
    we have $\expuseful{\nodew}{\solution} = \expuseful{\nodew}{\solution'}.$
    Putting everything together, we get
    \begin{align*}
    \partialbenefit{\nodeu}{\solution}  
    & = \sum_{\nodew\in\solution \cap \subtree{\nodeu}} \expuseful{\nodew}{\solution} \utility{\nodew} \\ 
    & = \sum_{\nodew\in\solution' \cap \subtree{\nodeu}} \expuseful{\nodew}{\solution'} \utility{\nodew} = \partialbenefit{\nodeu}{\solution'}. 
    \end{align*}
  \end{proof}
}

Let \nocap be a special node, which we will use to denote that no ancestor of a node \nodeu is included in a solution \solution. We define $\extendedancestors{\nodeu} = \ancestors{\nodeu} \cup \{\nocap\}$ as the {\em extended set of ancestors} of \nodeu, which adds \nocap into \ancestors{\nodeu}. Notice that $\nodepath{\nodeu}{\nocap}$ corresponds to the set of ancestors of \nodeu including the root \rout, i.e., $\nodepath{\nodeu}{\nocap}=\ancestors{\nodeu}$. 

\mpara{Optimal substructure.}
Next, 
% In Lemma~\ref{lemma:optimal-substructure} 
we present the optimal-substru\-cture property for Problem~\ref{problem:qtm}.
% which is later used as the basis of our dynamic-programming algorithm.
Lemma~\ref{lemma:optimal-substructure}
% builds upon Lemma~\ref{lemma:optimal-substructure-prep} and
% states that among nodes of the optimal solution, 
states that 
the subset of nodes of an optimal solution that fall in a given subtree 
depends only on the nodes of the subtree
and the lowest ancestor of the subtree that is included 
in the optimal solution.

\begin{lemma}[Optimal Substructure]
\label{lemma:optimal-substructure}
Given an elimination tree $\tree=(\nodes,\edges)$ and an integer budget $\budget$, 
let $\solution^*$ denote the optimal solution to Problem~\ref{problem:qtm}. 
Consider a node $\nodeu \in \nodes$ and let $\nodev \in \extendedancestors{\nodeu}$ 
be the lowest ancestor of $\nodeu$ that is included in $\solution^*$. 
Let $\solution_\nodeu^* = \subtree{\nodeu} \cap \solution^*$ 
denote the set of nodes in the optimal solution that reside in $\subtree{\nodeu}$ and 
let $\treebudgetallocation = \lvert\subtree{\nodeu} \cap \solution^*\rvert$. 
Then, 
\[
\solution_\nodeu^* = 
\underset{\substack{\solution_\nodeu \subseteq \subtree{\nodeu} \\ 
\lvert\solution_\nodeu\rvert = \treebudgetallocation}}
{\arg\max} \left\{ \partialbenefit{\nodeu}{\solution_\nodeu \cup \{\nodev\}} \right\}.
\]
\end{lemma}

\FullOnly{
  \begin{proof}
  First, notice that the sets $\solution^*$ and $\solution_\nodeu^* \cup \{\nodev\}$
  satisfy the pre-conditions of Lemma~\ref{lemma:optimal-substructure-prep}, and thus, 
  \begin{align}
  \label{eq:optStPrep1}
  \partialbenefit{\nodeu}{\solution^*} = \partialbenefit{\nodeu}{\solution_\nodeu^* \cup \{\nodev\}}.
  \end{align}
  Now, to achieve a contradiction, assume that there exists a set 
  $\solution_\nodeu' \neq \solution_\nodeu^*$ 
  such that $\lvert\solution_\nodeu'\rvert = \treebudgetallocation$ and 
  \begin{align}
  \partialbenefit{\nodeu}{\solution_\nodeu^* \cup \{\nodev\}} < 
  \partialbenefit{\nodeu}{\solution_\nodeu' \cup \{\nodev\}}.
  \end{align}
  Let $\solution' = (\solution^* \setminus \solution_\nodeu^*) \cup \solution_\nodeu'$ 
  denote the solution obtained by replacing the node set $\solution_\nodeu^*$ in $\solution^*$ by $\solution_\nodeu'$.
  Again $\solution'$ and $\solution_\nodeu' \cup \{\nodev\}$
  satisfy the preconditions of Lemma~\ref{lemma:optimal-substructure-prep}, and thus,
  \begin{align}
  \label{eq:optStPrep2}
  \partialbenefit{\nodeu}{\solution'} =  \partialbenefit{\nodeu}{\solution_\nodeu' \cup \{\nodev\}}.
  \end{align}
  As before, 
  for $\nodew \in \solution^* \setminus \solution_\nodeu^*$
  we define $\lsa{\nodew}{\solution^*}$ and $\lsa{\nodew}{\solution'}$ 
  to be the lowest ancestor of $\nodew$ in $\solution^*$ and in $\solution'$, respectively. 
  Given that $\nodew \not\in \subtree{\nodeu}$, we have $\lsa{\nodew}{\solution^*} = \lsa{\nodew}{\solution'}$, hence, for all $\nodew \in \solution^* \setminus \solution_\nodeu^*$:
  \begin{align}
  \label{eq:optStPrep3}
  \expuseful{\nodew}{\solution^*} = \expuseful{\nodew}{\solution'}.
  \end{align}
  Putting together Equations (\ref{eq:optStPrep1}-\ref{eq:optStPrep3}) we get
  \begin{eqnarray*}
  \benefit{\solution'} 
  & = & \sum_{\nodew \in\solution'}\expuseful{\nodew}{\solution'} \utility{\nodew} \\
  & = & \sum_{\nodew \in\solution_\nodeu'} \expuseful{\nodew}{\solution'} \utility{\nodew} + 
        \!\!\!\sum_{\nodew \in \solution^* \setminus \solution_\nodeu^*}\!\!\! 
        \expuseful{\nodew}{\solution'} \utility{\nodew} \\
  & = & \partialbenefit{\nodeu}{\solution_\nodeu' \cup \{\nodev\}}  + 
        \!\!\sum_{\nodew \in \solution^* \setminus \solution_\nodeu^*}\!\! 
        \expuseful{\nodew}{\solution'} \utility{\nodew} \\
  & = & \partialbenefit{\nodeu}{\solution_\nodeu' \cup \{\nodev\}} + 
        \!\!\sum_{\nodew \in \solution^* \setminus \solution_\nodeu^*}\!\! 
        \expuseful{\nodew}{\solution^*} \utility{\nodew} \\
  & > & \partialbenefit{\nodeu}{\solution_\nodeu^* \cup \{\nodev\}} + 
        \!\!\sum_{\nodew \in \solution^* \setminus \solution_\nodeu^*}\!\! 
        \expuseful{\nodew}{\solution^*} \utility{\nodew} \\
  & = & \sum_{\nodew \in\solution_\nodeu^*} \expuseful{\nodew}{\solution^*} \utility{\nodew} + 
        \!\!\!\sum_{\nodew \in \solution^* \setminus \solution_\nodeu^*}\!\!\! 
        \expuseful{\nodew}{\solution^*} \utility{\nodew} \\
  & = & \benefit{\solution^*}
  \end{eqnarray*}
  which is a contradiction since $\solution^*$ is the optimal solution of Problem~\ref{problem:qtm}, 
  and thus, 
  $\benefit{\solution'} \le \benefit{\solution^*}$. 
  \end{proof}
}

%The following lemma provides a bottom-up way to combine partial solutions computed on subtrees. We note that, to keep the presentation simple, we provide the exposition of results on binary trees in the rest of the section. 
The following lemma provides a bottom-up approach to combine partial solutions computed on subtrees. 
We note that in the rest of the section,
we present our results on \emph{binary trees}.
This assumption is made \emph{without any loss of generality}
as any $d$-ary tree can be converted into a binary tree by introducing dummy nodes; 
furthermore, by assigning appropriate cost to dummy nodes,
we can ensure that they will not be selected by the algorithm.

\begin{lemma}[Additivity]
\label{lemma:additivity}
Consider an elimination tree $\tree=(\nodes,\edges)$, a node $\nodeu\in\nodes$, and a set $\solution_{\nodeu}$ of nodes in $\subtree{\nodeu}$. 
Let \rightsub{\nodeu} and \leftsub{\nodeu} be the right and left children of \nodeu,
and let $\solution_{\rchild(\nodeu)} =  \subtree{\rightsub{\nodeu}} \cap \solution_{\nodeu} $ and $\solution_{\lchild(\nodeu)} = \subtree{\leftsub{\nodeu}} \cap \solution_{\nodeu}$.
Then, for $\nodev \in \extendedancestors{\nodeu}$:
\begin{align*}
& \partialbenefit{\nodeu}{\solution_{\nodeu} \cup \{\nodev\}} \\
&~~ = \begin{cases} 
   \partialbenefit{\nodeu}{\{\nodeu,\nodev\}} + 
   \partialbenefit{\rightsub{\nodeu}}{\solution_{\rchild(\nodeu)} \cup \{\nodeu\}}  & \\ 
   \quad\quad\quad\quad~\, + \partialbenefit{\leftsub{\nodeu}}{\solution_{\lchild(\nodeu)} \cup \{\nodeu\}}, & 
   \!\!\!\!\text{if } \nodeu \in \solution_{\nodeu} \\
   \partialbenefit{\rightsub{\nodeu}}{\solution_{\rchild(\nodeu)} \cup \{\nodev\}} + 
   \partialbenefit{\leftsub{\nodeu}}{\solution_{\lchild(\nodeu)} \cup \{\nodev\}}, & 
   \!\!\!\!\text{otherwise}.
  \end{cases}
\end{align*}
\end{lemma}

\FullOnly{
  \begin{proof}
  We show the result in the case of $\nodeu \in \solution_{\nodeu}$:
  notice that since $\nodeu \in \solution_{\nodeu}$ 
  the node $\nodev$ cannot be the lowest solution ancestor of any node in 
  $\solution_{\rchild(\nodeu)} \cup \solution_{\lchild(\nodeu)}$. 
  Given also that no node in $\solution_{\rchild(\nodeu)}$ can have an ancestor in $\solution_{\lchild(\nodeu)}$ and vice versa, following Lemma~\ref{lemma:lowestAncestor}, we have: 
  \begin{align*}
  & \partialbenefit{\nodeu}{\solution_{\nodeu}  \cup \{\nodev\}} \\
  & \quad =~ \sum_{\nodew \in \solution_{\nodeu}} \expuseful{\nodew}{\solution_{\nodeu} \cup \{\nodev\}} \utility{\nodew} & \\
  & \quad =~ \expuseful{\nodeu}{\nodev} \utility{\nodeu} + \sum_{\nodew \in \solution_{\rchild(\nodeu)}} \expuseful{\nodew}{\solution_{\rchild(\nodeu)} \cup \{\nodeu\}} \utility{\nodew} & \\ 
  & \quad\quad\quad + \sum_{\nodew \in \solution_{\lchild(\nodeu)}} \expuseful{\nodew}{\solution_{\lchild(\nodeu)}\cup \{\nodeu\}} \utility{\nodew}&  \\
  & \quad =~ \partialbenefit{\nodeu}{\{\nodeu,\nodev\}} + \partialbenefit{\rightsub{\nodeu}}{\solution_{\rchild(\nodeu)} \cup \{\nodeu\}} + \partialbenefit{\leftsub{\nodeu}}{\solution_{\lchild(\nodeu)} \cup \{\nodeu\}}.
  \end{align*}
  The case $\nodeu \not\in \solution_{\nodeu}$ is similar 
  and we omit the details for~brevity. 
  \end{proof}
}

\mpara{Dynamic programming.}
Finally, we discuss how to use the structural properties
shown above in order to devise the dynamic-programming algorithm.
We first define the data structures that we use. 
Consider a node \nodeu in the elimination tree, 
a node $\nodev \in\extendedancestors{\nodeu}$, and 
an integer \partialbudget between 1 and $\min\{\budget,|\subtree{\nodeu}|\}$. 
% For such nodes \nodeu and \nodev and integer \partialbudget 
We define
\DPtable{\nodeu}{\partialbudget}{\nodev}
to be the optimal value of partial benefit
\partialbenefit{\nodeu}{\solution} 
over all sets of nodes \solution that
satisfy the following three conditions:
% \squishlist
% \item[($i$)]   
(i) $|\subtree{\nodeu} \cap \solution| \le \partialbudget$;
% \item[($ii$)]  
(ii) $\nodev\in\solution$; and
% \item[($iii$)] 
(iii) $\nodepath{\nodeu}{\nodev} \cap \solution = \emptyset$.
% \squishend
% In the definition of \DPtable{\nodeu}{\partialbudget}{\nodev}
Condition ($i$) states that the node set \solution has at most \partialbudget nodes in the subtree \subtree{\nodeu};
condition ($ii$) states that node \nodev is contained in \solution; and 
condition ($iii$) states that no other node between \nodeu and \nodev is contained in \solution, 
i.e., node \nodev is the lowest ancestor of \nodeu in \solution. 

For all \nodeu, \nodev, \partialbudget, and sets \solution that satisfy conditions ($i$)--($iii$) 
we also define \DPtableInc{\nodeu}{\partialbudget}{\nodev} and 
\DPtableExc{\nodeu}{\partialbudget}{\nodev} 
to denote the optimal partial benefit \partialbenefit{\nodeu}{\solution} 
for the cases when $\nodeu \in \solution$ and $\nodeu \not\in \solution$, respectively. 
Hence, we have 
\begin{align*}
\DPtable{\nodeu}{\partialbudget}{\nodev} = 
\max \left\{\DPtableInc{\nodeu}{\partialbudget}{\nodev}, \DPtableExc{\nodeu}{\partialbudget}{\nodev}  \right\}.
\end{align*}

We assume that the special node \nocap belongs in all solution sets \solution
but does not count towards the size of \solution. 
Notice that \DPtable{\nodeu}{\partialbudget}{\nocap}
is the optimal partial benefit 
\partialbenefit{\nodeu}{\solution} 
for all sets \solution that have at most \partialbudget nodes in \subtree{\nodeu}
and no  ancestor of \nodeu belongs to \solution. We now show how to compute 
\DPtable{\nodeu}{\partialbudget}{\nodev}
for all $\nodeu\in\nodes$, 
$\partialbudget \in\{1,\ldots,\min\{\budget,|\subtree{\nodeu}|\}\}$, 
and $\nodev\in\extendedancestors{\nodeu}$ 
by a bottom-up dynamic-programming algorithm:

% \begin{enumerate}
\squishlist
\item[1.]
If \nodeu is a leaf of the elimination tree then
\begin{align*}
\DPtableExc{\nodeu}{1}{\nodev} = 0,  
\mbox{ for all } \nodev\in\extendedancestors{\nodeu},
\\
%and 
\DPtableInc{\nodeu}{1}{\nodev} = -\infty,  
\mbox{ for all } \nodev\in\extendedancestors{\nodeu}. 
\end{align*}
This initialization enforces leaf nodes not to be selected, 
as they correspond to factors that define the Bayesian network and are part of the input.
%to be materialized,  as they are considered part of the input.
% \item
% If \nodeu is a leaf of the elimination tree then
% \[
% \DPtableExc{\nodeu}{1}{\nodev} = 0,  
% \mbox{ for all } \nodev\in\extendedancestors{\nodeu},
% \]
% and 
% \[
% \DPtableInc{\nodeu}{1}{\nodev} = \partialbenefit{\nodeu}{\{\nodeu, \nodev\}},  
% \mbox{ for all } \nodev\in\extendedancestors{\nodeu}. 
% \]
\item[2.]
If \nodeu is not a leaf of the elimination tree then
\begin{align*}
& \DPtableInc{\nodeu}{\partialbudget}{\nodev} ~=~ \partialbenefit{\nodeu}{\{\nodeu,\nodev\}} + \\
& \quad\quad\quad\quad  \max_{\partialbudgetl+\partialbudgetr = \partialbudget-1} 
\left\{
\DPtable{\leftsub{\nodeu}}{\partialbudgetl}{\nodeu} +
\DPtable{\rightsub{\nodeu}}{\partialbudgetr}{\nodeu}
\right\}, 
\end{align*}
and
\[
\DPtableExc{\nodeu}{\partialbudget}{\nodev} = \max_{\partialbudgetl+\partialbudgetr = \partialbudget} 
\left\{
\DPtable{\leftsub{\nodeu}}{\partialbudgetl}{\nodev} +
\DPtable{\rightsub{\nodeu}}{\partialbudgetr}{\nodev}
\right\}.
\]
\squishend
% \end{enumerate}
%
The value of the optimal solution is returned by 
$\DPtable{\rout}{\budget}{\nocap} = \max \{\DPtableInc{\rout}{\budget}{\nocap}, \DPtableExc{\rout}{\budget}{\nocap} \}$,
where \rout is the root of the elimination tree.
To compute the entries of the table 
\DPtable{\nodeu}{\partialbudget}{\nodev}, 
for all 
$\nodeu\in\nodes$, 
$\partialbudget \in\{1,\ldots,\min\{\budget,|\subtree{\nodeu}|\}\}$, 
and $\nodev\in\extendedancestors{\nodeu}$, 
we proceed in a bottom-up fashion.
For each node \nodeu, 
once all entries for the nodes in the subtree of \nodeu have been computed, 
we compute 
\DPtable{\nodeu}{\partialbudget}{\nodev}, 
for all $\partialbudget \in\{1,\ldots,\min\{\budget,|\subtree{\nodeu}|\}\}$, 
and all $\nodev\in\extendedancestors{\nodeu}$. Hence, computing each entry \DPtable{\nodeu}{\partialbudget}{\nodev}, requires only entries that are already computed.
Once all the entries \DPtable{\nodeu}{\partialbudget}{\nodev} are computed, we 
construct the optimal solution by backtracking --
specifically, invoking the subroutine 
ConstructSolu\-tion$({\rout},{\budget},{\nocap})$;
pseudocode as Algorithm~\ref{alg:constructSolution}.  

\iffalse
\begin{algorithm}[t]
\caption{ConstructSolution$({\nodeu},{\partialbudget},{\nodev})$}
\label{alg:constructSolution}
\Indm
\Indp
\BlankLine
\If{$\DPtable{\nodeu}{\partialbudget}{\nodev} = \DPtableInc{\nodeu}{\partialbudget}{\nodev}$}{
	{\bf print } $\nodeu$\;
	\If{$\partialbudget = 1$} {
{\bf return}\;
}
	$(\partialbudgetl^*,\partialbudgetr^*) \leftarrow \underset{\partialbudgetl + \partialbudgetr = \partialbudget - 1}{\arg\max~} \DPtable{\leftsub{\nodeu}}{\partialbudgetl}{\nodeu} + \DPtable{\rightsub{\nodeu}}{\partialbudgetr}{\nodeu}$ \;
	ConstructSolution$({\leftsub{\nodeu}},{\partialbudgetl^*},{\nodeu})$ \;
	ConstructSolution$({\rightsub{\nodeu}},{\partialbudgetr^*},{\nodeu})$ \;
}
\Else{
	$(\partialbudgetl^*,\partialbudgetr^*) \leftarrow \underset{\partialbudgetl + \partialbudgetr = \partialbudget}{\arg\max~} \DPtable{\leftsub{\nodeu}}{\partialbudgetl}{\nodev} + 	\DPtable{\rightsub{\nodeu}}{\partialbudgetr}{\nodev}$ \;
	ConstructSolution$({\leftsub{\nodeu}},{\partialbudgetl^*},{\nodev})$ \;
	ConstructSolution$({\rightsub{\nodeu}},{\partialbudgetr^*},{\nodev})$ \;
}
\end{algorithm}
\fi

\begin{theorem}
\label{theorem:DP-optimality}
The dynamic-programming algorithm 
desc\-rib\-ed above computes correctly the optimal solution $\solution^*$.
\end{theorem}

\FullOnly{
  \begin{proof}
  The correctness of the bottom-up computation of 
  $\DPtableInc{\nodeu}{\partialbudget}{\nodev}$ and 
  $\DPtableExc{\nodeu}{\partialbudget}{\nodev}$ follows from Lemmas~\ref{lemma:optimal-substructure} 
  and~\ref{lemma:additivity}. 
  Once we fill the table, we have, for each node $\nodeu$, 
  the optimal partial benefit for all possible combinations of partial solution size 
  $\partialbudget \in\{1,\ldots,\min\{\budget,|\subtree{\nodeu}|\}\}$ and 
  lowest solution ancestor $\nodev\in\extendedancestors{\nodeu}$ of $\nodeu$
  \begin{align*}
  \DPtable{\nodeu}{\partialbudget}{\nodev} = \underset{\substack{\solution_\nodeu \subseteq \subtree{\nodeu} \\ 
  \lvert\solution_\nodeu\rvert = \partialbudget}}{\max~} \partialbenefit{\nodeu}{\solution_\nodeu \cup \{\nodev\}}.
  \end{align*}
  Moreover, each entry $\DPtable{\nodeu}{\partialbudget}{\nodev}$ indicates whether $\nodeu$ would be included in any solution $\solution$ in which ($i$) $\partialbudget$ nodes are selected from $\subtree{\nodeu}$ into $\solution$ and ($ii$) $\nodev \in \solution$  and $\nodepath{\nodeu}{\nodev} \cap \solution = \emptyset$. Once we fill all the entries of the table, the optimal solution is constructed by Algorithm~\ref{alg:constructSolution} that performs a BFS traversal of the tree: the decision to select each visited node into $\solution^*$ is given based on its inclusion state indicated by the entry $\DPtable{\nodeu}{\partialbudget_\nodeu^*}{\lsa{\nodeu}{\solution^*}}$, where $\lsa{\nodeu}{\solution^*}$ is the lowest ancestor of $\nodeu$ in solution $\solution^*$ that is added to the solution before visiting $\nodeu$ and $\partialbudget_\nodeu^*$  is the optimal partial budget allowance for $\subtree{\nodeu}$, which are both determined by the decisions taken in previous layers before visiting node $\nodeu$. 
  \end{proof}
}

We note that optimality holds for the elimination order \eliminationOrder which, as explained in Section~\ref{sec:setting} is given as an input to the problem.  %remains fixed.
Notice that for each node \nodeu
the computation of the entries \DPtable{\nodeu}{\partialbudget}{\nodev} 
requires the computation of partial benefit values \partialbenefit{\nodeu}{\{\nodeu, \nodev\}} 
for pairs of nodes (\nodeu, \nodev), which in turn, 
require access to or computation of values \expuseful{\nodeu}{\nodev}. 
As Lemma~\ref{lemma:usefulnessRelation} below shows, the latter quantity can be computed from 
$\expuseful{\nodeu}{\emptyset}$ and $\expuseful{\nodev}{\emptyset}$, 
for all $\nodeu \in \nodes$ and $\nodev \in \ancestors{\nodeu}$.
In practice, it is reasonable to consider a setting 
where one has used historical query logs to learn empirical values 
for $\expuseful{\nodeu}{\emptyset}$ % for all $\nodeu \in \nodes$
and thus for \expuseful{\nodeu}{\nodev}. %  for pairs of nodes \nodeu and \nodev. 

\begin{algorithm}[!t]
\begin{small}
\begin{algorithmic}[1]
\If{$\DPtable{\nodeu}{\partialbudget}{\nodev} = \DPtableInc{\nodeu}{\partialbudget}{\nodev}$}
	\State {\bf print } $\nodeu$
	\If{$\partialbudget = 1$}
		\State \Return 
	\EndIf
	\State $(\partialbudgetl^*,\partialbudgetr^*) \leftarrow \underset{\partialbudgetl + \partialbudgetr = \partialbudget - 1}{\arg\max~} \DPtable{\leftsub{\nodeu}}{\partialbudgetl}{\nodeu} + \DPtable{\rightsub{\nodeu}}{\partialbudgetr}{\nodeu}$
	\State	ConstructSolution$({\leftsub{\nodeu}},{\partialbudgetl^*},{\nodeu})$
	\State	ConstructSolution$({\rightsub{\nodeu}},{\partialbudgetr^*},{\nodeu})$ 
\Else
	\State 
	$(\partialbudgetl^*,\partialbudgetr^*) \leftarrow \underset{\partialbudgetl + \partialbudgetr = \partialbudget}{\arg\max~} \DPtable{\leftsub{\nodeu}}{\partialbudgetl}{\nodev} + 	\DPtable{\rightsub{\nodeu}}{\partialbudgetr}{\nodev}$ 
	\State	ConstructSolution$({\leftsub{\nodeu}},{\partialbudgetl^*},{\nodev})$ 
	\State ConstructSolution$({\rightsub{\nodeu}},{\partialbudgetr^*},{\nodev})$ 
\EndIf
\end{algorithmic}
\caption{\label{alg:constructSolution}ConstructSolution$({\nodeu},{\partialbudget},{\nodev})$}
\end{small}
\end{algorithm}

%Notice that for each node \nodeu, computation of entries \DPtable{\nodeu}{\partialbudget}{\nodev} requires the computation of partial benefit values \partialbenefit{\nodeu}{\{\nodeu, \nodev\}} for pairs of nodes (\nodeu, \nodev) -- which, in turn, require access to or computation of values \expuseful{\nodeu}{\nodev}. In practice, it is reasonable to consider a setting where one has used historical query logs to learn a probability distribution \prob{\query} that represents the anticipated query-workload -- and where \prob{\isuseful{\query}{\nodeu}{\{\nodeu, \nodev\}}=1} is computed directly from \prob{\query} for pairs (\nodeu, \nodev).
%In such a setting, the complexity of computing \prob{\isuseful{\query}{\nodeu}{\{\nodeu, \nodev\}}=1} might vary with the form of \prob{\query}, but how exactly remains an open question for this paper.
%As the focus of the paper is on a different optimization (Problem~\ref{problem:qtm}), for the purposes of presentation and analysis of algorithms, we simply assume that access to \prob{\isuseful{\query}{\nodeu}{\{\nodeu, \nodev\}}=1} is provided by an {\it oracle} in constant time for each value request.
%For the purposes of evaluation of algorithms, we experiment in settings where \prob{\isuseful{\query}{\nodeu}{\{\nodeu, \nodev\}}=1} is directly accessible or given by a closed form expression.
%

\begin{lemma}
\label{lemma:usefulnessRelation} 
Let $\nodeu \in \nodes$ be a given node in an elimination tree \tree and 
let $\nodev \in \ancestors{\nodeu}$ denote an ancestor of $\nodeu$. Then,
\begin{align*}
\expuseful{\nodeu}{\nodev} = \expuseful{\nodeu}{\emptyset} - \expuseful{\nodev}{\emptyset}.
\end{align*}
%\prob{\isuseful{\query}{\nodeu}{\nodeu} = 1} = \prob{\isuseful{\query}{\nodev}{\nodev} = 1} + \prob{\isuseful{\query}{\nodeu}{\{\nodeu, \nodev\}} = 1}.
\end{lemma}

\FullOnly{
  \begin{proof}
  Notice that for any possible query $\query$, 
  whenever $\vars{\nodev} \subseteq \qvarmiss$, 
  we also have $\vars{\nodeu} \subseteq \qvarmiss$, 
  since $\subtree{\nodeu} \subseteq \subtree{\nodev}$. 
  This suggests that given any query $\query$ for which $\isuseful{\query}{\nodev}{\emptyset} = 1$, 
  we also have $\isuseful{\query}{\nodeu}{\emptyset} = 1$. 
  On the other hand, when $\isuseful{\query}{\nodeu}{\emptyset} = 1$, 
  there can be two cases: 
  ($i$) $\vars{\nodev} \subseteq \qvarmiss$, which implies ${\isuseful{\query}{\nodeu}{\nodev} = 0}$, and 
  ($ii$) there exists a node $\nodew \in \subtree{\nodev} \setminus \subtree{\nodeu}$ 
  such that $ \vars{\nodew} \not\subseteq \qvarmiss$, 
  which implies ${\isuseful{\query}{\nodeu}{\nodev} = 1}$. 
  The latter suggests that the event $[\isuseful{\query}{\nodev}{\emptyset} = 1]$ 
  occurs for a subset of queries \query for which the event $[\isuseful{\query}{\nodeu}{\emptyset} = 1]$ occurs. 
  The lemma follows. 
  \end{proof}
}

Finally, the running time of the algorithm can be easily derived by the time needed to compute all entries of the dynamic-programming table.
We note that the efficiency of the algorithm, as analyzed in Theorem~\ref{theorem:DP-running-time}, makes it practical to update the materialized factors whenever the input Bayesian network is updated. 

%\ReviewOnly{As with all lemmas and theorems, the proof can be found in the extended version~\cite{us-arxiv}.}

\begin{theorem}
\label{theorem:DP-running-time}
The running time of the dynamic-program\-ming algorithm 
is $\bigO{\nonodes\height\budget^2}$, 
where \nonodes is the number of nodes in the elimination tree,
\height is its height, 
and \budget is the number of nodes to materialize.
\end{theorem}

\FullOnly{
\begin{proof}
Notice that we have $\bigO{\nonodes \height \budget}$ subproblems, 
where each subproblem corresponds to an entry \DPtable{\nodeu}{\partialbudget}{\nodev} of the three-dimensional table. To fill each entry of the table, we need to compute the two distinct values of \partialbudgetr and \partialbudgetl that maximize \DPtableInc{\nodeu}{\partialbudget}{\nodev} (subject to $\partialbudgetr + \partialbudgetl = \partialbudget - 1$) and \DPtableExc{\nodeu}{\partialbudget}{\nodev} (subject to $\partialbudgetr + \partialbudgetl = \partialbudget$), respectively. 
Thus, it takes $\bigO{\budget}$ time to fill each entry of the table in a bottom-up fashion, hence, 
the overall running time is $\bigO{\nonodes\height\budget^2}$.  
\end{proof}
}

\ReviewOnly{
%\note[cigdem]{if there is enough space, include medium-greedy.tex otherwise keep short-greedy.tex}
\eat{\subsection{Greedy algorithm}
\label{sec:greedy}
\CA{In this section, 
we first point out that the benefit function 
$\benefitfn: 2^{\nodes} \rightarrow \posreals$ is monotone and submodular. 
We then exploit these properties to provide a greedy algorithm 
that achieves an approximation guarantee of $(1 - \frac{1}{e})$. 
In the discussion that follows, we'll be using the notion of 
{\it marginal benefit} to refer to the benefit we gain by adding 
one extra node to the solution set.}

\CA{\begin{definition}[Marginal Benefit]
\label{definition:partial-benefit}
Consider an elimination tree $\tree=(\nodes,\edges)$, 
a set of nodes $\solution\subseteq\nodes$, 
a node $\nodeu \in \nodes \setminus \solution$, and 
a probability distribution $\prob{\query}$
over the set of all possible queries. 
The marginal benefit $\mgbenefit{\nodeu}{\solution}$ of the node $\nodeu$ 
with respect to the solution set $\solution$ is defined as:
\[
\mgbenefit{\nodeu}{\solution} = \benefit{\solution \cup \{\nodeu\}} - \benefit{\solution}.
\]
\end{definition}
Marginal benefits can be computed via the closed-form expression provided by the following Lemma. 
\begin{lemma}\label{lemma:mgBenefit}
Consider an elimination tree $\tree=(\nodes,\edges)$, 
a set of nodes $\solution\subseteq\nodes$, 
a node $\nodeu \in \nodes \setminus \solution$, 
and a probability distribution $\prob{\query}$
over the set of all possible queries. 
Let $\hsd{\nodeu}{\solution} = \{\nodev \mid \nodev \in \subtree{\nodeu} \cap \solution 
\text{ and } \nodepath{\nodev}{\nodeu} \cap \solution = \emptyset \}$ 
denote the set of descendants of \nodeu in set \solution 
whose lowest ancestor in $\solution \cup \{\nodeu\}$ is \nodeu, and let 
$\lsa{\nodeu}{\solution} \in \extendedancestors{\nodeu}$ denote the lowest ancestor of \nodeu in \solution. 
Then, the marginal benefit $\mgbenefit{\nodeu}{\solution}$ of node \nodeu 
with respect to the set \solution is given by: 
\begin{align}\label{eq:mgBenefit}
\mgbenefit{\nodeu}{\solution} = 
\expuseful{\nodeu}{\lsa{\nodeu}{\solution}} \left(\utility{\nodeu} - 
\sum_{\nodev \in  \hsd{\nodeu}{\solution}} \utility{\nodev} \right). 
\end{align}
\end{lemma}}

\CA{The main result of this section is the following.
\begin{lemma}\label{lemma:submodular}
The benefit function $\benefitfn: 2^{\nodes} \rightarrow \posreals$ is monotone and submodular. 
\end{lemma}}

\CA{Consider now the greedy algorithm that creates a solution set incrementaly,
each time adding the node with the highest marginal benefit into the solution set 
until the cardinality budget is consumed. Then, it follows from Lemma~\ref{lemma:submodular} and the classic result of Nemhauser et al.~\cite{nemhauser1978analysis} that the greedy algorithm for Problem~\ref{problem:qtm} offers a $(1 - 1/e)$ approximation guarantee.}

}
\subsection{Greedy algorithm}
\label{sec:greedy}

% The greedy algorithm creates a solution set incrementally,
% each time adding the node that leads to the highest benefit 
% until the cardinality budget is consumed. %, as shown in Algorithm~\ref{alg:greedy}. 
The benefit function is non-negative, non-decreasing, %  for larger solution sets \solution, 
and submodular. Hence, it follows from the classic result of Nemhauser et al.~\cite{nemhauser1978analysis} that the greedy algorithm for Problem~\ref{problem:qtm} offers a $(1 - 1/e)$ approximation guarantee.
Technical details and proofs are in the extended version of the paper~\cite{aslay2020query}.

}
\FullOnly{
  \subsection{Greedy algorithm}
\label{sec:greedy}

In this section, 
we first point out that the benefit function 
$\benefitfn: 2^{\nodes} \rightarrow \posreals$ is monotone and submodular. 
We then exploit these properties to provide a greedy algorithm 
that achieves an approximation guarantee of $(1 - \frac{1}{e})$. 
In the discussion that follows, we'll be using the notion of 
{\it marginal benefit} to refer to the benefit we gain by adding 
one extra node to the solution set.

\begin{definition}[Marginal Benefit]
\label{definition:partial-benefit}
Consider an elimination tree $\tree=(\nodes,\edges)$, 
a set of nodes $\solution\subseteq\nodes$, 
a node $\nodeu \in \nodes \setminus \solution$, and 
a probability distribution $\prob{\query}$
over the set of all possible queries. 
The marginal benefit $\mgbenefit{\nodeu}{\solution}$ of the node $\nodeu$ 
with respect to the solution set $\solution$ is defined as:
\[
\mgbenefit{\nodeu}{\solution} = \benefit{\solution \cup \{\nodeu\}} - \benefit{\solution}.
\]
\end{definition}
Marginal benefits can be computed via the closed-form expression provided by the following Lemma. 
\begin{lemma}\label{lemma:mgBenefit}
Consider an elimination tree $\tree=(\nodes,\edges)$, 
a set of nodes $\solution\subseteq\nodes$, 
a node $\nodeu \in \nodes \setminus \solution$, 
and a probability distribution $\prob{\query}$
over the set of all possible queries. 
Let $\hsd{\nodeu}{\solution} = \{\nodev \mid \nodev \in \subtree{\nodeu} \cap \solution 
\text{ and } \nodepath{\nodev}{\nodeu} \cap \solution = \emptyset \}$ 
denote the set of descendants of \nodeu in set \solution 
whose lowest ancestor in $\solution \cup \{\nodeu\}$ is \nodeu, and let 
$\lsa{\nodeu}{\solution} \in \extendedancestors{\nodeu}$ denote the lowest ancestor of \nodeu in \solution. 
Then, the marginal benefit $\mgbenefit{\nodeu}{\solution}$ of node \nodeu 
with respect to the set \solution is given by: 
\begin{align}\label{eq:mgBenefit}
\mgbenefit{\nodeu}{\solution} = 
\expuseful{\nodeu}{\lsa{\nodeu}{\solution}} \left(\utility{\nodeu} - 
\sum_{\nodev \in  \hsd{\nodeu}{\solution}} \utility{\nodev} \right). 
\end{align}
\end{lemma}

\FullOnly{
\begin{proof}
  Notice that for all nodes $\nodev \in \solution \setminus \hsd{\nodeu}{\solution}$ 
  the lowest ancestor $\lsa{\nodev}{\solution}$ of \nodev  in \solution remains unchanged in 
  $\solution \cup \{\nodeu\}$. 
  On the other hand, for each node $\nodev \in \hsd{\nodeu}{\solution}$ 
  we have $\lsa{\nodev}{\solution} = \lsa{\nodeu}{\solution}$. 
  Thus, using Lemmas~\ref{lemma:lowestAncestor} and~\ref{lemma:usefulnessRelation}, we have:  
  \begin{align*}
  \mgbenefit{\nodeu}{\solution} 
  & = \benefit{\solution \cup \{\nodeu\}} - \benefit{\solution} \\
  & = \sum_{\nodev \in \solution \cup \{\nodeu\}} \expuseful{\nodev}{\solution \cup \{\nodeu\}} \utility{\nodev} \\
  & \;\;- \sum_{\nodev \in\solution} \expuseful{\nodev}{\solution} \utility{\nodev} \\ 
  & =\expuseful{\nodeu}{\solution \cup \{\nodeu\}} \utility{\nodeu} \\ 
  & \;\;+ \sum_{\nodev \in  \hsd{\nodeu}{\solution}} \left(\expuseful{\nodev}{\solution \cup \{\nodeu\}} - \expuseful{\nodev}{\solution}\right) \utility{\nodev}  \\
  & \;\;+ \sum_{\nodev \in  \solution \setminus \hsd{\nodeu}{\solution}} \left(\expuseful{\nodev}{\solution \cup \{\nodeu\}} - \expuseful{\nodev}{\solution}\right) \utility{\nodev} \\
  & \stackrel{L.\ref{lemma:lowestAncestor}}{=} \expuseful{\nodeu}{\lsa{\nodeu}{\solution}} \utility{\nodeu} \\ 
  & \;\;+ \sum_{\nodev \in  \hsd{\nodeu}{\solution}} \left(\expuseful{\nodev}{\nodeu} - \expuseful{\nodev}{\lsa{\nodeu}{\solution}}\right) \utility{\nodev} \\
  & \stackrel{L.\ref{lemma:usefulnessRelation}}{=} 
   \expuseful{\nodeu}{\lsa{\nodeu}{\solution}} \left(\utility{\nodeu} - \sum_{\nodev \in  \hsd{\nodeu}{\solution}} \utility{\nodev}\right).
  \end{align*}
  %spacesaving
  %Notice that,  following Lemma~\ref{lemma:usefulnessRelation} we have: 
  %$$\expuseful{\nodev}{\nodeu} = \expuseful{\nodev}{\emptyset} - \expuseful{\nodeu}{\emptyset},$$ and $$\expuseful{\nodev}{\lsa{\nodeu}{\solution}} = \expuseful{\nodev}{\emptyset} - \expuseful{\lsa{\nodeu}{\solution}}{\emptyset}.$$ Thus, we have:  
  %\begin{align*}
  %\mgbenefit{\nodeu}{\solution} &= \expuseful{\nodeu}{\lsa{\nodeu}{\solution}} \utility{\nodeu} \\ 
  %&+ \sum_{\nodev \in \hsd{\nodeu}{\solution}} \left(\expuseful{\nodev}{\nodeu} - \expuseful{\nodev}{\lsa{\nodeu}{\solution}} \right) \utility{\nodev} \\
  %&= \expuseful{\nodeu}{\lsa{\nodeu}{\solution}} \left(\utility{\nodeu} - \sum_{\nodev \in  \hsd{\nodeu}{\solution}} \utility{\nodev}\right). 
  %\end{align*}
  \end{proof}
}

The main result of this section is the following.
\begin{lemma}\label{lemma:submodular}
The benefit function $\benefitfn: 2^{\nodes} \rightarrow \posreals$ is monotone and submodular. 
\end{lemma}

\FullOnly{
  \begin{proof}
  We will first show that the benefit function $\benefitfn$ is monotone, i.e., $\mgbenefit{\nodeu}{\solution} \ge 0$ for any given $\solution \subseteq \nodes$ and $\nodeu \in \nodes \setminus \solution$. 
  % First consider the case $\hsd{\nodeu}{\solution} = \emptyset$. In this case, it is easy to see that we always have 
  % $$\mgbenefit{\nodeu}{\solution} = \expuseful{\nodeu}{\lsa{\nodeu}{\solution}} \utility{\nodeu} \ge 0.$$
  % % for any $\hsd{\nodeu}{\solution} \in \extendedancestors{\nodeu}$. 
  % Now consider the case $\hsd{\nodeu}{\solution} \neq \emptyset$. 
  % In this case, 
  In light of Lemma~\ref{lemma:mgBenefit},
  it suffices to show that $\utility{\nodeu} \ge \sum_{\nodev \in  \hsd{\nodeu}{\solution}} \utility{\nodev}$ for any $\nodeu$ and for any possible set $\hsd{\nodeu}{\solution}$ of its descendants whose lowest ancestor in $\solution \cup \{\nodeu\}$ is $\nodeu$. Notice that for any node $\nodev \in \hsd{\nodeu}{\solution}$, by definition we have $\nodepath{\nodev}{\nodeu} \cap \solution = \emptyset$, which means that no ascendant or descendant of $\nodev$ can be in $\hsd{\nodeu}{\solution}$. Now, remember that by definition we have: 
  $$\utility{\nodeu} = \sum_{\substack{\nodex\in\subtree{\nodeu}}} \cost{\nodex}.$$
  Using this definition, we equivalently have 
  \begin{align} \label{eq:utilityRelation}
  \utility{\nodeu} = \cost{\nodeu} + \utility{\leftsub{\nodeu}} + \utility{\rightsub{\nodeu}},
  \end{align}
  which implies that the utility of a parent node is always greater than the sum of the utilities of its children. Given also that for any $\nodev \in \hsd{\nodeu}{\solution}$, no ascendant or descendant of $\nodev$ can be in $\hsd{\nodeu}{\solution}$, we have 
  $$\sum_{\nodev \in \hsd{\nodeu}{\solution} \cap \subtree{\leftsub{\nodeu}}} \utility{\nodev} \le \utility{\leftsub{\nodeu}} \text{ and } \sum_{\nodev \in \hsd{\nodeu}{\solution} \cap \subtree{\rightsub{\nodeu}}} \utility{\nodev} \le \utility{\rightsub{\nodeu}}.$$
  %and 
  %$$\sum_{\nodev \in \hsd{\nodeu}{\solution} \cap \subtree{\rightsub{\nodeu}}} \utility{\nodev} \le \utility{\rightsub{\nodeu}}.$$ 
  and therefore
  \begin{align*}
  \utility{\nodeu} 
  & \ge \utility{\leftsub{\nodeu}} + \utility{\rightsub{\nodeu}} \\
  & \ge \sum_{\nodev \in \hsd{\nodeu}{\solution} \cap \subtree{\leftsub{\nodeu}}} \utility{\nodev} + \sum_{\nodev \in \hsd{\nodeu}{\solution} \cap \subtree{\rightsub{\nodeu}}} \utility{\nodev} 
    = \sum_{\nodev \in  \hsd{\nodeu}{\solution}} \utility{\nodev},
  \end{align*}
  concluding the proof of monotonicity.

  We proceed to show that $\benefitfn$ is submodular, i.e., that for any $\solution \subseteq S \subseteq \nodes$ and $\nodeu \in \nodes \setminus S$, we have $\mgbenefit{\nodeu}{\solution} \ge \mgbenefit{\nodeu}{S}$. For any given $\solution$ and node $\nodew \not \in \solution$, let $S = \solution \cup \{\nodew\}$. We consider two cases: (i) $\nodew \in \ancestors{\nodeu}$, or (ii) $\nodew \in \subtree{\nodeu}$. Notice that the case of $\nodew$ being neither an ancestor or descendant of $\nodeu$ is trivial, since we would then have $\mgbenefit{\nodeu}{\solution} = \mgbenefit{\nodeu}{S}$. 

  First consider the case $\nodew \in \ancestors{\nodeu}$. In this case, it could be that either (i) $\nodew \in \nodepath{\nodeu}{\lsa{\nodeu}{\solution}}$, which implies $\nodew = \lsa{\nodeu}{S}$, or (ii) $\nodew \in \nodepath{\lsa{\nodeu}{\solution}}{\nocap}$ which means that the lowest ancestor of node $\nodeu$ in $S$ is the same as in $\solution$. It is easy to see that in the latter case we have $\mgbenefit{\nodeu}{\solution} = \mgbenefit{\nodeu}{S}$, hence, we only consider the case in former. Notice that $\nodew = \lsa{\nodeu}{S}$ implies that $\lsa{\nodeu}{\solution} \in \ancestors{\nodew}$, which, by Lemma~\ref{lemma:usefulnessRelation}, further implies that $\expuseful{\lsa{\nodeu}{\solution}}{\emptyset} \le \expuseful{\nodew}{\emptyset}$. Then we have:
  \begin{eqnarray*}
  \expuseful{\nodeu}{\lsa{\nodeu}{\solution}} 
  &  =  &  \expuseful{\nodeu}{\emptyset} - \expuseful{\lsa{\nodeu}{\solution}}{\emptyset} \\ 
  & \ge & \expuseful{\nodeu}{\emptyset}  - \expuseful{\nodew}{\emptyset} \\
  &  =  & \expuseful{\nodeu}{\nodew}. 
  \end{eqnarray*}
  Hence, by Lemma~\ref{lemma:mgBenefit} we have: 
  \begin{eqnarray*}
  \mgbenefit{\nodeu}{\solution} 
  &  =  & \expuseful{\nodeu}{\lsa{\nodeu}{\solution}} \left(\utility{\nodeu} - \sum_{\nodev \in  \hsd{\nodeu}{\solution}} \utility{\nodev} \right) \\ 
  & \ge &\expuseful{\nodeu}{\nodew} \left(\utility{\nodeu} - \sum_{\nodev \in  \hsd{\nodeu}{\solution}} \utility{\nodev} \right) \\ 
  & = &  \mgbenefit{\nodeu}{S}.
  \end{eqnarray*}
  %\item 

  Now we consider the case $\nodew \in \subtree{\nodeu}$. In this case, if $\nodew \not\in \hsd{\nodeu}{S}$ then it trivially follows that $\mgbenefit{\nodeu}{\solution} = \mgbenefit{\nodeu}{S}$, hence, we only consider the case in which $\nodew \in \hsd{\nodeu}{S}$. Notice that if $\nodew \in \hsd{\nodeu}{S}$ then either (i) $\subtree{\nodew} \cap \hsd{\nodeu}{\solution} = \emptyset$, or (ii) $\subtree{\nodew} \cap \hsd{\nodeu}{\solution} \neq \emptyset$. First, consider the case $\subtree{\nodew} \cap \hsd{\nodeu}{\solution} = \emptyset$ which implies that $\hsd{\nodeu}{S} = \hsd{\nodeu}{\solution} \cup \{\nodew\}$. Then we have: 
  \begin{eqnarray*}
  \mgbenefit{\nodeu}{\solution}  
  &  =  & \expuseful{\nodeu}{\lsa{\nodeu}{\solution}} \left(\utility{\nodeu} - \sum_{\nodev \in  \hsd{\nodeu}{\solution}} \utility{\nodev} \right) \\ 
  & \ge & \expuseful{\nodeu}{\lsa{\nodeu}{\solution}} \left(\utility{\nodeu} - \sum_{\nodev \in  \hsd{\nodeu}{\solution} \cup \{\nodew\}} \utility{\nodev} \right) \\
  & = & \mgbenefit{\nodeu}{S}. 
  \end{eqnarray*}
  Next, consider the case $\subtree{\nodew} \cap \hsd{\nodeu}{\solution} \neq \emptyset$. 
  In this case, it holds that $\hsd{\nodew}{\solution} = \{\nodev \in \hsd{\nodeu}{\solution}: \nodepath{\nodev}{\nodeu} \cap S = \nodew\}$ 
  and $\hsd{\nodeu}{S} = (\hsd{\nodeu}{\solution} \setminus \hsd{\nodew}{\solution}) \cup \{\nodew\}$. Remember that, as given by Eq.~\ref{eq:utilityRelation}, the utility of a parent node is always greater than the sum of utilities of its children. Given also that, for any $\nodev \in \hsd{\nodew}{\solution}$, no ascendant or descendant of $\nodev$ can be in $\hsd{\nodew}{\solution}$, we have 
  $$\utility{\nodew} \ge \sum_{\nodev \in \hsd{\nodew}{\solution}} \utility{\nodev},$$ which implies that $\sum_{\nodev \in  \hsd{\nodeu}{\solution}} \utility{\nodev} \le \sum_{\nodev \in  \hsd{\nodeu}{S}} \utility{\nodev}$, since we have $\hsd{\nodeu}{S} = (\hsd{\nodeu}{\solution} \setminus \hsd{\nodew}{\solution}) \cup \{\nodew\}$. Thus, we have: 
  \begin{eqnarray*}
  \mgbenefit{\nodeu}{\solution}  
  &  =  & \expuseful{\nodeu}{\lsa{\nodeu}{\solution}} \left(\utility{\nodeu} - \sum_{\nodev \in  \hsd{\nodeu}{\solution}} \utility{\nodev} \right) \\ 
  & \ge & \expuseful{\nodeu}{\lsa{\nodeu}{\solution}} \left(\utility{\nodeu} - \sum_{\nodev \in  \hsd{\nodeu}{S}} \utility{\nodev} \right) \\ 
  &  =  &\mgbenefit{\nodeu}{S}.
  \end{eqnarray*}
  This concludes the proof.
  \end{proof}
}

Consider now the greedy algorithm that creates a solution set incrementaly,
each time adding the node with the highest marginal benefit into the solution set 
until the cardinality budget is consumed, as shown in Algorithm~\ref{alg:greedy}. 
It is easy to show that the algorithm comes with a constant factor approximation guarantee.
\begin{theorem} 
\label{theo:greedy}
Algorithm~\ref{alg:greedy} achieves an approximation guarantee of $(1 - 1/e)$. 
\end{theorem}

\FullOnly{
\begin{proof}
As shown in Lemma~\ref{lemma:submodular}, the non-negative
benefit function \benefitfn is monotone and submodular. 
Hence, the $(1 - 1/e)$ approximation guarantee for the greedy method follows from the classic result of 
Nemhauser et al.~\cite{nemhauser1978analysis}.
\end{proof}
}

\begin{algorithm}[t]
\begin{algorithmic}[1]
\State $\solution \leftarrow \emptyset$
\While{$|\solution| < k$}
	\State $\nodeu \leftarrow \underset{\nodev \in \nodes \setminus \solution}{\arg\max~} 
  \benefit{\solution \cup \{\nodev\}} -  \benefit{\solution}$ 
	\State $\solution \leftarrow \solution \cup \{ \nodeu \}$
\EndWhile
\State \Return \solution
\end{algorithmic}
\caption{\label{alg:greedy}Greedy Algorithm}
\end{algorithm}

%\note[Cigdem]{Will add a toy example showing greedy is suboptimal.}
%\note[Michael]{This might also be shown in the experiments, we can do what's easier.}
}

\section{Extensions}
\label{sec:extensions}

\subsection{Space budget constraints}
\label{sec:space-budget}

The algorithms presented in the previous section 
address Problem~\ref{problem:qtm}, 
where budget \budget specifies the number of nodes to materialize. 
A more practical scenario is Problem~\ref{problem:qtm-space}, 
where a budget \spacebudget is given on the total space required to materialize the selected nodes. 
In this case, for each node \nodeu of the elimination tree \tree
the space~$\weight{\nodeu}$ required to materialize the probability table at node \nodeu
is given as input. 
Both algorithms, 
dynamic-programming and greedy, 
are extended to address this more general version of the problem.
The extension is fairly standard, 
and thus we describe it here only briefly.

For the dynamic-programming algorithm the idea is to create an entry
\DPtable{\nodeu}{\partialbudget}{\nodev}
for nodes \nodeu and \nodev,
and with index \partialbudget taking values from 1 to $\min\{\spacebudget,\totalweight{\nodeu}\}$, 
where \totalweight{\nodeu} is the total space required to materialize 
the probability tables of all nodes in~\subtree{\nodeu}.
We then evaluate 
\DPtable{\nodeu}{\partialbudget}{\nodev}
as the maximum benefit over all possible values 
\partialbudgetl and \partialbudgetr such that 
$\partialbudgetl+\partialbudgetr = \partialbudget-\weight{\nodeu}$, 
where \weight{\nodeu} is the space required to materialize node \nodeu.

The modified algorithm provides the exact solution in $\bigO{\nonodes\height\spacebudget^2}$ running time.
Note, however, that unlike the previous case (Problem~\ref{problem:qtm}) 
where \budget is bounded by \nonodes, 
the value of \spacebudget is not bounded by \nonodes. 
As the running time is polynomial in the \emph{value} of \spacebudget, 
which is specified by $\bigO{\log\spacebudget}$ bits, 
it follows that the algorithm is \emph{pseudo-polynomial}.
However, the technique can be used to obtain
a fully-polynomial approximation scheme (FPTAS)
by rounding all space values into a set of smaller granularity 
and executing the dynamic programming algorithm using these rounded values. 

For the greedy algorithm, 
in each iterative step
we select to materialize the node \nodeu that maximizes the 
\emph{normalized marginal gain}
$\left(\benefit{\solution \cup \{\nodeu\}} - \benefit{\solution}\right)/\weight{\nodeu}$.
The modified greedy algorithm has the same running time
and approximation guarantee
$\left(1-\frac 1e\right)$~\cite{sviridenko2004note}.

\subsection{Accounting for redundant variables}
\label{sec:redundant}
\ReviewOnly{
So far we have considered 
a \emph{fixed} elimination tree \tree and elimination order \eliminationOrder.
The elimination~tree~\tree specifies the order in which sum-of-product operations are performed, 
with one summation for \emph{every variable} in the Bayesian network~\bayesianNet.
One can observe, however,  that it is \emph{not} necessary to involve \emph{every variable} in the evaluation of a query: previous work~\cite{geiger1990d,lauritzen1988local,zhang1994simple}
provides methods to determine the variables that are \emph{redundant} for the evaluation of a query \query\  
allowing to perform computations based on a ``shrunk'' Bayesian network.

Accordingly, one can devise a \emph{redundancy-aware scheme} by materializing different probability tables for a set of ``shrunk'' Bayesian networks obtained through removal of redundant variables. {Here, we provide a brief discussion for it and provide the details in the extended version of the paper\cite{aslay2020query}.} Let \lattice be the set that contains these shrunk Bayesian networks and the input Bayesian network. The set \lattice can be represented as a lattice with edges between each network and its maximal subnetworks in \lattice. When a query arrives, it is efficiently mapped to one network in the lattice from which its value can be computed exactly. In this scheme, optimization considerations involve the choice of networks to include in the lattice, as well as the materialized factors for each network. 
}

\FullOnly{
In our algorithms so far we have considered 
a \emph{fixed} elimination tree \tree and elimination order \eliminationOrder.
The elimination~tree \tree specifies 
% the computations that are performed for the evaluation of a query \query\ 
% --- essentially, 
the order in which sums-of-products evaluations are performed, 
with one summation for \emph{every variable} in~\bayesianNet.
% as described in Section~\ref{sec:setting}.
One can observe, however,  that 
it is \emph{not} necessary to involve \emph{every variable} in the evaluation of a query.
For example, for the Bayesian network \bayesianNet shown in Figure~\ref{fig:redundant}, 
the query $\query_1 = \prob{\varb = \varbvalue, \varc}$
can be computed from the sub-network $\bayesianNet_1 \subseteq \bayesianNet$, 
while the query $\query_2 = \prob{\varc \mid \varb = \varbvalue}$
can be computed from the sub-network $\bayesianNet_2 \subseteq \bayesianNet$.

\begin{figure}
\begin{tabular}{c||c|c|c|c}
$\mathcal N$ & $\mathcal N_1$ &  $\mathcal N_2$ &$\mathcal N_3$ & $\mathcal N_4$\\ %\hline
\includegraphics[width=0.14\columnwidth]{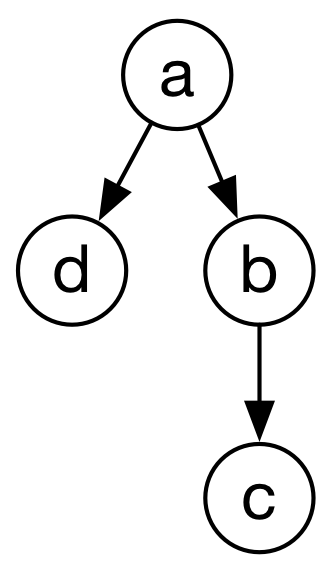} &
\includegraphics[width=0.14\columnwidth]{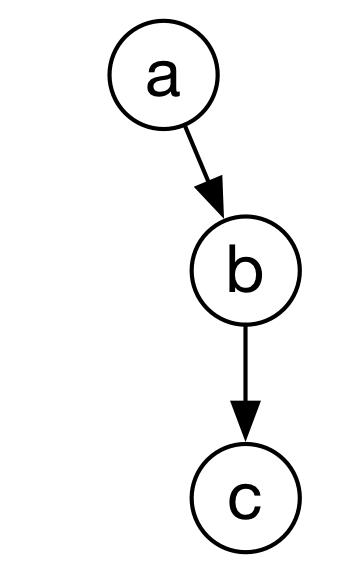} &
\includegraphics[width=0.14\columnwidth]{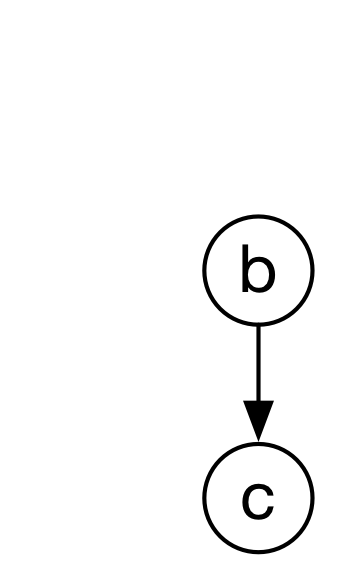} &
\includegraphics[width=0.14\columnwidth]{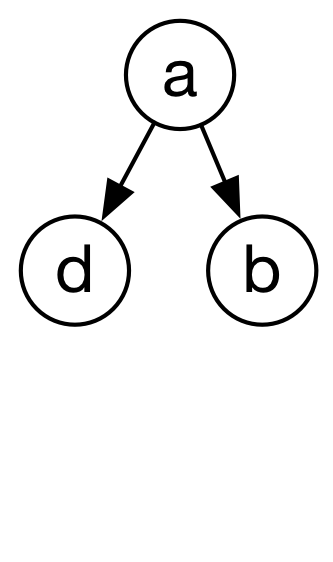} &
\includegraphics[width=0.14\columnwidth]{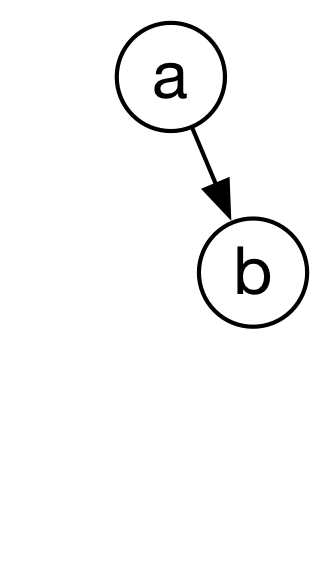} % \\ \hline
\end{tabular}
\caption{The input Bayesian network \bayesianNet and some of its sub-networks. Some variables in \bayesianNet may be redundant for the evaluation of some queries, allowing us to perform computations over ``shrunk'' networks.}
\label{fig:redundant}
\end{figure}

Previous work~\cite{geiger1990d,lauritzen1988local,zhang1994simple}
provides methods to determine the variables that are \emph{redundant} for the evaluation of a query \query\  
allowing us to perform computations based on a ``shrunk'' Bayesian network.
% i.e., the sub-network induced from \bayesianNet by the non-redundant variables.
The characterization of variables into redundant and non-redundant is given in Theorem~\ref{thm:redundant}, 
based on the following two definitions.
\begin{definition}[Moral graph~\cite{zhang1994simple}] 
The moral graph \moralGraph of a Bayesian network \bayesianNet 
is the undirected graph that results from \bayesianNet 
after dropping edge directions and 
adding one edge for all pairs of nodes that share a common child. % in~\bayesianNet.
\end{definition} 
\begin{definition}[\separated variables~\cite{zhang1994simple}] 
Two\\vari\-ables \vara and \varb in a Bayesian network \bayesianNet
are said to be \separated by variables \somevars 
if removing \somevars from the moral graph \moralGraph of \bayesianNet leaves no 
(undirected) path between \vara and \varb in \moralGraph. 
This property is denoted as \isseparated{\vara}{\varb}{\somevars}.
\end{definition}
\begin{theorem}[Redundant Variables~\cite{zhang1994simple}]
Let \bayesianNet be~a Bayesian network  and 
$\query = \prob{\qvarfree, \qvarbound=\qvarvals \mid \qvarbound'=\qvarvals'}$ a query. 
Let \ancestors{} be the union of ancestors in \bayesianNet of all variables in \query:
\[
\ancestors{} = \cup_{x\in\qvarfree\cup\qvarbound\cup\qvarbound'}\ancestors{x}.
\]
Also, let 
\redundantmarginal be all variables outside \ancestors{}, 
i.e., $\redundantmarginal = \allvars \setminus \ancestors{}$,
and \redundantconditional all ancestor nodes \ancestors{} that are 
\separated from $\qvarfree\cup\qvarbound$ by $\qvarbound'$, i.e., 
\[
\redundantconditional = \{\vara\in\ancestors{} \mid \isseparated{\vara}{\varb}{\qvarbound'} \text{ for all } \varb\in\qvarfree\cup\qvarbound\}.
\]
The variables in $\redundantvars = \redundantmarginal\cup\redundantconditional$ are redundant, 
and no other variables are redundant.
\label{thm:redundant}
\end{theorem}
Given a Bayesian network \bayesianNet and a query \query
we write $\shrink{\query, \bayesianNet}$ 
(or $\shrink{\query}$ when \bayesianNet is understood from the context)
to denote the Bayesian network 
that results from the removal of all redundant variables as per Theorem~\ref{thm:redundant}.
We can evaluate the query \query on a shrunk Bayesian network \resultBN, 
such that $\shrink{\query} \subseteq \resultBN \subseteq \bayesianNet$,
by building an elimination tree $\tree'$ on \resultBN 
and obtain immediate efficiency gains.
However, the elimination tree $\tree'$ that is built on $\resultBN$ 
codifies different computations than the tree \tree built on \bayesianNet, 
even if $\tree'\subseteq\tree$.
Therefore, the tables of factors we materialize for \tree 
using the algorithms of Section~\ref{sec:algorithms} 
do not generally correspond to factor tables for $\tree'$.
In the next section we discuss how to address
the issue of evaluating different queries while accounting for redundant variables.

\spara{Redundancy-aware scheme.} The main idea of our \emph{redundancy-aware scheme} 
is to materialize different probability tables for a set of ``shrunk'' Bayesian networks 
obtained through removal of redundant variables.
The scheme consists of the following components:

\begin{figure}
\begin{center}
\includegraphics[width=0.3\columnwidth]{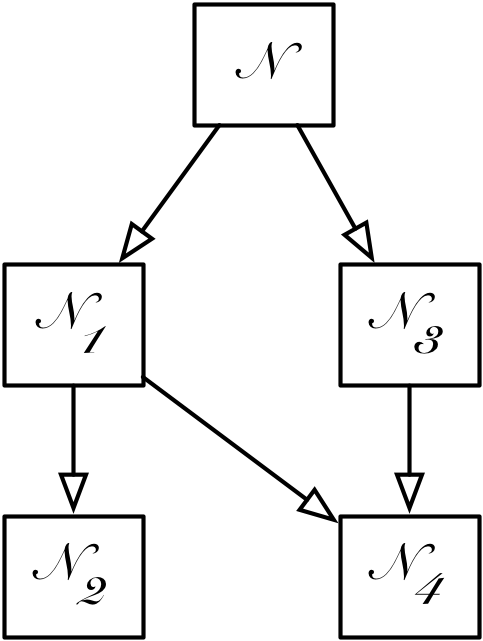}
\end{center}
\caption{Instance of a lattice of Bayesian networks, for the networks that appear in Figure~\ref{fig:redundant}.}
\label{fig:lattice}
\end{figure}

\spara{Lattice of Bayesian networks}. Consider a set of Bayesian networks 
$\lattice = \{\bayesianNet_0, \bayesianNet_1, \bayesianNet_2, \ldots, \bayesianNet_\latticeSize\}$ that includes the input Bayesian network $\bayesianNet = \bayesianNet_0$ and $\latticeSize$ of its subnetworks, 
each of which is induced by a subset of variables:
\[
\bayesianNet_0 = \bayesianNet;\; \bayesianNet_i \subseteq \bayesianNet, \text{ for all } i = 1,\ldots,\latticeSize.
\]
The set \lattice can be represented as a lattice where edges are added
between each network and its maximal subnetworks in \lattice 
(see an example in Figure~\ref{fig:lattice}).

\spara{Query-network mapping}. 
Consider a function $\mapping: \queriesset\rightarrow \lattice$
(where \queriesset is the set of all possible queries)
that maps a query $\query\in\queriesset$ to a Bayesian network $\bayesianNet_i\in\lattice$
from which the answer to \query can be computed exactly. 
Notice that there is always such a Bayesian network in the lattice, namely the input Bayesian network $\bayesianNet_0 = \bayesianNet$.

\spara{Query workloads}. 
Each Bayesian network $\bayesianNet_i\in\lattice$ 
is associated with a query workload, characterized by: 
($i$) the probability $\bnProb_i$ that a random query is mapped to $\bayesianNet_i$;
%  with $\pi_i\geq 0$ and $\sum_{i=0}^\latticeSize{ \bnProb_i } = 1$; 
($ii$) a probability distribution $\netprob{i}{\query} = \prob{\query \mid \bayesianNet_i}$ 
over the queries that are mapped to~$\bayesianNet_i$.

\smallskip
We now discuss how the scheme operates and how its components are built.
When a query arrives it is mapped to one network in the lattice 
from which its value is computed exactly. 
As we discuss below, this mapping operation can be performed efficiently. 
In this scheme, offline optimization considerations include 
the choice of networks to include in the lattice, 
as well as the materialization of factors for each network. 
We discuss them below.

\spara{An algorithm for query-network mapping.}
Algorithm~\ref{algo:map} finds the smallest % (in terms of node-size) 
Bayesian network \resultBN 
in the lattice such that $\shrink{\query} \subseteq \resultBN \subseteq \bayesianNet$
that can be used to answer a query \query. 
The algorithm proceeds as follows: 
first, at line~\ref{line:shrink}, it computes the smallest 
% (in terms of node-size) 
shrunk network \shrunkBayesianNet = \shrink{\query} 
that can be used to answer query \query exactly; 
% (this can be done in time $\bigO{|\edges|}$ as per Theorem~\ref{thm:redundant}); 
then, it performs a breadth-first-search on the lattice \lattice 
starting from the top element % (input Bayesian network \bayesianNet) 
but does not extend search paths on which it encounters networks \otherBayesianNet 
that do not include $\shrunkBayesianNet$ as subnetwork.
% (line~\ref{line:subnetworkTest}). 

To test whether $\shrunkBayesianNet\subseteq\otherBayesianNet$, it is sufficient to test whether the intersection of the (labeled) edge-sets of the two networks is not empty,
which can be done in time $\bigO{|\edges|\log(|\edges|)}$.
% Notice also (line~\ref{line:noDuplicates}) that the algorithm does not visit twice the same Bayesian network in the lattice, as it requires each tested Bayesian network to be strictly smaller than the best found up to that point.
The algorithm finds the correct network in the lattice since, by construction, 
if it has visited a Bayesian network \otherBayesianNet that contains the target \resultBN as subnetwork, 
there is a path from \otherBayesianNet to \resultBN, 
and this condition holds for the best (smallest Bayesian network that contains \shrunkBayesianNet) 
discovered up to any point during the execution of the algorithm.
The total running time in terms of subnetwork tests
%  (evaluations of $\shrunkBayesianNet\subseteq\otherBayesianNet$ in line~\ref{line:subnetworkTest}) 
is $\bigO{\latticeSize\,|\edges|\log(|\edges|)}$.

\begin{algorithm}[t]
\begin{small}
\begin{algorithmic}[1]
\caption{\label{algo:map}Map(\lattice, \query)}
\State Let $\shrunkBayesianNet \assign \shrink{\query}$\label{line:shrink}
\State Let $\queue \assign [\bayesianNet]$
\State Let $\resultBN \assign \bayesianNet$
\While{$\queue\not = \emptyset$}
	\State \otherBayesianNet = \deque{\queue}
	\If{$\|\otherBayesianNet\|<\resultBN$ \label{line:noDuplicates} }
		\If{$\shrunkBayesianNet\subseteq\otherBayesianNet$ \label{line:subnetworkTest} }
			\State $\resultBN\assign\otherBayesianNet$
			\For{$\yetAnotherBayesianNet\in\children{\otherBayesianNet}$}
				\State \enque{\yetAnotherBayesianNet, \queue}
			\EndFor
		\EndIf
	\EndIf
\EndWhile
\State \Return \resultBN
\end{algorithmic}
\caption{\label{algo:map}Map(\lattice, \query)}
\end{small}
\end{algorithm}

\mpara{Building the lattice.} 
We build the lattice \lattice off-line, in three phases.
During the first phase, we consider the full lattice \fullLattice that includes all sub-networks of \bayesianNet\ 
and estimate the probability $\fullBNProb_i$ that a random query \query has 
$\bayesianNet_i = \shrink{\query} \in\fullLattice$ as its corresponding ``shrunk'' network. 
Notice that, for the full lattice, $\fullBNProb_i$ is also the probability that a random query is mapped by Algorithm~\ref{algo:map} to Bayesian network $\bayesianNet_i\in\fullLattice$.
In practice, we consider a sample of queries \query 
(either from a query-log or a probabilistic model) and estimate $\fullBNProb_i$ 
as the relative frequency with which network $\bayesianNet_i$ is the ``shrunk'' network 
that can be used to evaluate \query.

During the second phase, we choose a small number \latticeSize of networks from \fullLattice 
% (other than the input Bayesian network) 
to form lattice \lattice. 
We want to build a lattice of networks that captures well the distribution \fullBNProb.
%  --- for example, to choose the \latticeSize networks with highest \fullBNProb. 
%% Another approach is to consider some notion of \emph{utility} for the $\latticeSize$ chosen networks that form lattice \lattice: a chosen lattice will have higher utility if it allows us to evaluate a workload of queries more efficiently compared to evaluating all queries directly from the input network \bayesianNet. 
In practice, we use a greedy approach, 
successively choosing to add to \lattice the network that optimizes the utility of the lattice. 
% Such an approach is used by \cite{harinarayan1996datacubes} to build a lattice for the efficient computation of data cubes.
% 
During the third phase, % having already built a lattice \lattice, 
we follow an approach similar to the first phase to estimate anew the probability $\bnProb_i$ 
that a random query \query has $\bayesianNet_i = \shrink{\query} \in\lattice$ 
as its corresponding ``shrunk'' Bayesian network, 
as well as the probability distribution $\netprob{i}{\query} = \prob{\query \mid \bayesianNet_i}$ 
over the queries \query that are mapped to $\bayesianNet_i$.

\mpara{Optimal materialization.} 
Given the set \lattice of networks contained in the lattice, 
a query workload $(\bnProb_i, \netprob{i}{\query})$ over the networks, and a budget $\budget$, 
we wish to materialize $\budget_i$ factors for Bayesian network $\bayesianNet_i$, 
with $\sum_{i=0}^{\latticeSize}\budget_i\leq\budget$,
so that
$\globalbenefit{\{\budget_i\}}=\sum_{i=0}^{\latticeSize}\bnProb_i \optimalbenefit{i}{\budget_i}$ is maximized, 
where $\optimalbenefit{i}{\budget_i}$ is the optimal benefit obtained by solving problem~\ref{problem:qtm} for Bayesian network $\bayesianNet_i$ with budget $\budget_i$.
Let \globallyoptimalbenefit{m, \budget} be the optimal value of 
% \globalbenefit{\{\budget_i;\ \text{s.t.}\ \budget_i = 0\ \text{for}\ i = (m+1)..\latticeSize\}} 
\globalbenefit{\cdot}
for the first $m$ networks of \lattice, with budget \budget. Then the following equation holds
\[
\globallyoptimalbenefit{m+1, \budget} = \max_\partialbudget\{\bnProb_{m+1}\optimalbenefit{i}{\partialbudget} + \globallyoptimalbenefit{m, \budget - \partialbudget}\}\,,
\]
and defines a dynamic-programming algorithm to compute the optimal materialization over a set of networks \lattice.
% 	\item This can be done optimally with a DP programming.
% 	\item DP formula for optimal benefit: 
% \[
% 	OPT[i;k] = \max_\kappa\{\pi_i opt[i;\kappa] + OPT[i+1; k - \kappa]\}
% \]
}
\section{Experiments}
\label{sec:experiments}

Our empirical evaluation has two parts.
In the first part, we evaluate the benefits of materialization for variable elimination.
As explained in Section~\ref{sec:setting}, 
we consider materializing only factors resulting from joins and variable summations. 
In the second part, we compare our approach with two junction tree-based inference algorithms, which also rely on materialization.
% Our experiments are performed on real-world Bayesian networks.
% while a more general choice would consider factors resulting from natural joins combined any possible operation per variable (variable summation, row selection, or no operation), 
% Naturally, this self-imposed constraint limits the benefit of materialization, since we do not materialize any type of factor.
In what follows, we first describe the experimental setup (Sec.~\ref{sec:setup}) 
and then the results (Sec.~\ref{sec:results}).

\subsection{Setup}
\label{sec:setup}

\begin{table}[t]
\fontsize{8}{9}\selectfont
% \small
\centering
\caption{\label{table:dataset}Statistics of Bayesian networks.}
\begin{tabular}{ l  r  r  r  r  }
\toprule
Network 
& nodes 
& edges 
& parameters
% We note that the reported numbers are obtained after removing redundant parameters.} 
& avg.\ degree \\ 
% & & & & degree \\
\midrule
\mildew & $35$ & $46$ & $547$\,K & $2.63$ \\ 
\bnpathfinder~\cite{heckerman1992toward} & $109$ & $195$ & $98$\,K & $ 2.96$   \\ 
\munins & $186$ & $273$ & $19$\,K & $2.94$ \\ 
\andes~\cite{conati1997line}  & $220$ & $338$ & $2.3$\,K & $ 3.03$  \\ 
\diabetes~\cite{andreassen1991model} & $413$ & $602$ & $461$\,K & $2.92$ \\ 
\link~\cite{jensen1999blocking}  & $714$ & $1\,125$ & $20$\,K & $3.11$  \\ 
\muninm~\cite{andreassen1989munin} & $1\,003$ & $1\,244$ & $84$\,K & $2.94$ \\ 
\muninb~\cite{andreassen1989munin} & $1\,041$ & $1\,397$ & $98$\,K & $ 2.68$    \\ 
\revision{\tpchsmall} & \revision{$17$} & \revision{$17$} & \revision{$1.5$\,K} & \revision{$2.00$} \\ 
\revision{\tpchmedium} & \revision{$31$} & \revision{$31$} & \revision{$7.4$\,K} & \revision{$2.00$} \\ 
\revision{\tpchverylarge} & \revision{$38$} & \revision{$39$} & \revision{$355$\,K} & \revision{$2.05$} \\ 
\revision{\tpchlarge} & \revision{$35$} & \revision{$37$} & \revision{$27$\,K} & \revision{$2.11$} \\ 
\bottomrule
%\insurance~\cite{binder1997adaptive} & $27$ & $52$ & $1.4$\,K & $3.85$   \\ 
%\hailfinder~\cite{abramson1996hailfinder} & $56$ & $66$ & $3.7$\,K & $2.36$   \\ 
%\water~\cite{jensen1989expert} & $32$ & $66$ & $13.4$\,K & $4.12$  \\ 
\end{tabular}
\end{table}

\spara{Datasets.}
We use real-world Bayesian networks (see Table~\ref{table:dataset} for statistics).
Column ``parameters'' refers to the number of entries 
of the factors that define the corresponding Bayesian network.
\bnpathfinder~\cite{heckerman1992toward} is used in an expert system that assists surgical pathologists with the diagnosis of lymph-node diseases. 
\diabetes~\cite{andreassen1991model} models insulin dose adjustment. 
\mildew is used to predict the necessary amount of fungicides against mildew in wheat.
%\footnote{The model was developed by Finn Jensen, J{\o}rgen Olesen, and Uffe Kj{\ae}rulff.} 
\link~\cite{jensen1999blocking} models the linkage between a gene associated with a rare heart disease 
(the human LQT syndrome) and a genetic marker gene. 
\muninb~\cite{andreassen1989munin} is used in an expert electromyography assistant. 
\munins and \muninm are two subnetworks of \muninb. 
\andes~\cite{conati1997line} is used in an intelligent tutoring system that teaches Newtonian physics to students. 
\revision{The \tpch Bayesian networks were learned from \tpch data, following Tzoumas et al.~\cite{tzoumas2013adapting}}.
\revision{When necessary due to the page limit, we show results only on a subset of datasets, implying that the results are similar on the rest.
\ReviewOnly{The full plots are placed in the extended version~\cite{aslay2020query} of the paper.}}
All datasets are publicly available online.%
\footnote{\revision{See \url{https://github.com/aslayci/qtm} for the \tpch datasets and \url{http://www.bnlearn.com/bnrepository/}~\cite{scutari2014bayesian} for the rest.}}

\spara{Elimination order.} 
As explained in Section~\ref{sec:setting}, \revision{elimination trees are determined by the given variable-elimination order}.
However, finding the optimal order is \nphard~\cite{koller2009probabilistic},
and several heuristics have been proposed.
% to construct an order that does not lead to exponential space blowup.  
% in the size of the factors. 
%that leads to convenient elimination trees, i.e., trees that do not lead to the computation of large multi-joins of factors. 
Among these heuristics, greedy algorithms 
% that minimize a heuristic cost function at each iteration 
perform well in practice~\cite{fishelson2004optimizing}. 
Given a Bayesian network~\bayesianNet,  
a greedy algorithm begins by initializing a graph \ordergraph from the ``moralization'' of \bayesianNet, 
i.e., by connecting the parents of each node and dropping the direction of the edges. 
Then at the $i$-th iteration, 
a node that minimizes a heuristic cost function is selected as the $i$-th variable in the ordering. 
The selected variable is then removed from \ordergraph and 
undirected edges are introduced between all its neighbors in \ordergraph.
In this paper, we consider heuristics where the cost of a node is:
{\em min-neighbors} (\mn): the number of neighbors it has in \ordergraph; 
{\em min-weight} (\mw): the product of domain cardinalities of its neighbors in \ordergraph;
{\em min-fill} (\mf): the number of edges that need to be added to \ordergraph due to its removal; 
and
{\em weighted-min-fill} (\wmf): the sum of the weights of the edges that need to be added to \ordergraph due to its removal, where the weight of an edge is the product of the domain cardinalities of its endpoints~\cite{koller2009probabilistic}. 

Table~\ref{table:factorTableSizes} shows statistics for the factors in the elimination trees created under orders generated by the aforementioned heuristics.
%As can be seen, the choice of elimination order has direct impact on the size of the factor tables and is important to consider. 
For each Bayesian network, we select the elimination order that induces the smallest average parameter size and use the maximum parameter size as a tie-breaker. 
\ReviewOnly{\revision{In separate experiments,
we have confirmed that this heuristic works well, 
i.e., it chooses an elimination order with the best or close-to-the-best performance and avoids very costly orders~\cite{aslay2020query}.}}
Table~\ref{table:treeStats} reports statistics of the elimination trees obtained from the chosen elimination order for each dataset. 
% In all the trees, the average number of children is found to be consistently close to $1$ and so we do not report it in the table.

\iffalse
\begin{table}[t]
\begin{center}
\scriptsize
\begin{tabular}{|c|c|c|c|c|}
    \hline
     {Order \eliminationOrder } & \mn & \mw & \mf & \wmf  \\ \hline
\mildew & 15\,K & 170\,K & 966\,K & 19\,M & {\bf 10\,K} & {\bf 170\,K} & 57\,K & 1\,M \\ 
\bnpathfinder &  570 & 16\,K & $10^{16}$ & $10^{19}$ & {\bf 568} & {\bf 16\,K} &  643 & 16\,K  \\ 
\munins & 749\,K & 59\,M & $10^{15}$ & $10^{19}$ & 375\,K & 39\,M & {\bf 367\,K} & {\bf 39\,M} \\ 
\andes  & 2\,K & 131\,K & $10^{12}$ & $10^{13}$ & {\bf 1.4\,K}  & {\bf 66\,K } & {\bf 1.4\,K}  & {\bf  66\,K } \\
\diabetes & 9\,K & 	194\,K & $10^{16}$ & $10^{19}$ &  {\bf 4\,K}  & {\bf 194\,K} & 325\,K & 33\,M \\
\link & 109\,K & 17\,M & $10^{16}$ & $10^{19}$ & {\bf 31\,K}  & {\bf 4\,M} &  633\,K & 268\,M \\
\muninm & 40\,K & 31\,M & $10^{15}$ & $10^{19}$ & {\bf 1.7\,K}  & {\bf 168\,K} & 1.8\,K  & 168\,K  \\
\muninb & 9.5\,K & 588\,K & $10^{16}$ & $10^{19}$ & 5\,K  & 392\,K & {\bf 3\,K}  & {\bf 112\,K} \\
\bottomrule
  \end{tabular}
  \caption{Average and maximum parameter size of factors created by different elimination orders.}
\end{center}
 \label{table:factorTableSizes}
\end{table}
\fi

\begin{table}[t]
\setlength\tabcolsep{4pt}
\fontsize{7.5}{8.5}\selectfont
\begin{center}
% \small
\caption{\label{table:factorTableSizes}Parameter size of factors created
with different elimination orders (K: thousand, M: million, T: trillion).}
% \begin{tabular}{lcccccc}
% \begin{tabular}{p{1.4cm}p{0.6cm}p{0.6cm}p{0.75cm}p{0.75cm}p{0.75cm}p{0.75cm}p{0.75cm}}
\begin{tabular}{lrrrrrrrr}
\toprule
 % & \multicolumn{6}{c}{Elimination order heuristic} \\ \cmidrule(lr{.99em}){2-7} 
 & \multicolumn{2}{c}{\mn} & \multicolumn{2}{c}{\mf} & \multicolumn{2}{c}{\wmf} & \multicolumn{2}{c}{\mw} \\
\cmidrule(lr{0.7em}){2-3}
\cmidrule(lr{0.7em}){4-5}
\cmidrule(lr{0.7em}){6-7}
\cmidrule(lr{0.7em}){8-9}
Network & avg & max & avg & max  & avg & max & avg & max  \\ \midrule
%\mildew & 15\,K & 170\,K & {\bf 10\,K} & {\bf 170\,K} & 57\,K & 1\,M & 966\,K & 19\,M \\ 
%\bnpathfinder\!\!\!\! &  570~~\, & 16\,K & {\bf 568~~\,} & {\bf 16\,K} &  643~~\, & 16\,K & $10^{16}$ & $10^{19}$  \\ 
%\munins & 749\,K & 59\,M & 375\,K & 39\,M & {\bf 367\,K} & {\bf 39\,M} & $10^{15}$ & $10^{19}$  \\ 
%\andes  & 2\,K & 131\,K & {\bf 1.4\,K}  & {\bf 66\,K} & {\bf 1.4\,K}  & {\bf  66\,K} & $10^{12}$ & $10^{13}$  \\
%\diabetes & 9\,K & 	194\,K &  {\bf 4\,K}  & {\bf 194\,K} & 325\,K & 33\,M & $10^{16}$ & $10^{19}$ \\
%\link & 109\,K & 17\,M  & {\bf 31\,K}  & {\bf 4\,M} &  633\,K & 268\,M & $10^{16}$ & $10^{19}$ \\
%\muninm & 40\,K & 31\,M  & {\bf 1.7\,K}  & {\bf 168\,K} & 1.8\,K  & 168\,K & $10^{15}$ & $10^{19}$  \\
%\muninb & 9.5\,K & 588\,K  & 5\,K  & 392\,K & {\bf 3\,K}  & {\bf 112\,K} & $10^{16}$ & $10^{19}$ \\
\mildew & 15\,K & 170\,K & {\bf 10\,K} & {\bf 170\,K} & 57\,K & 1\,M & 966\,K & 19\,M \\ 
\bnpathfinder\!\!\!\! &  570~~\, & 16\,K & {\bf 568~~\,} & {\bf 16\,K} &  643~~\, & 16\,K & $>1$\,T & $>1$\,T  \\ 
\munins & 749\,K & 59\,M & 375\,K & 39\,M & {\bf 367\,K} & {\bf 39\,M} & $>1$\,T & $>1$\,T   \\ 
\andes  & 2\,K & 131\,K & {\bf 1.4\,K}  & {\bf 66\,K} & {\bf 1.4\,K}  & {\bf  66\,K} & $>1$\,T & $>1$\,T   \\
\diabetes & 9\,K & 	194\,K &  {\bf 4\,K}  & {\bf 194\,K} & 325\,K & 33\,M & $>1$\,T & $>1$\,T  \\
\link & 109\,K & 17\,M  & {\bf 31\,K}  & {\bf 4\,M} &  633\,K & 268\,M & $>1$\,T & $>1$\,T  \\
\muninm & 40\,K & 31\,M  & {\bf 1.7\,K}  & {\bf 168\,K} & 1.8\,K  & 168\,K & $>1$\,T & $>1$\,T  \\
\muninb & 9.5\,K & 588\,K  & 5\,K  & 392\,K & {\bf 3\,K}  & {\bf 112\,K} & $>1$\,T & $>1$\,T \\
\revision{\tpchsmall} & \revision{144}~~\, & \revision{48}~~\, & \revision{144}~~\, & \revision{48}~~\, & \revision{144}~~\, & \revision{48}~~\, & \revision{{\bf 144}}~~\, & \revision{{\bf 39}}~~\,  \\ 
\revision{\tpchmedium} & \revision{400}~~\, & \revision{60}~~\, & \revision{400}~~\, & \revision{60}~~\, & \revision{400}~~\, & \revision{60}~~\, &  \revision{{\bf 280}}~~\, & \revision{{\bf 39}}~~\, \\ 
\revision{\tpchverylarge} & \revision{937\,K} & \revision{30\,K} & \revision{937\,K} & \revision{30\,K} & \revision{112\,K} & \revision{6\,K} & \revision{{\bf 4\,K}} & \revision{{\bf 300}} \\ 
\revision{\tpchlarge} & \revision{1\,K} & \revision{230} & \revision{1\,K}  & \revision{230}  & \revision{1\,K} & \revision{195} & \revision{{\bf 1\,K}} & \revision{{\bf 176}}   \\ 
\bottomrule
\end{tabular}
\end{center}
\end{table}

%\begin{table}[t]
%\setlength\tabcolsep{4pt}
%\fontsize{8}{9}\selectfont
%\begin{center}
%% \small
%\caption{\label{table:factorTableSizes}Parameter size of factors created
%with different elimination orders (K: thousand, M: million).}
%% \begin{tabular}{lcccccc}
%% \begin{tabular}{p{1.4cm}p{0.6cm}p{0.6cm}p{0.75cm}p{0.75cm}p{0.75cm}p{0.75cm}p{0.75cm}}
%\begin{tabular}{lrrrrrr}
%\toprule
% % & \multicolumn{6}{c}{Elimination order heuristic} \\ \cmidrule(lr{.99em}){2-7} 
% & \multicolumn{2}{c}{\mn} & \multicolumn{2}{c}{\mf} & \multicolumn{2}{c}{\wmf} \\
%\cmidrule(lr{0.7em}){2-3}
%\cmidrule(lr{0.7em}){4-5}
%\cmidrule(lr{0.7em}){6-7}
%Network & avg & max & avg & max  & avg & max  \\ \midrule
%\mildew & 15\,K & 170\,K & {\bf 10\,K} & {\bf 170\,K} & 57\,K & 1\,M \\ 
%\bnpathfinder\!\!\!\! &  570~~\, & 16\,K & {\bf 568~~\,} & {\bf 16\,K} &  643~~\, & 16\,K  \\ 
%\munins & 749\,K & 59\,M & 375\,K & 39\,M & {\bf 367\,K} & {\bf 39\,M} \\ 
%\andes  & 2\,K & 131\,K & {\bf 1.4\,K}  & {\bf 66\,K} & {\bf 1.4\,K}  & {\bf  66\,K} \\
%\diabetes & 9\,K & 	194\,K &  {\bf 4\,K}  & {\bf 194\,K} & 325\,K & 33\,M \\
%\link & 109\,K & 17\,M  & {\bf 31\,K}  & {\bf 4\,M} &  633\,K & 268\,M \\
%\muninm & 40\,K & 31\,M  & {\bf 1.7\,K}  & {\bf 168\,K} & 1.8\,K  & 168\,K  \\
%\muninb & 9.5\,K & 588\,K  & 5\,K  & 392\,K & {\bf 3\,K}  & {\bf 112\,K} \\
%\revision{\tpchsmall} & & & &  \\ 
%\revision{\tpchmedium} & & & & \\ 
%\revision{\tpchverylarge} & & & \\ 
%\revision{\tpchlarge} & & & & \\ 
%\bottomrule
%\end{tabular}
%\end{center}
%\end{table}

\begin{table}[t]
\setlength\tabcolsep{4pt}
\fontsize{8}{9}\selectfont
\centering
% \small
\caption{\label{table:treeStats}Statistics of elimination trees.}
\begin{tabular}{lrrr}
\toprule
Tree & nodes & height & max.\,\# \\ 
     &       &        & children \\ \midrule 
 % &  &  & children \\ \midrule
\mildew (\mf) &  70 & 17 & 3 \\
\bnpathfinder (\mf) & 218 & 12 & 54  \\ 
\munins (\wmf)  &  372  & 23   & 7  \\
\andes (\mf)  & 440  & 38  & 5 \\
\diabetes (\mf)  & 826  & 77 & 4  \\ 
\link (\mf)   & 1\,428  & 56  & 15 \\
\muninm (\mf)  & 2\,006  & 23  & 8  \\
\muninb (\wmf)  &  2\,082 &  24 & 8 \\ 
\revision{\tpchsmall (\mw)} & \revision{34} & \revision{8} &  \revision{3} \\ 
\revision{\tpchmedium (\mw)} & \revision{62} & \revision{11} & \revision{5} \\ 
\revision{\tpchverylarge (\mw)} & \revision{76} & \revision{13} & \revision{5}  \\ 
\revision{\tpchlarge (\mw)} & \revision{70} & \revision{11} & \revision{4} \\ 
\bottomrule
\end{tabular}
\end{table}

\spara{Query workload.} 
The problem we consider assumes a query work\-load, i.e., 
a probability distribution $\prob{\query}$ of queries~\query.
In practice, it is reasonable to consider a setting where one has access to a historical query log. 
%  to learn a distribution $\prob{\query}$.
In the absence of such a log for the networks of Table~\ref{table:dataset}, 
% we consider a general setting for $\prob{\query}$. In particular, for simplicity of presentation, 
we consider queries $\query = \probability(\qvarfree,\allowbreak\qvarbound=\qvarvals)$, 
where $\qvarbound = \emptyset$
% i.e., no variables are bound, 
and all variables are either free (\qvarfree) or summed-out (\qvarmiss). 
Note that the setting $\qvarbound = \emptyset$ is a worst-case scenario that leads to computationally intensive queries, since, by not selecting any subsets of rows associated with $\qvarbound=\qvarvals$
we need to generate larger factors. 
We consider two workload schemes. 
In the first scheme we consider \emph{uniform} workloads, 
where each variable has equal probability to be a member of \qvarfree. 
For each dataset, we generate a total of $250$ random queries, 
with $50$ queries for each query size \qsize, i.e., $\qsize = |\qvarfree| \in [1,5]$. 
% cigdem: for the skewed workload scheme, should we add that we experimented with it "in order to simulate a setting in which an elimination order that is mindful of the benefit to be gained from materialization under a given query workload" ?
In the second scheme we consider \emph{skewed} workloads, 
where the variables appearing earlier in the elimination ordering
%  i.e., the variables associated with the nodes closer to the leaves of the elimination tree \tree, 
are more likely to appear among the summed-out variables \qvarmiss. 
%  rather than the free variables \qvarfree. 
Specifically,  
a variable that appears $\ell$ levels above another in \tree is $\ell$ times more likely 
to be placed among the free variables \qvarfree.
We do not show results for workload distributions with opposite skew: since we focus on materialized tables that involve only summed-out variables, it is easy to see that having free variables close to the leaves of the elimination tree \tree would usually render them not useful. 

\spara{Cost values.} To solve Problem~\ref{problem:qtm}, 
we must assign partial cost values \cost{\nodeu} to the nodes \nodeu of elimination trees.
Cost must represent the running time of computing the corresponding factor 
from its children in the elimination tree.
% (see Section~\ref{sec:setting}, Definition~\ref{definition:cost}).
Following the time-complexity analysis of Koller et al.~\cite{koller2009probabilistic} 
for tabular-factor representations, 
we estimate \cost{\nodeu} to be proportional to the cost of the corresponding natural join operation. 
We implement joins using the one-dimensional representation of factor tables described 
by Murphy~\cite{murphy2002fast}.
%  which makes use of a constant-time arithmetic mapping between the indices of the one-dimensional factor table and the combinations of values of variables in its multi-dimensional counterpart.
For this implementation, the cost of the natural join operation is twice the resulting size of join, 
which can be calculated from the sizes of the joined tables without actually performing the join.
We confirm empirically that the cost estimates align almost perfectly with the corresponding execution times 
(Pearson $\rho \ge 0.99$).

\spara{Algorithms.} We compare with other inference algorithms:

\spara{(i) Junction tree (\jt).} 
Following Lauritzen et al.~\cite{lauritzen1988local} 
the tree is calibrated by precomputing and materializing the joint probability distributions for each of its 
%clique- and separator-
nodes. As discussed in Section~\ref{section:related}, this supports efficient evaluation of all the queries in which the query variables belong to the same tree node. 

\spara{(ii) Indexed junction tree (\kanagal).} 
We implement the approach of Kanagal and Deshpande~\cite{kanagal2009indexing}, 
which makes use of a hierarchical index built on the calibrated junction tree. 
As discussed in Section~\ref{section:related}, this index contains additional (materialized) joint probability distributions for speeding up ``out-of-clique'' queries. 
%In comparison with the standard junction tree algorithm (\jt), the index contains additional (materialized) joint probability distributions, named shortcut potentials, that can be used to speed up ``out-of-clique'' queries.
%
Index construction requires to specify a parameter denoting the maximum possible size for a potential to be materialized.  
% a parameter  that provides an upper bound on the size of each node of the index to determine the space used for materialization. %, reported in Table~\ref{table:offlineStats} .
%
%In this experimental evaluation, w
%We specify this parameter in terms of total number of potential entries (i.e., rows in the probability tables) stored in correspondence of each node. 
We set this parameter in terms of total number of entries in a potential by considering candidate values $\{250, 10^3, 10^5\}$ and selecting the one resulting in smallest median query processing cost for each dataset under uniform workload. 
%
%We consider a grid of candidate parameter values $\{250, 10^3, 10^5\}$ and for each dataset, we select the value leading to the smallest median query processing cost under the uniform query workload scheme.
We compute the cost of these algorithms as done for variable elimination following Koller et al.~\cite{koller2009probabilistic} and confirm that the cost estimates align perfectly with the execution times (Pearson $\rho \ge 0.98$).

\FullOnly{
\begin{figure*}[t]
\begin{center}
\begin{tabular}{cccc}
    \includegraphics[width=.20\textwidth]{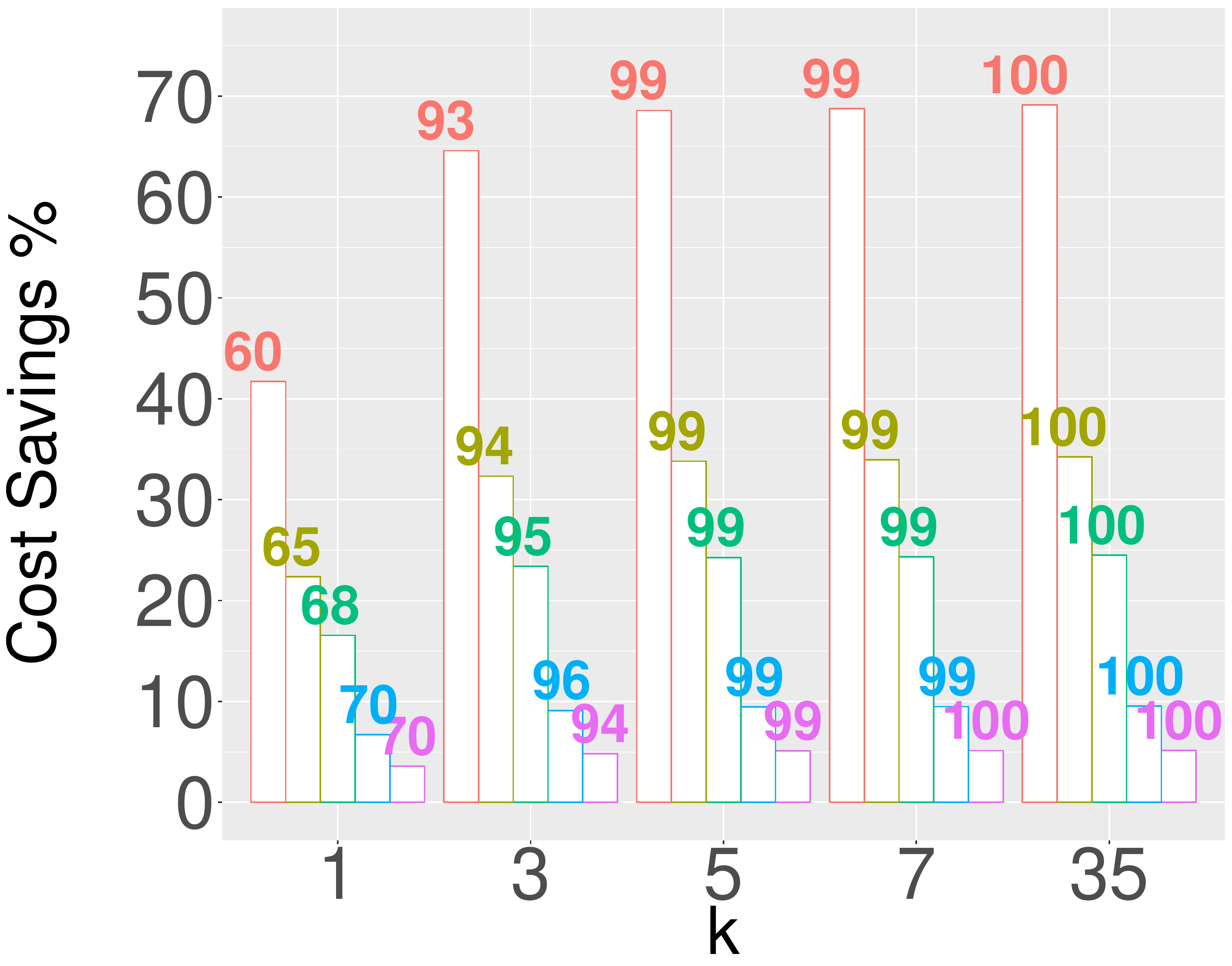}&
    \includegraphics[width=.20\textwidth]{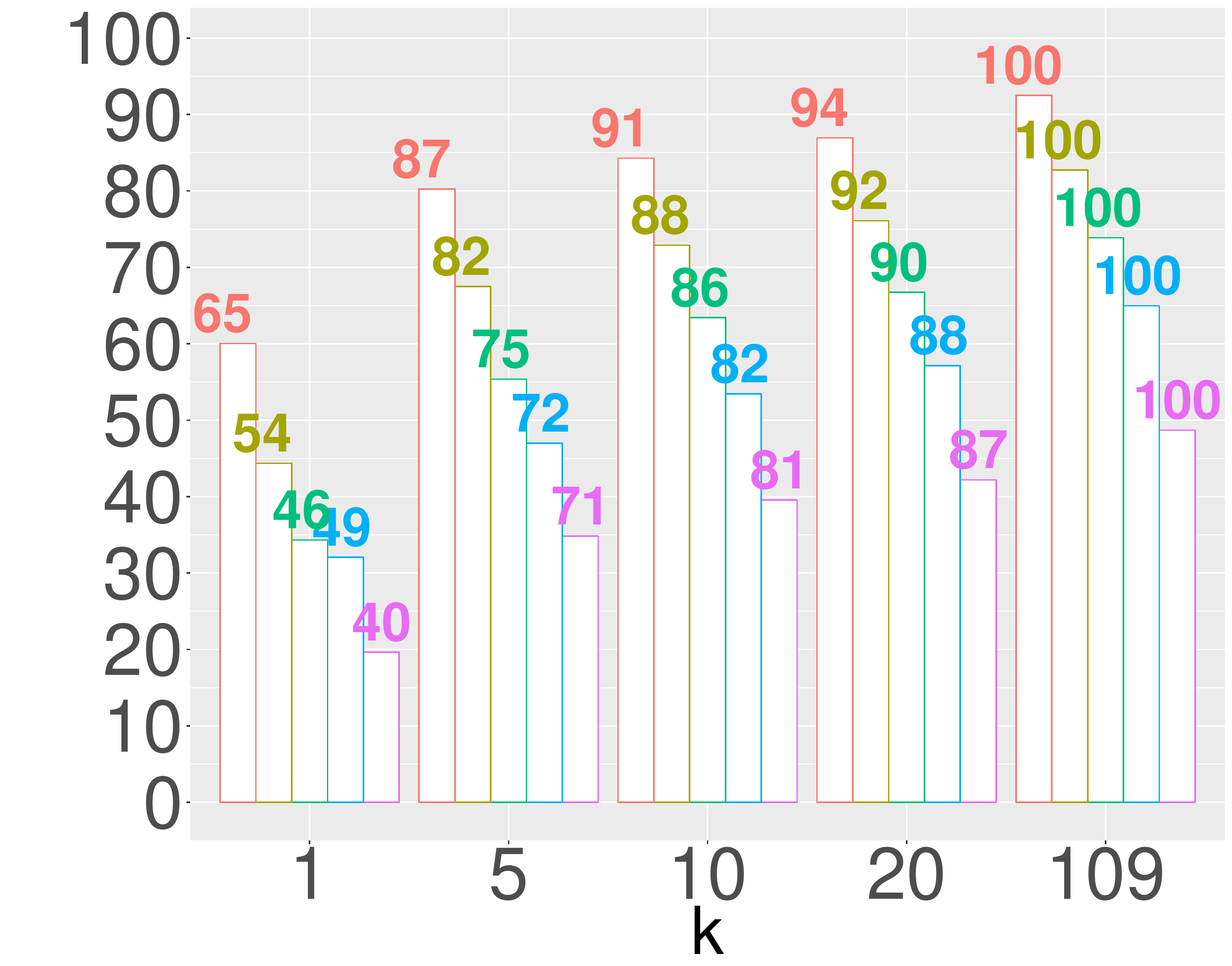}&
    \includegraphics[width=.20\textwidth]{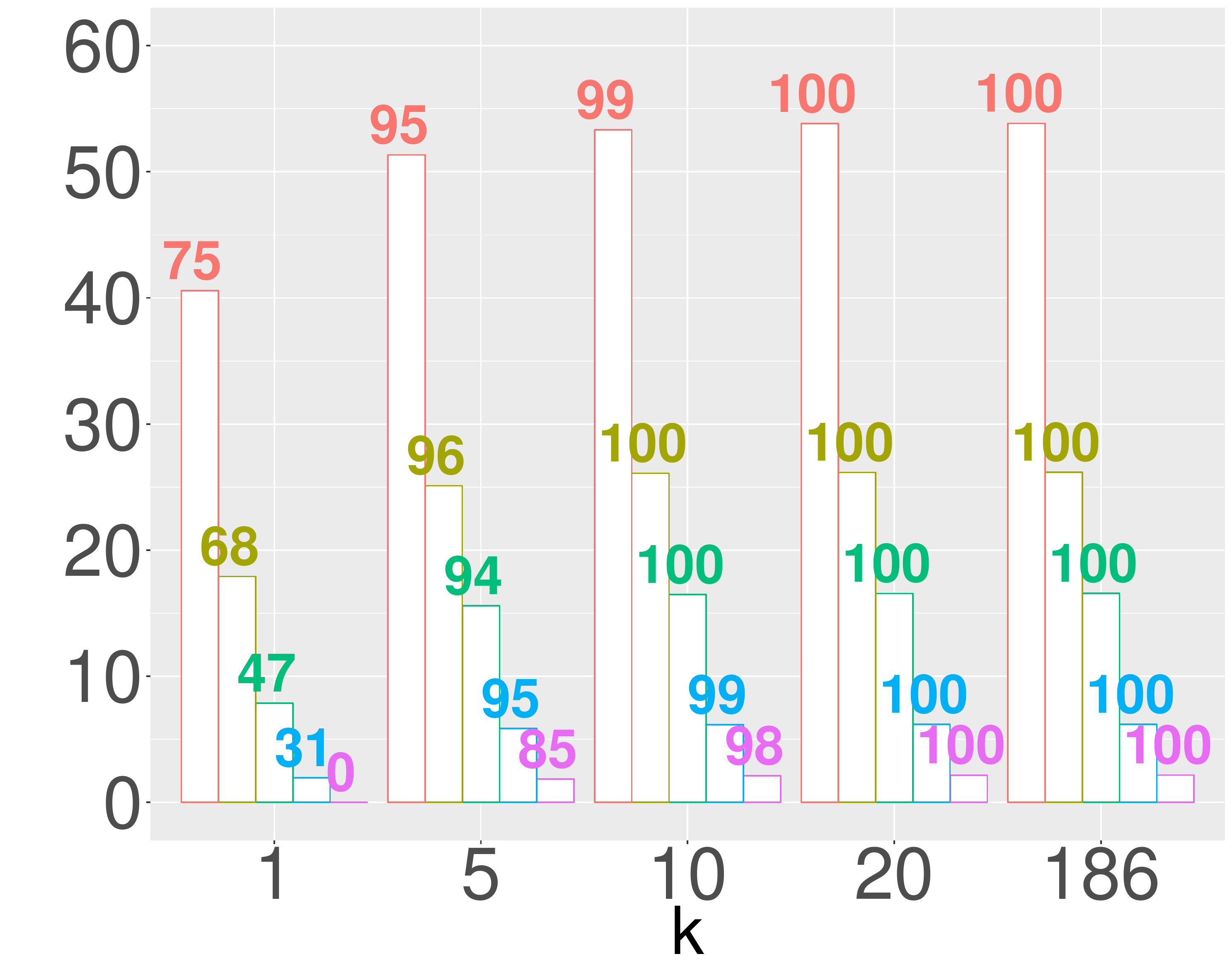}&
    \includegraphics[width=.20\textwidth]{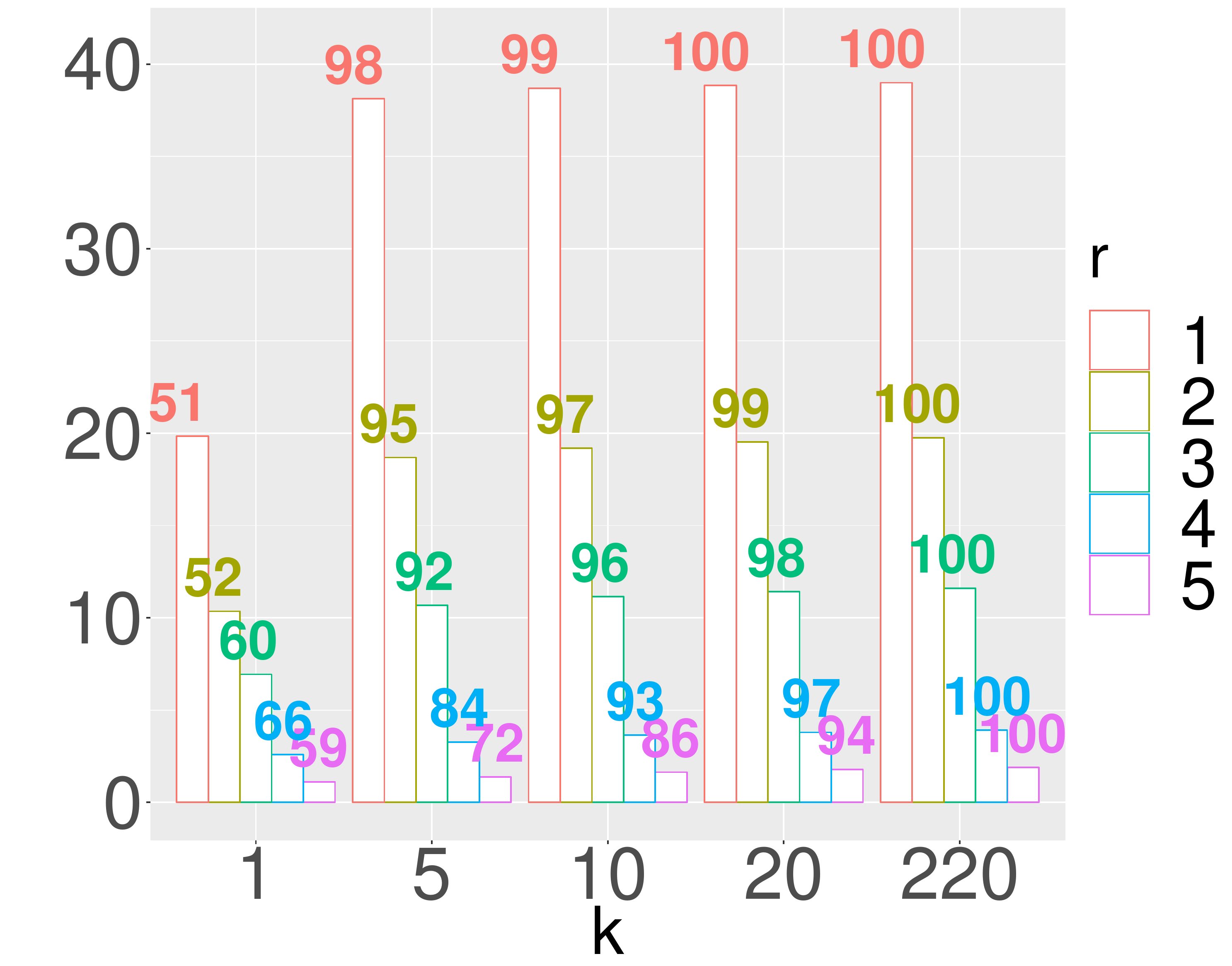}\\
	(a) \mildew (\mf) & (b) \bnpathfinder (\mf)  & (c) \munins (\wmf) & (d) \andes (\mf)   \\
	\includegraphics[width=.20\textwidth]{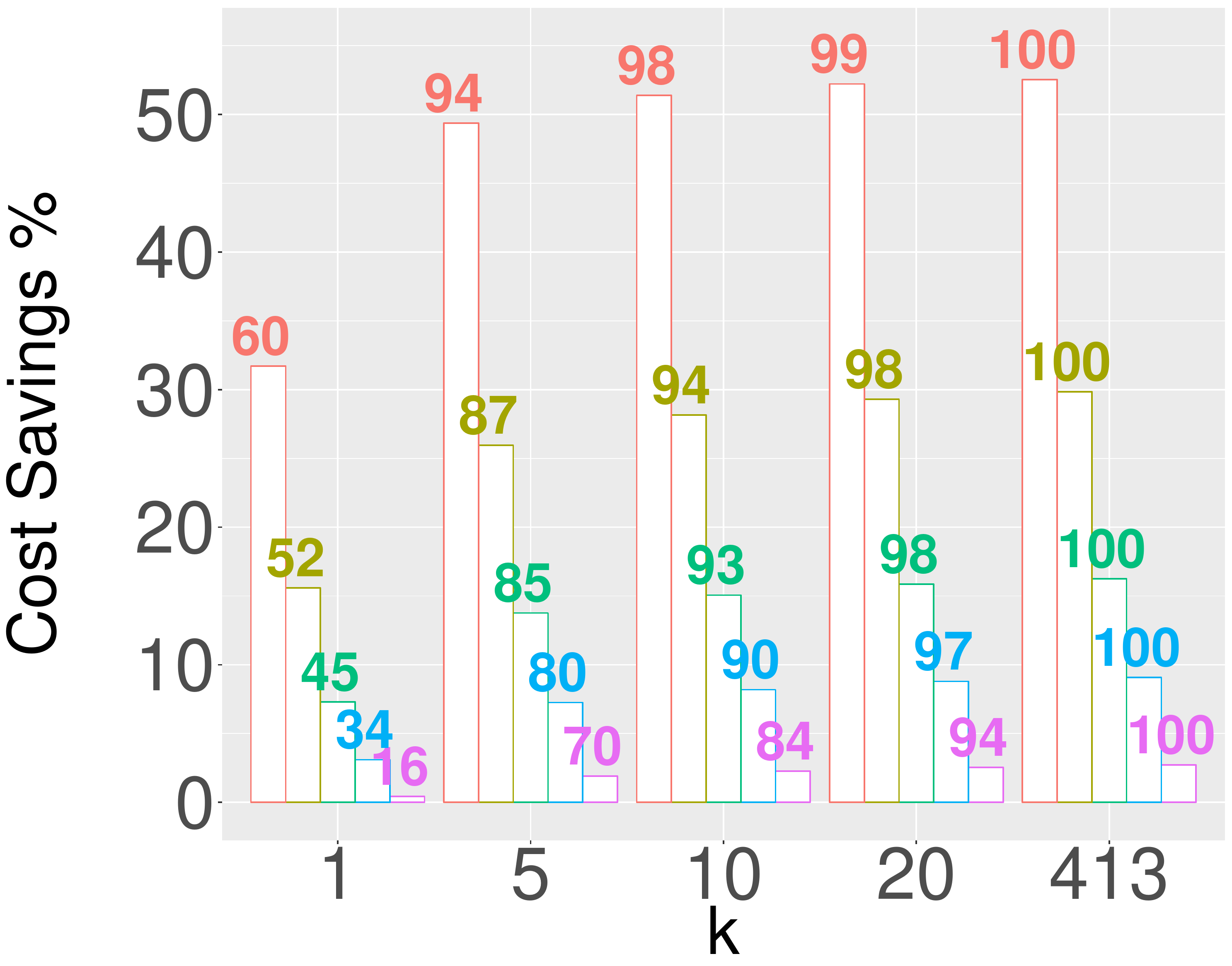} &
    \includegraphics[width=.20\textwidth]{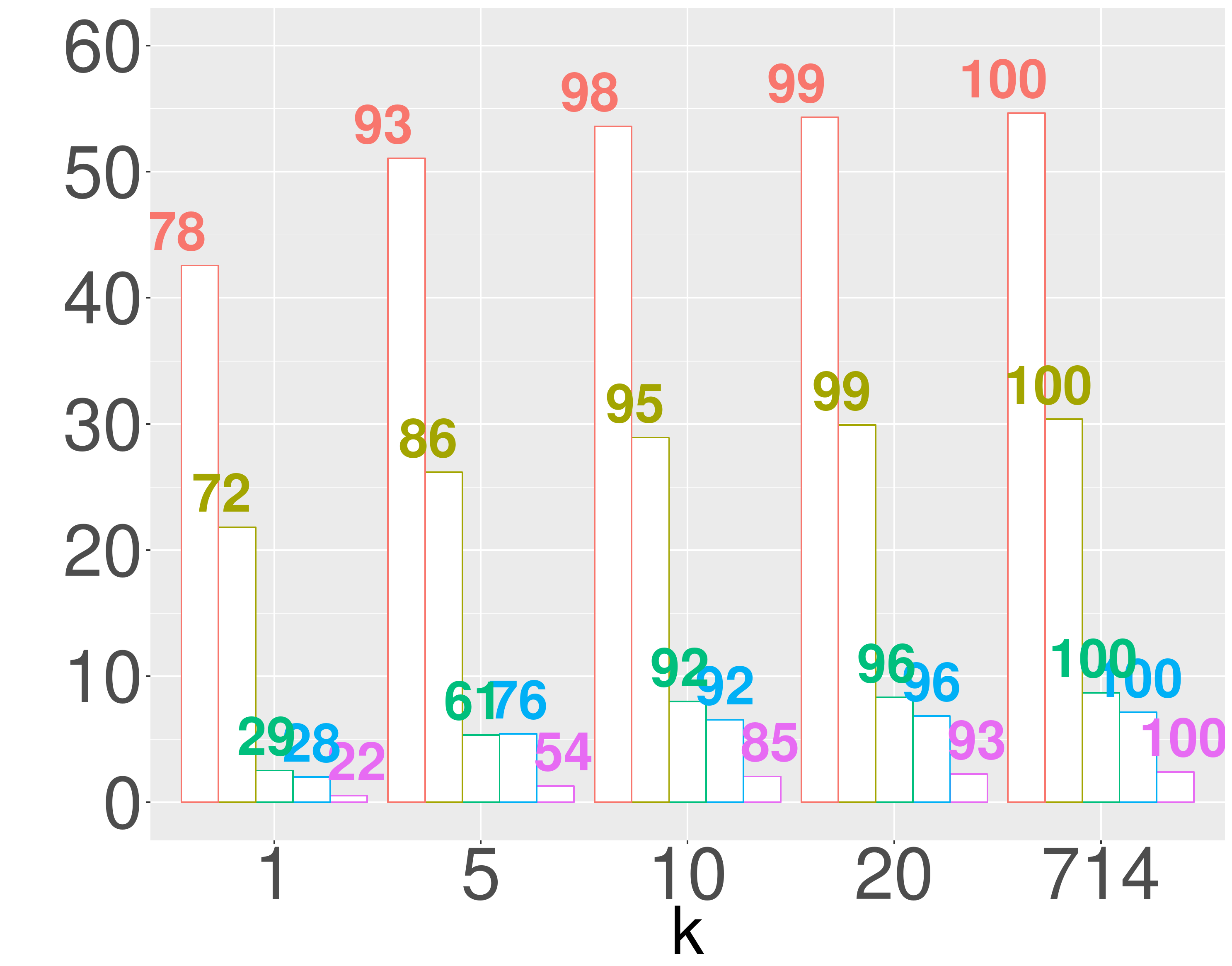} & 
    \includegraphics[width=.20\textwidth]{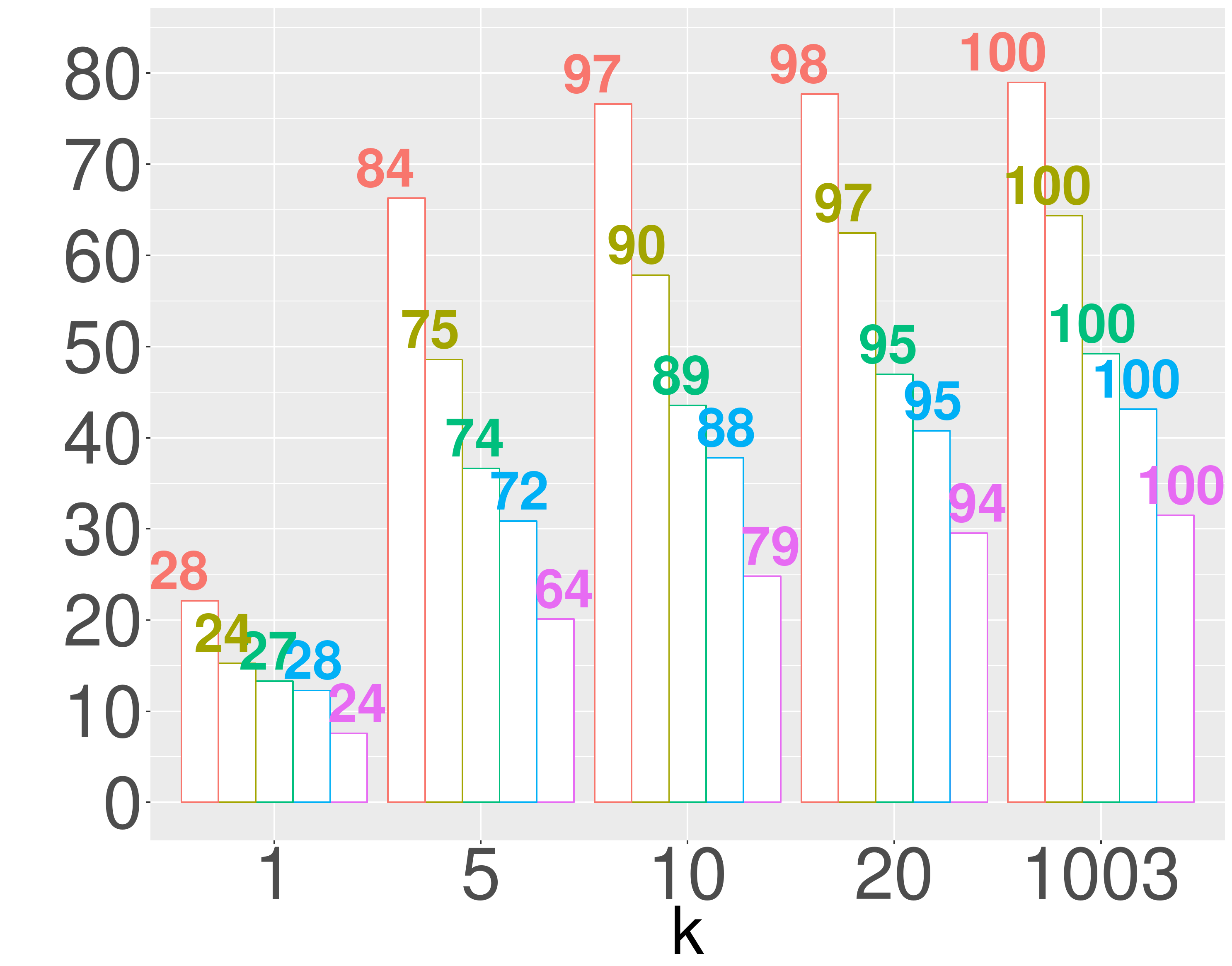} & 
    \includegraphics[width=.20\textwidth]{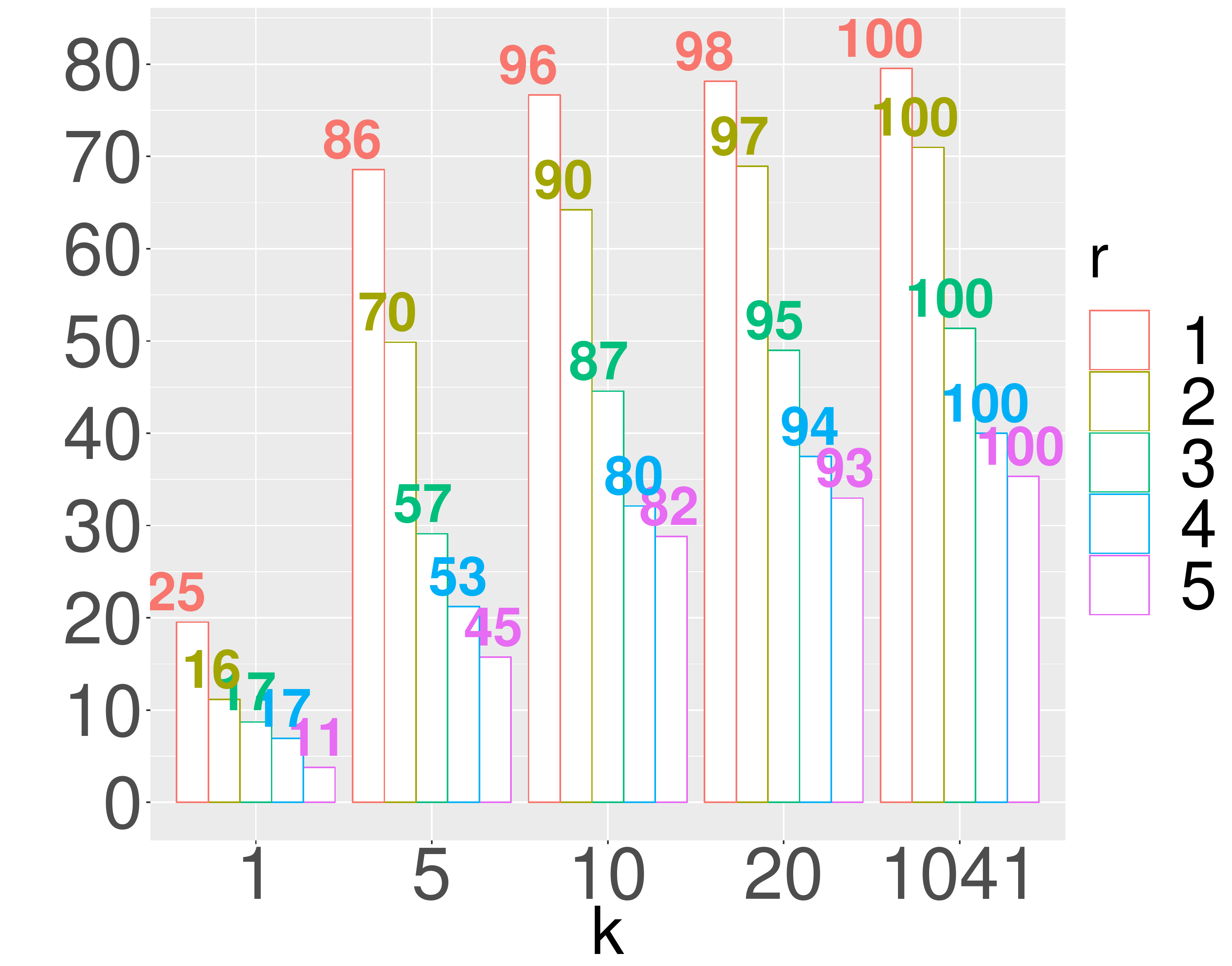}
    \\
   (e) \diabetes (\mf)  & (f) \link (\mf) & (g)  \muninm (\mf)   & (h) \muninb (\wmf)  \\ 
   \includegraphics[width=.20\textwidth]{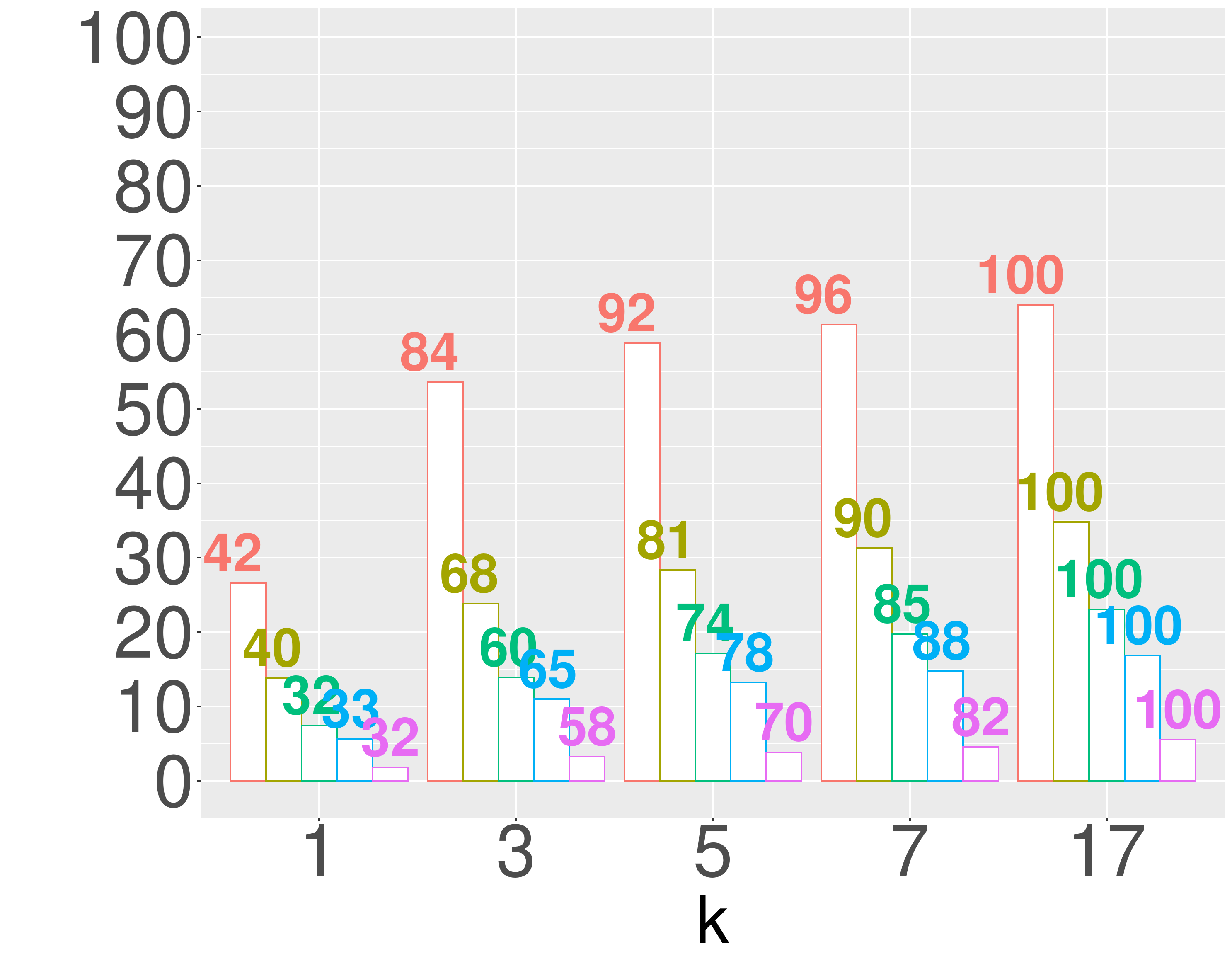} &
    \includegraphics[width=.20\textwidth]{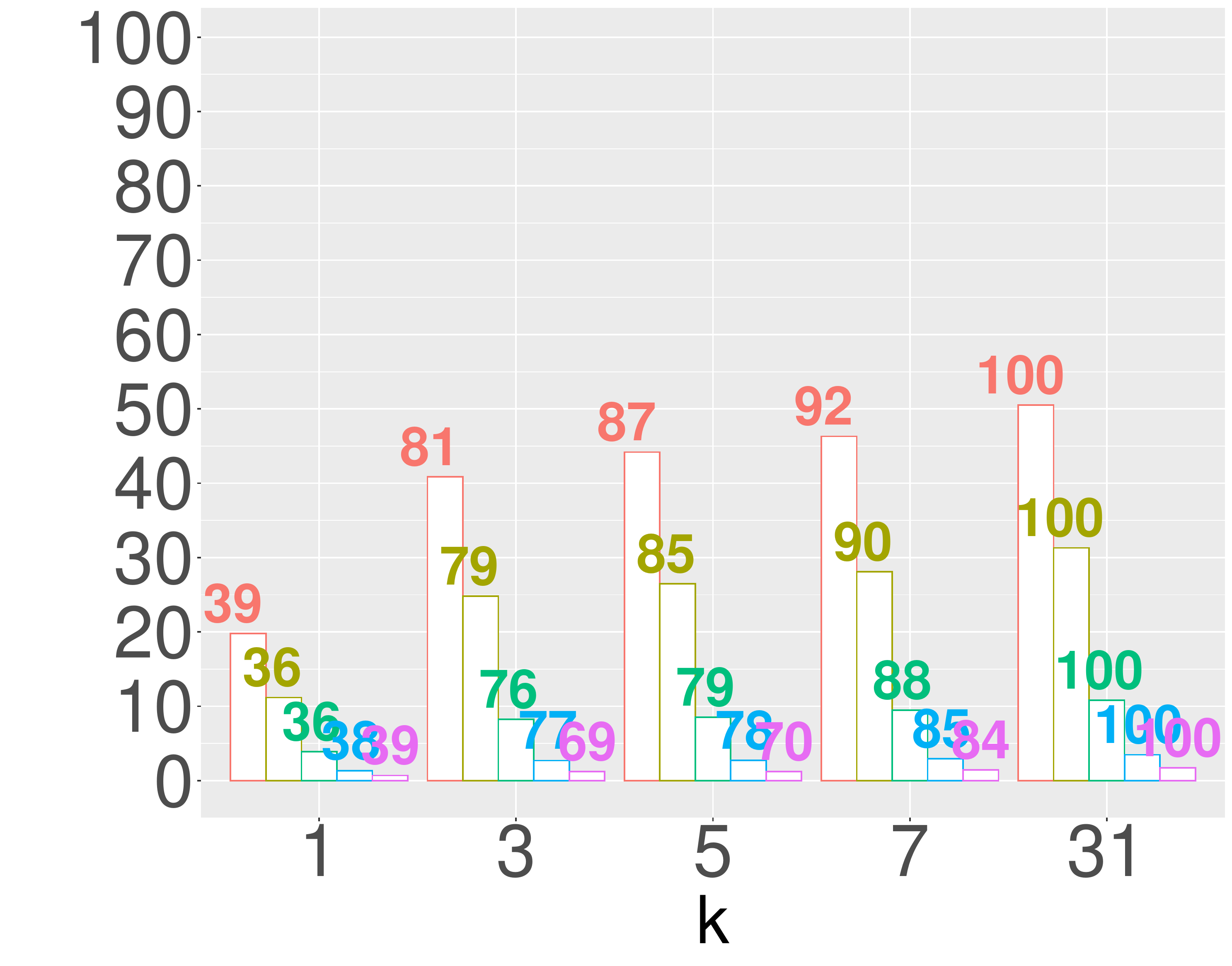} & 
    \includegraphics[width=.20\textwidth]{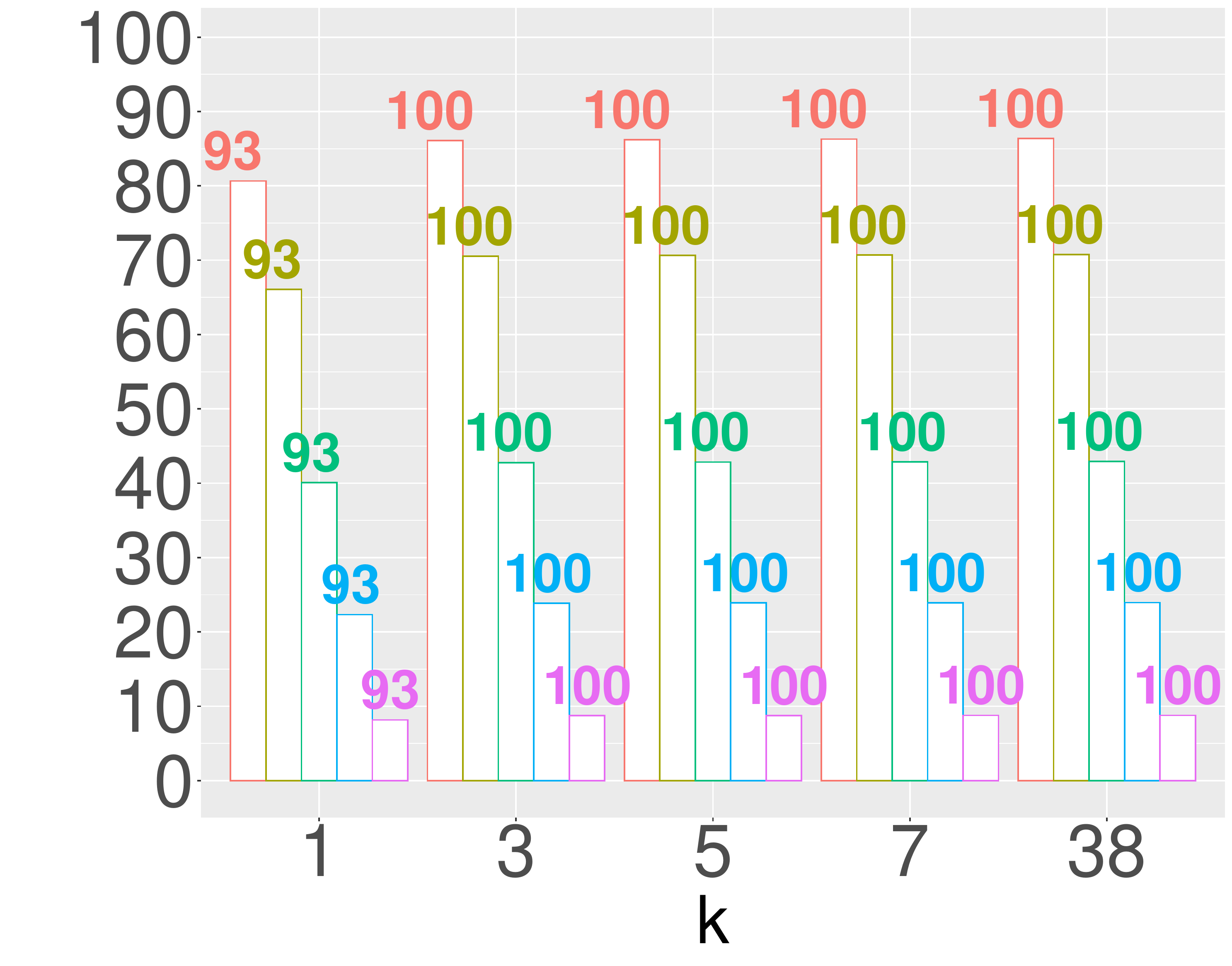} & 
    \includegraphics[width=.20\textwidth]{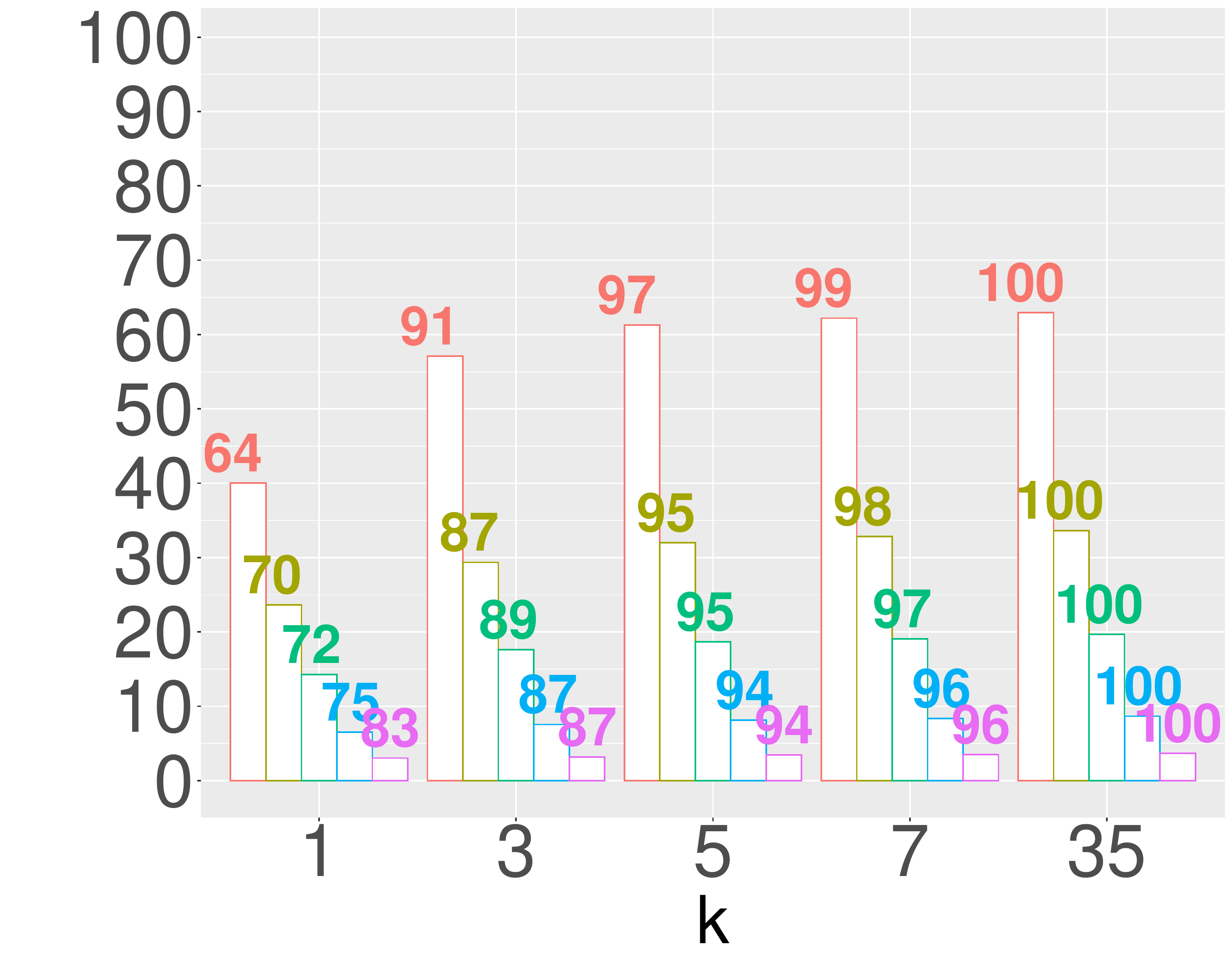}
    \\
   \revision{(i) \tpchsmall (\mw)}  & \revision{(j) \tpchmedium (\mw)} & \revision{(k) \tpchverylarge (\mw)}   & \revision{(l) \tpchlarge (\mw)}  \\ 
\end{tabular}
\caption{\label{fig:unif_per_r}Cost savings per query size $\qsize$ in uniform-workload scheme.
$x$-axis: number of materialized factors (budget \budget). $y$-axis: cost savings in query running time compared to no materialization. Numbers on the bars: the percentage of cost savings relative to the materialization of all factors.}
\end{center}
\end{figure*}
}

\FullOnly{
	\begin{figure*}[t]
		\begin{center}
			\begin{tabular}{cccc}
				\includegraphics[width=.20\textwidth]{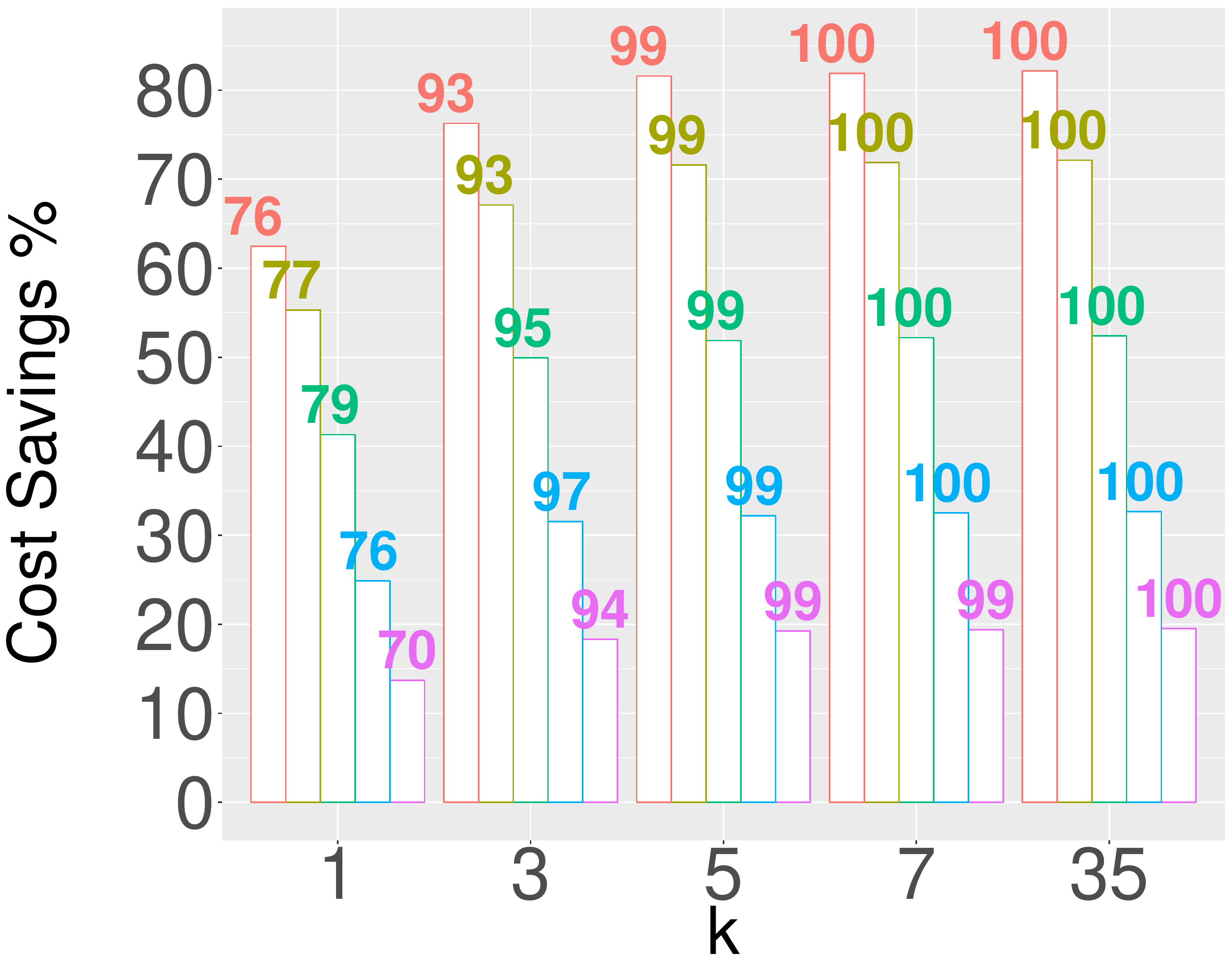}&
				\includegraphics[width=.20\textwidth]{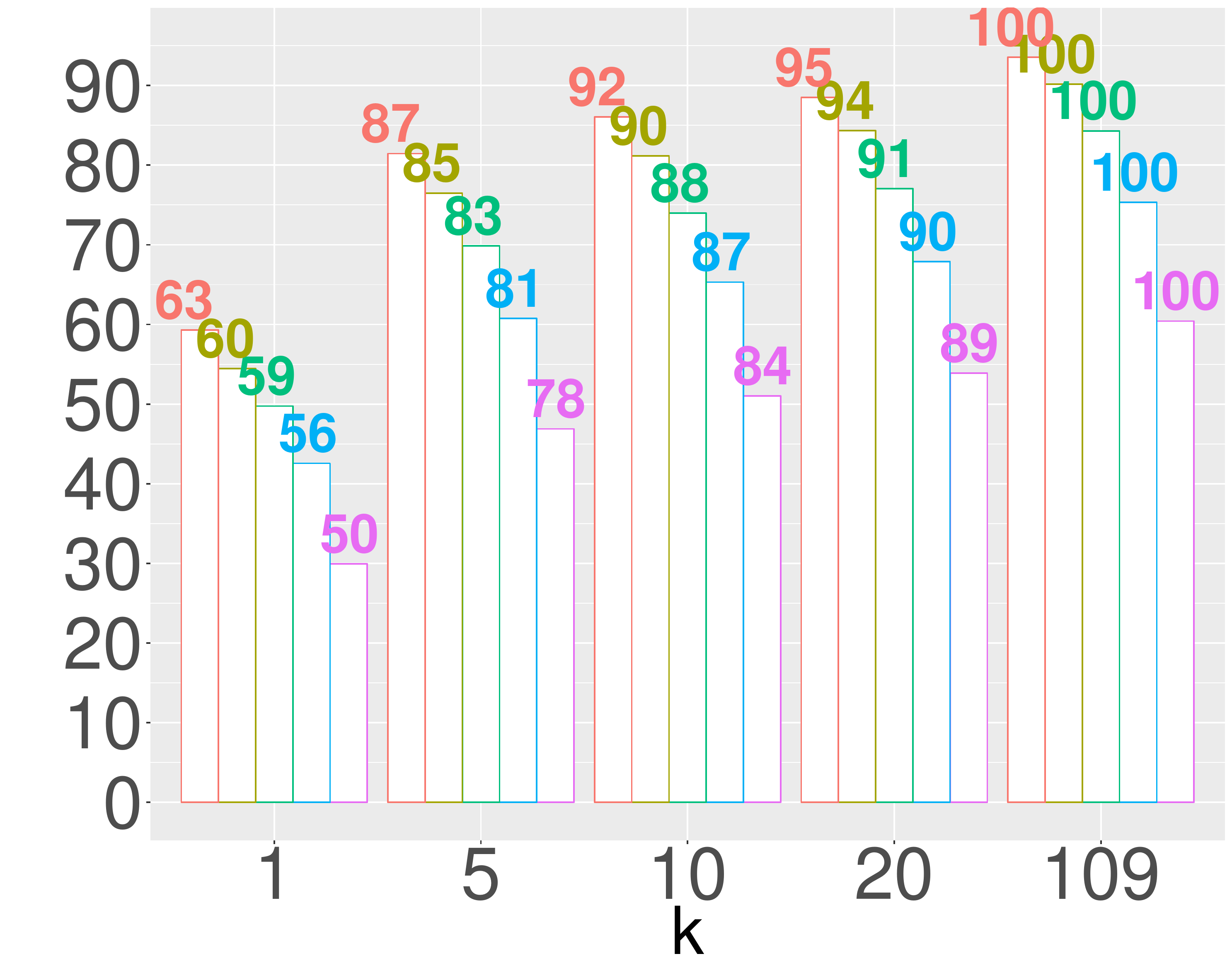}&
				\includegraphics[width=.20\textwidth]{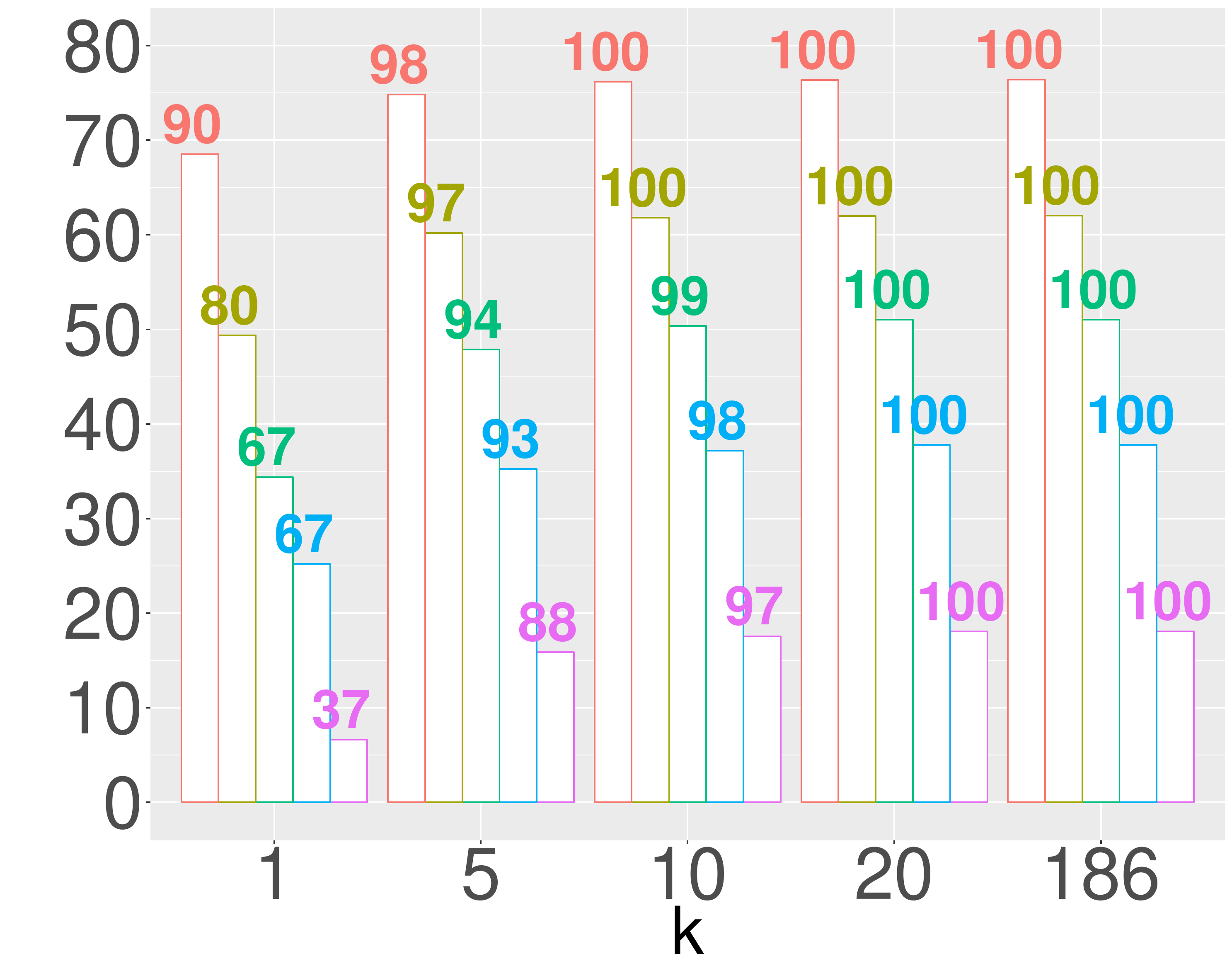}&
				\includegraphics[width=.20\textwidth]{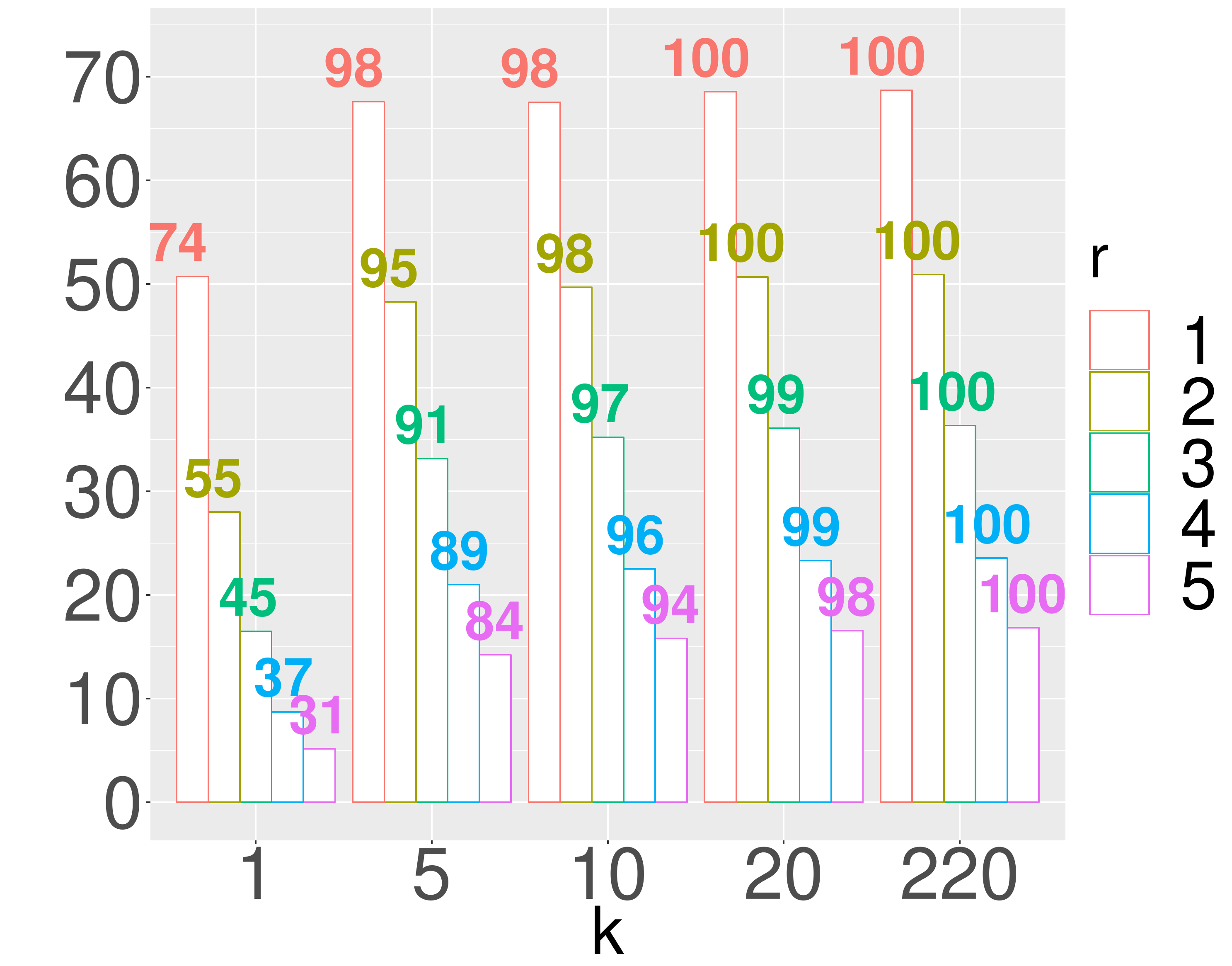}\\
				(a) \mildew (\mf) & (b) \bnpathfinder (\mf)  & (c) \munins (\wmf) & (d) \andes (\mf)   \\
				\includegraphics[width=.20\textwidth]{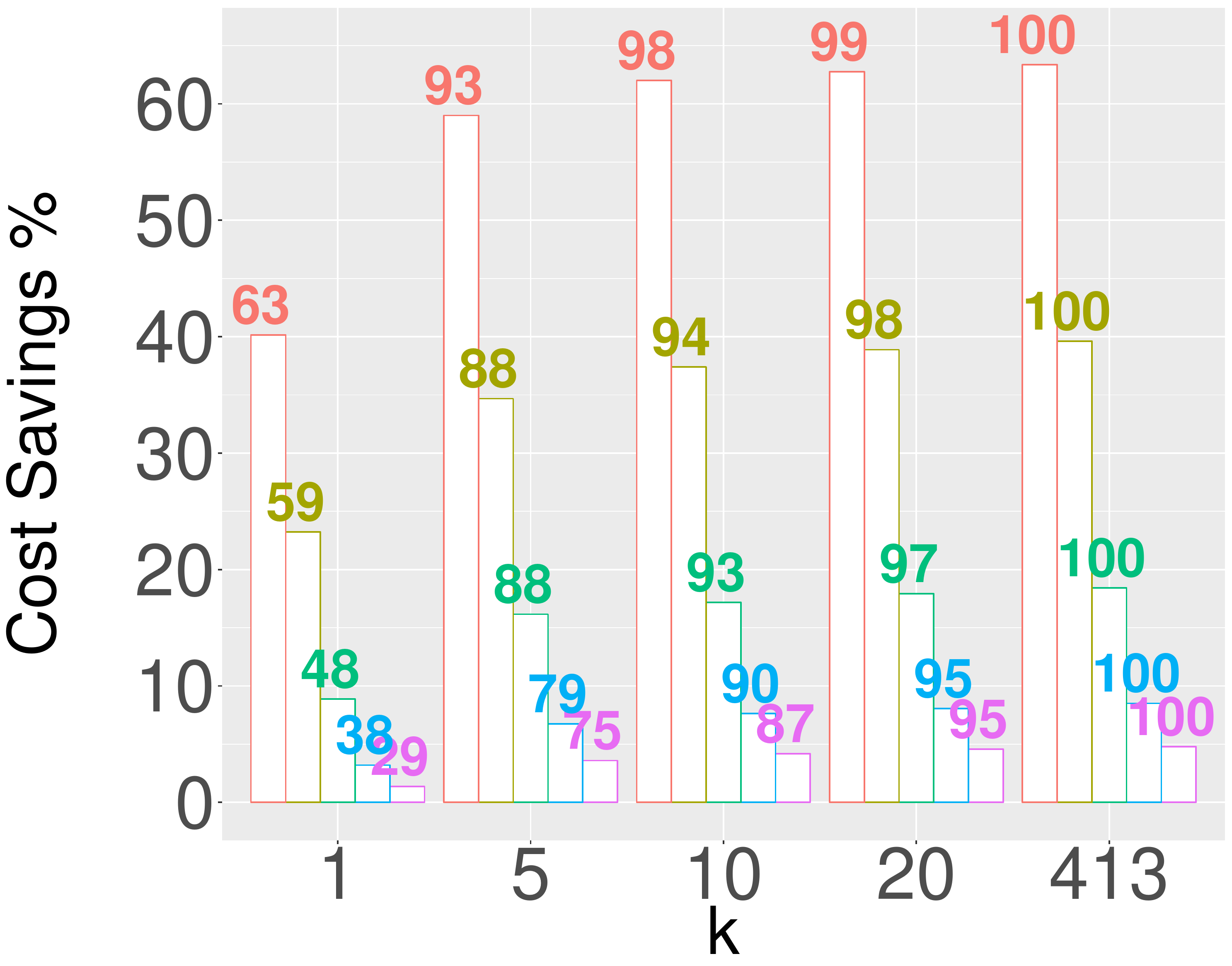} &
				\includegraphics[width=.20\textwidth]{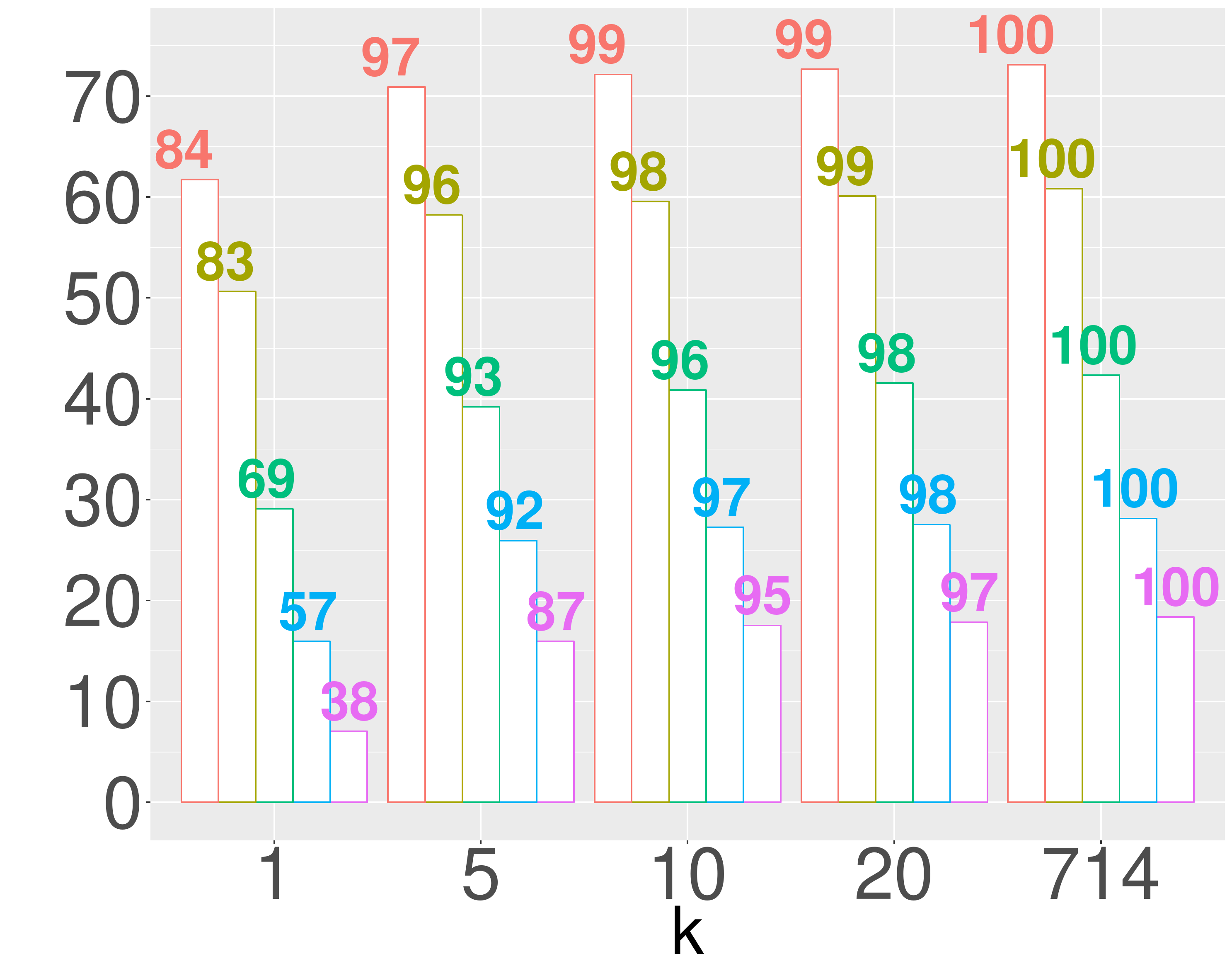} & 
				\includegraphics[width=.20\textwidth]{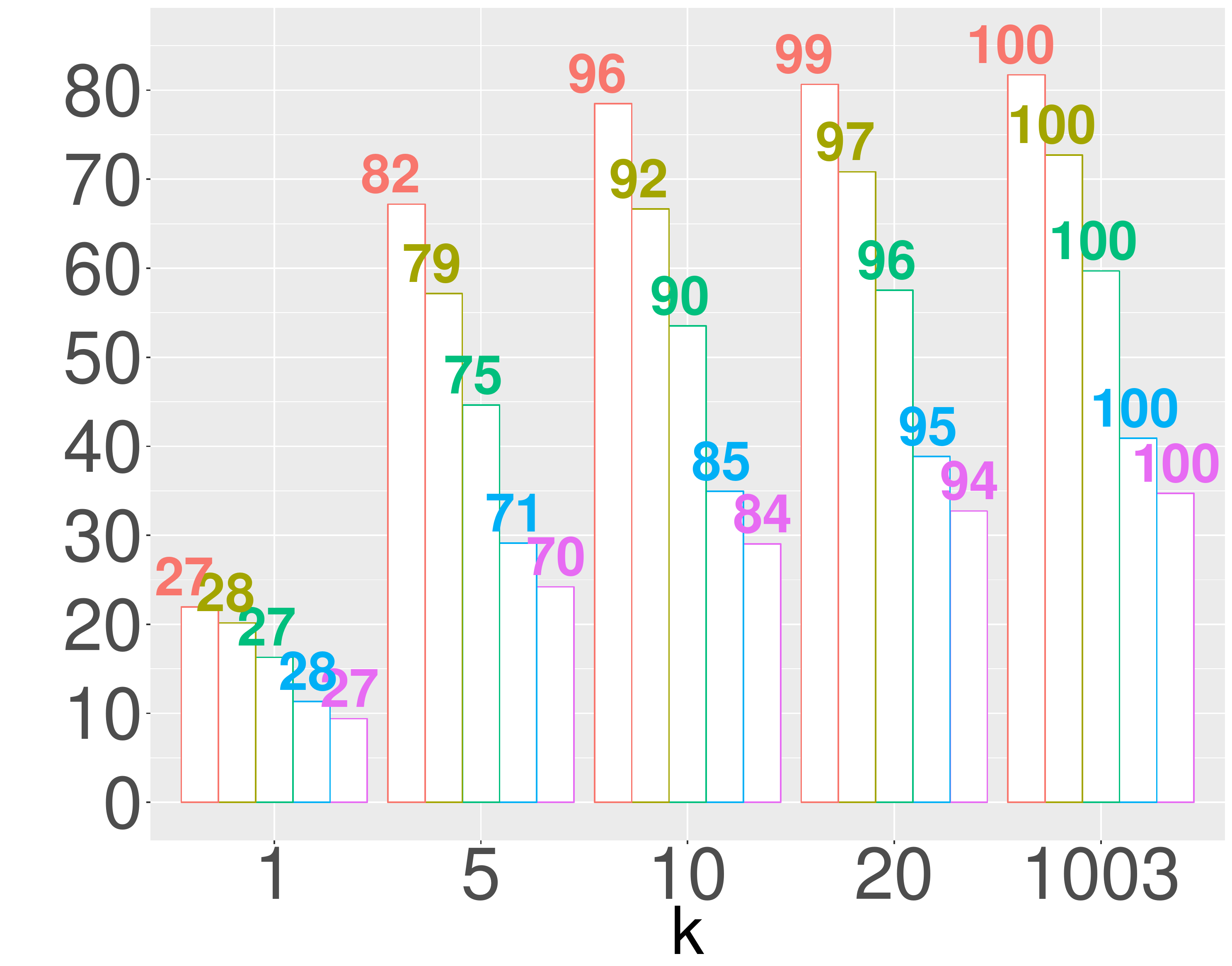} & 
				\includegraphics[width=.20\textwidth]{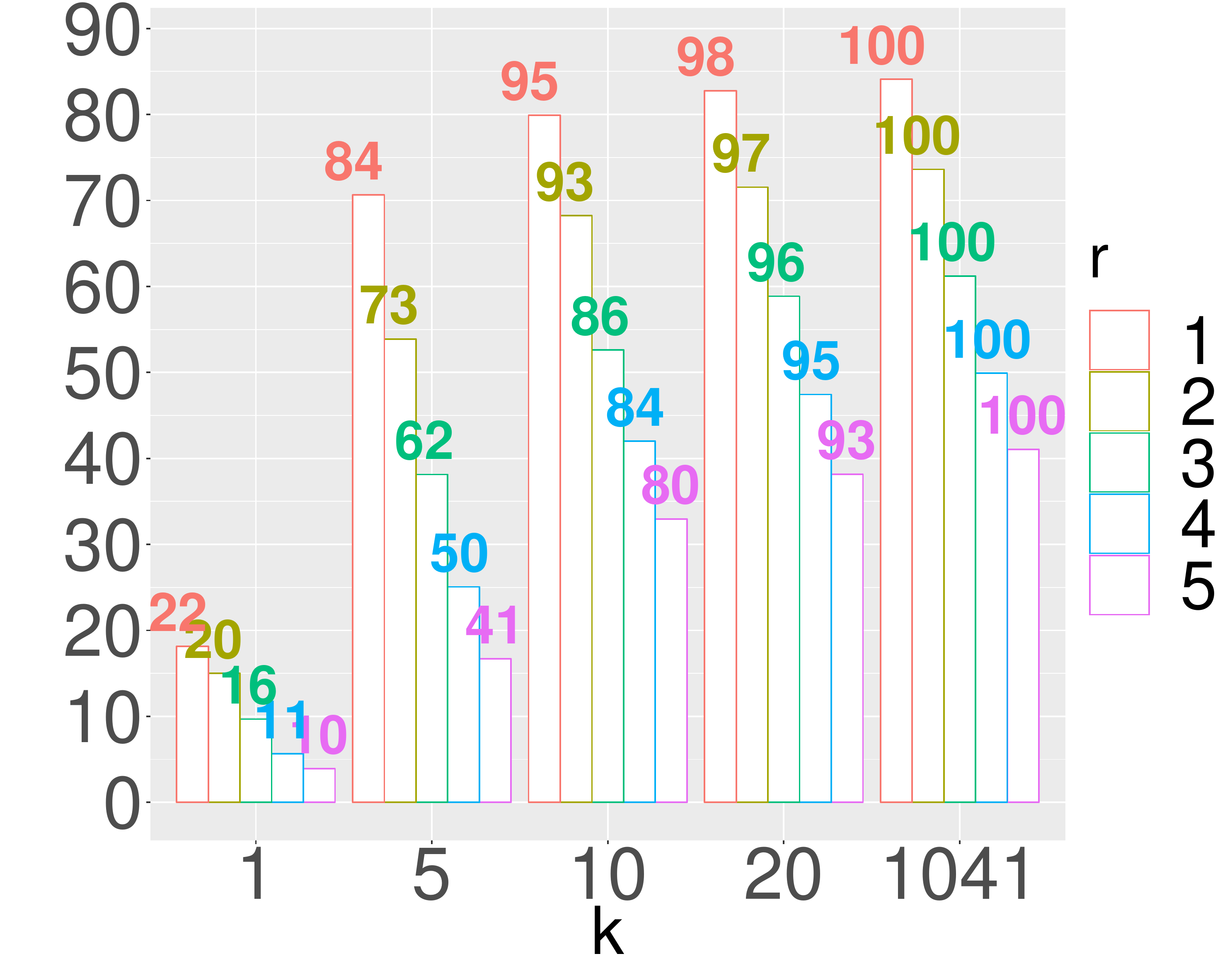}
				\\
				(e) \diabetes (\mf)  & (f) \link (\mf) & (g)  \muninm (\mf)   & (h) \muninb (\wmf)  \\ 
				\includegraphics[width=.20\textwidth]{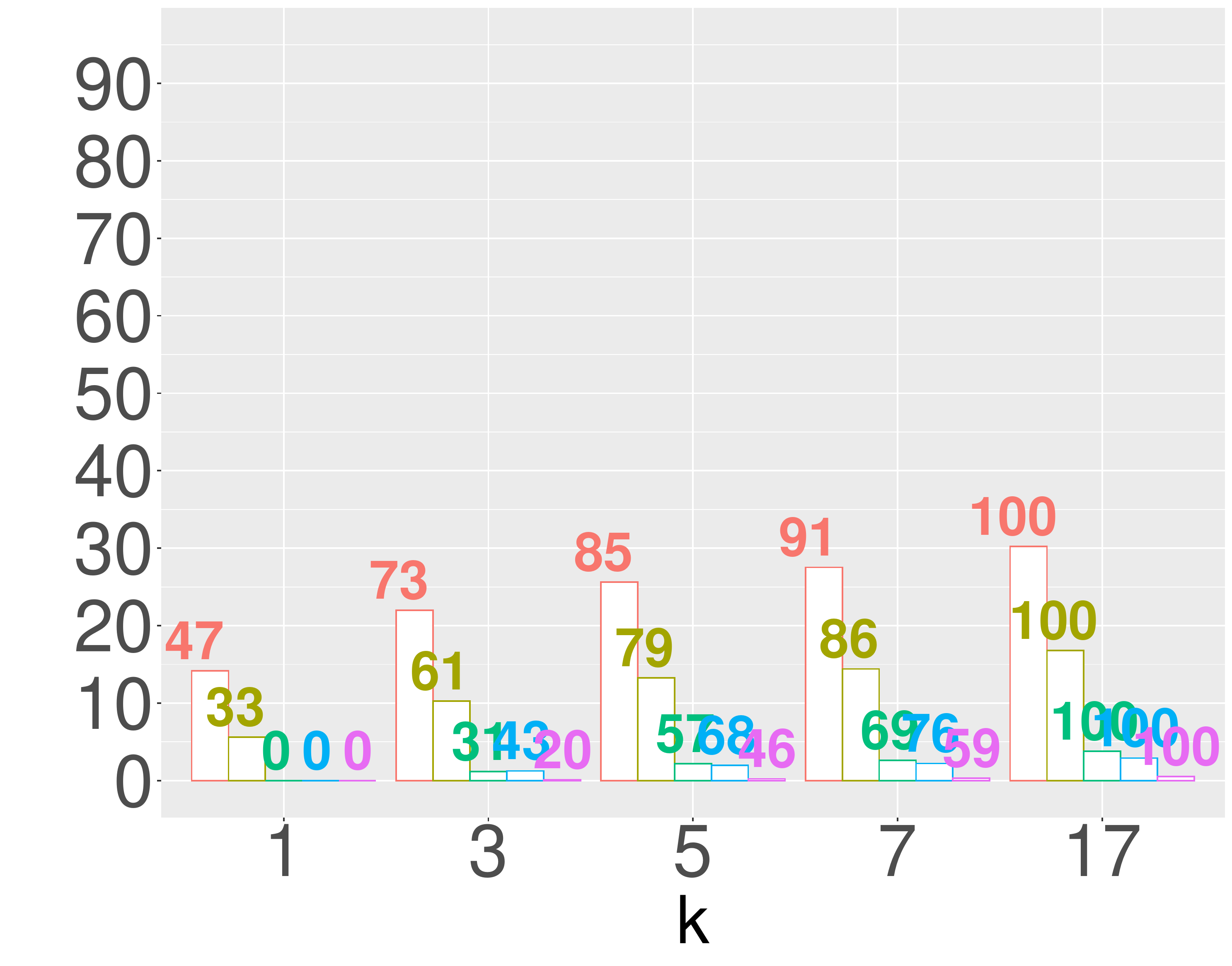} &
				\includegraphics[width=.20\textwidth]{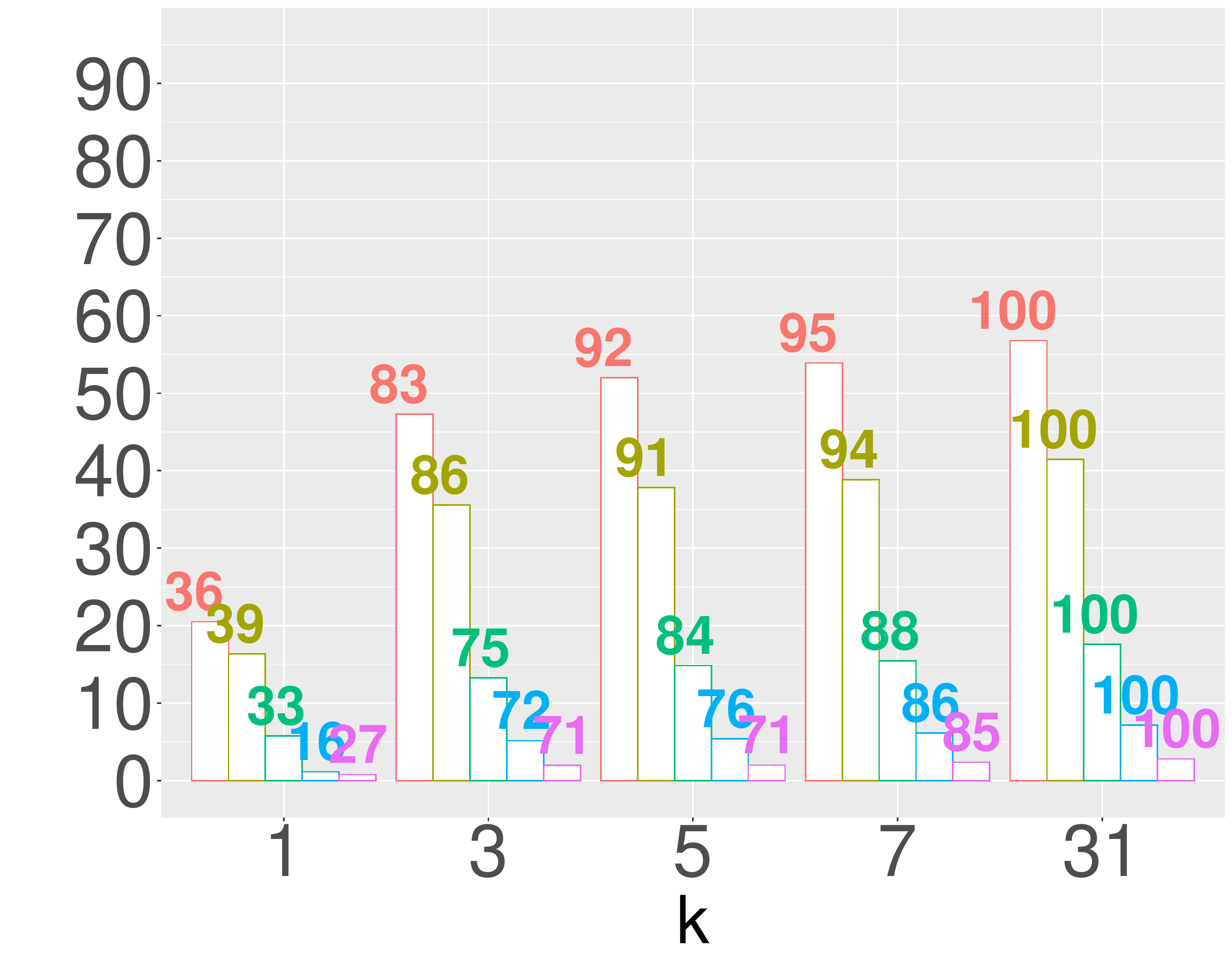} & 
				\includegraphics[width=.20\textwidth]{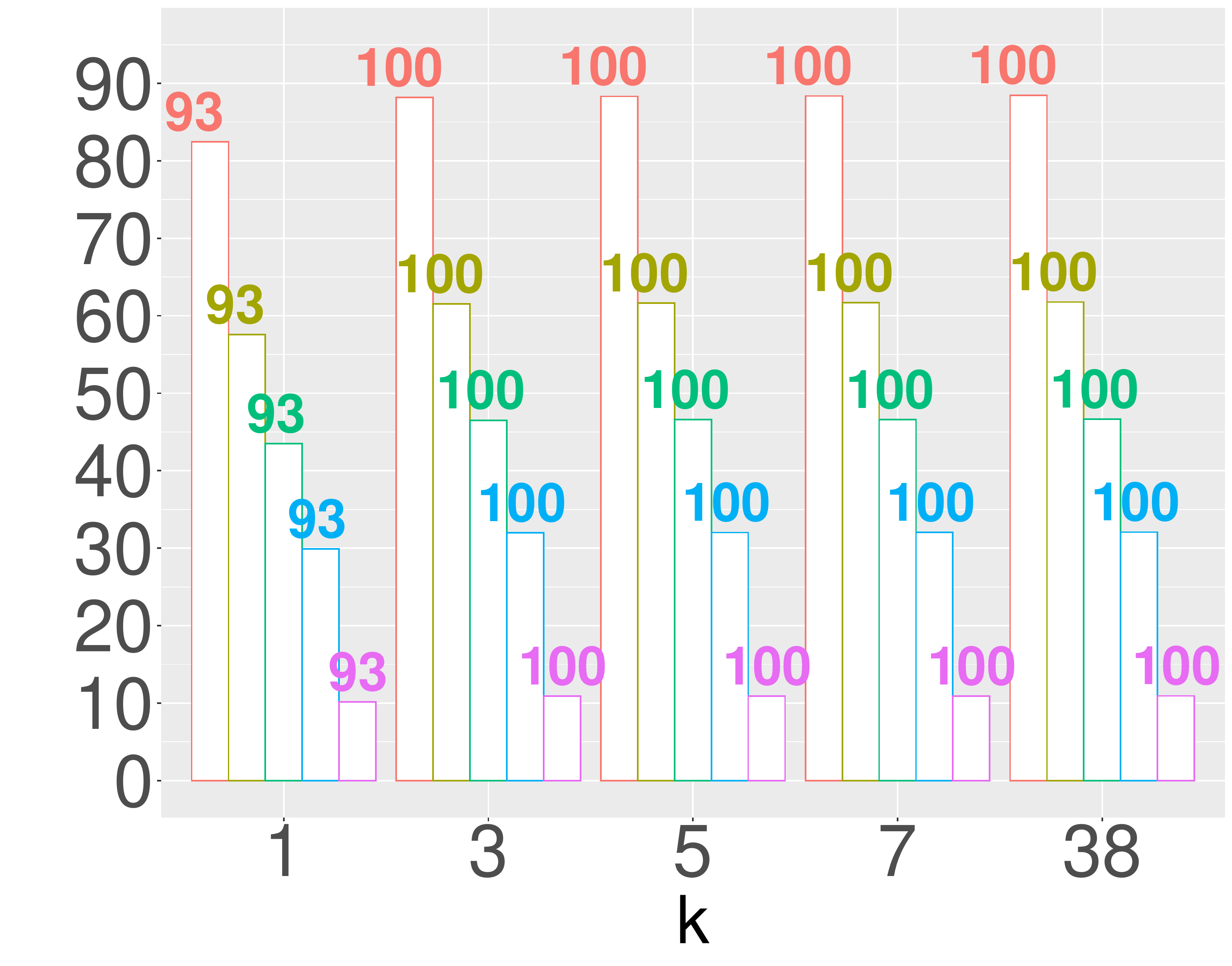} & 
				\includegraphics[width=.20\textwidth]{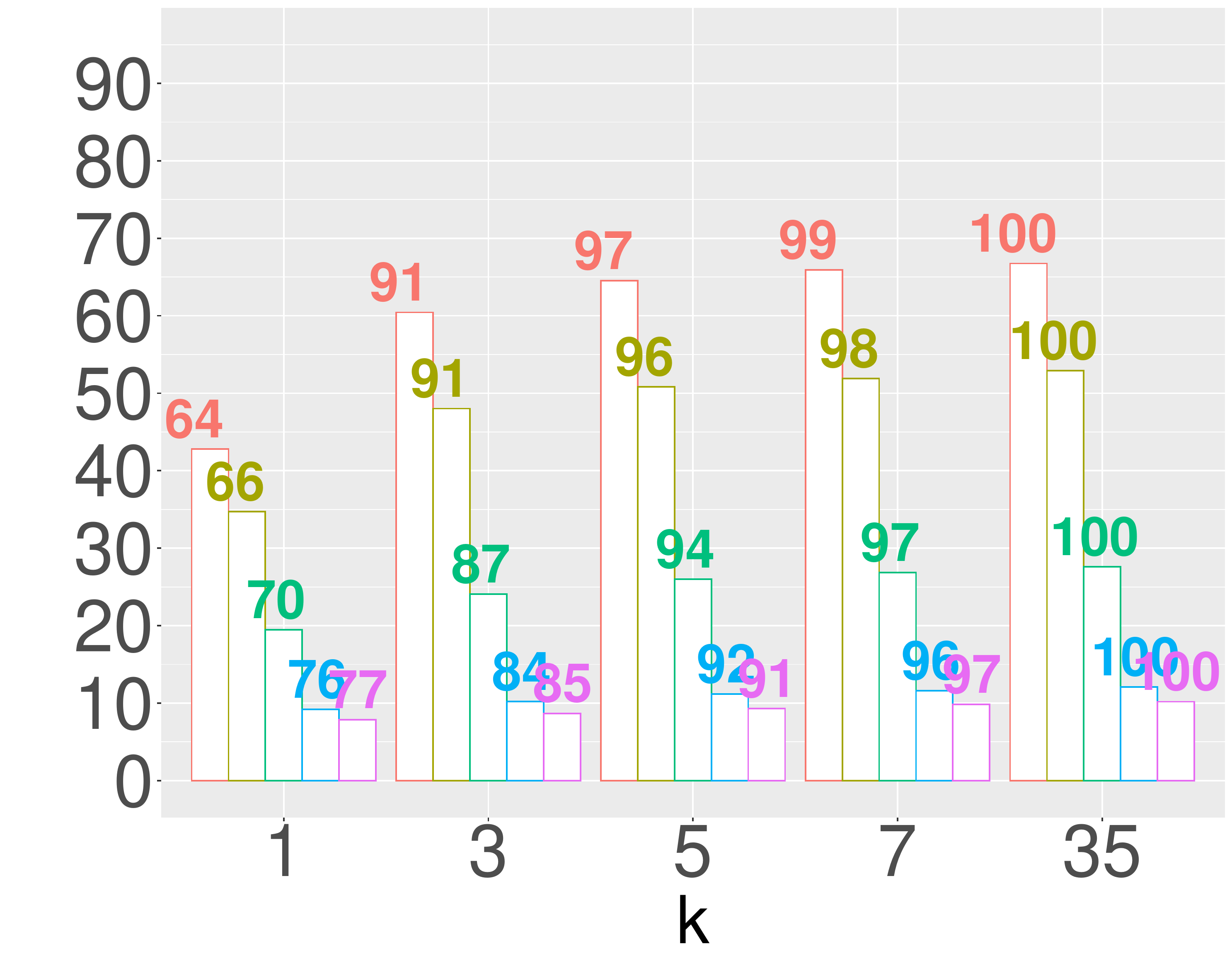}
				\\
				\revision{(i) \tpchsmall (\mw)}  & \revision{(j) \tpchmedium (\mw)} & \revision{(k) \tpchverylarge (\mw)}   & \revision{(l) \tpchlarge (\mw)}  \\ 
			\end{tabular}
			\caption{\label{fig:biasr_per_r}Cost savings per query size $\qsize$ in skewed-workload scheme.
				$x$-axis: number of materialized factors (budget \budget). $y$-axis: cost savings in query running time compared to no materialization. Numbers on the bars: the percentage of cost savings relative to the materialization of all factors.}
		\end{center}
	\end{figure*}
}

\ReviewOnly{
	\begin{figure*}[t]
		\begin{center}
			\begin{tabular}{cccc}
				\includegraphics[width=.20\textwidth]{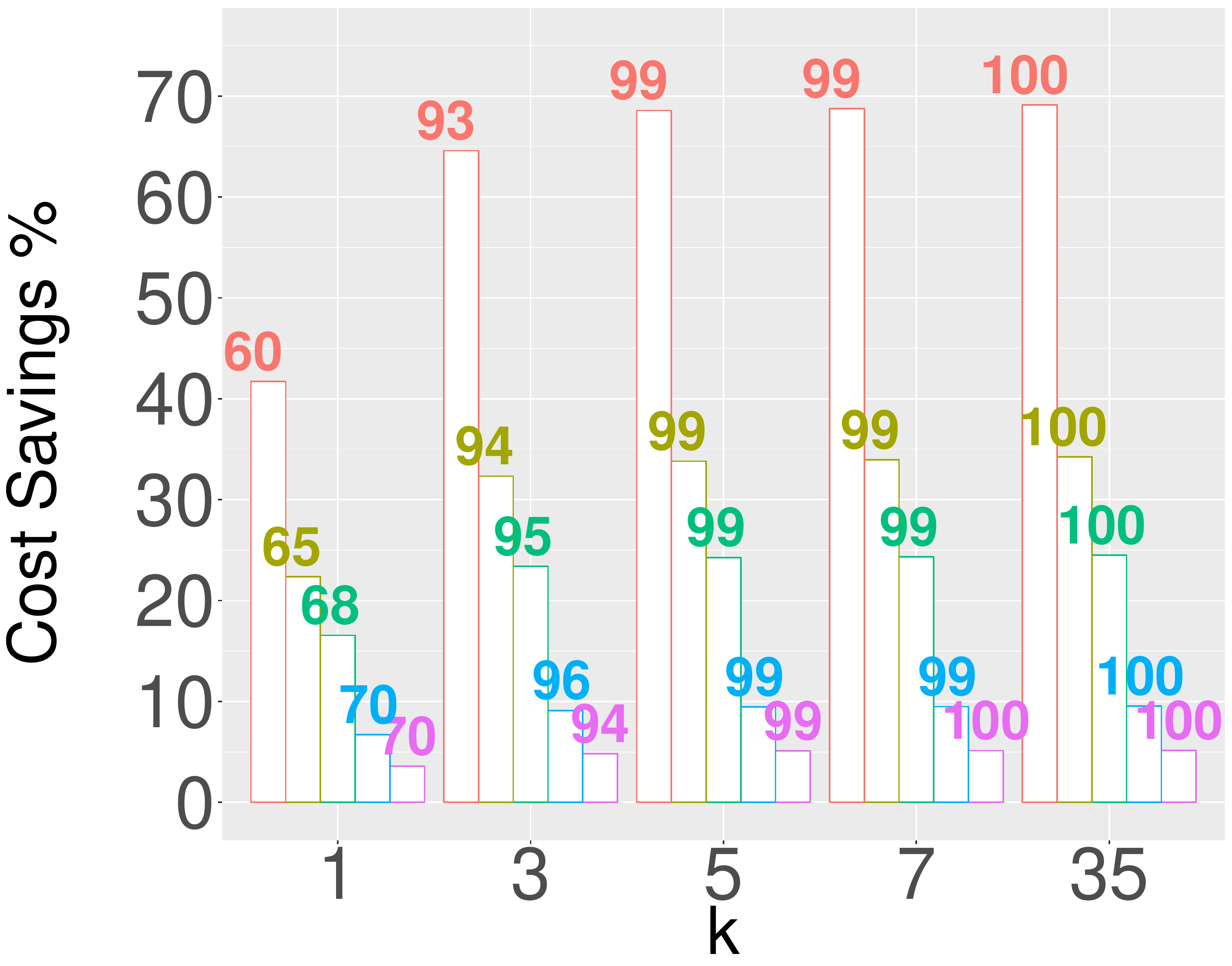}&
				\includegraphics[width=.20\textwidth]{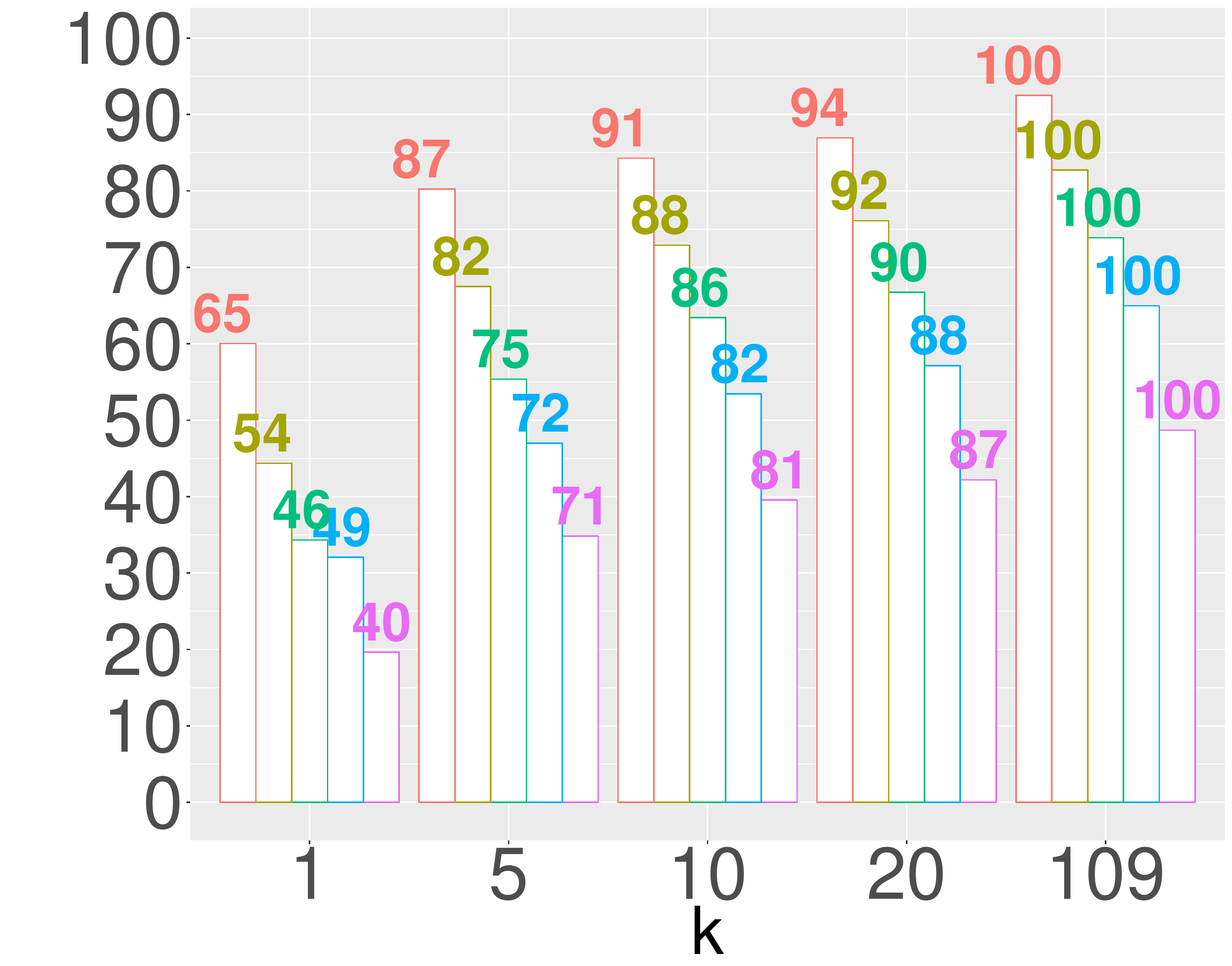}&
				\includegraphics[width=.20\textwidth]{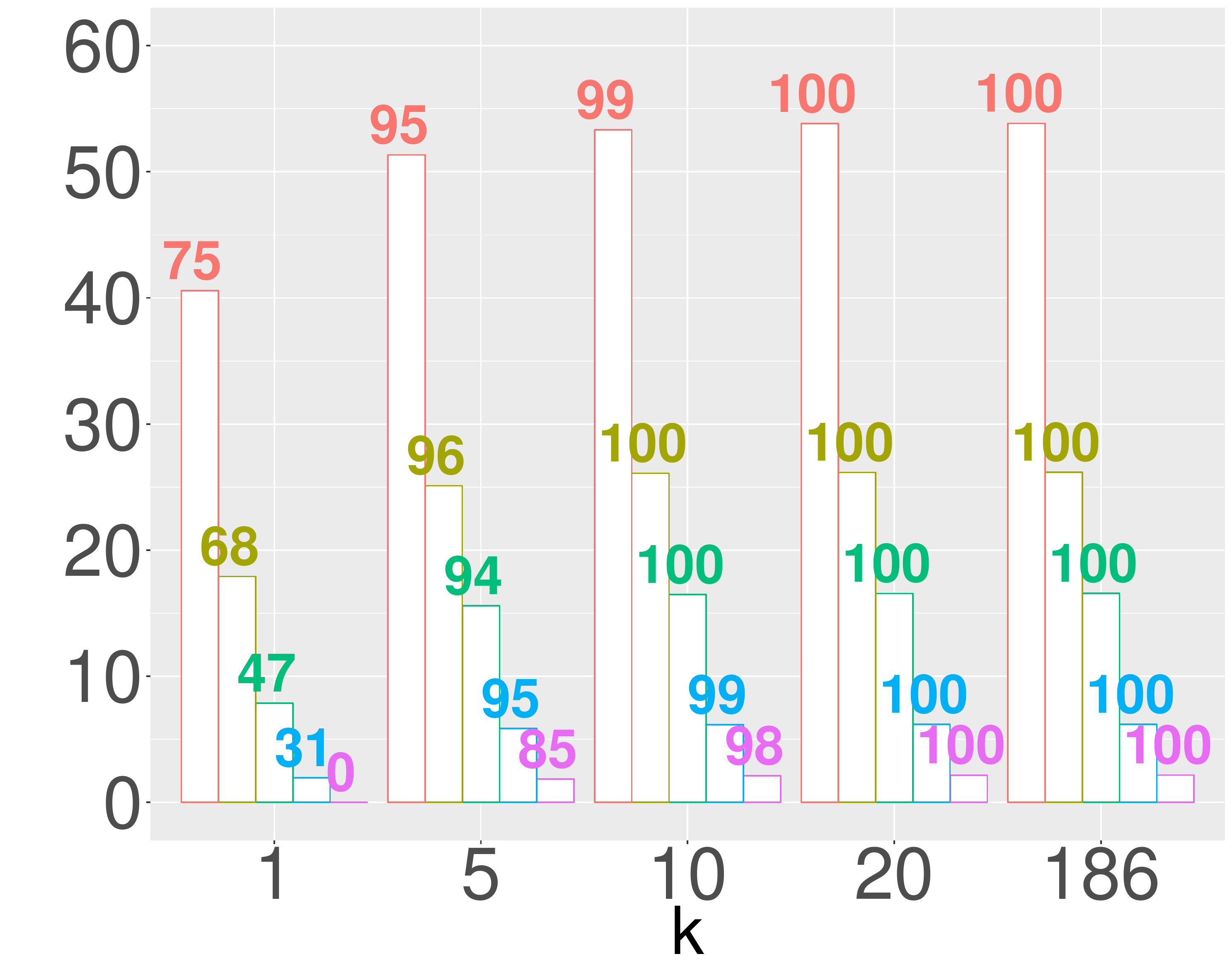}&
				\includegraphics[width=.20\textwidth]{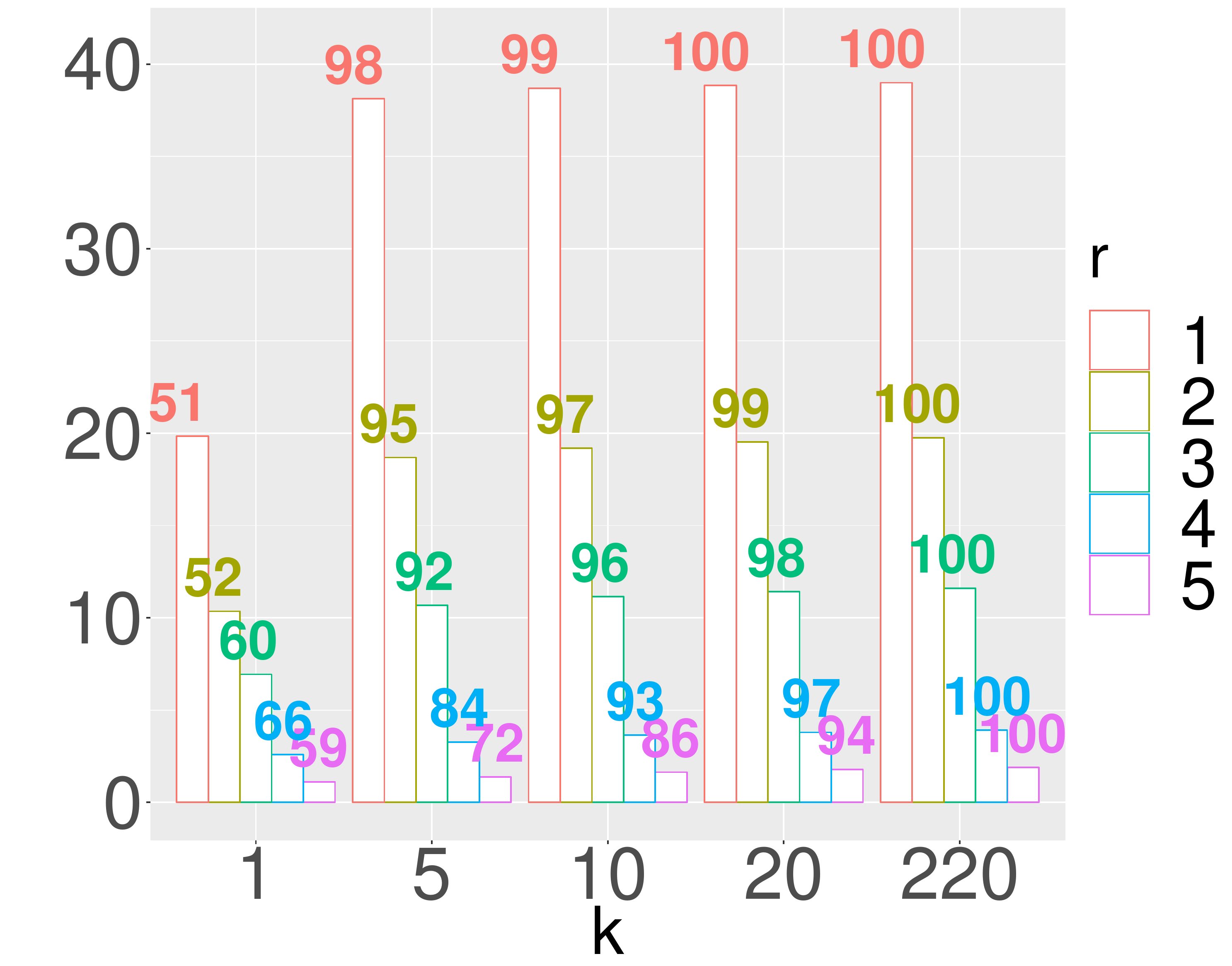}\\
				(a) \mildew (\mf) & (b) \bnpathfinder (\mf)  & (c) \munins (\wmf) & (d) \andes (\mf)   \\
				\includegraphics[width=.20\textwidth]{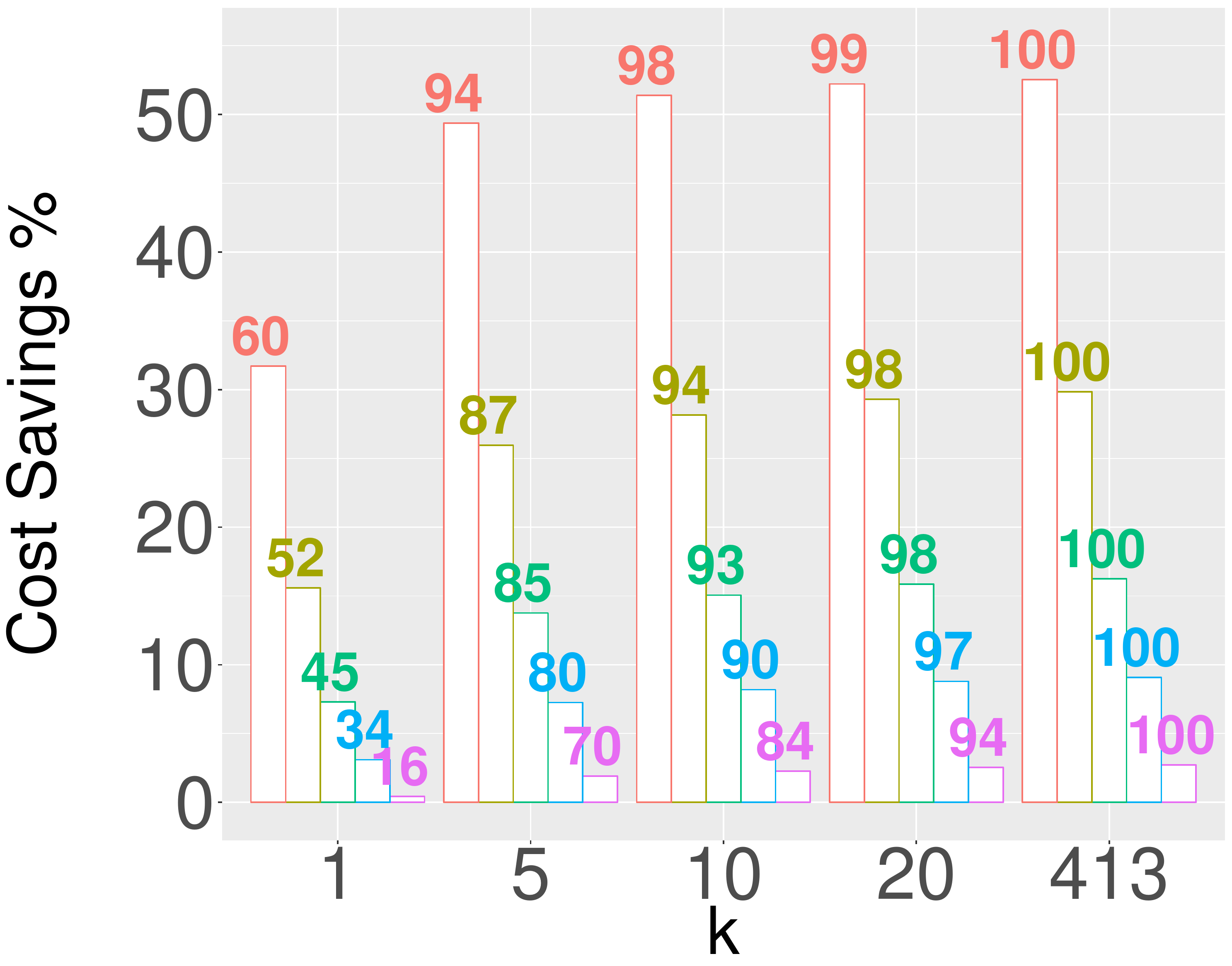} &
				\includegraphics[width=.20\textwidth]{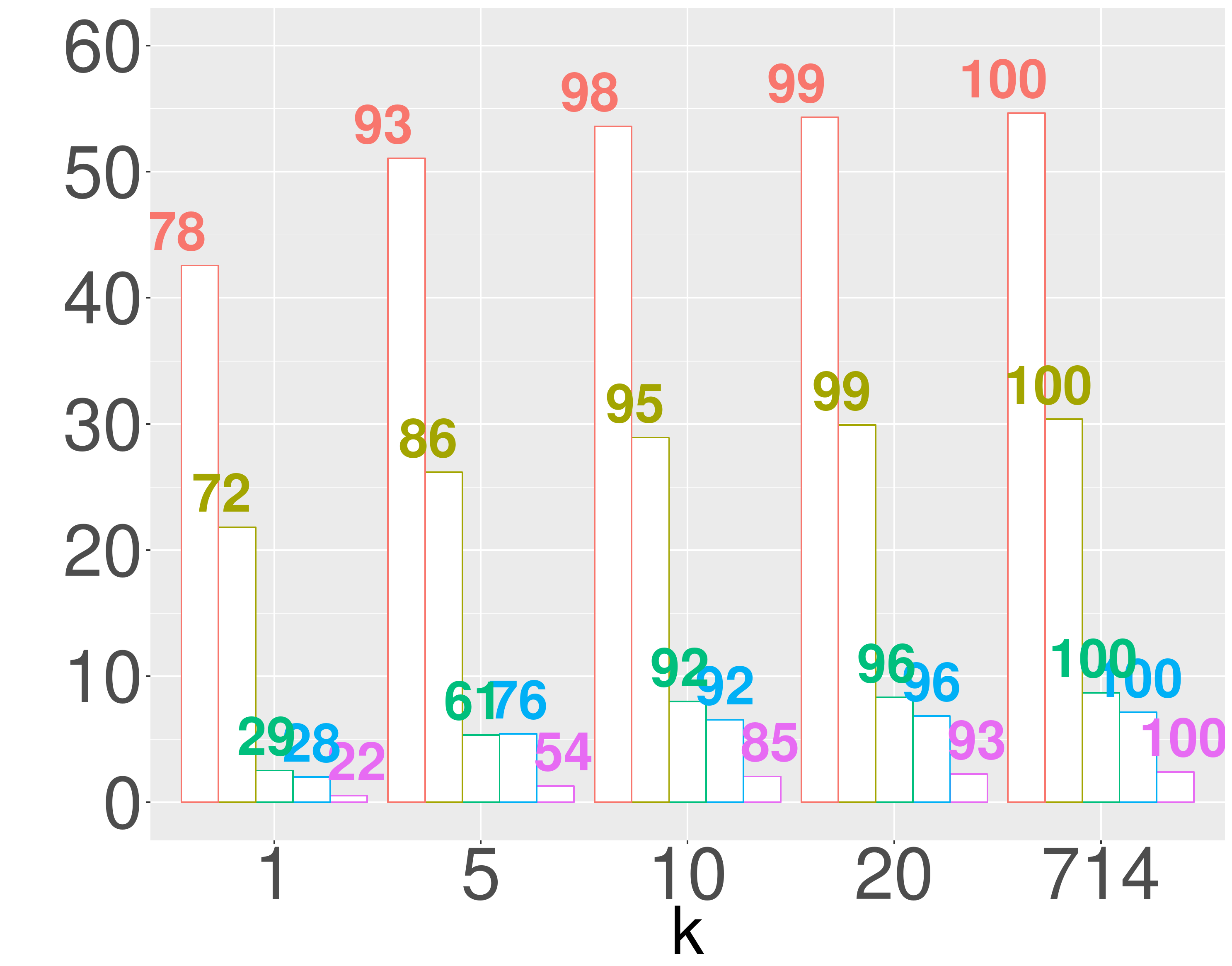} & 
				\includegraphics[width=.20\textwidth]{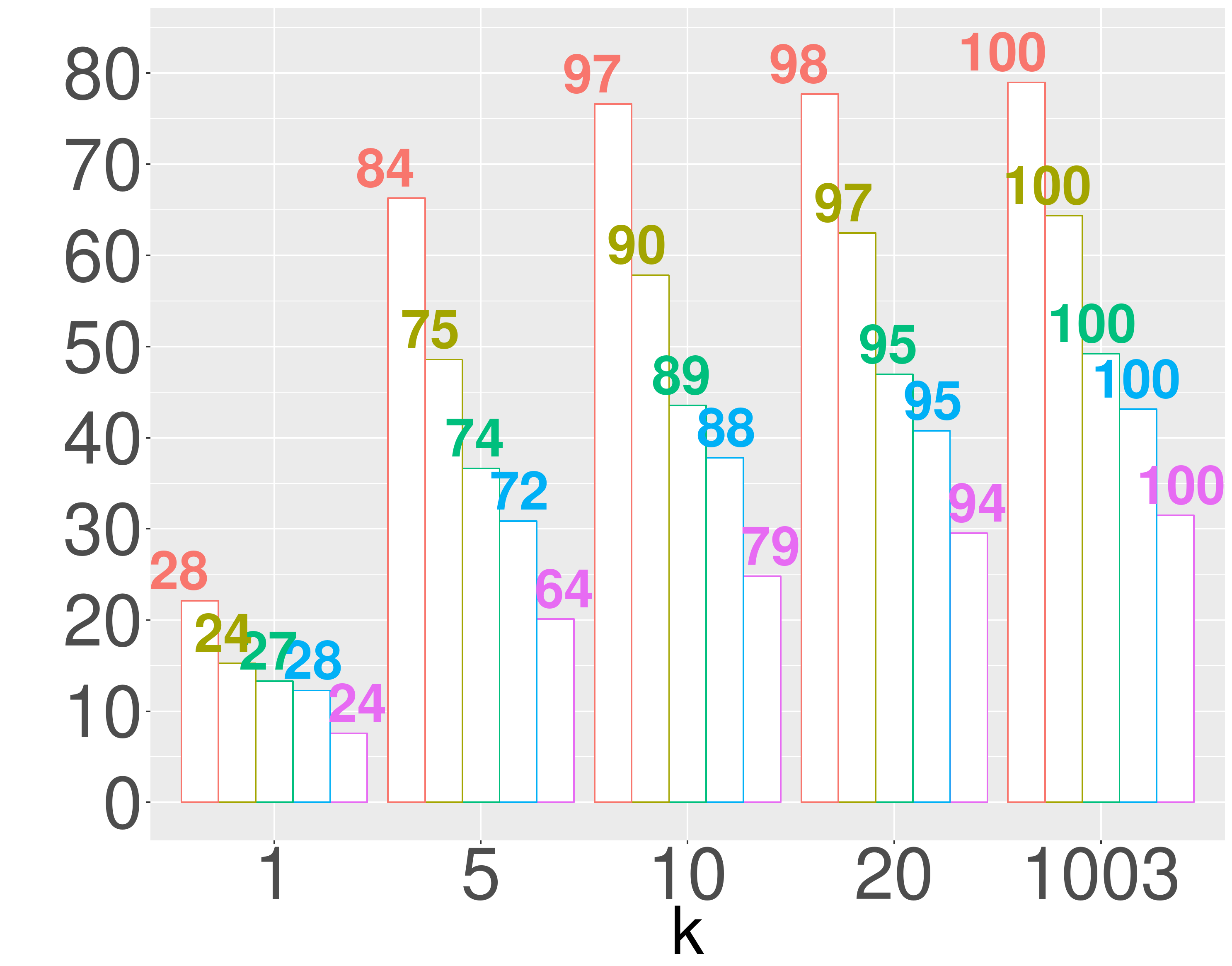} & 
				\includegraphics[width=.20\textwidth]{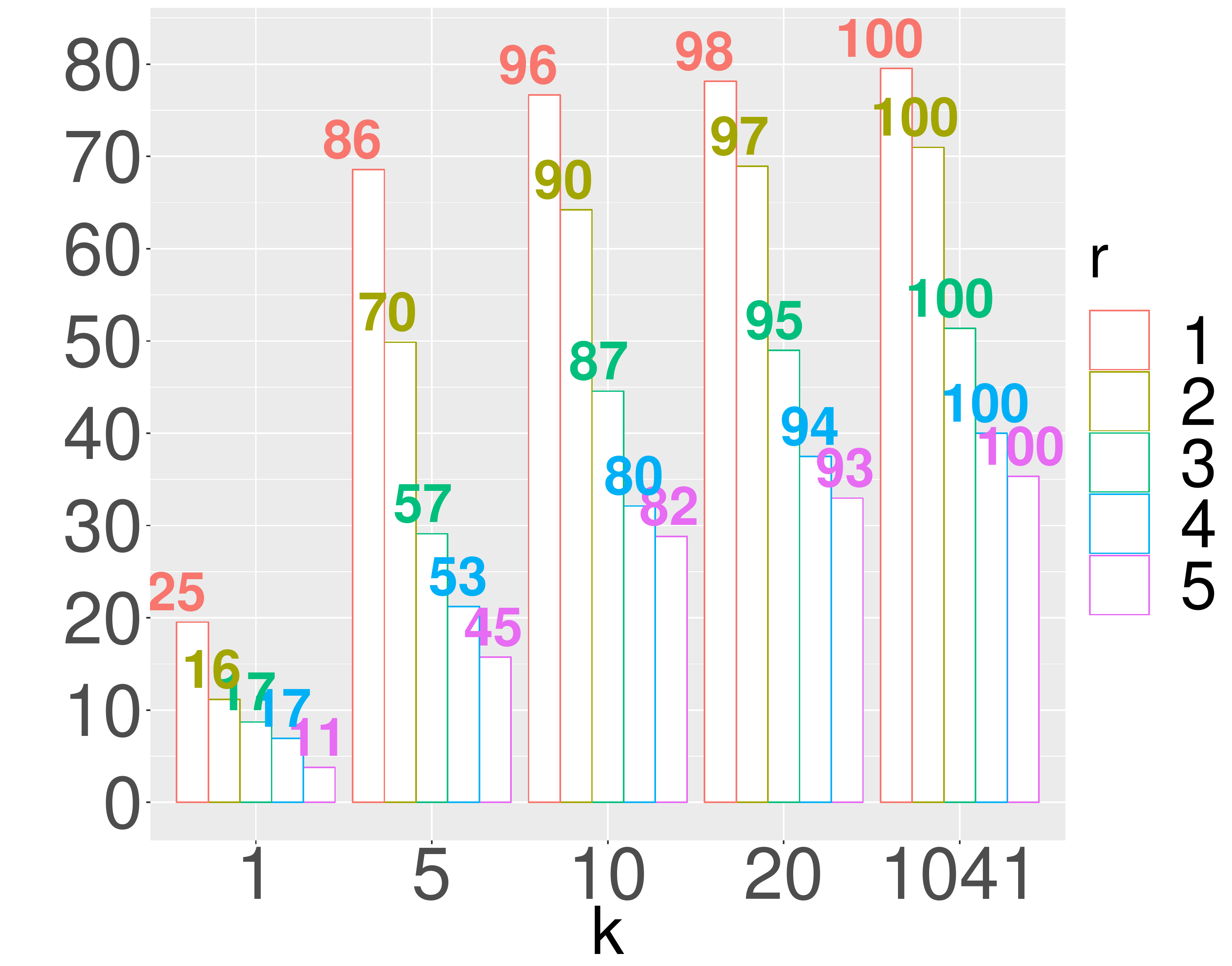}
				\\
				(e) \diabetes (\mf)  & (f) \link (\mf) & (g)  \muninm (\mf)   & (h) \muninb (\wmf)  \\ 
			\end{tabular}
			\caption{Cost savings per query size $\qsize$ in uniform-workload scheme.
				$x$-axis: number of materialized factors (budget \budget). $y$-axis: cost savings in query running time compared to no materialization. Numbers on the bars: the percentage of cost savings relative to the materialization of all factors.}
		\end{center}
		\label{fig:unif_per_r}
	\end{figure*}
}

\begin{table}[t]
\setlength\tabcolsep{4pt}
\fontsize{8}{9}\selectfont
\begin{center}
% \begin{footnotesize}
\caption{Average query processing times in seconds in uniform-workload scheme, when no materialization is used.}
\label{table:runtimes}
\begin{tabular}{lrrrrrr}
%\begin{tabular}{p{1.4cm}p{0.75cm}p{0.75cm}p{0.75cm}p{0.75cm}p{0.75cm}p{0.5cm}}
\toprule
Network & \qsize\!=\,1 & \qsize\!=\,2 & \qsize\!=\,3 & \qsize\!=\,4  & \qsize\!=\,5 & all   \\ \midrule
\mildew &   11.2    & 33.5 & 97.5 & 122.7 &177.3 & 77.0 \\
\bnpathfinder &  0.2 & 0.2 & 0.2 & 0.3 & 0.4  & 0.2  \\
\munins &  319.4  & 408.4  & 474.7 & 656.6  & 767.2 & 438.2  \\
\andes  &  1.1 &  1.6  & 2.6 & 4.8     & 7.6 & 3.5 \\
\diabetes &  18.0 &  95.7 & 332.3 & 621.0 & 801.9 & 162.3  \\
\link &  119.4  &  215.8 &   313.8 & 391.1 & 532.0 & 287.1 \\
\muninm &  7.1  & 11.6 &     14.5   & 12.9 &     25.8   & 14.4  \\
\muninb &   15.1 & 16.7 & 25.8 & 30.3 & 38.5 & 25.3  \\
\revision{\tpchsmall} & \revision{0.1} & \revision{0.1} & \revision{0.1} & \revision{0.2} & \revision{2.1} & \revision{0.5}  \\ 
\revision{\tpchmedium} & \revision{0.1} & \revision{0.1} & \revision{0.6} & \revision{2.9} & \revision{15.5} & \revision{3.8}  \\
\revision{\tpchverylarge} & \revision{3.7} & \revision{5.9} & \revision{37.3} & \revision{81.9} & \revision{183.7} & \revision{50.9} \\ 
\revision{\tpchlarge} & \revision{0.2} & \revision{0.7} & \revision{2.9} & \revision{25.5} & \revision{28.2} & \revision{8.7} \\ 
\bottomrule
\end{tabular}
% \end{footnotesize}
\end{center}
\end{table}

\ReviewOnly{
\begin{figure*}[t]
\begin{center}
\begin{tabular}{cccc}
    \includegraphics[width=.20\textwidth]{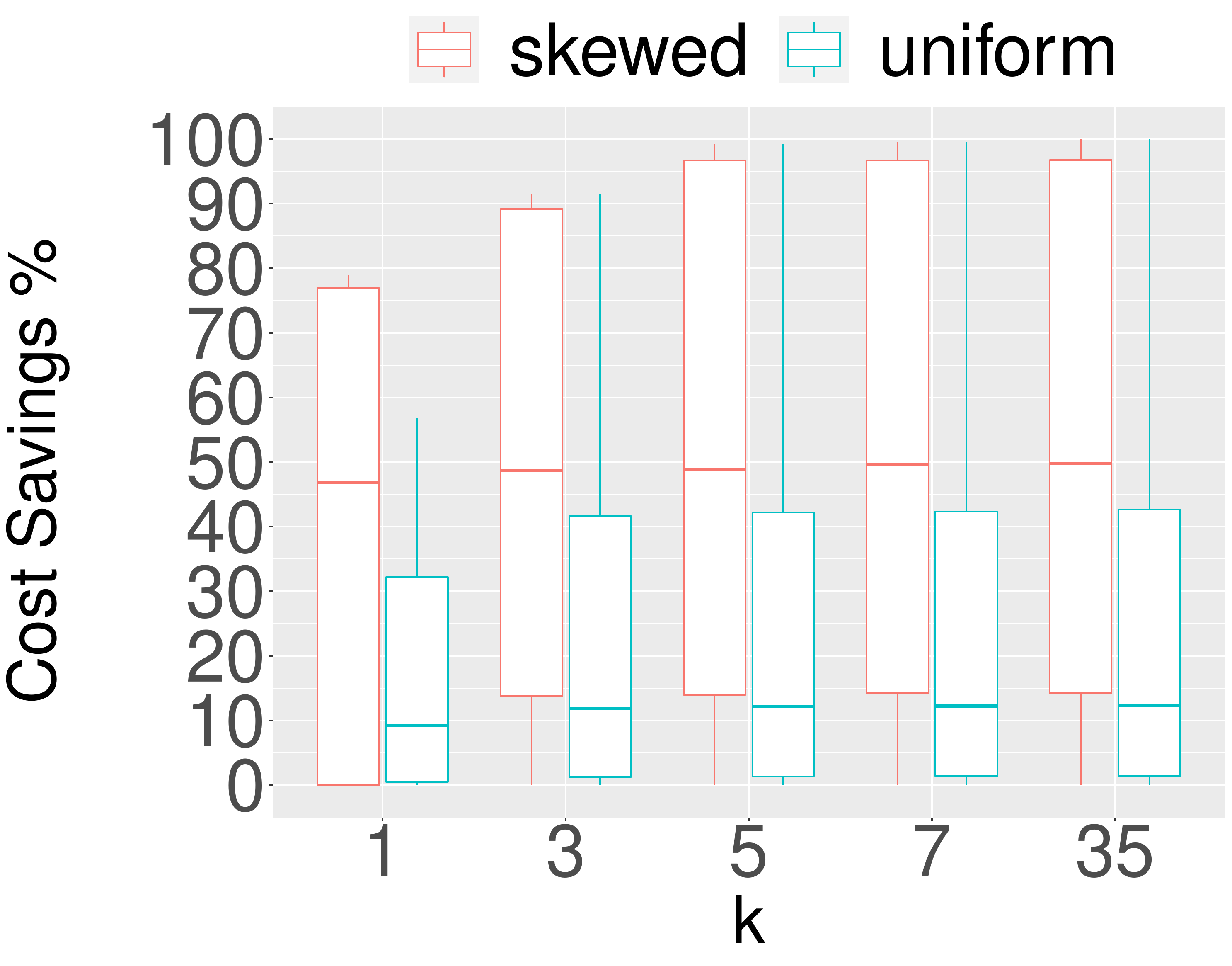}&
    \hspace{-0mm} \includegraphics[width=.20\textwidth]{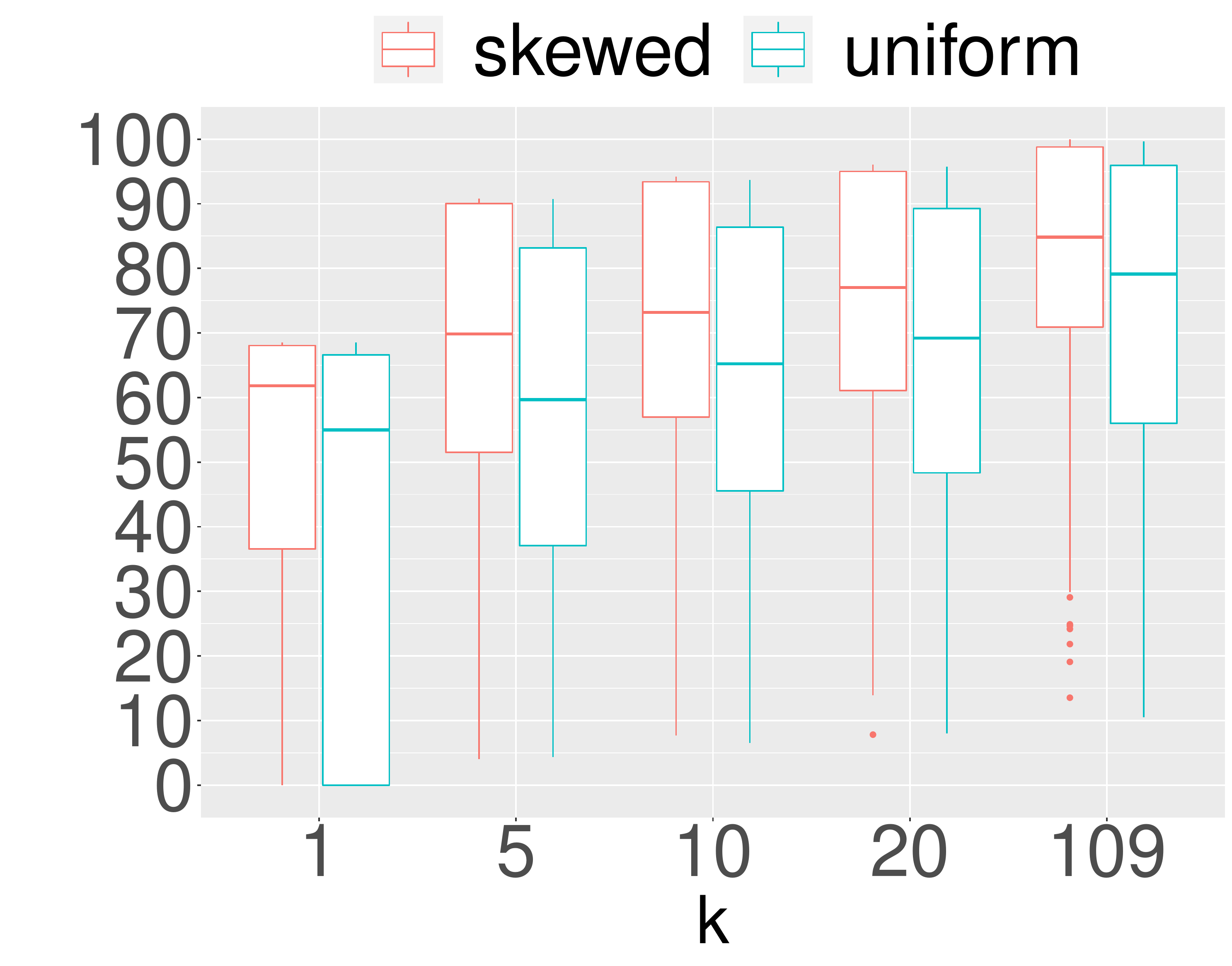}&
    \hspace{-0mm} \includegraphics[width=.20\textwidth]{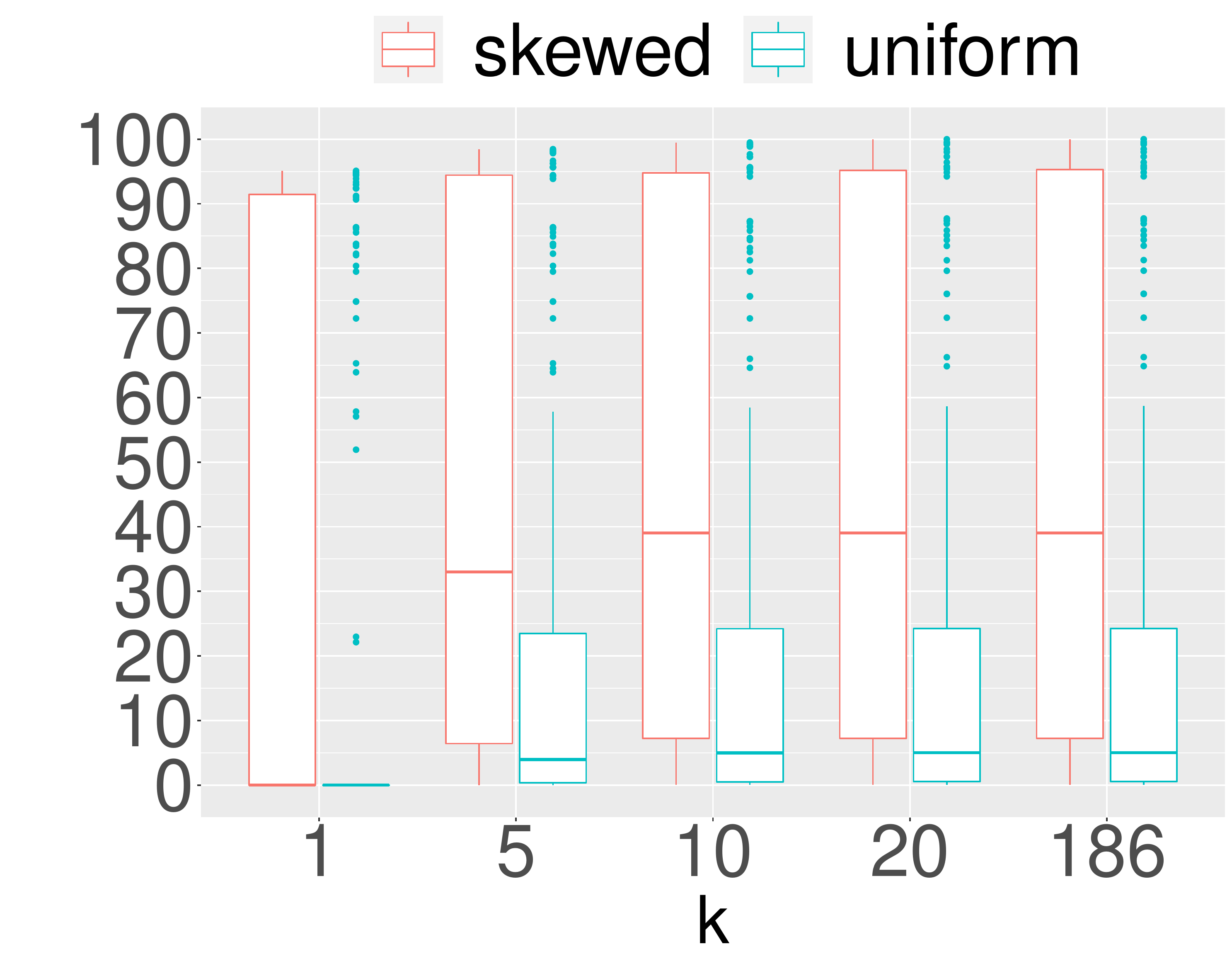}&
    \hspace{-0mm} \includegraphics[width=.20\textwidth]{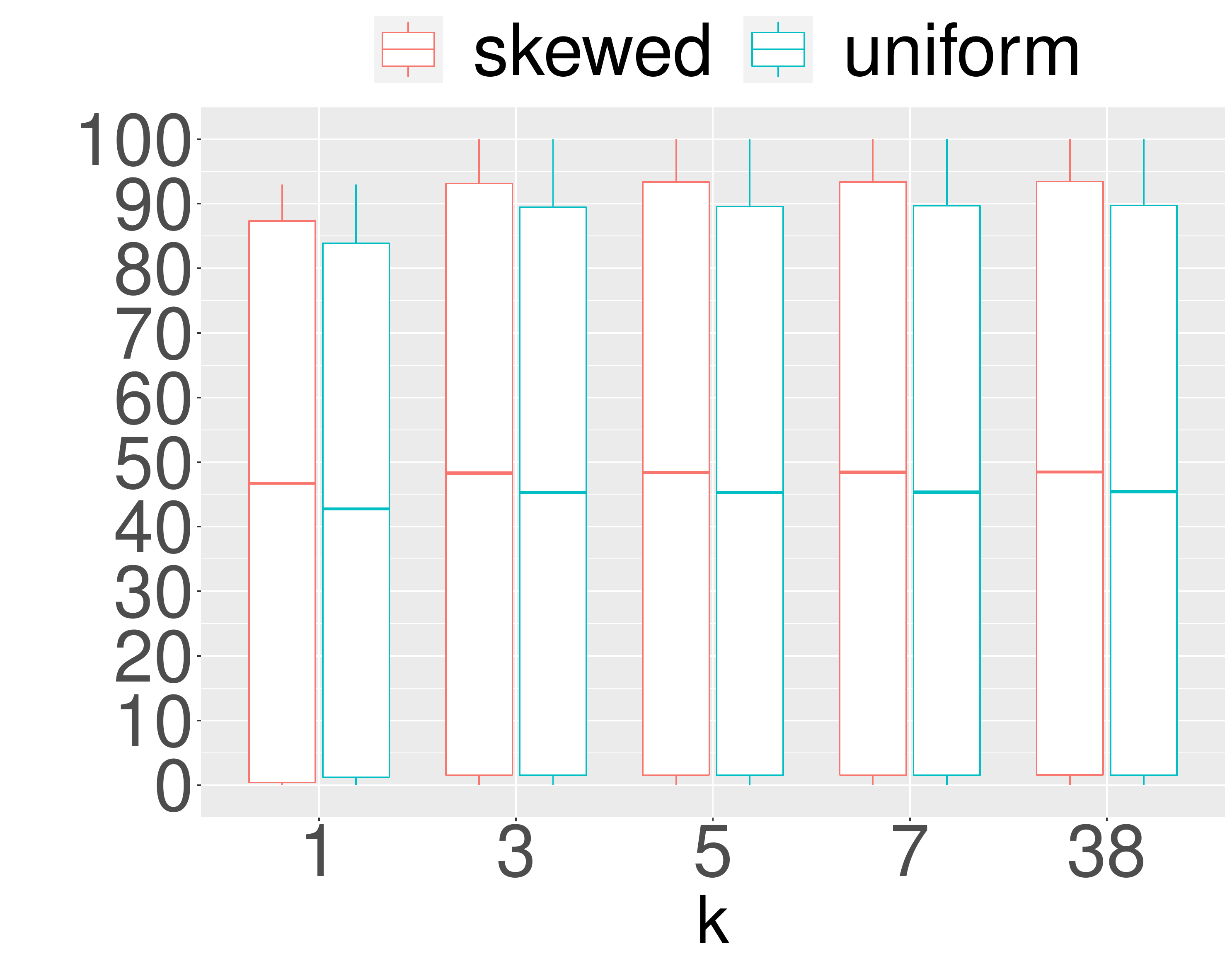}\\
  (a) \mildew (\mf) & (b) \bnpathfinder (\mf)  & (c) \munins (\wmf) & \revision{(d) \tpchverylarge (\mw)}   \\
  \includegraphics[width=.20\textwidth]{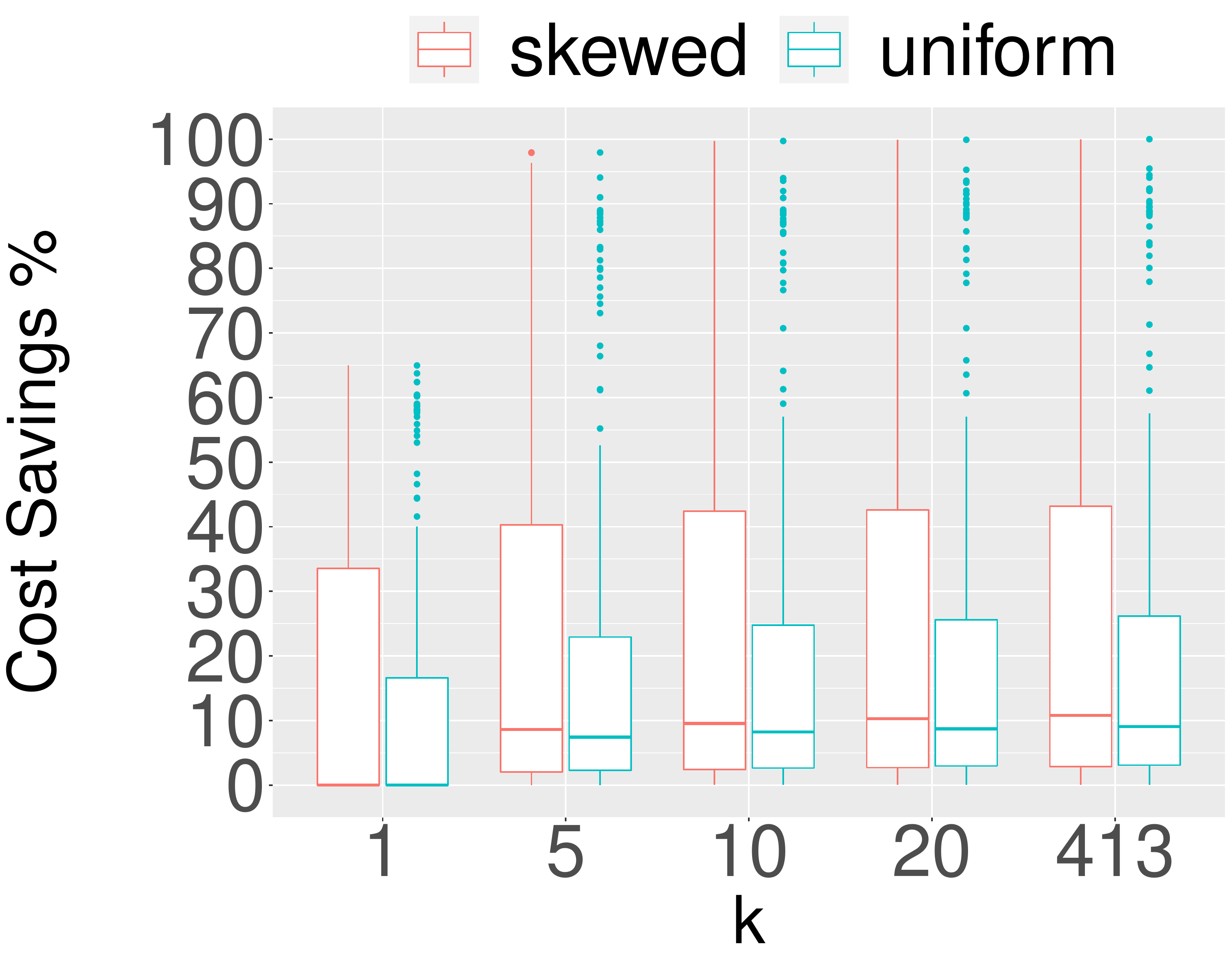} &
    \hspace{-0mm} \includegraphics[width=.20\textwidth]{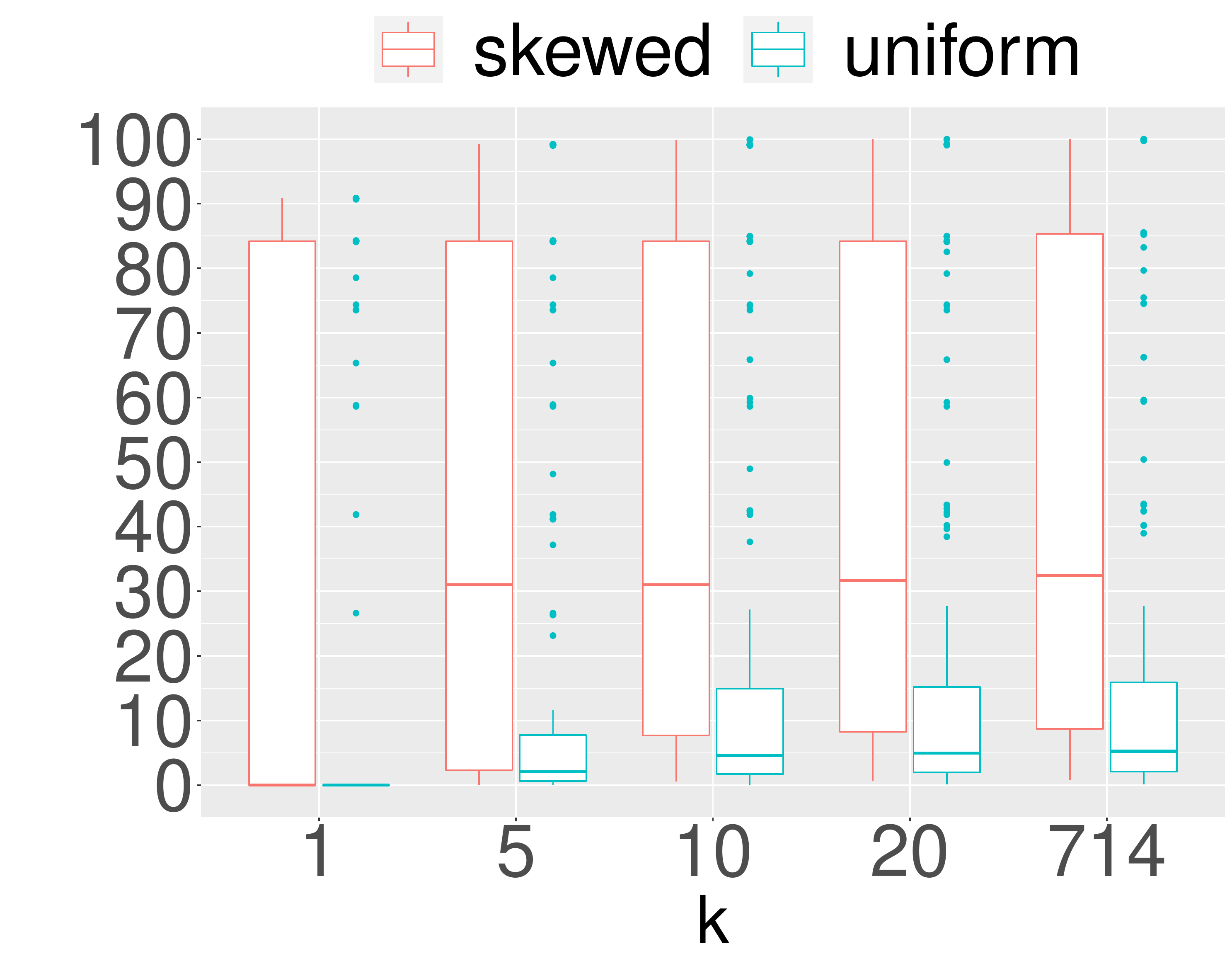} & 
    \hspace{-0mm} \includegraphics[width=.20\textwidth]{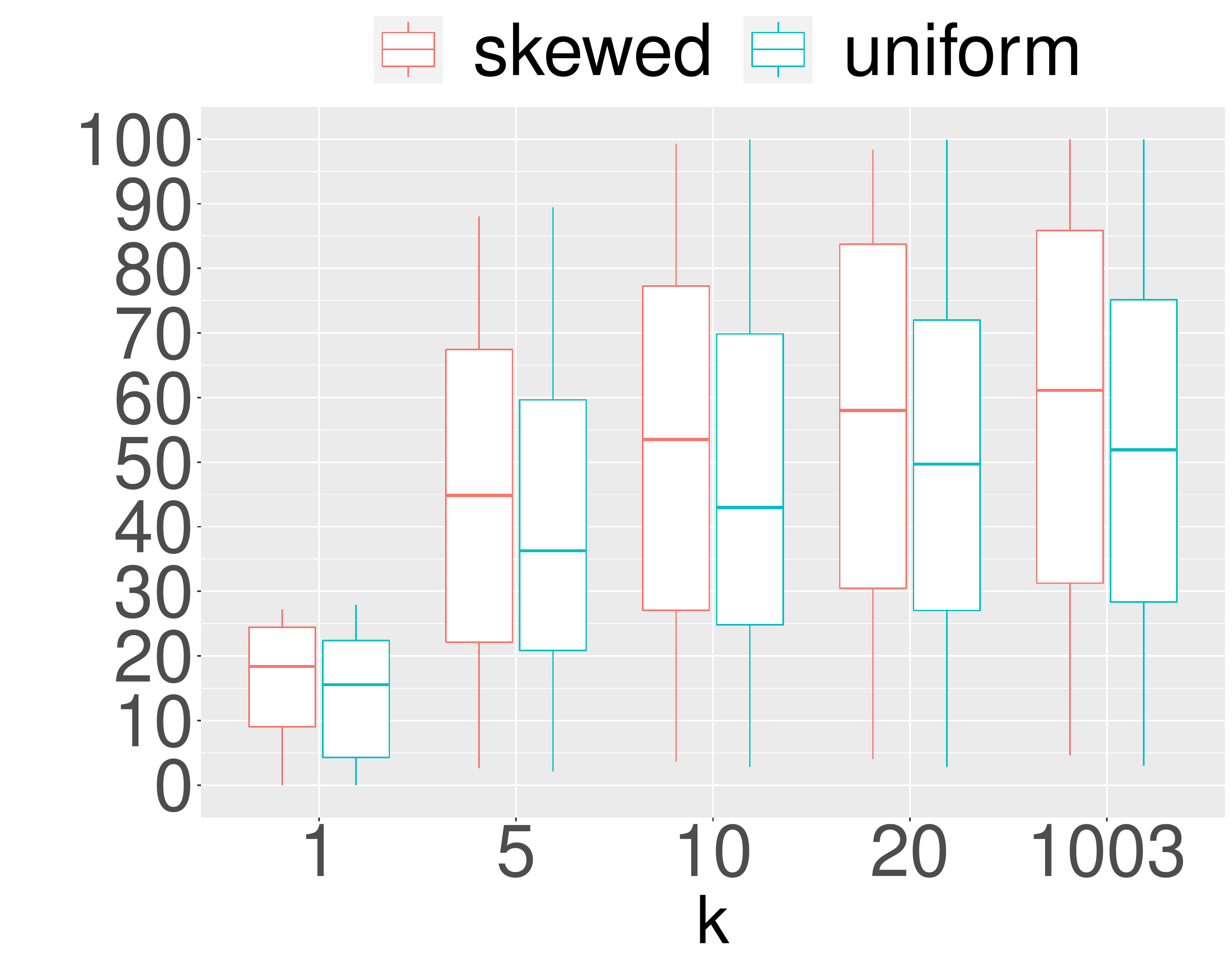} & 
    \hspace{-0mm} \includegraphics[width=.20\textwidth]{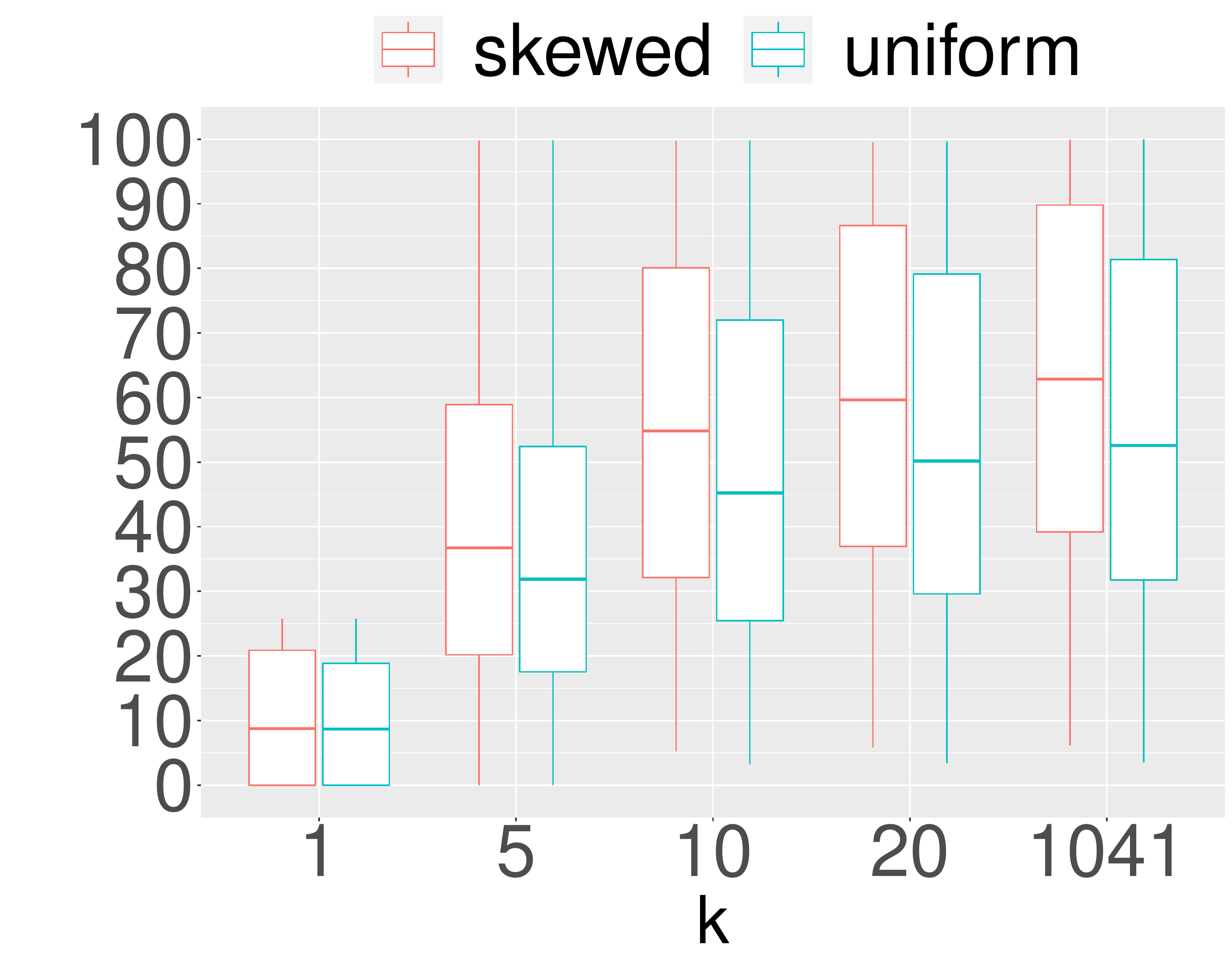}
    \\
   (e) \diabetes (\mf)  & (f) \link (\mf) & (g)  \muninm (\mf)   & (h) \muninb (\wmf)
\end{tabular}
\caption{Cost savings for uniform and skewed workloads. $x$-axis: budget \budget. $y$-axis: cost savings in query running time compared to no materialization.}
\end{center}
\label{fig:workloadMixed}
\end{figure*}
}

\FullOnly{
\begin{figure*}[t]
\begin{tabular}{cccc}
    \includegraphics[width=.225\textwidth]{plots/mildew/workload_mixed}&
    \hspace{-0mm} \includegraphics[width=.225\textwidth]{plots/pathfinder/workload_mixed}&
    \hspace{-0mm} \includegraphics[width=.225\textwidth]{plots/munin1/workload_mixed}&
    \hspace{-0mm} \includegraphics[width=.225\textwidth]{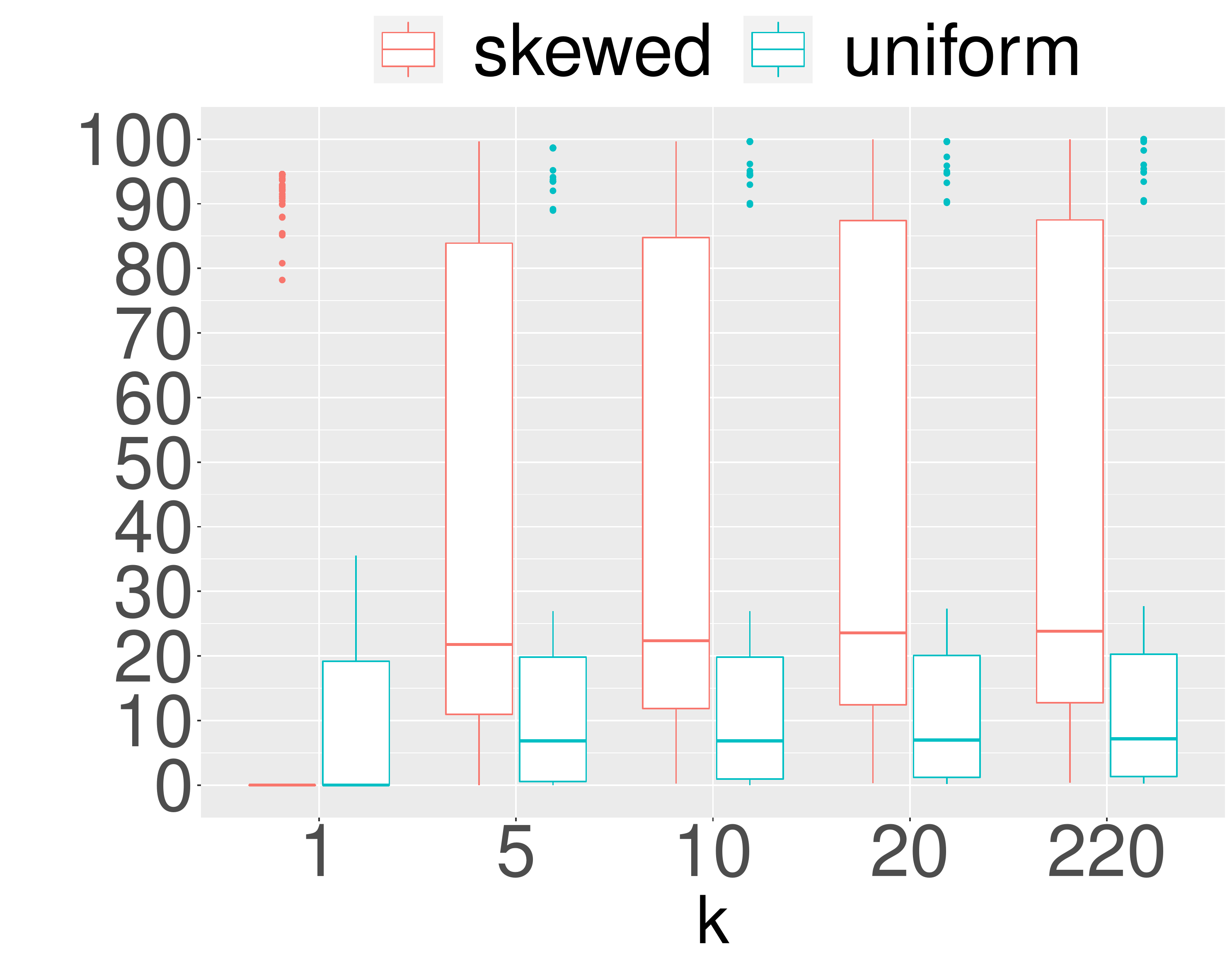}\\
  (a) \mildew (\mf) & (b) \bnpathfinder (\mf)  & (c) \munins (\wmf) & (d) \andes (\mf)   \\
  \includegraphics[width=.225\textwidth]{plots/diabetes/workload_mixed} &
    \hspace{-0mm} \includegraphics[width=.225\textwidth]{plots/link/workload_mixed} & 
    \hspace{-0mm} \includegraphics[width=.225\textwidth]{plots/munin2/workload_mixed} & 
    \hspace{-0mm} \includegraphics[width=.225\textwidth]{plots/munin/workload_mixed}
    \\
   (e) \diabetes (\mf)  & (f) \link (\mf) & (g)  \muninm (\mf)   & (h) \muninb (\wmf) \\
   \includegraphics[width=.225\textwidth]{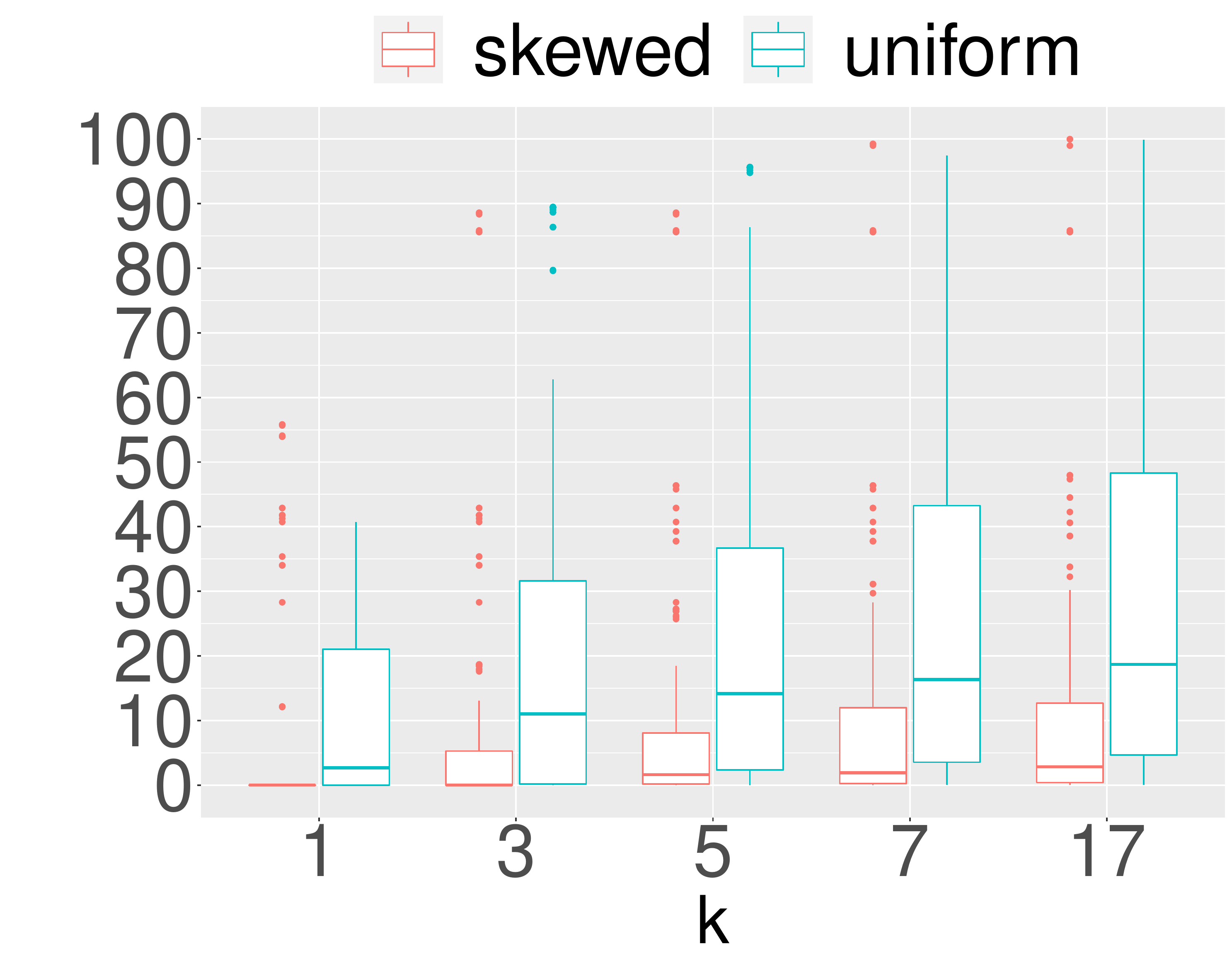} &
    \hspace{-0mm} \includegraphics[width=.225\textwidth]{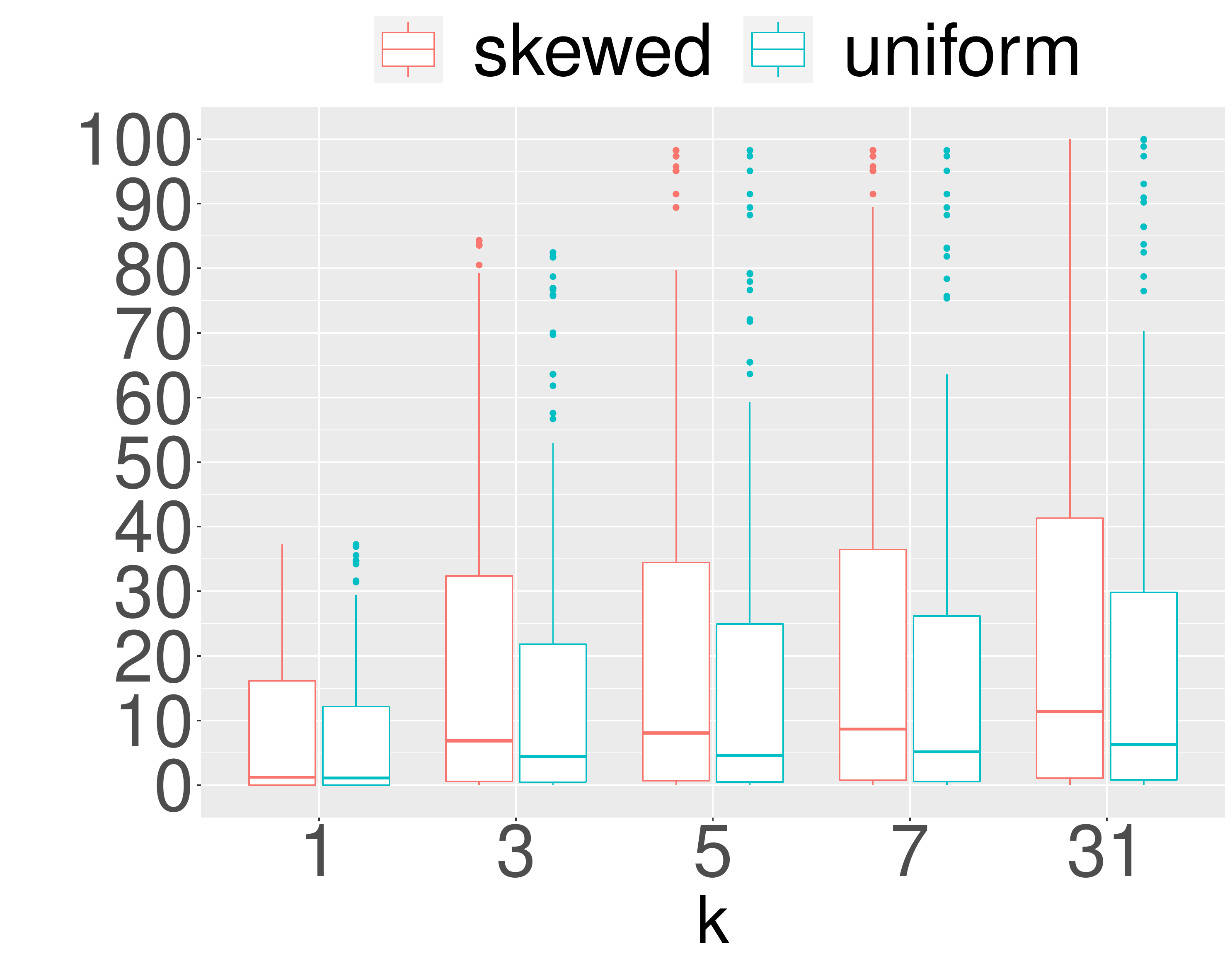} & 
    \hspace{-0mm} \includegraphics[width=.225\textwidth]{plots/tpch3/workload_mixed} & 
    \hspace{-0mm} \includegraphics[width=.225\textwidth]{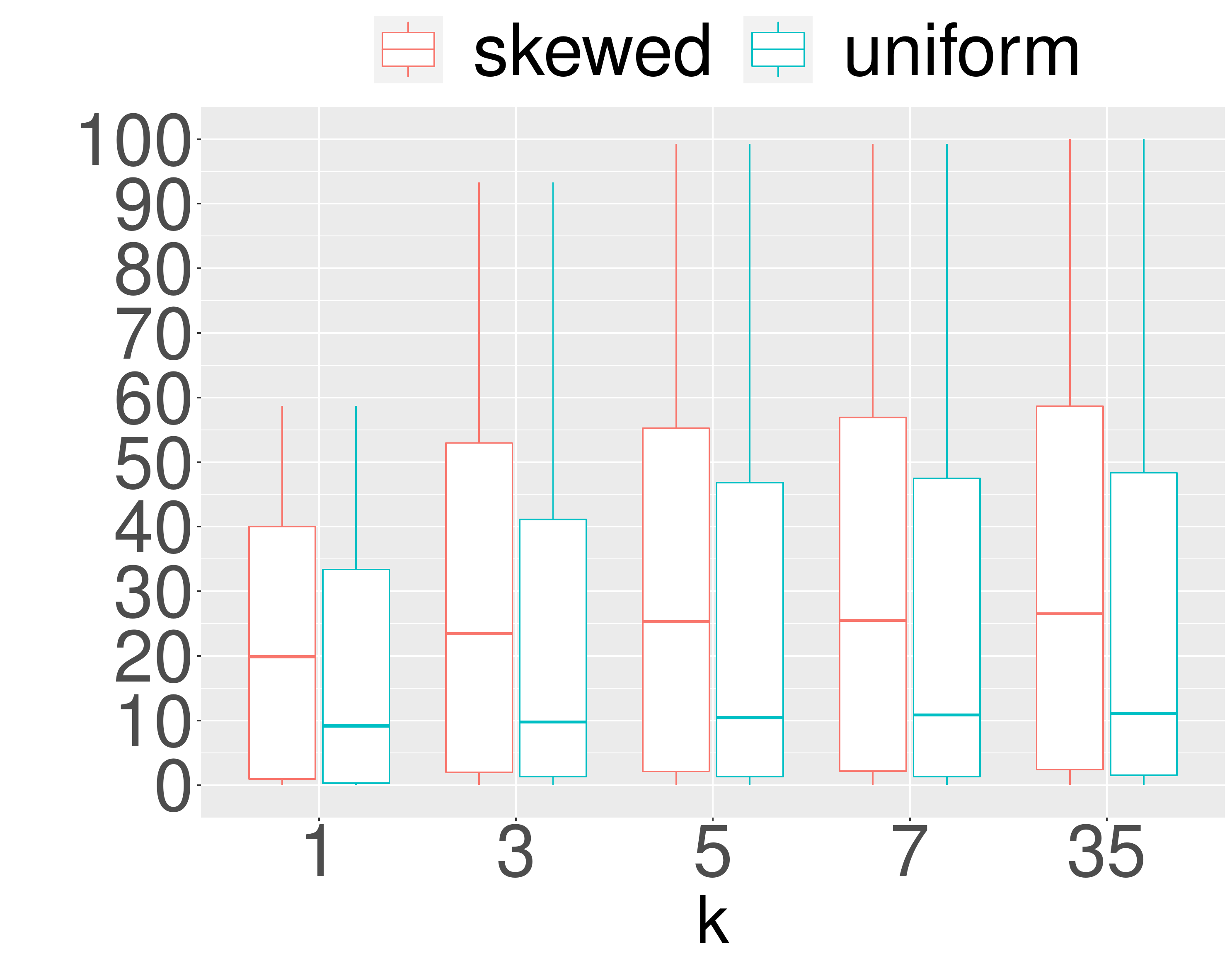}
    \\
   \revision{(i) \tpchsmall (\mw)}  & \revision{(j) \tpchmedium (\mw)} & \revision{(k)  \tpchverylarge (\mw)}   & \revision{(l) \tpchlarge (\mw)}  
\end{tabular}
\caption{Cost savings for uniform and skewed workloads. $x$-axis: budget \budget. $y$-axis: cost savings in query running time compared to no materialization. }
\label{fig:workloadMixed}
\end{figure*}
}

%% new comparison plots -- will be updated so just a spaceholder for now
\ReviewOnly{
\begin{figure*}[t]
\begin{center}
\begin{tabular}{cccc}
    \includegraphics[width=.20\textwidth]{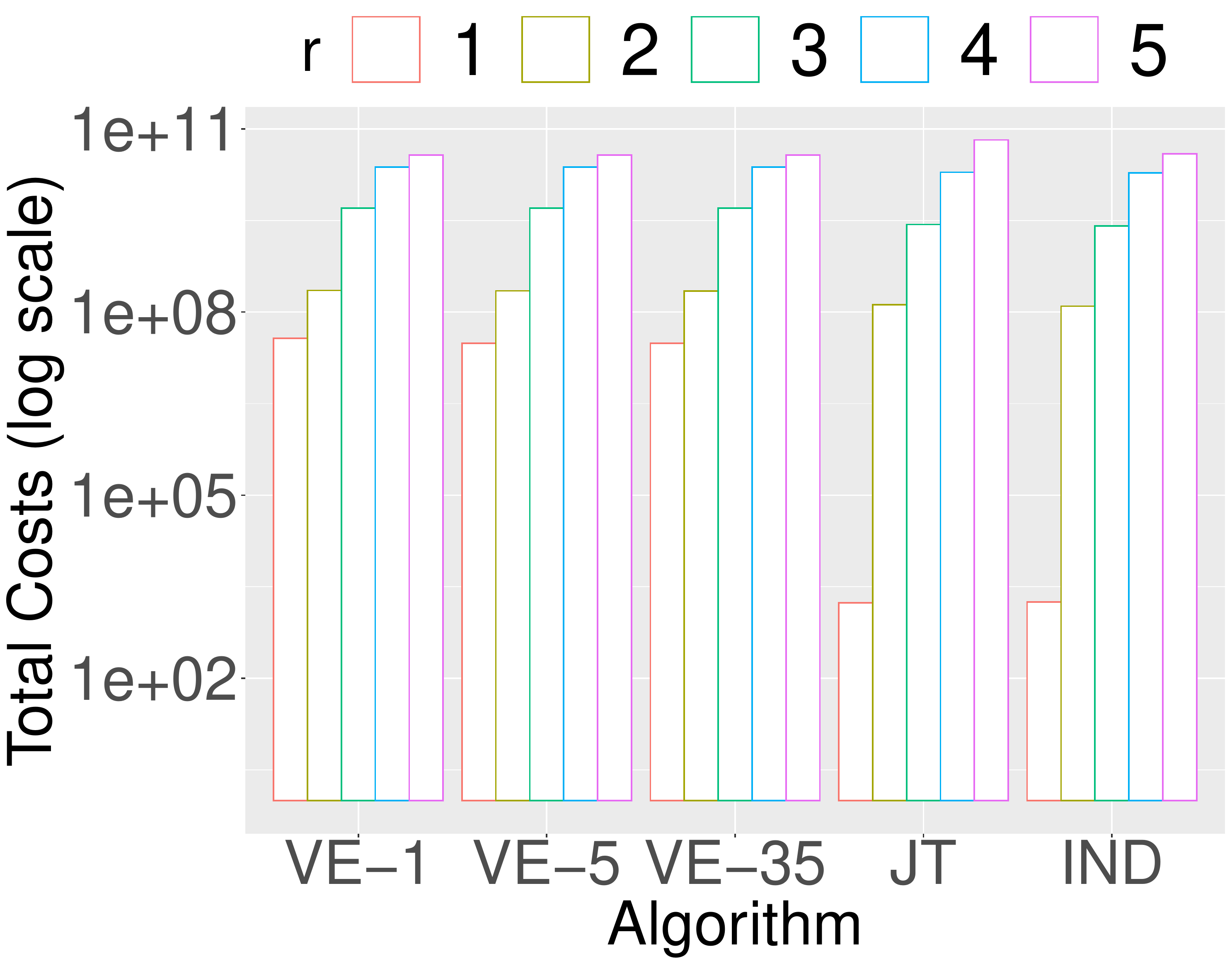}&
    \includegraphics[width=.20\textwidth]{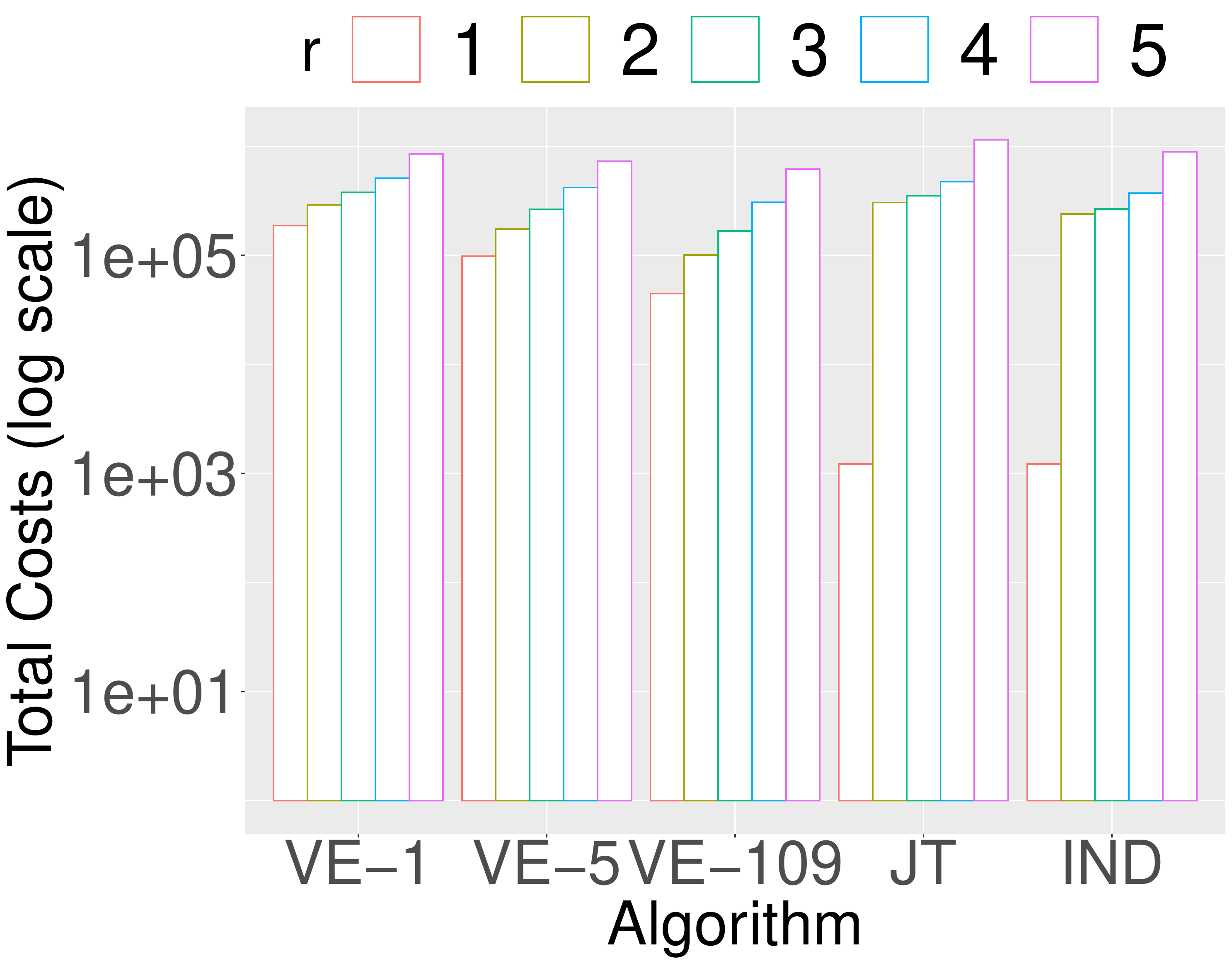}&
    \includegraphics[width=.20\textwidth]{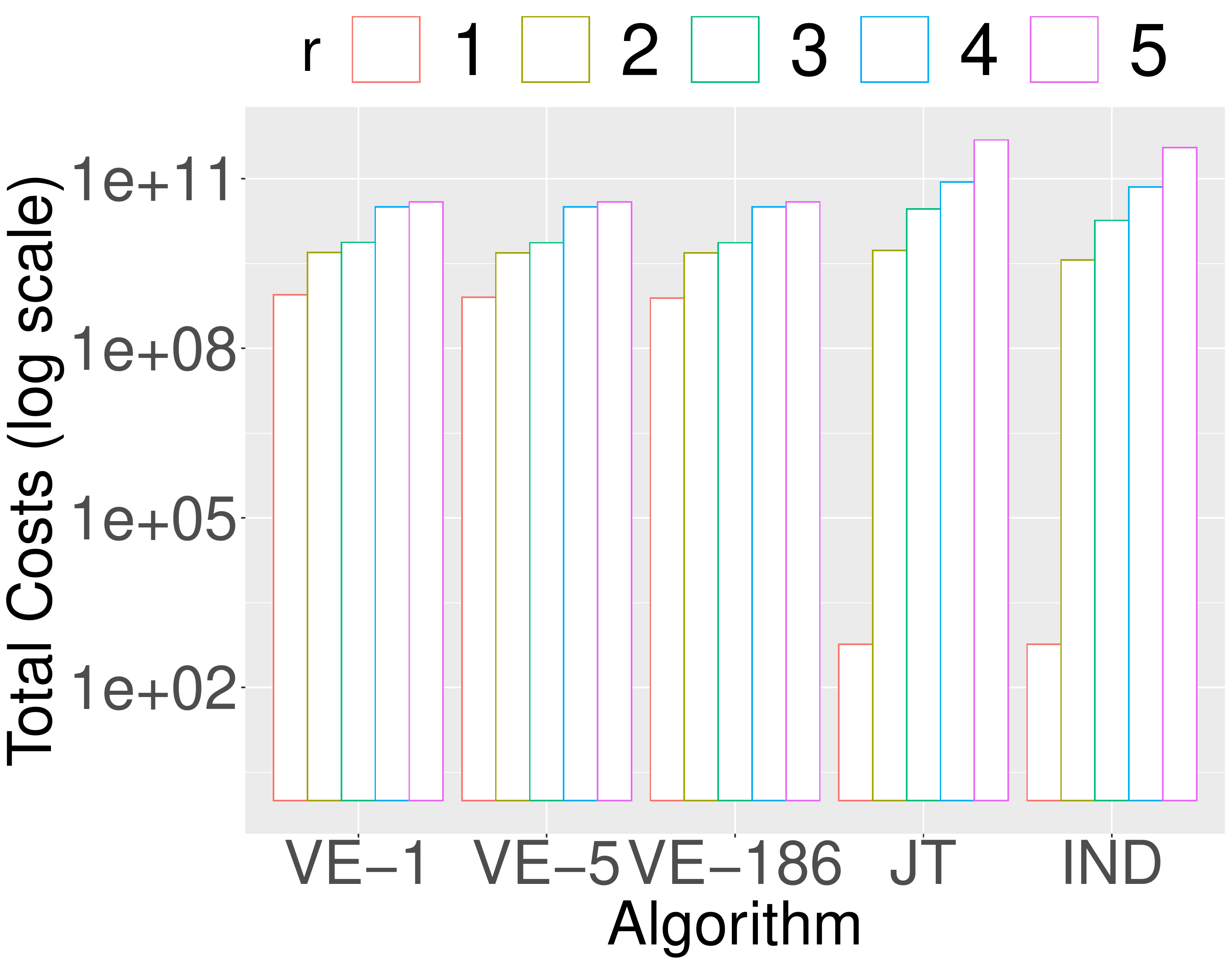}&
    \includegraphics[width=.20\textwidth]{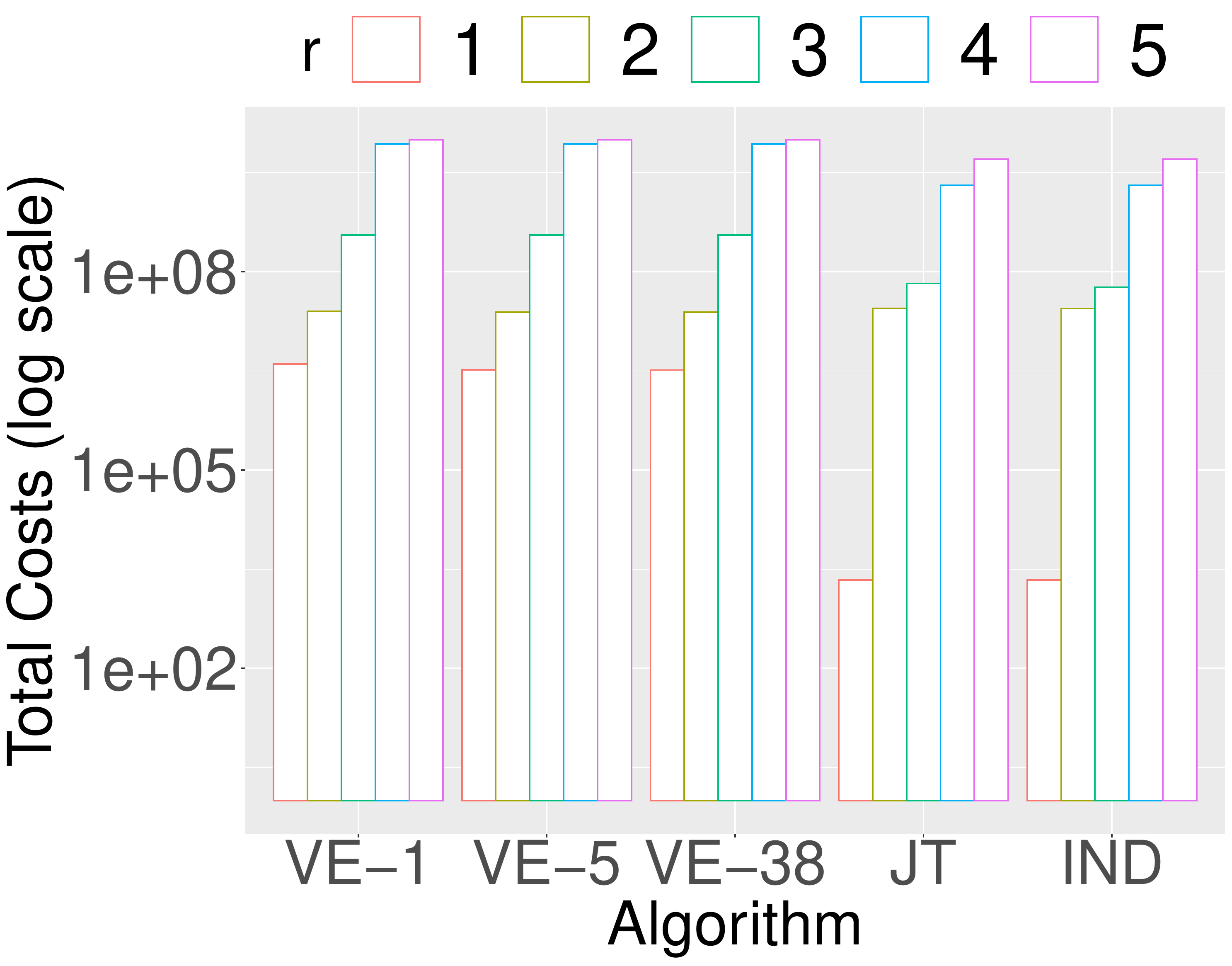}\\
  (a) \mildew (\mf) & (b) \bnpathfinder (\mf)  & (c) \munins (\wmf) & \revision{(d) \tpchverylarge (\mw)}   \\
  \includegraphics[width=.20\textwidth]{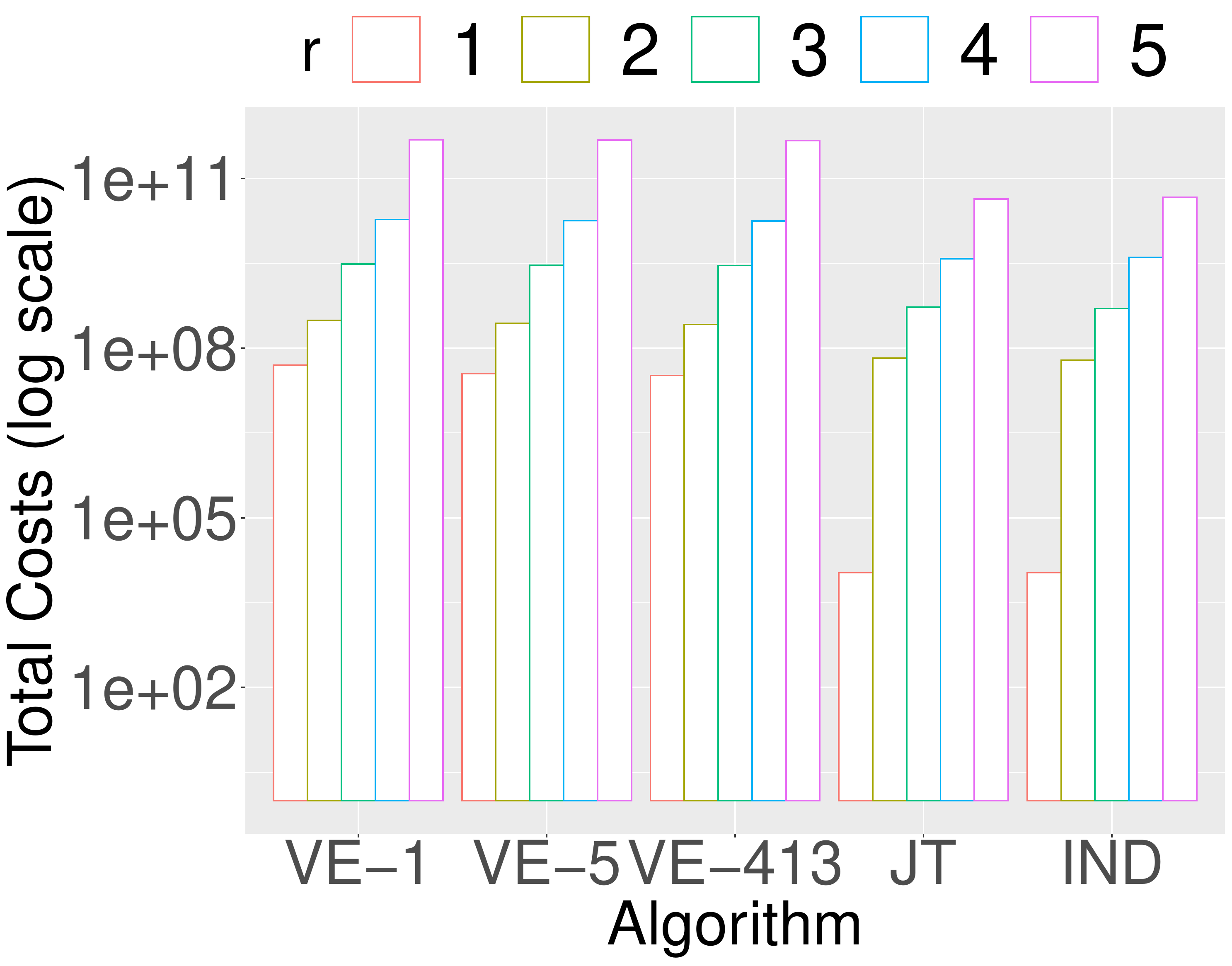} &
    \includegraphics[width=.20\textwidth]{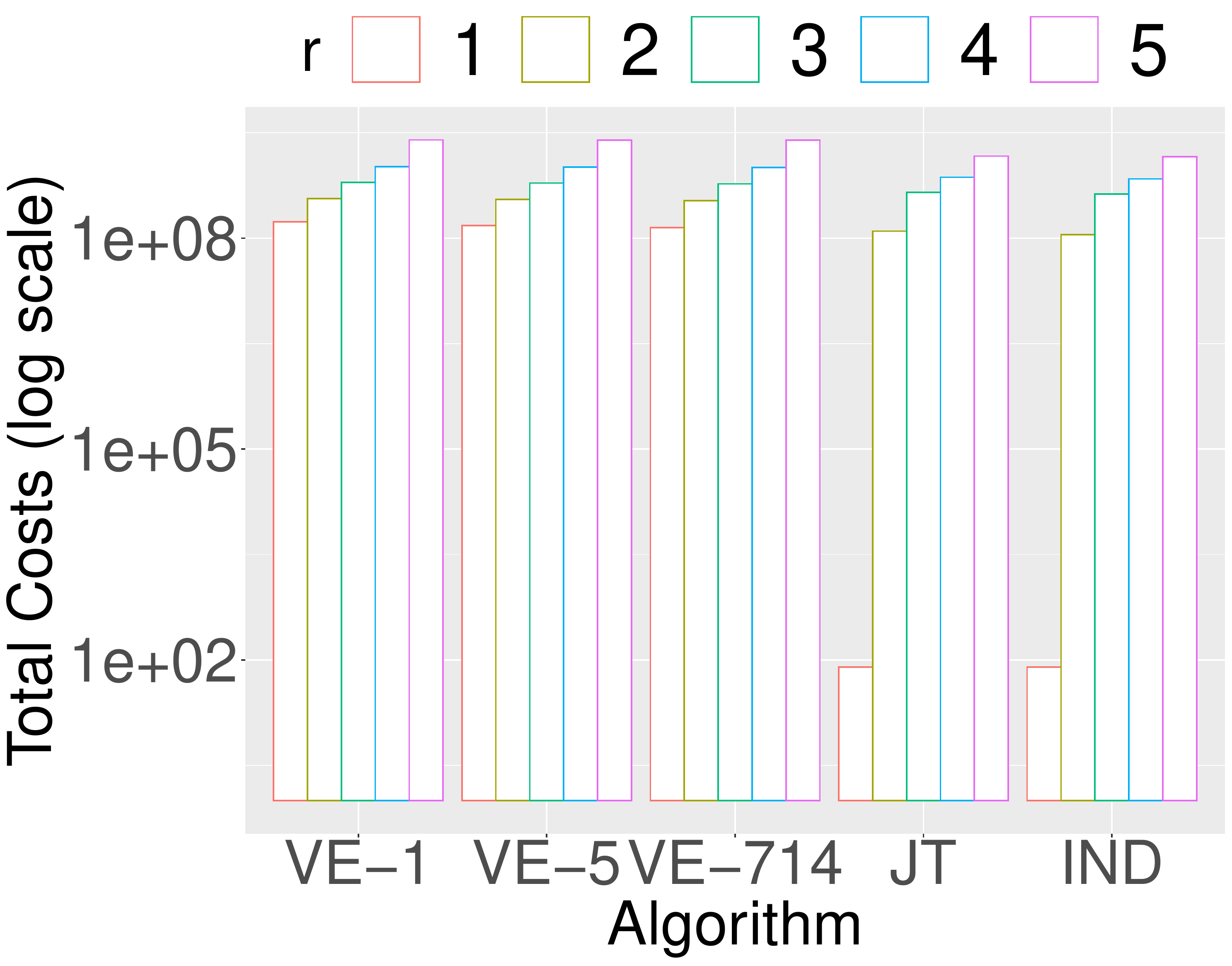} & 
    \includegraphics[width=.20\textwidth]{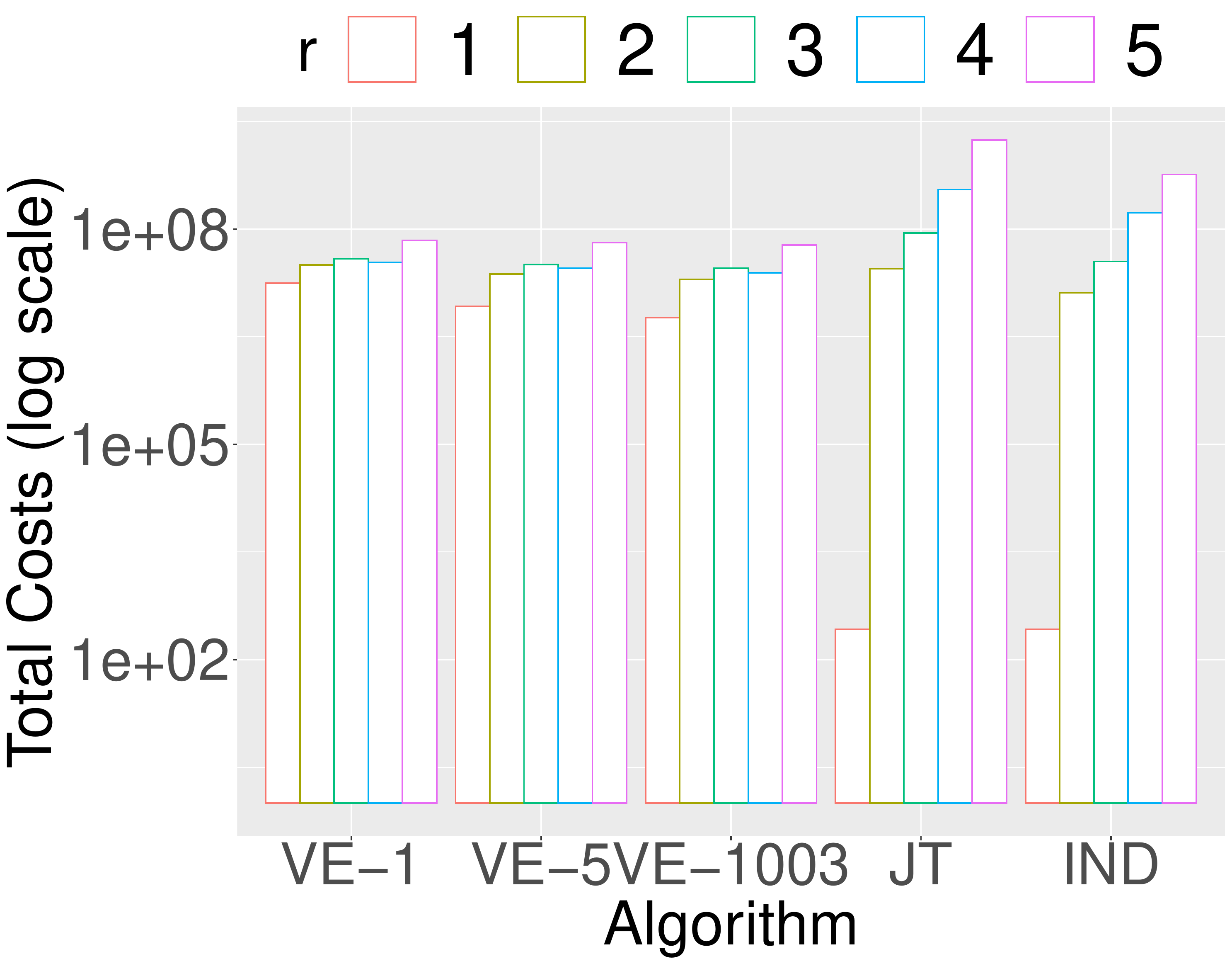} & 
    \includegraphics[width=.20\textwidth]{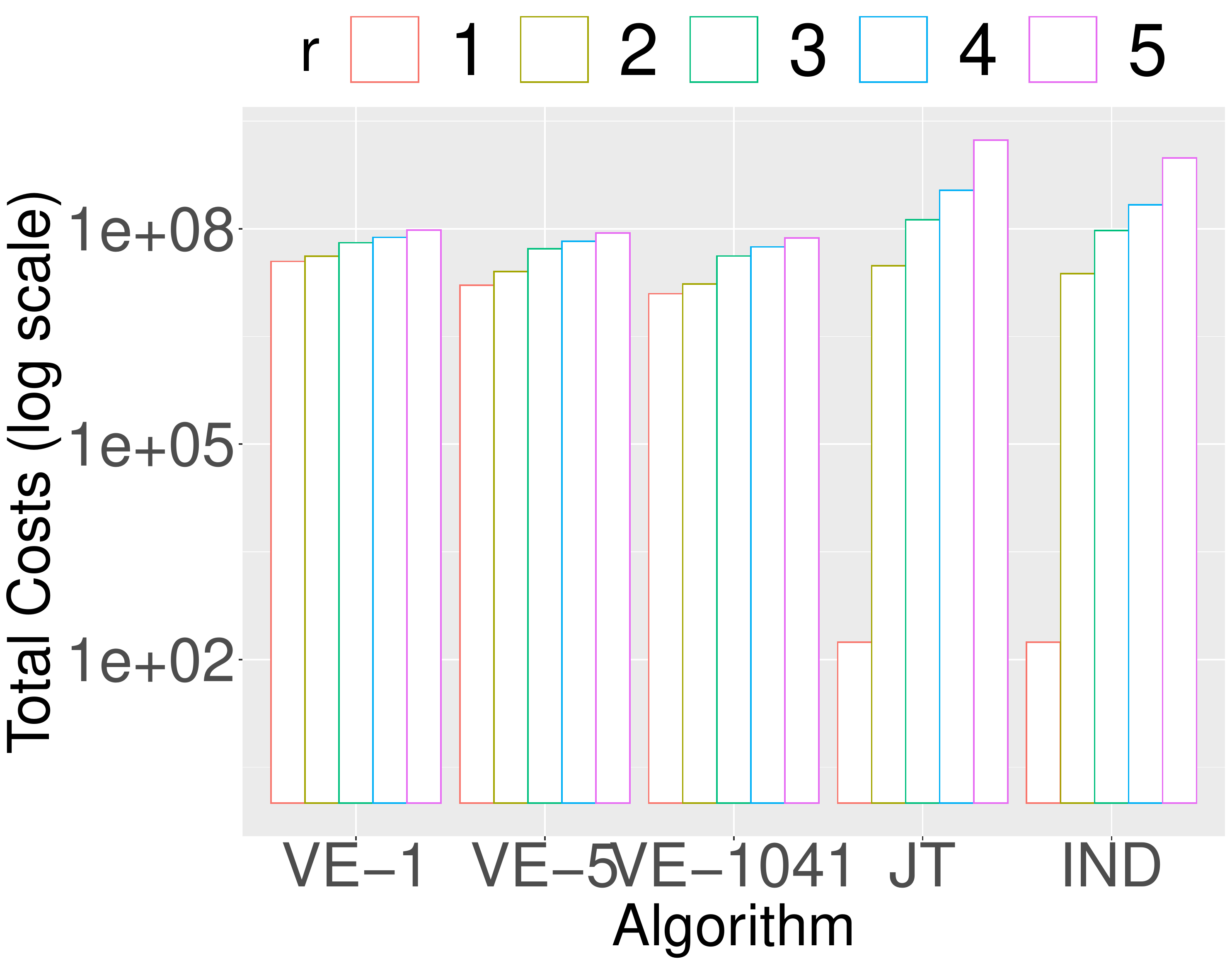} \\
    (e) \diabetes (\mf)  & (f) \link (\mf) & (g)  \muninm (\mf)   & (h) \muninb (\wmf)  \\ 
\end{tabular}
\caption{\label{fig:comp_unif_per_r}Total costs per query size $\qsize$ in uniform-workload scheme for different algorithms. }
\end{center}
\end{figure*}
}

\FullOnly{
\begin{figure*}[t]
\begin{tabular}{cccc}
    \includegraphics[width=.235\textwidth]{comparison_plots/mildew/uniform_per_r}&
    \includegraphics[width=.235\textwidth]{comparison_plots/pathfinder/uniform_per_r}&
    \includegraphics[width=.235\textwidth]{comparison_plots/munin1/uniform_per_r}&
    \includegraphics[width=.235\textwidth]{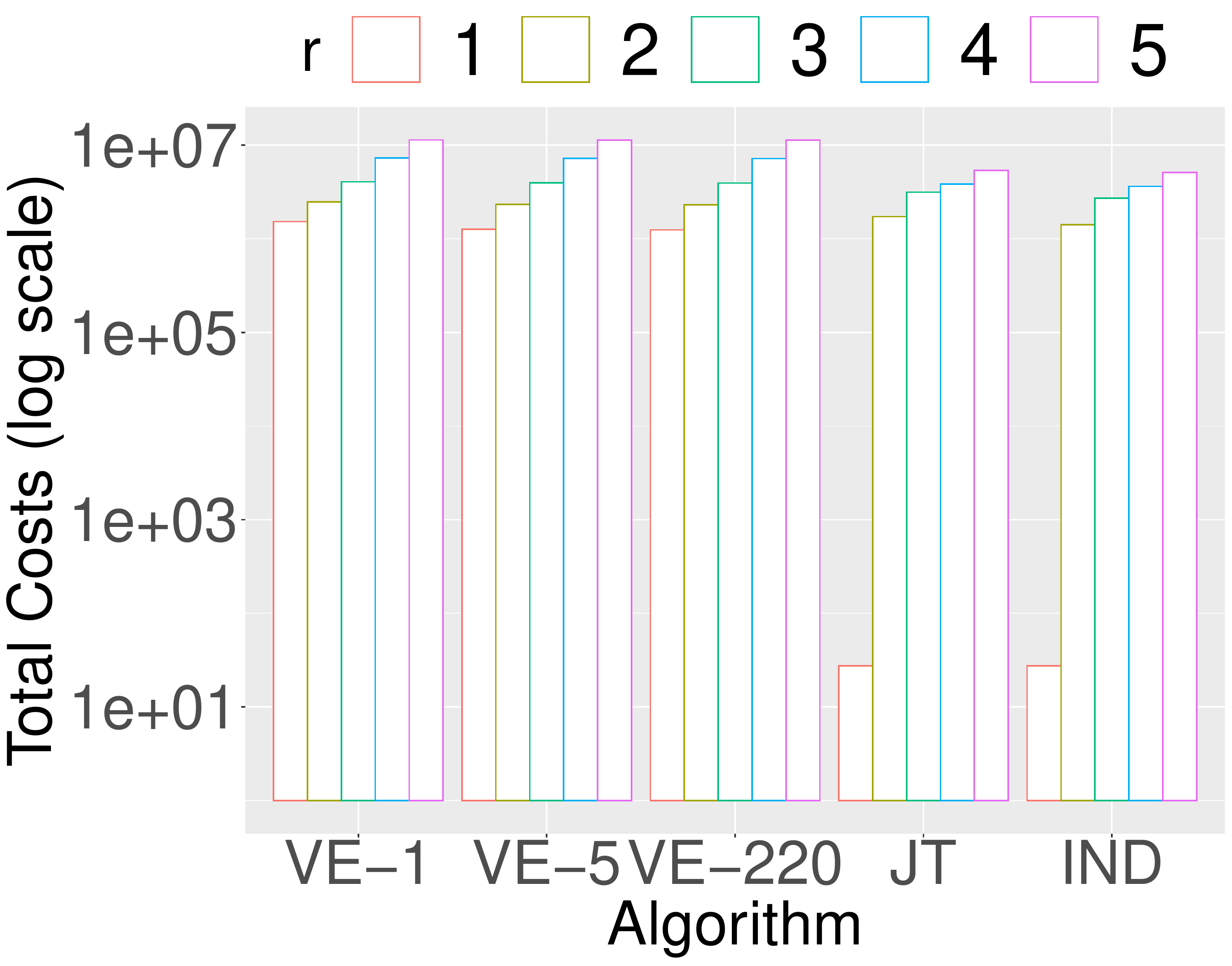}\\
  (a) \mildew (\mf) & (b) \bnpathfinder (\mf)  & (c) \munins (\wmf) & (d) \andes (\mf)   \\
  \includegraphics[width=.235\textwidth]{comparison_plots/diabetes/uniform_per_r} &
    \includegraphics[width=.235\textwidth]{comparison_plots/link/uniform_per_r} & 
    \includegraphics[width=.235\textwidth]{comparison_plots/munin2/uniform_per_r} & 
    \includegraphics[width=.235\textwidth]{comparison_plots/munin/uniform_per_r}
    \\
   (e) \diabetes (\mf)  & (f) \link (\mf) & (g)  \muninm (\mf)   & (h) \muninb (\wmf)  \\ 
   \includegraphics[width=.225\textwidth]{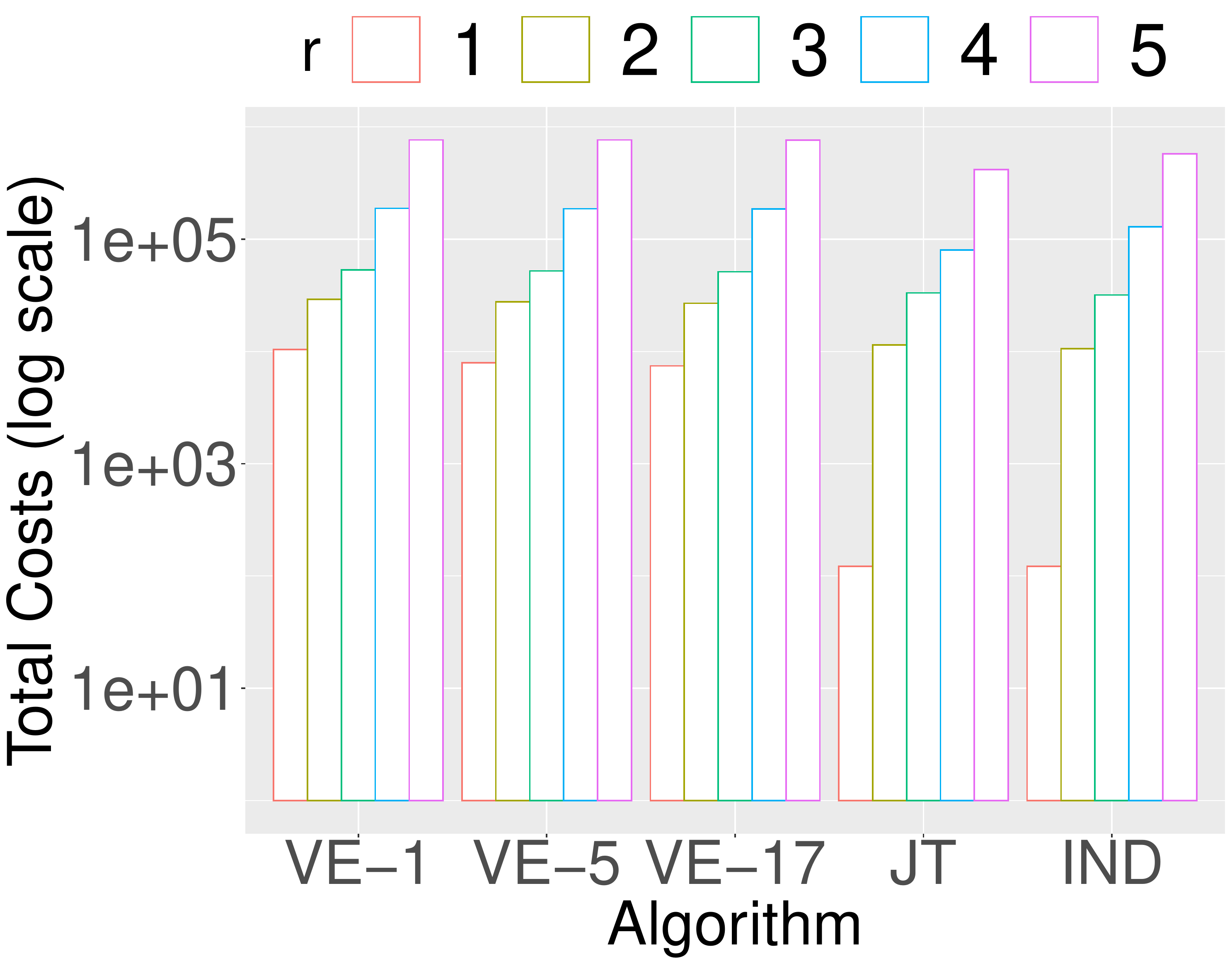} &
    \hspace{-0mm} \includegraphics[width=.225\textwidth]{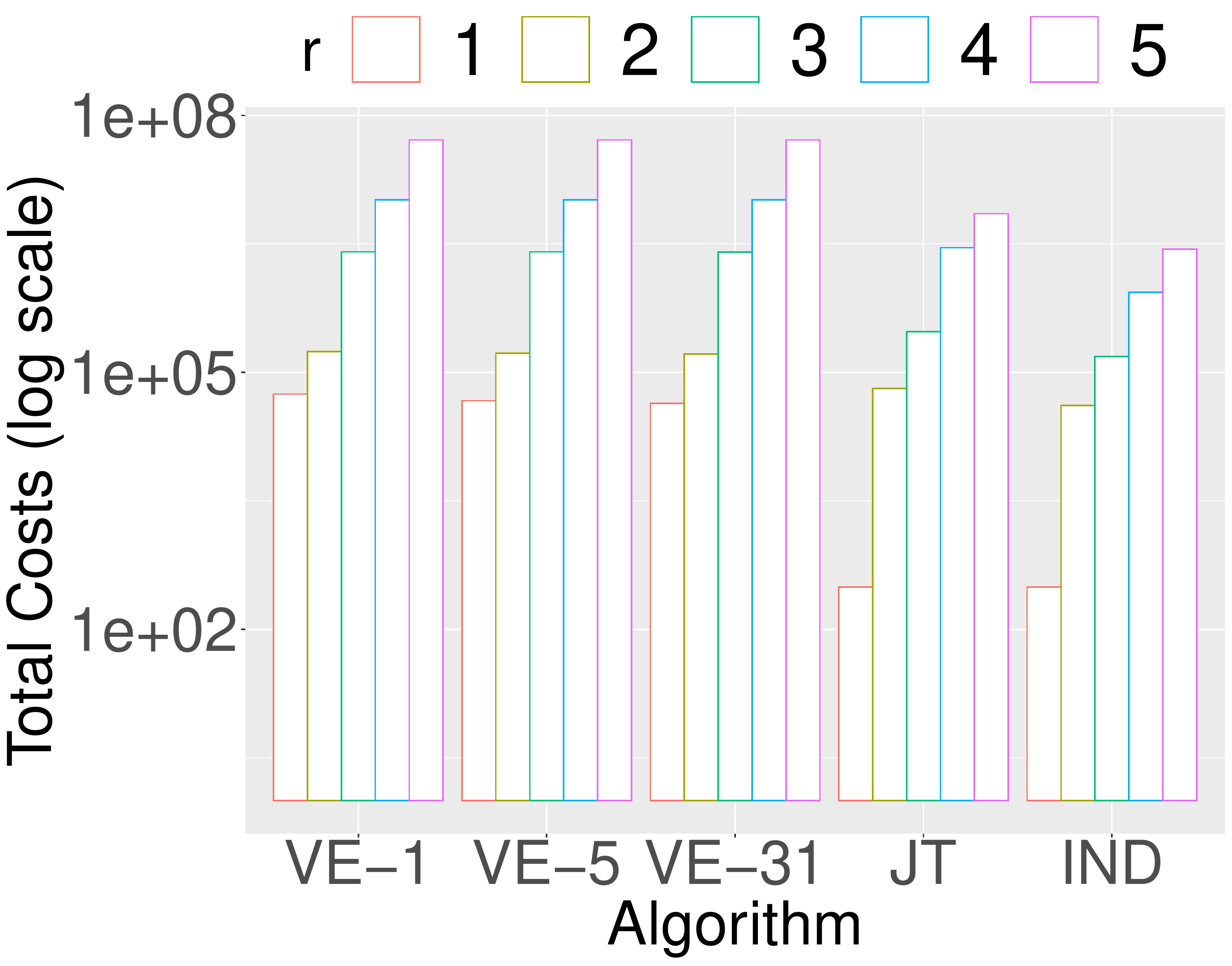} & 
    \hspace{-0mm} \includegraphics[width=.225\textwidth]{comparison_plots/tpch3/uniform_per_r} & 
    \hspace{-0mm} \includegraphics[width=.225\textwidth]{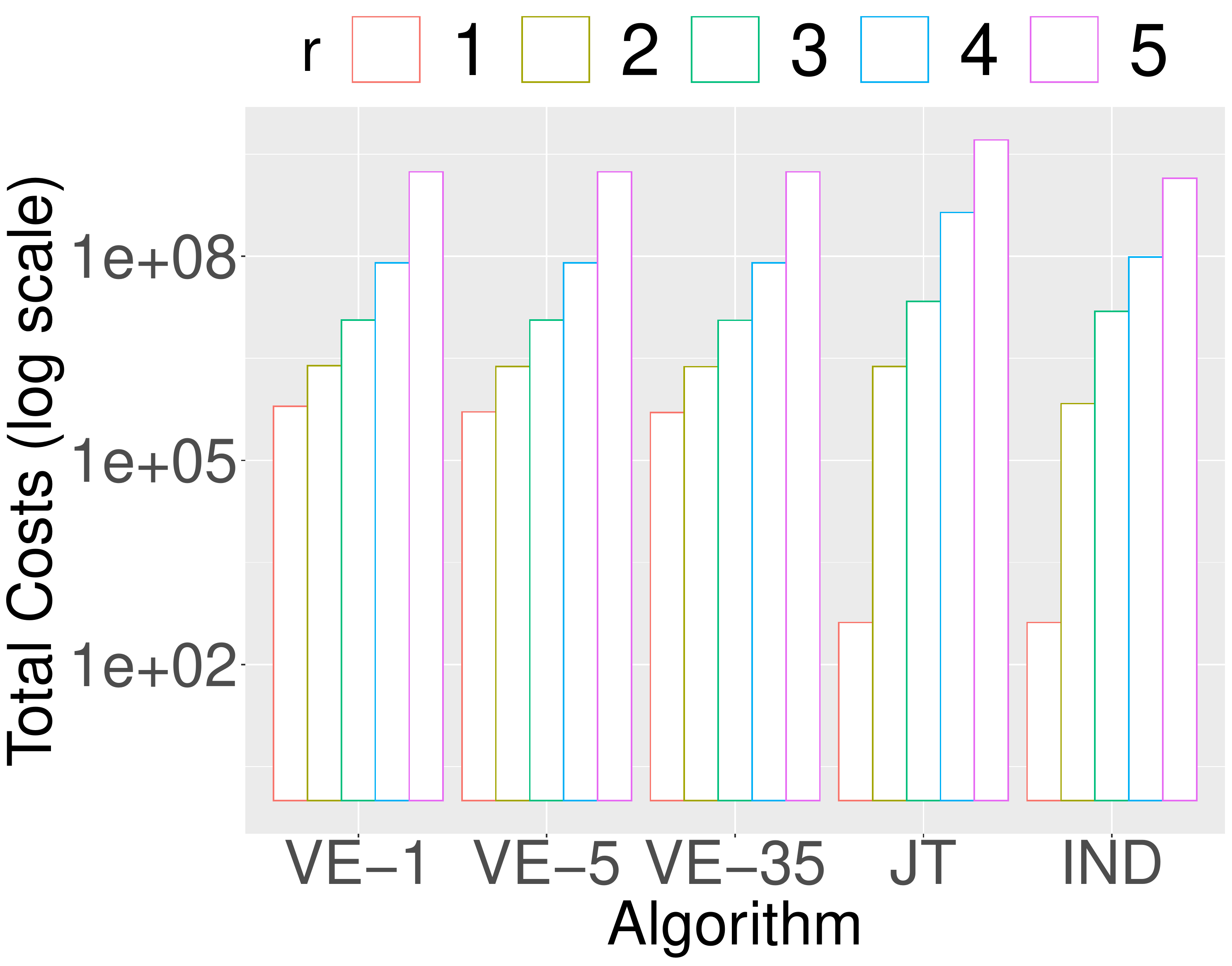}
    \\
   \revision{(i) \tpchsmall (\mw)}  & \revision{(j) \tpchmedium (\mw)} & \revision{(k)  \tpchverylarge (\mw)}   & \revision{(l) \tpchlarge (\mw)}
\end{tabular}
\caption{\label{fig:comp_unif_per_r} Total costs per query size $\qsize$ in uniform-workload scheme for different algorithms. }
\end{figure*}
}

\FullOnly{
\begin{figure*}[t]
\begin{tabular}{cccc}
    \includegraphics[width=.235\textwidth]{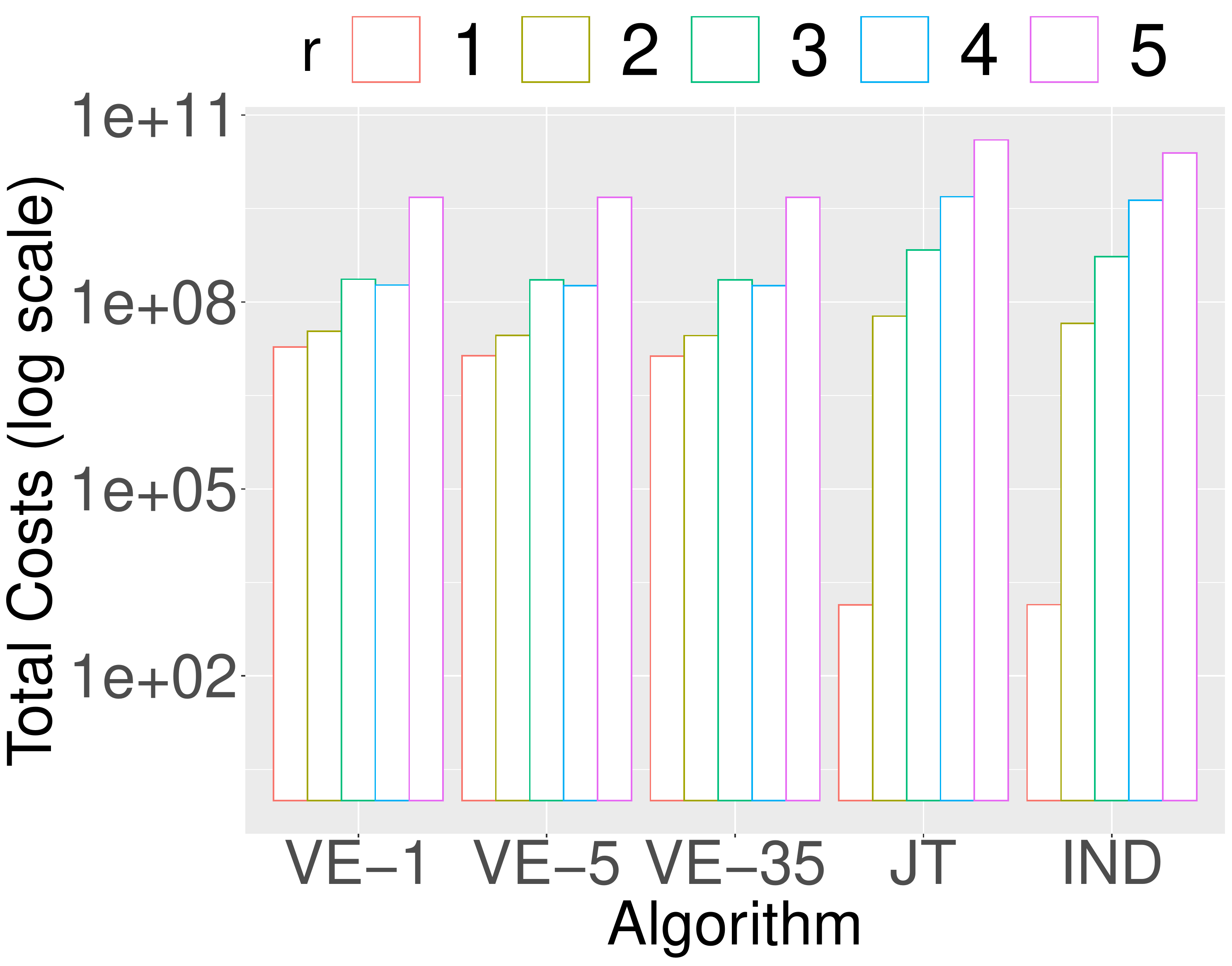}&
     \hspace{-0mm} \includegraphics[width=.235\textwidth]{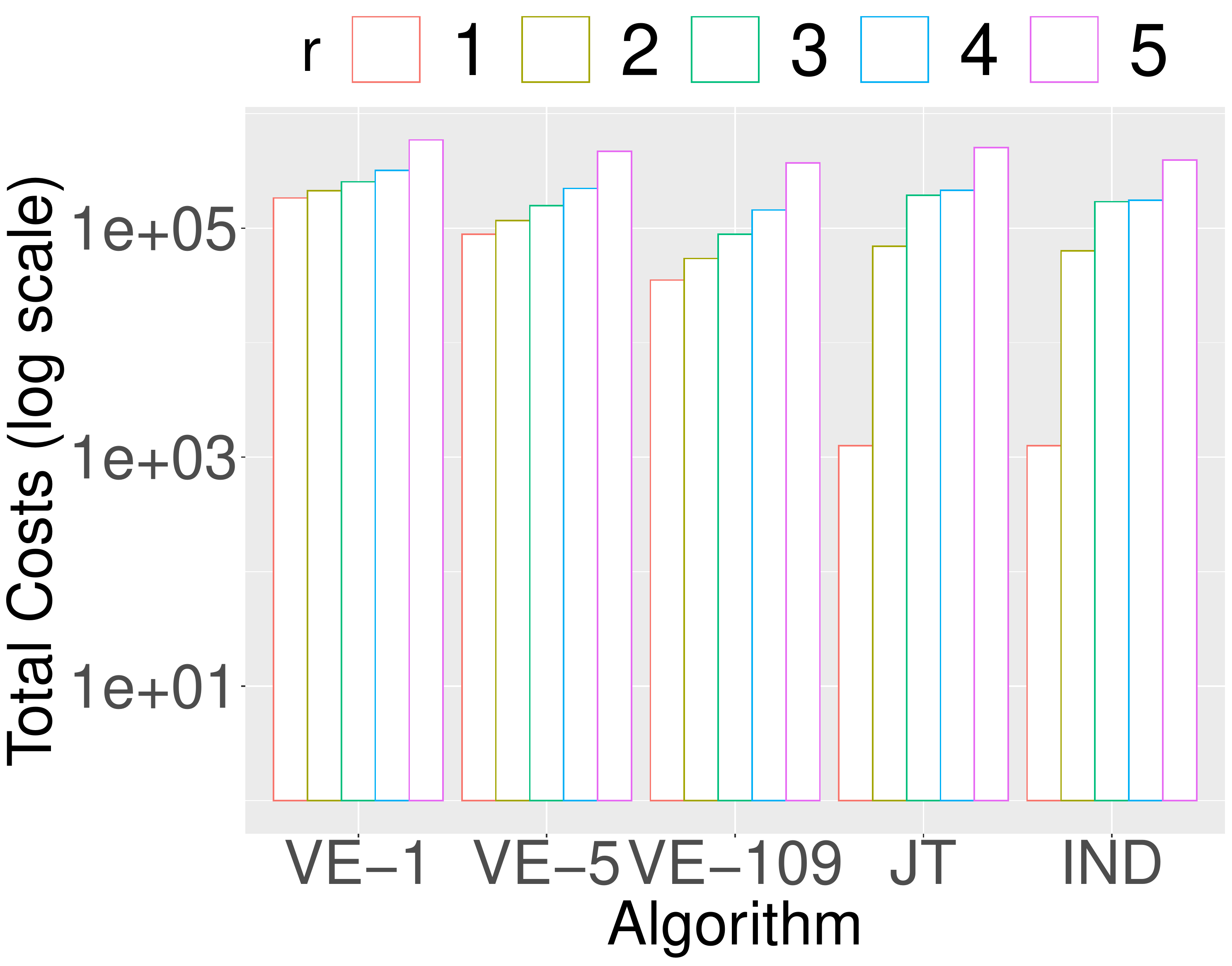}&
    \hspace{-0mm}\includegraphics[width=.235\textwidth]{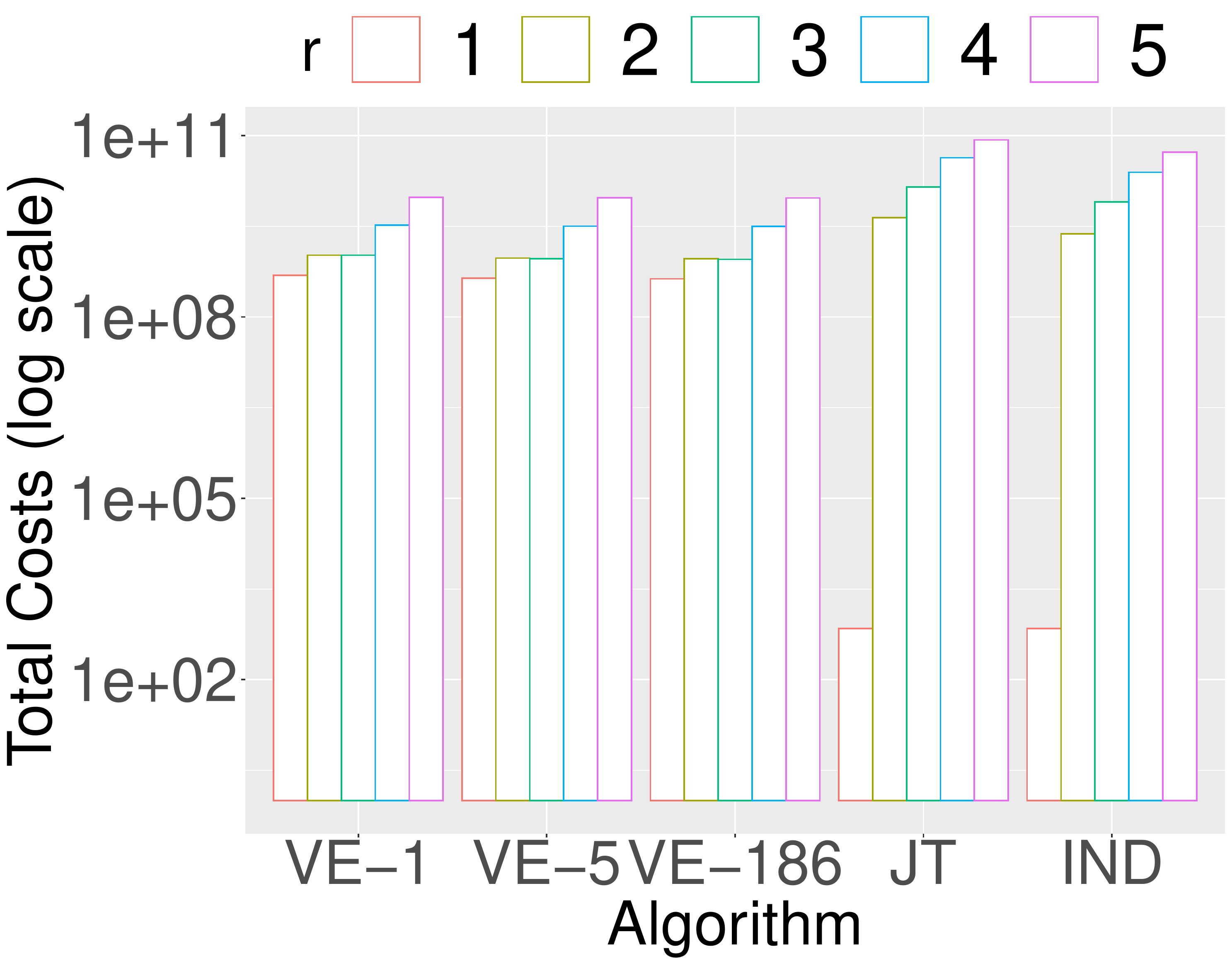}&
    \hspace{-0mm}\includegraphics[width=.235\textwidth]{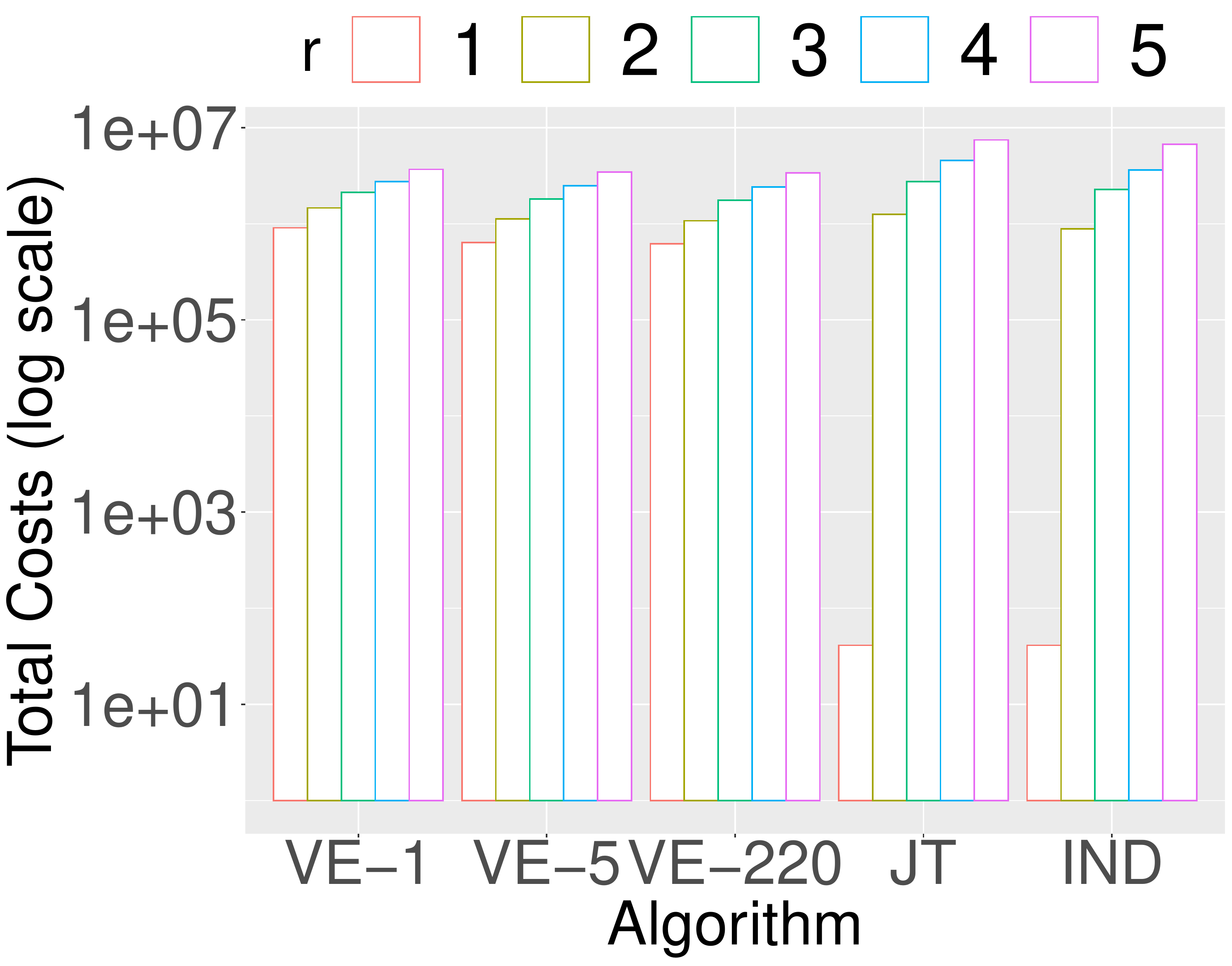}\\
  (a) \mildew (\mf) & (b) \bnpathfinder (\mf)  & (c) \munins (\wmf) & (d) \andes (\mf)   \\
     \includegraphics[width=.235\textwidth]{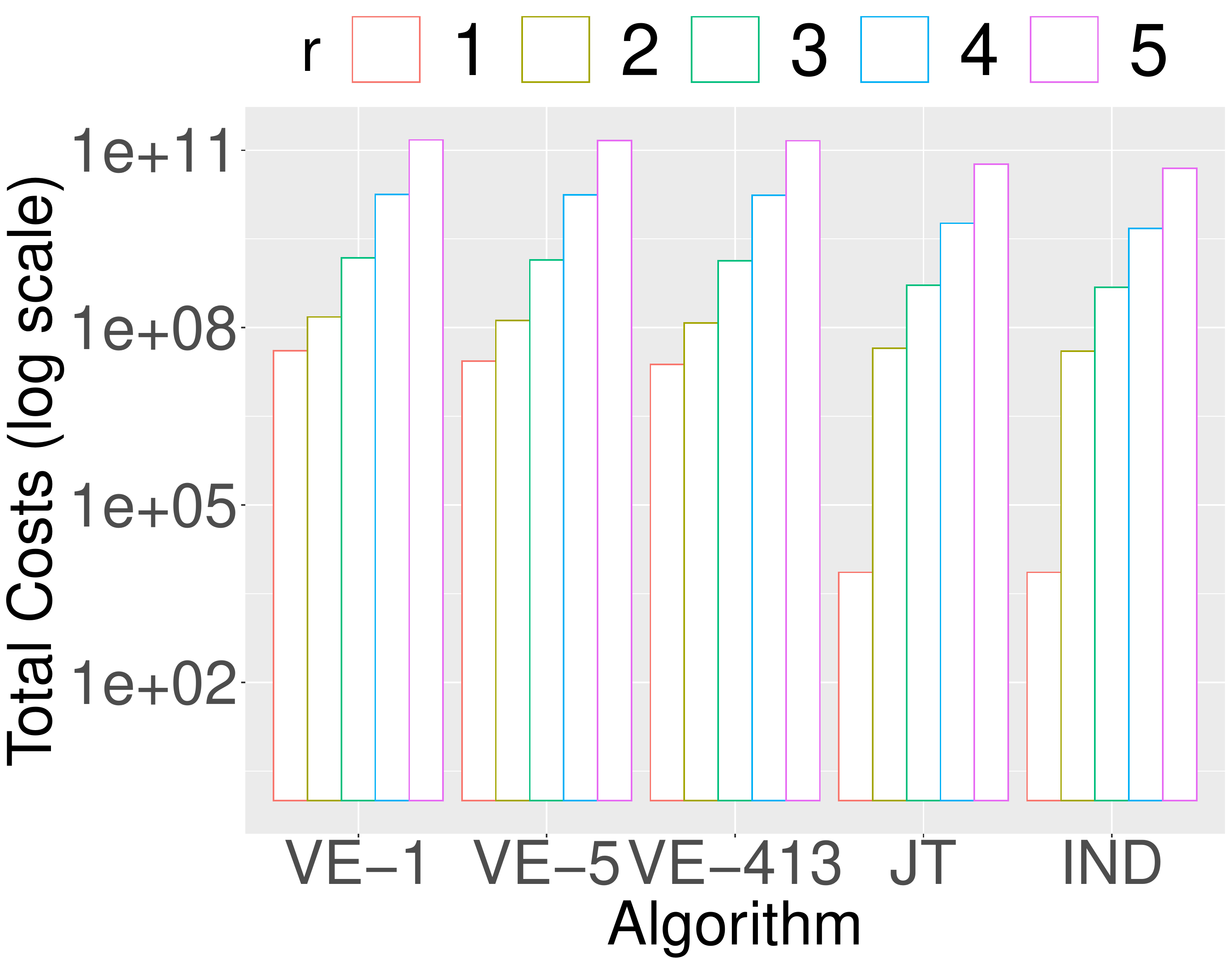} &
    \hspace{-0mm}\includegraphics[width=.235\textwidth]{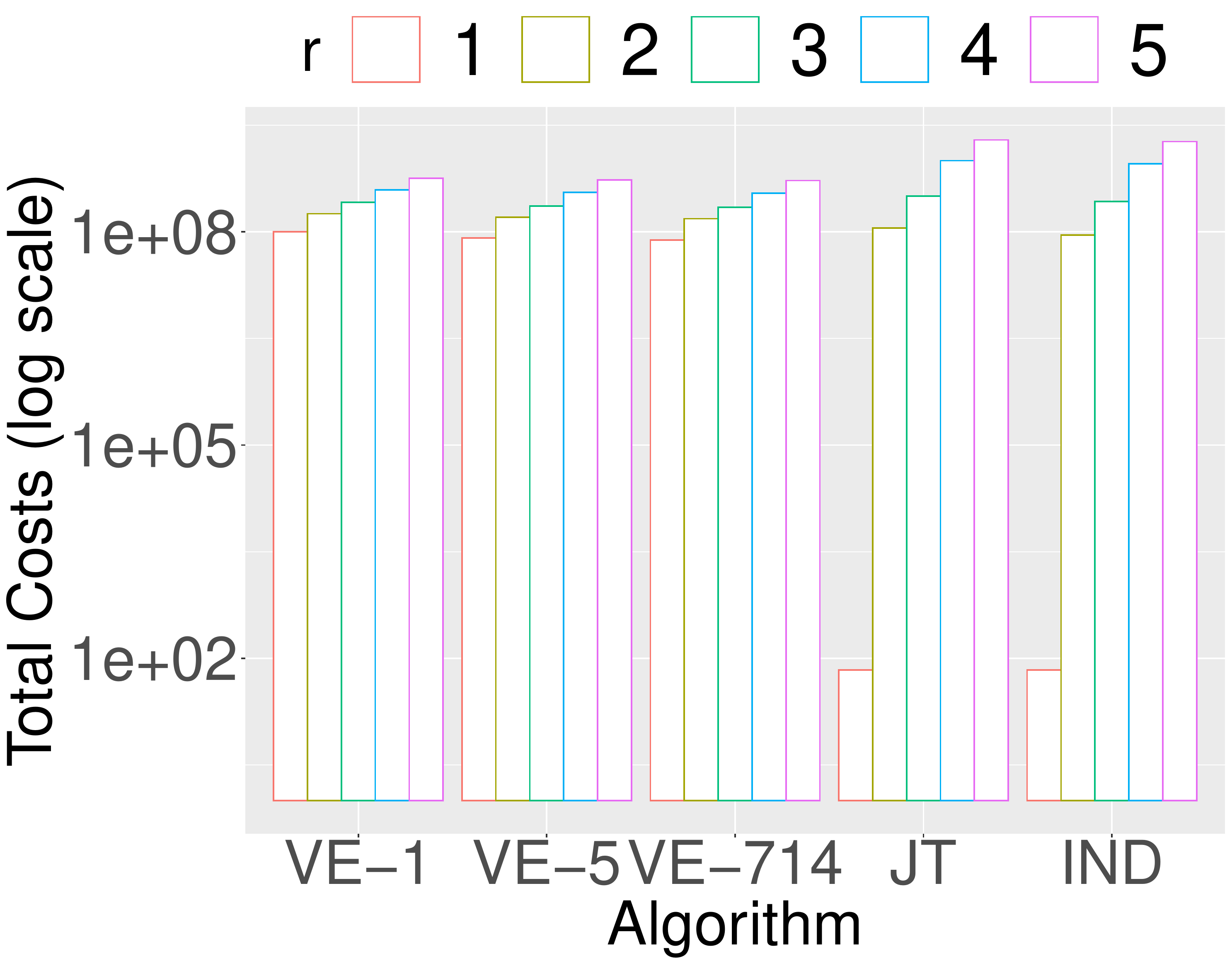} & 
    \hspace{-0mm}\includegraphics[width=.235\textwidth]{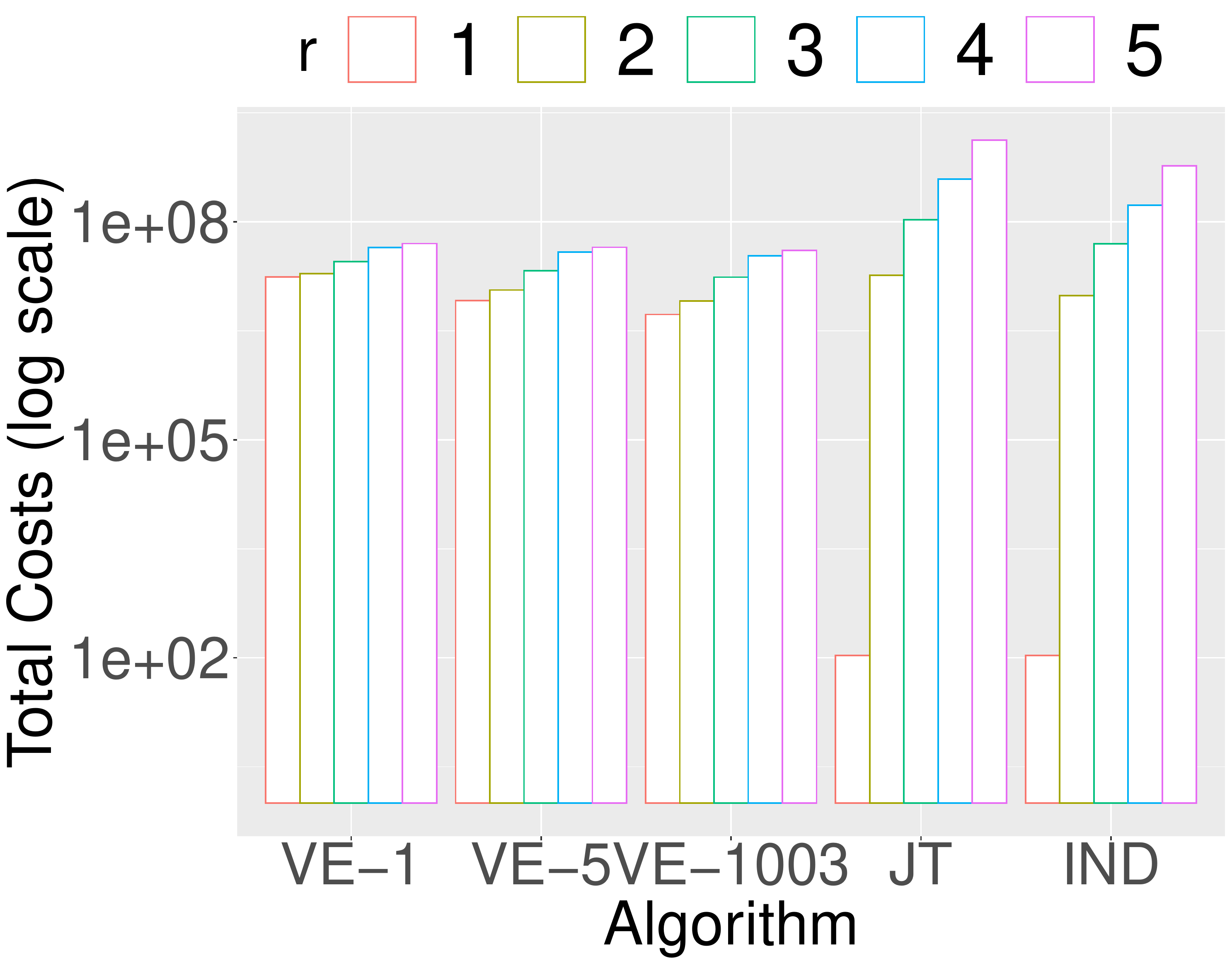} & 
    \hspace{-0mm}\includegraphics[width=.235\textwidth]{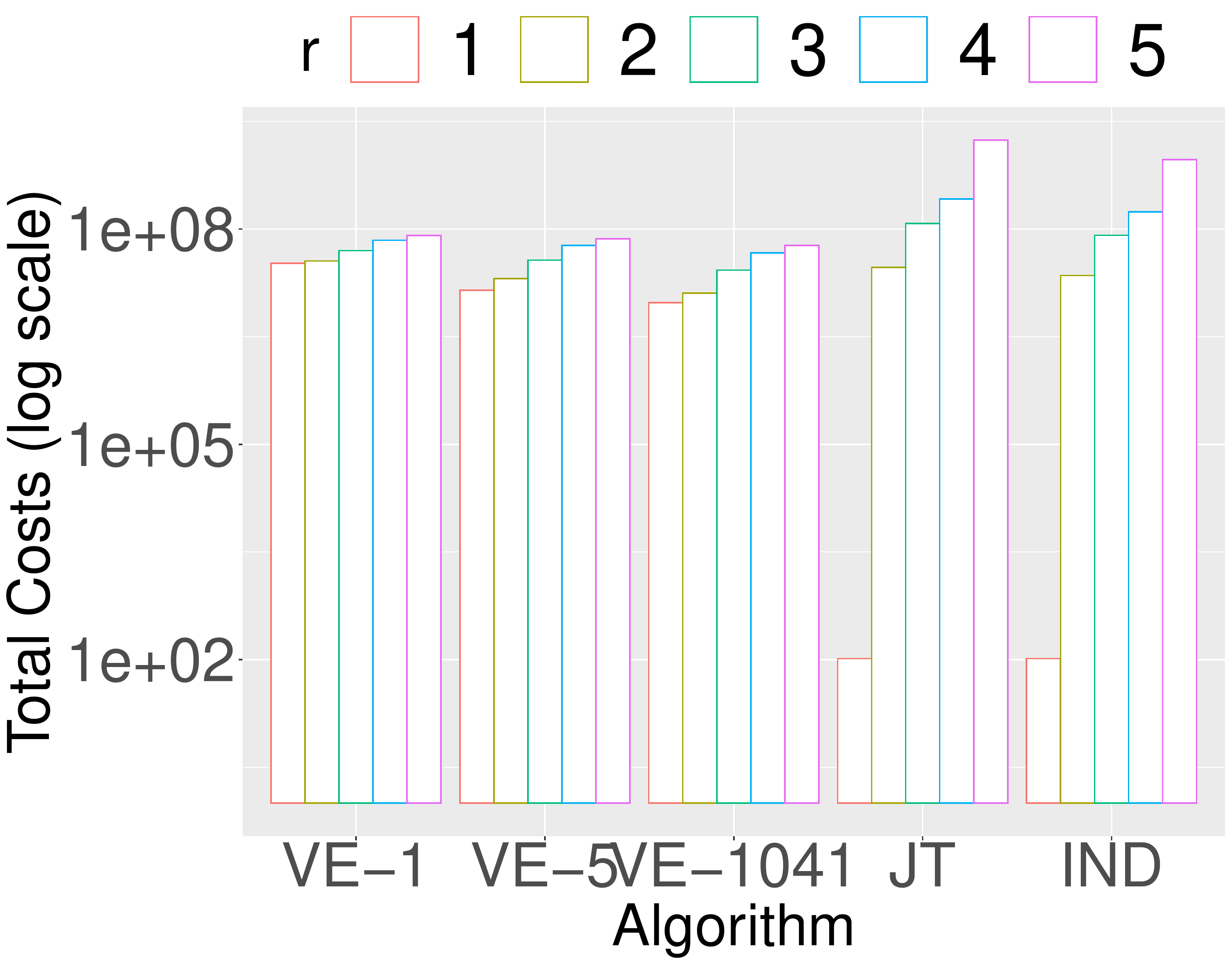}
    \\
   (e) \diabetes (\mf)  & (f) \link (\mf) & (g)  \muninm (\mf)   & (h) \muninb (\wmf)  \\ 
    \includegraphics[width=.225\textwidth]{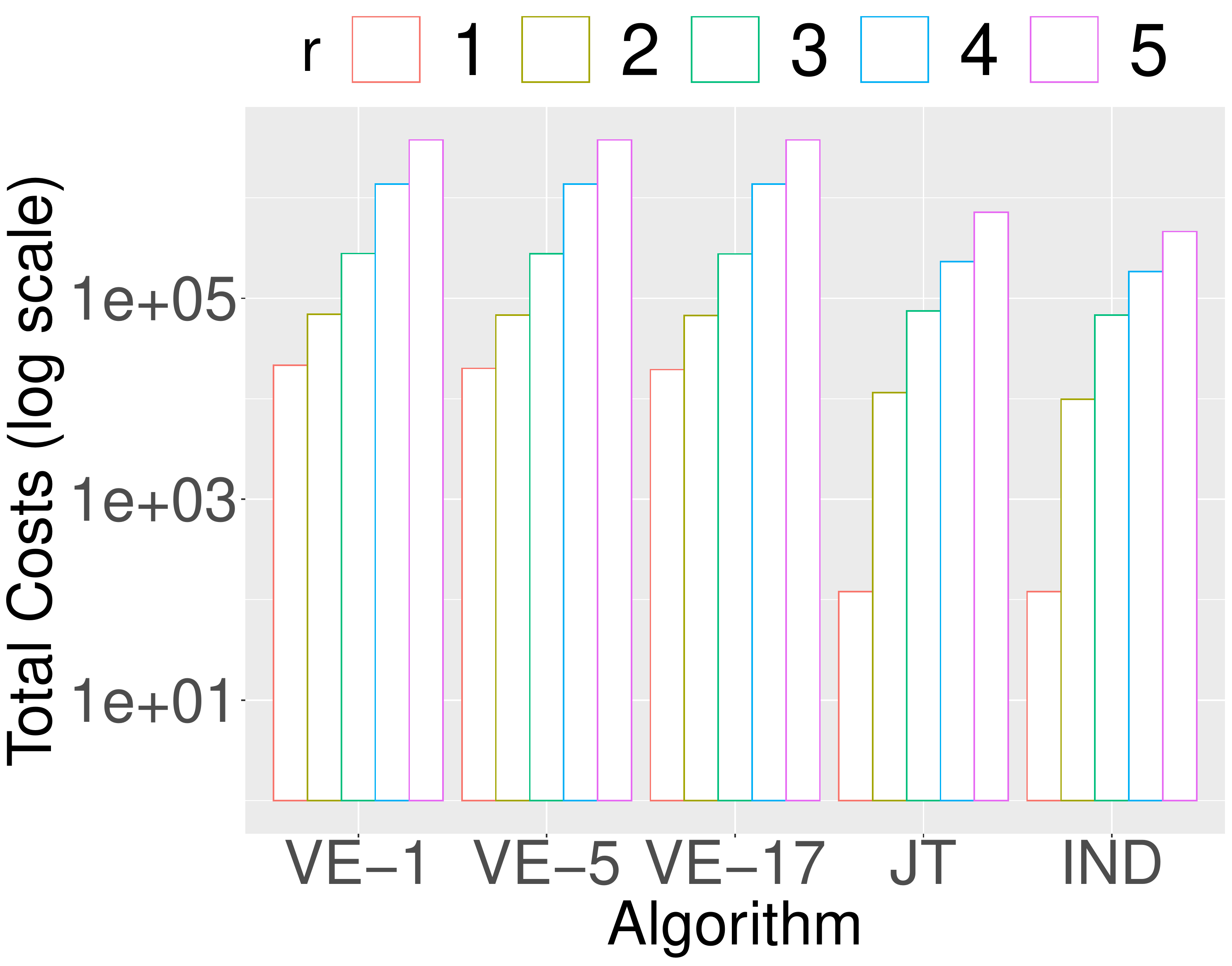} &
   \hspace{-0mm} \includegraphics[width=.225\textwidth]{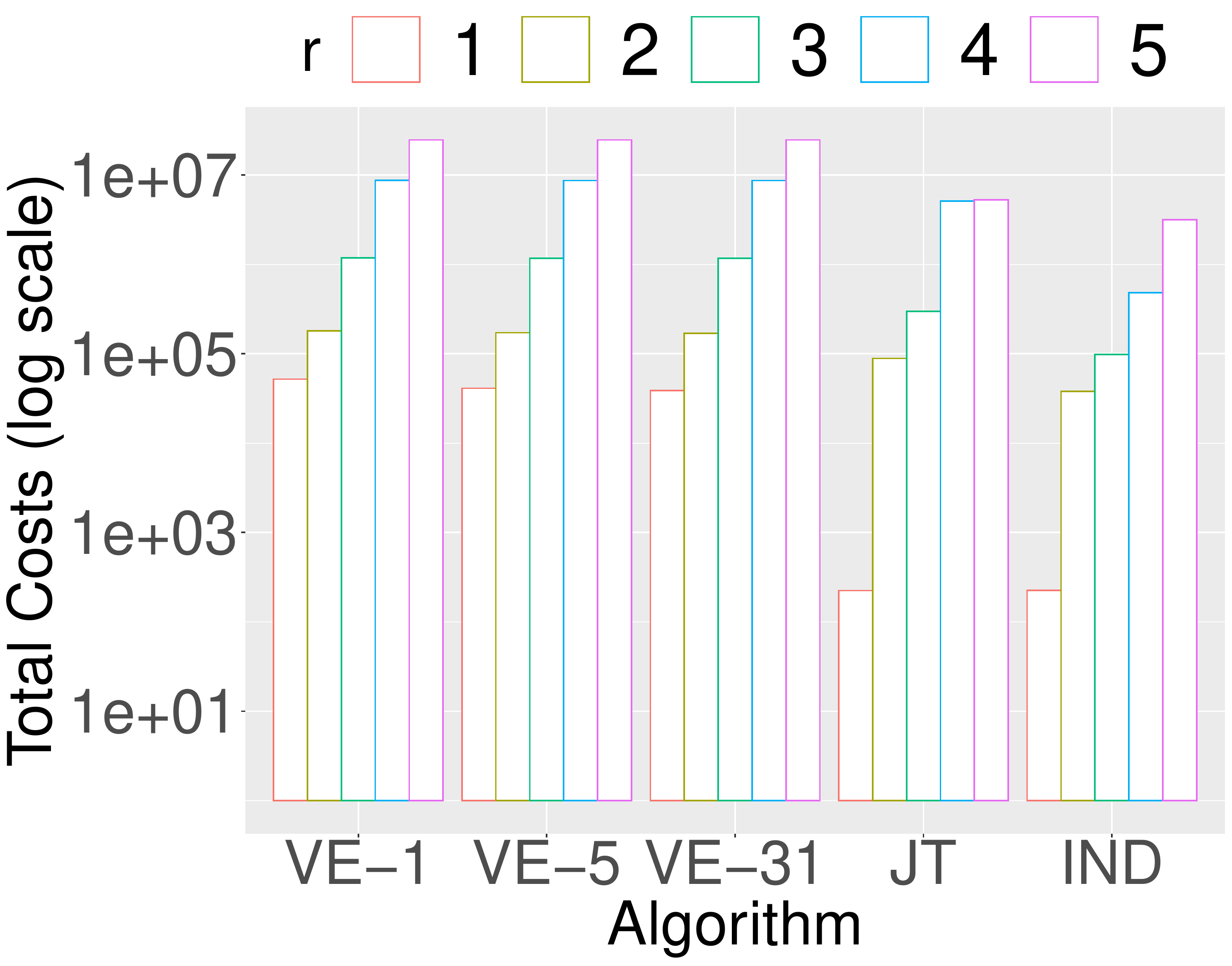} & 
   \hspace{-0mm} \includegraphics[width=.225\textwidth]{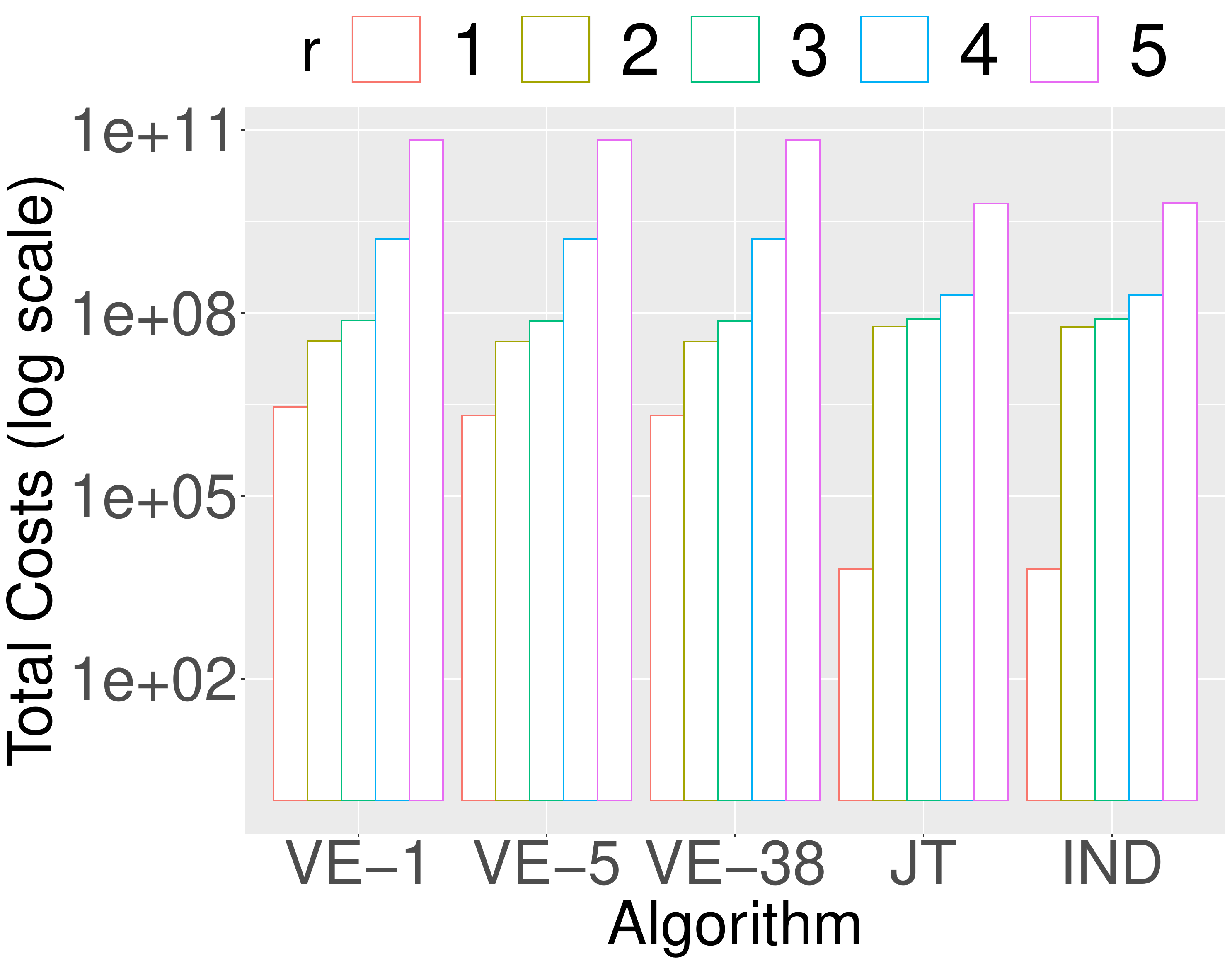} & 
   \hspace{-0mm} \includegraphics[width=.225\textwidth]{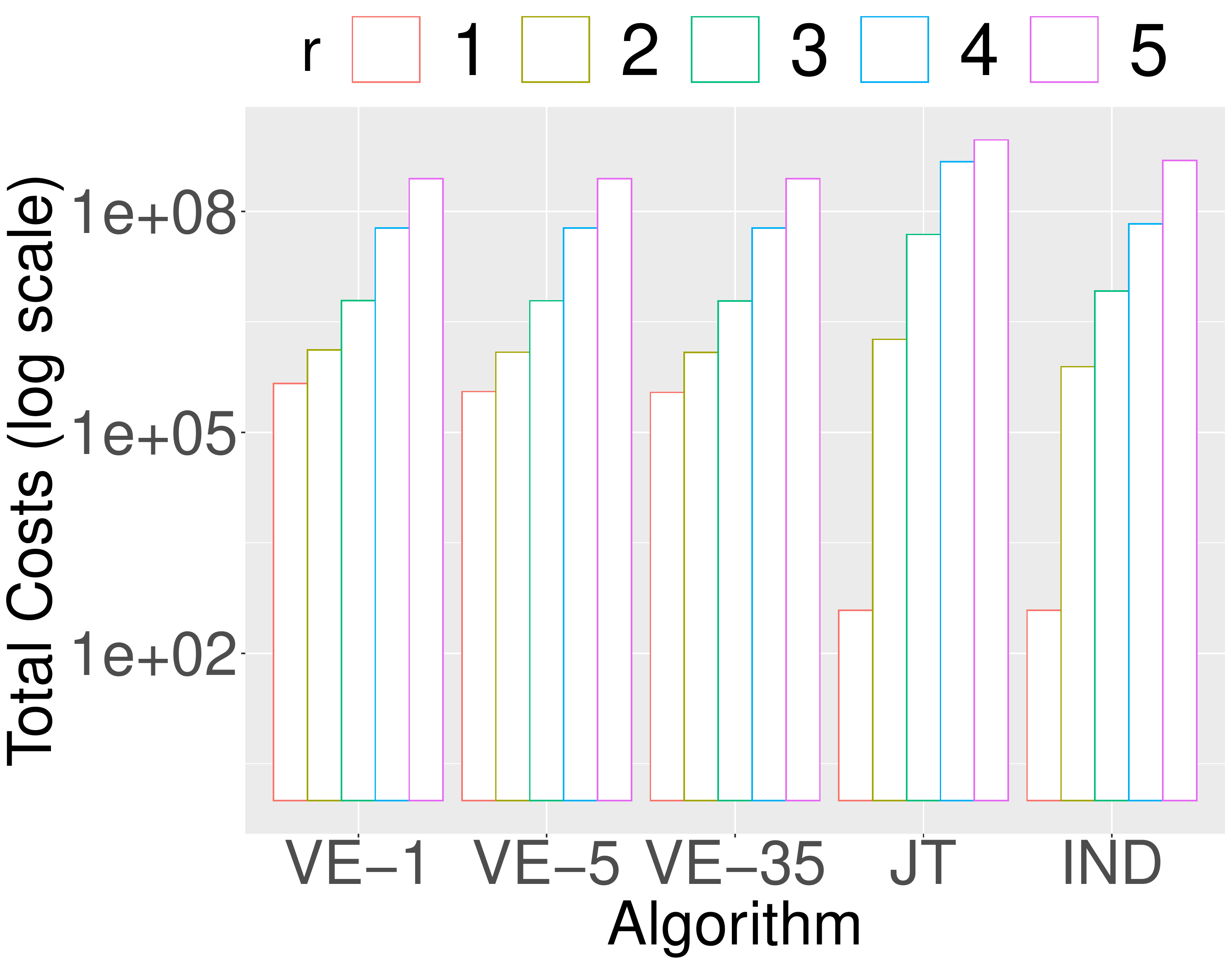}
   \\
   \revision{(i) \tpchsmall (\mw)}  & \revision{(j) \tpchmedium (\mw)} & \revision{(k)  \tpchverylarge (\mw)}   & \revision{(l) \tpchlarge (\mw)}
\end{tabular}
\caption{\label{fig:comp_biasr_per_r} Total costs per query size $\qsize$ in skewed-workload scheme for different algorithms.}
\end{figure*}
}

\ReviewOnly{
\begin{figure*}[t]
\begin{center}
\begin{tabular}{cccc}
    \includegraphics[width=.20\textwidth]{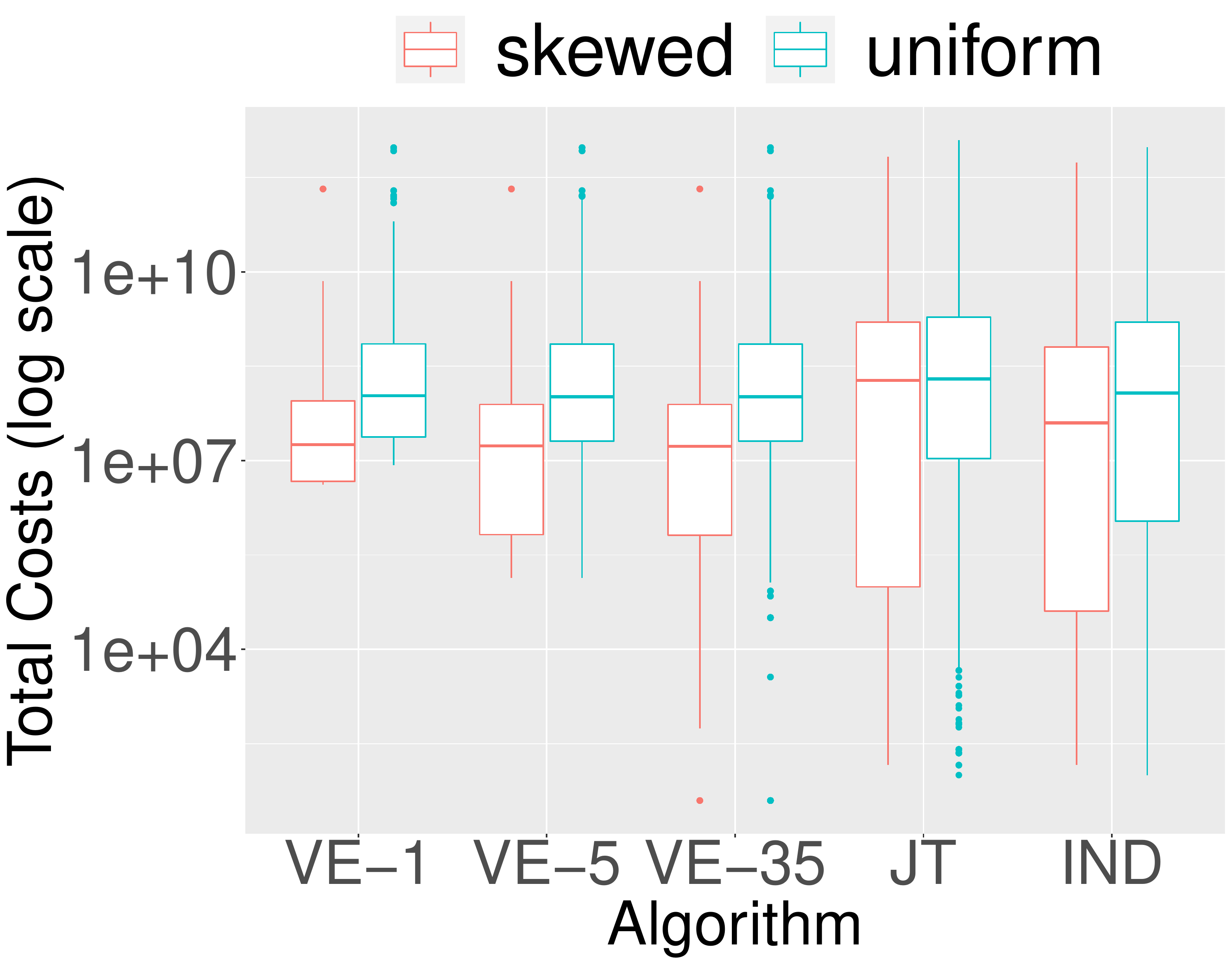}&
    \hspace{-0mm} \includegraphics[width=.20\textwidth]{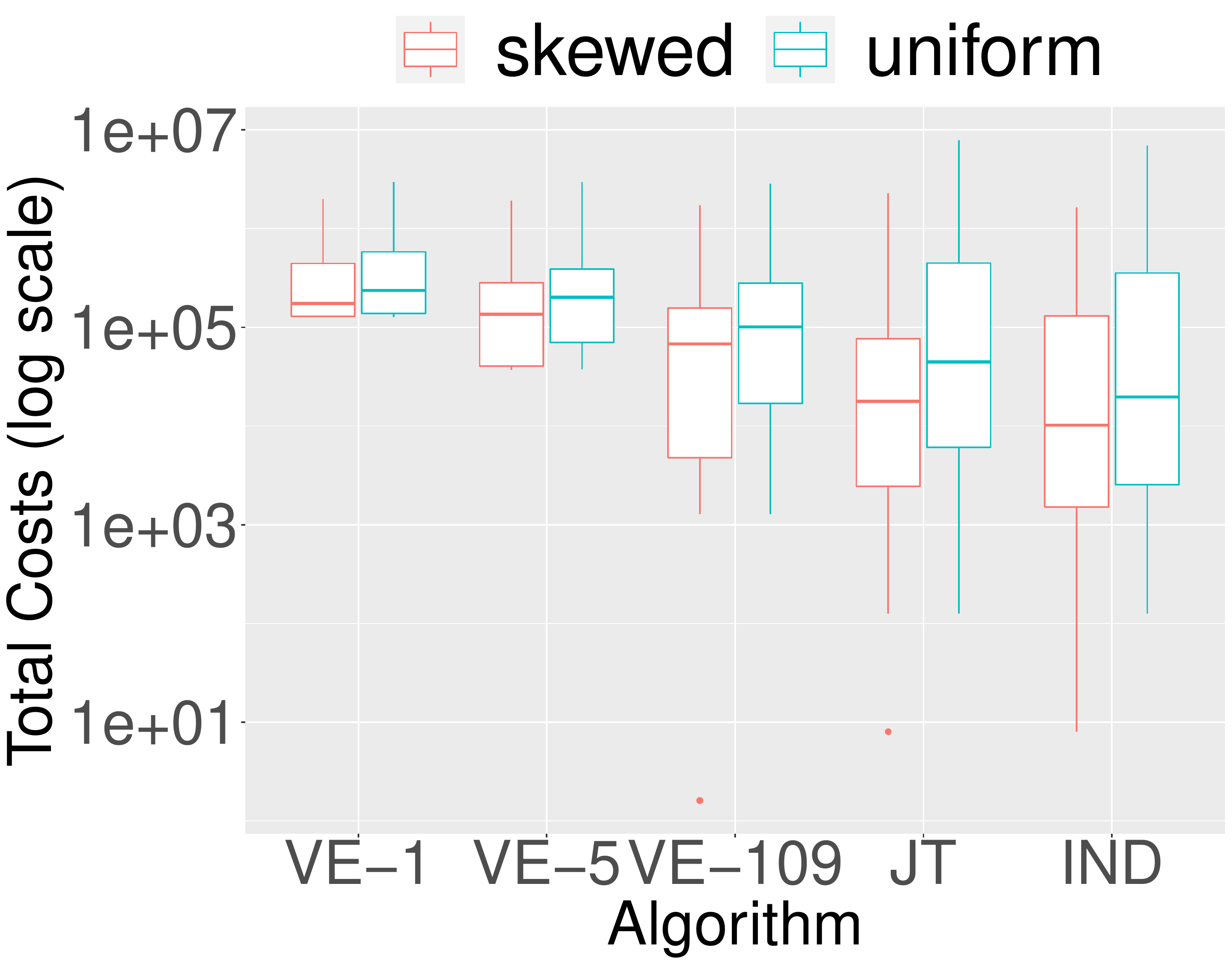}&
    \hspace{-0mm} \includegraphics[width=.20\textwidth]{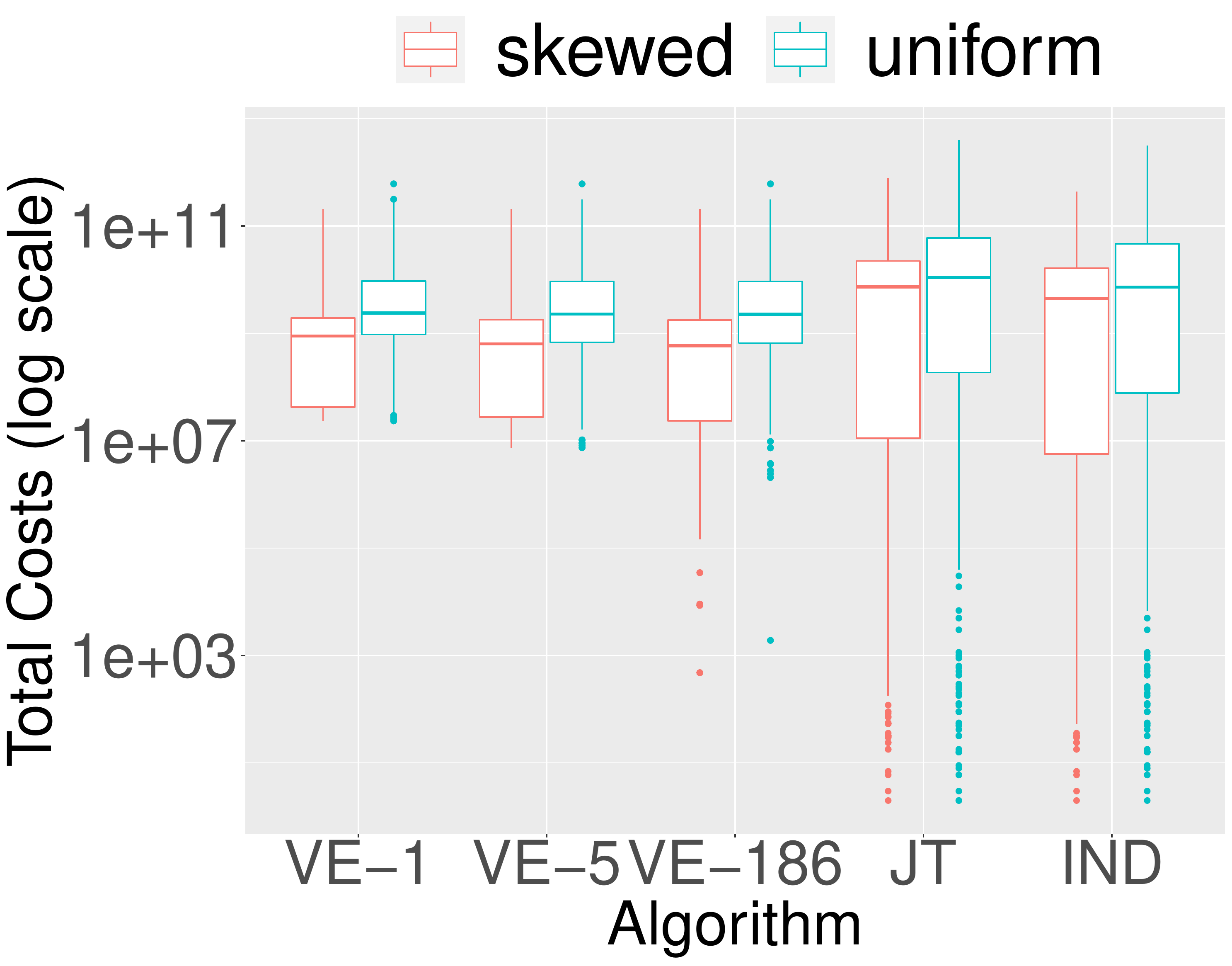}&
    \hspace{-0mm} \includegraphics[width=.20\textwidth]{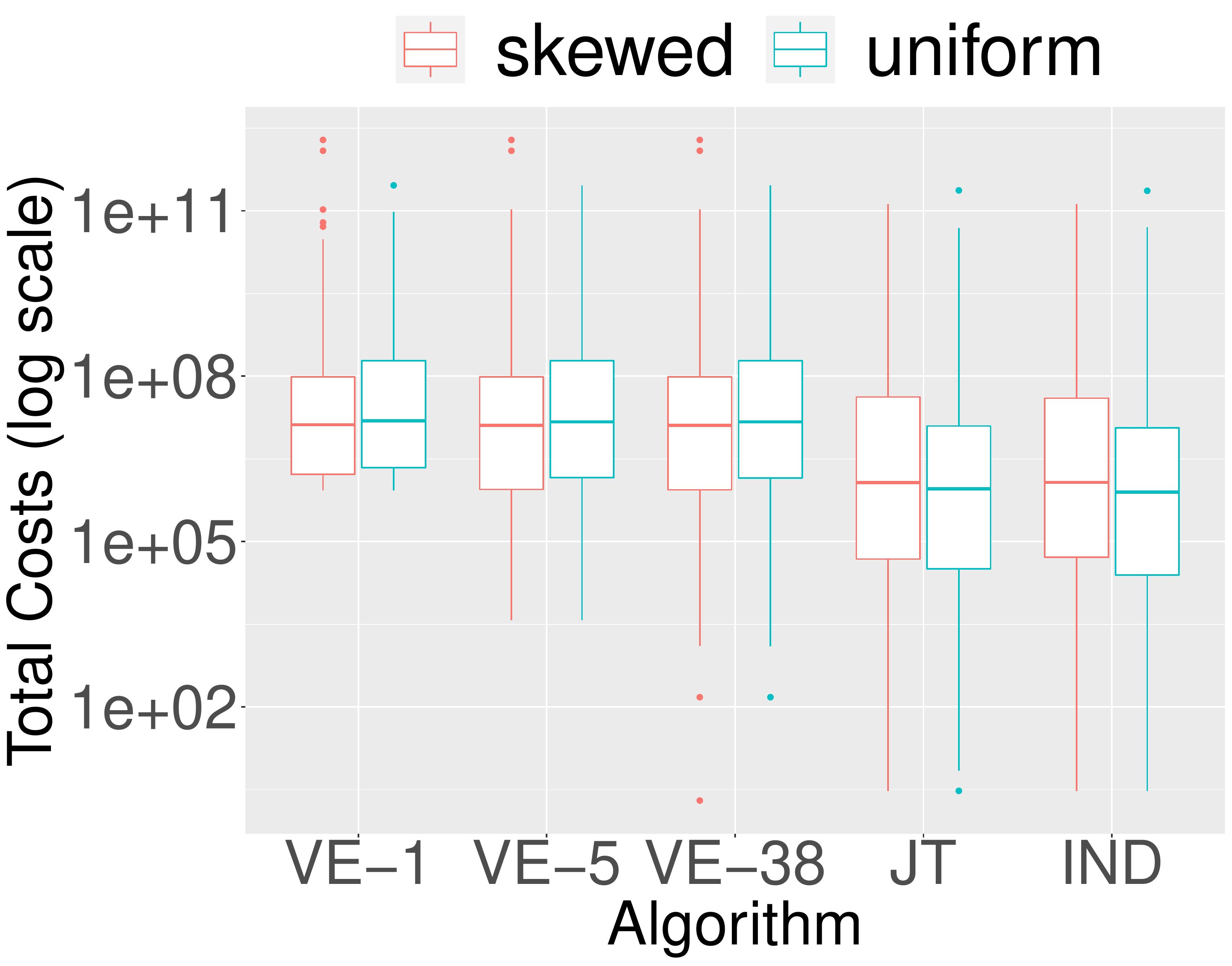}\\
  (a) \mildew (\mf) & (b) \bnpathfinder (\mf)  & (c) \munins (\wmf) & \revision{(d) \tpchverylarge (\mw)}   \\
  \includegraphics[width=.20\textwidth]{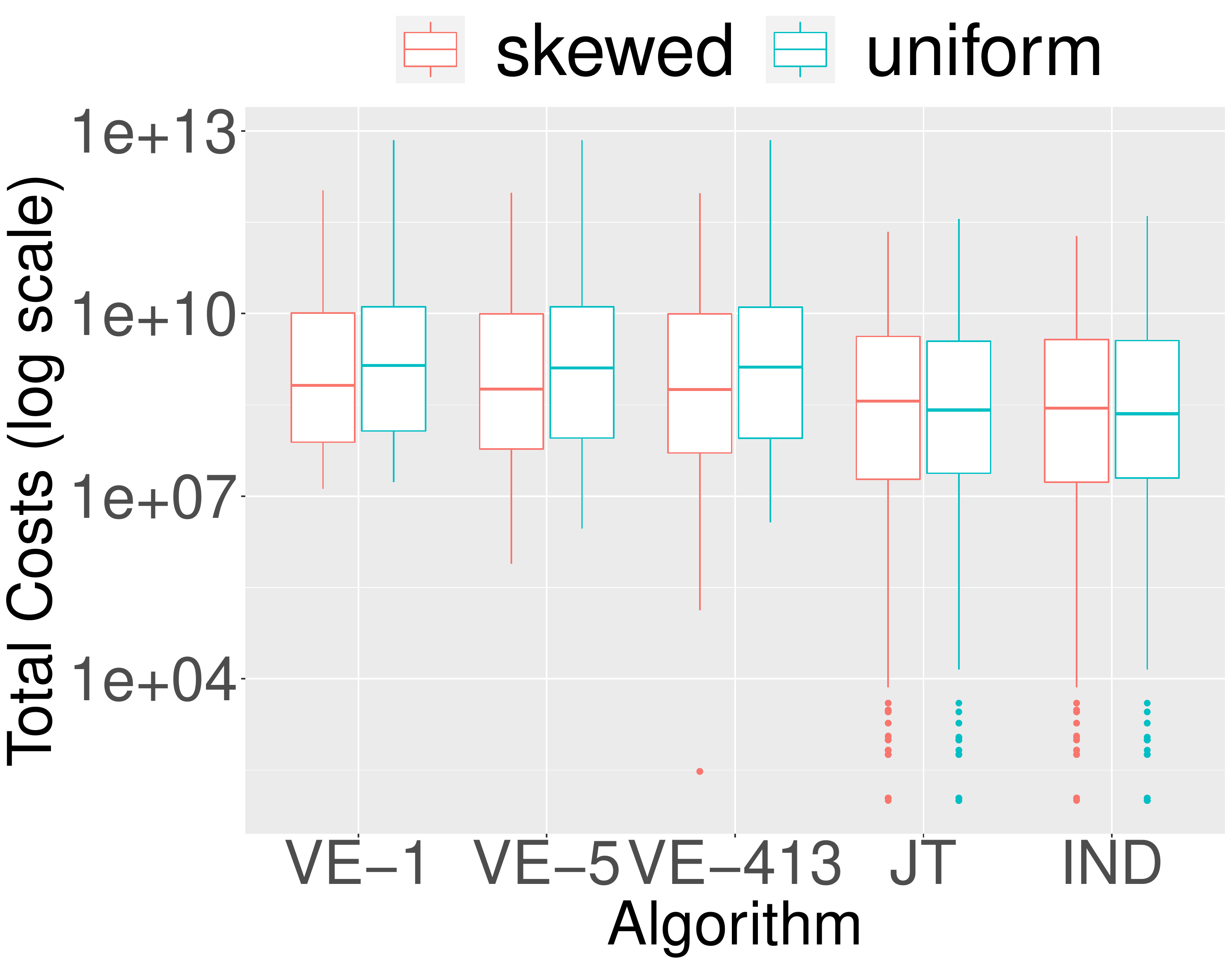} &
    \hspace{-0mm} \includegraphics[width=.20\textwidth]{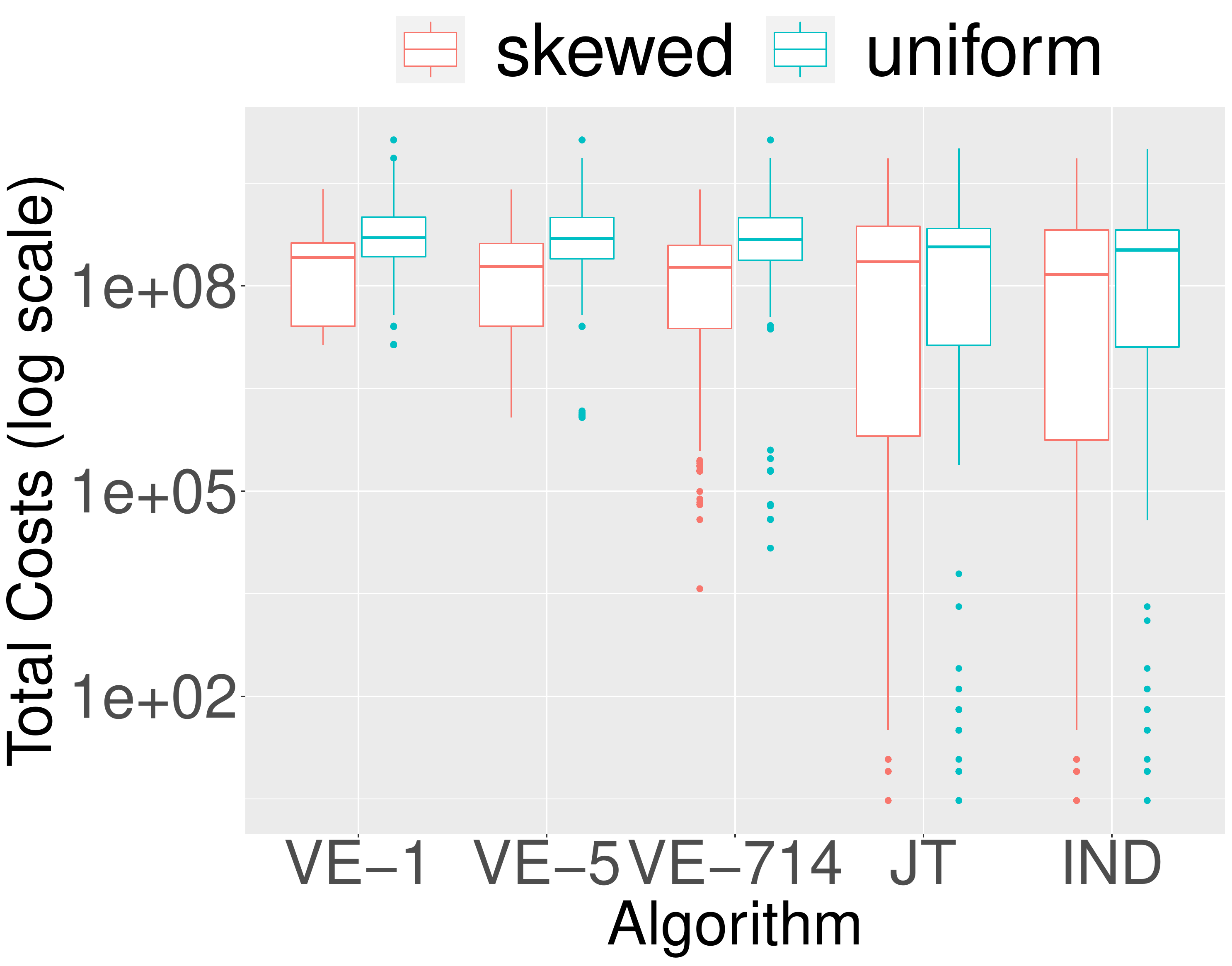} & 
    \hspace{-0mm} \includegraphics[width=.20\textwidth]{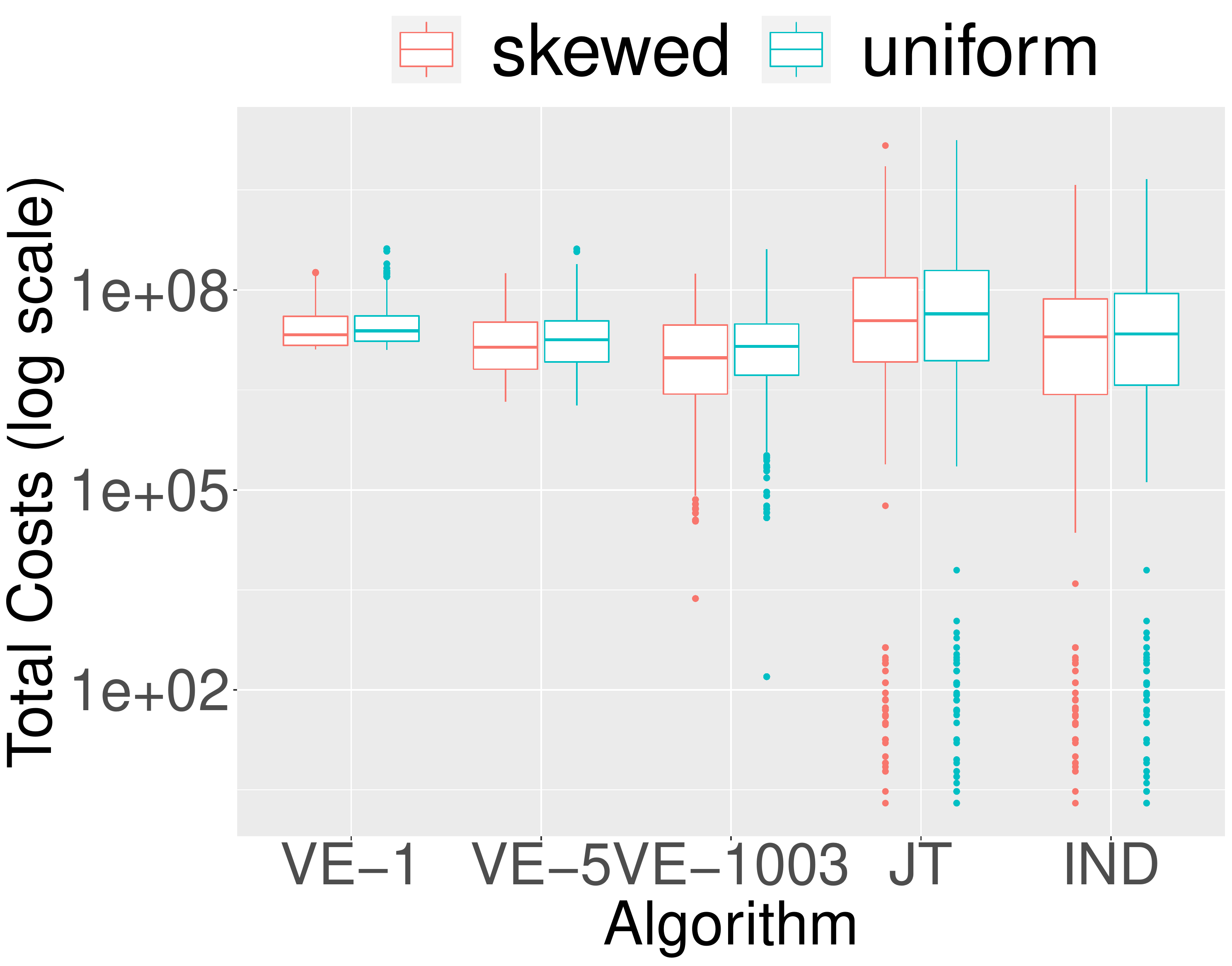} & 
    \hspace{-0mm} \includegraphics[width=.20\textwidth]{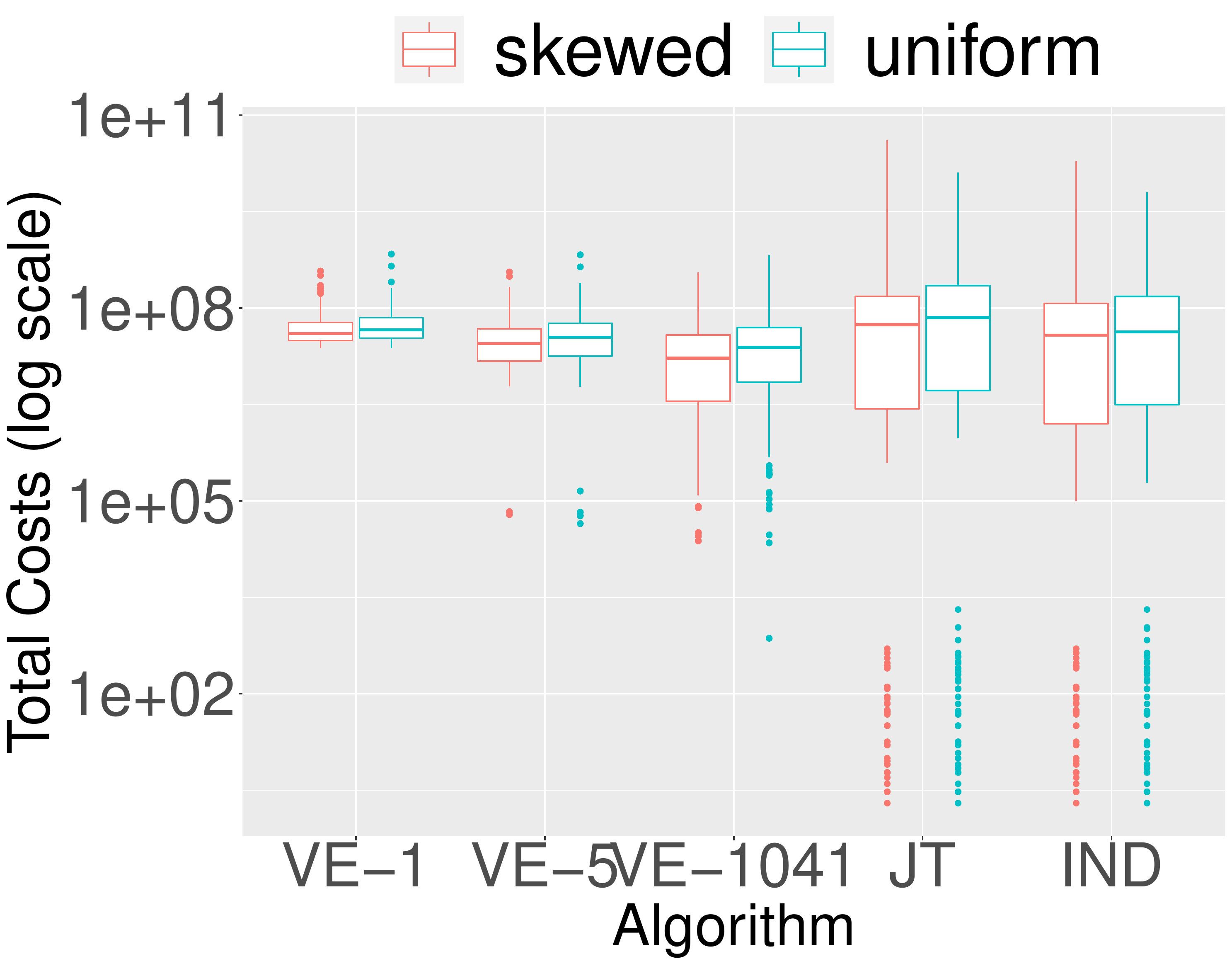}
    \\
   (e) \diabetes (\mf)  & (f) \link (\mf) & (g)  \muninm (\mf)   & (h) \muninb (\wmf)  \\ 
\end{tabular}
\caption{\label{fig:compMixed}Comparison of total costs under uniform and skewed workloads for different algorithms. }
\end{center}
\end{figure*}
}

\FullOnly{
\begin{figure*}[t]
\begin{tabular}{cccc}
    \includegraphics[width=.225\textwidth]{comparison_plots/mildew/workload_mixed}&
    \hspace{-0mm} \includegraphics[width=.225\textwidth]{comparison_plots/pathfinder/workload_mixed}&
    \hspace{-0mm} \includegraphics[width=.225\textwidth]{comparison_plots/munin1/workload_mixed}&
    \hspace{-0mm} \includegraphics[width=.225\textwidth]{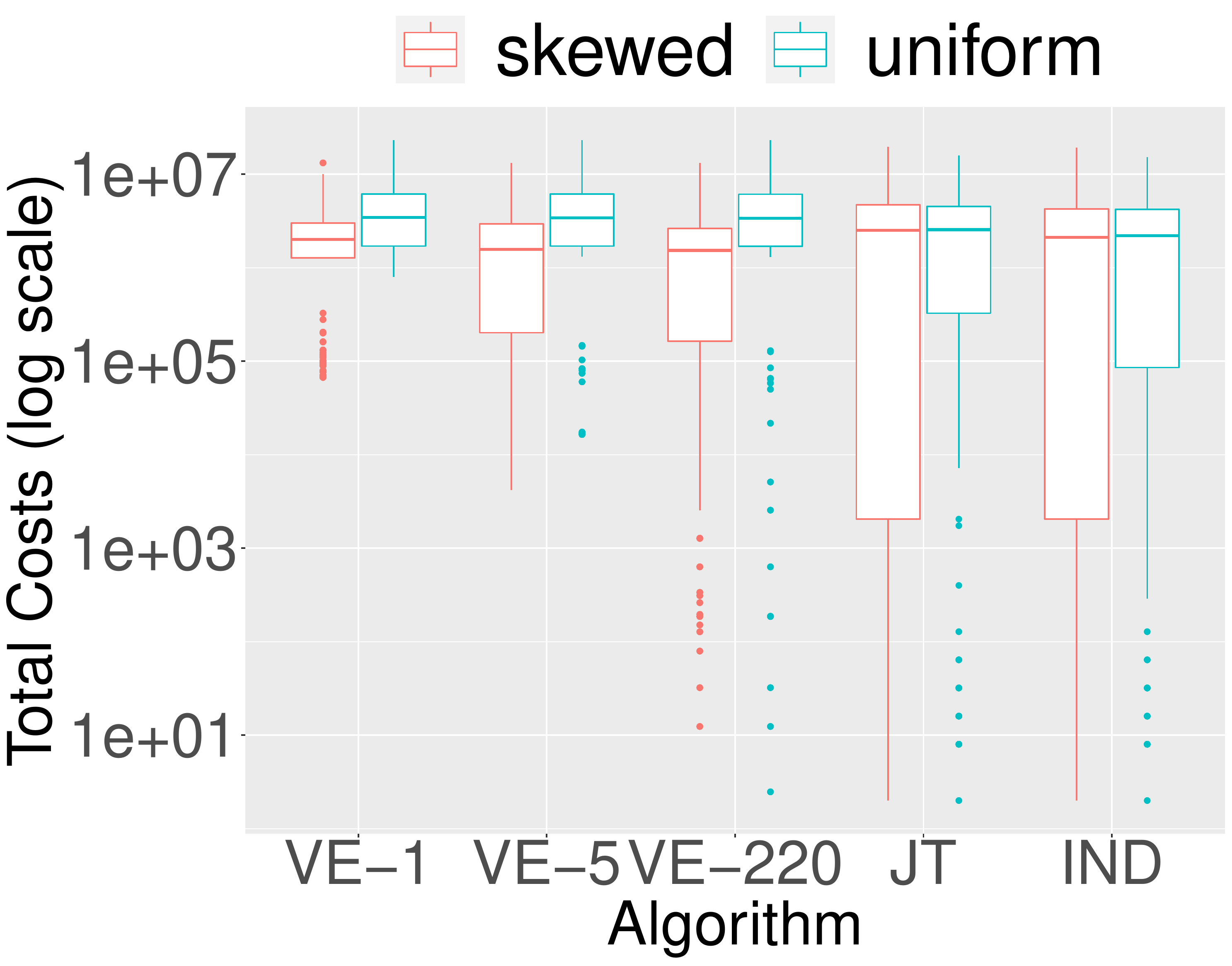}\\
  (a) \mildew (\mf) & (b) \bnpathfinder (\mf)  & (c) \munins (\wmf) & (d) \andes (\mf)   \\
  \includegraphics[width=.225\textwidth]{comparison_plots/diabetes/workload_mixed} &
    \hspace{-0mm} \includegraphics[width=.225\textwidth]{comparison_plots/link/workload_mixed} & 
    \hspace{-0mm} \includegraphics[width=.225\textwidth]{comparison_plots/munin2/workload_mixed} & 
    \hspace{-0mm} \includegraphics[width=.225\textwidth]{comparison_plots/munin/workload_mixed}
    \\
   (e) \diabetes (\mf)  & (f) \link (\mf) & (g)  \muninm (\mf)   & (h) \muninb (\wmf)  \\ 
   \includegraphics[width=.225\textwidth]{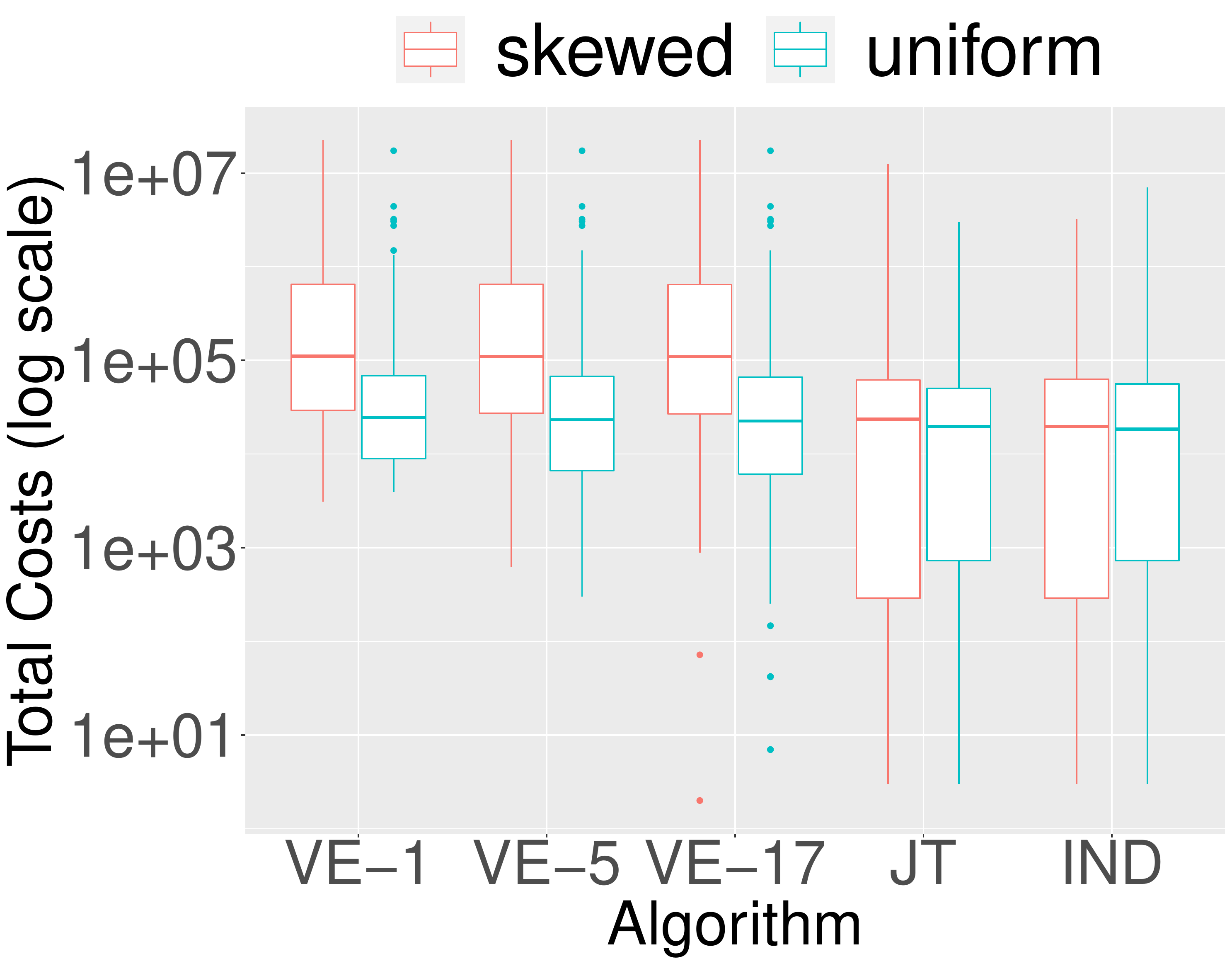} &
    \hspace{-0mm} \includegraphics[width=.225\textwidth]{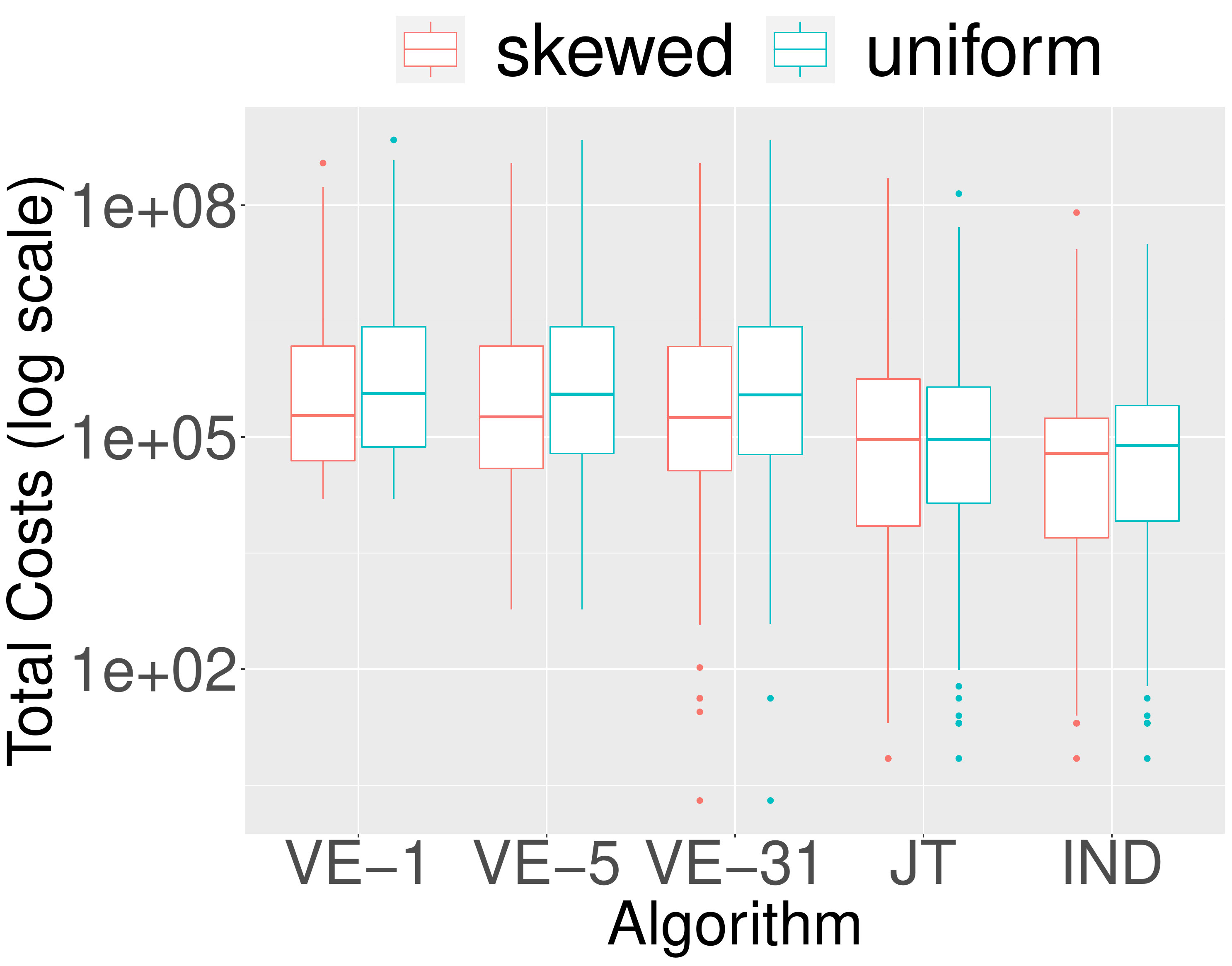} & 
    \hspace{-0mm} \includegraphics[width=.225\textwidth]{comparison_plots/tpch3/workload_mixed} & 
    \hspace{-0mm} \includegraphics[width=.225\textwidth]{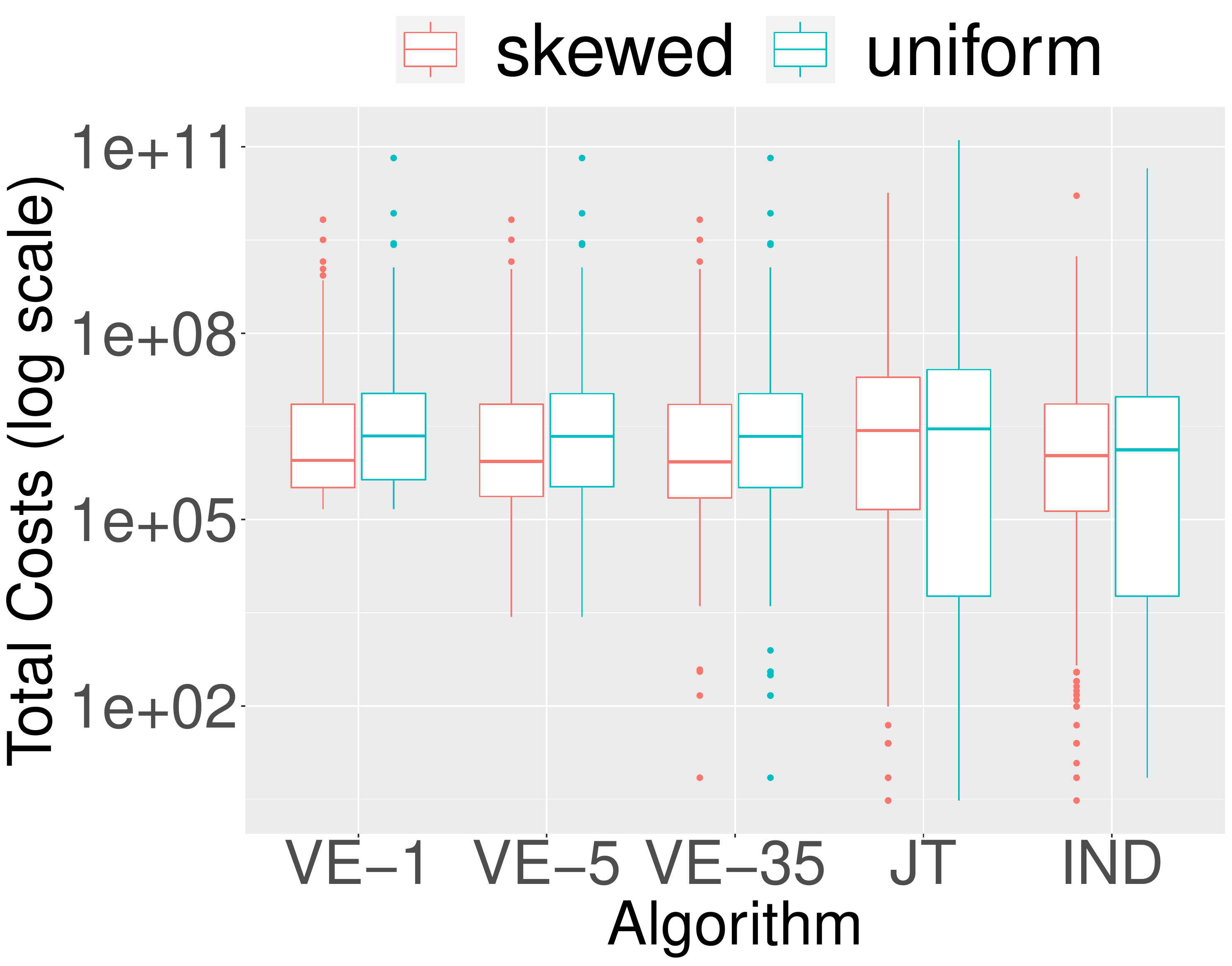}
    \\
   \revision{(i) \tpchsmall (\mw)}  & \revision{(j) \tpchmedium (\mw)} & \revision{(k)  \tpchverylarge (\mw)}   & \revision{(l) \tpchlarge (\mw)}
\end{tabular}
\caption{\label{fig:compMixed} Comparison of total costs under uniform and skewed workloads for different algorithms.}
\end{figure*}
}
%\ReviewOnly{\newpage}

\spara{Execution system.} 
Experiments were executed on a 64-bit SUSE Linux Enterprise Server with Intel Xeon 2.90\,GHz CPU and 264\,GB memory. 
Our implementation is online.\footnote{\url{https://github.com/aslayci/qtm}}
%\footnote{{https://github.com/aslayci/qtm}}

\subsection{Results}
\label{sec:results}

\spara{Improvement over Variable Elimination.} 
We first report the performance gains that materialization brings to variable elimination. Results for the uniform scheme are shown in Figure~\ref{fig:unif_per_r}. Each plot corresponds to one dataset, with the $x$-axis showing the number of factors that are materialized (budget \budget) and
the $y$-axis showing the cost savings in query time, 
expressed as a percentage of the query time when no materialization is used.
The reported savings are averages over the query workload and each bar within each plot corresponds to a different query size $\qsize$.
% cigdem: do we need to give an example? it takes a lot of space. 
%Let us see an example from Figure~\ref{fig:unif_per_r}(a): for queries on \mildew with $\qsize = 1$ free variable, a budget of $\budget = 1$ materialized factor leads on average to a decrease of about $40\%$ of query running time --- while a budget of $\budget = 7$ materialized factors leads to a decrease of about $65\%$ of query running time. 
The numbers on the bars indicate the percentage of cost savings relative to the materialization of all the factors in the tree in the uniform-workload scheme. 

We observe in Figure~\ref{fig:unif_per_r} that, consistently in all datasets, a small number of materialized factors can achieve cost savings almost as high as in the case of materializing all factors. 
% of the elimination tree. 
This result is expected as the submodularity property of the benefit function implies a diminishing-returns behavior. This result is also desirable as it shows that we can achieve significant benefit by ma\-te\-ri\-al\-iz\-ing only a small number of factors. Another observation, common to all datasets, is that, as the number \qsize of variables in a query increases, the savings from ma\-te\-ri\-al\-iza\-tion decrease. 
Given our choice of limiting the ma\-te\-ri\-al\-iza\-tion operation to factors that involve only joins and variable summation, this trend is expected: with higher \qsize, the probability that a materialized factor does not contain any free variable in its subtree decreases, limiting the benefit. 

%\todo[MM]{We have not defined the term `query variable' in this section or elsewhere. Do we mean `free variable'?}

The datasets where we observe considerably small savings are \munins, \andes, \link, and \diabetes. In all these datasets except \diabetes, we find that a small number of factors contribute to the largest part of the computational cost when $k = 0$, due to their large number of entries: we observe that 5 out of 372 factors in \munins and 6 out of 1428  factors in \link  contribute to almost $90\%$ of the computational cost, while for \andes, 5 out of 440 factors contribute to $75\%$. This suggests that the same computational burden of creating these large factor tables carries to the case of $k>0$ whenever none of the materialized factors 
%, which could help to skip such heavy computations, 
are useful.
%  during query processing.  
We indeed observe that the average cost savings, among the queries in which the variables associated to these large factors are summed out, is greater than $90\%$ in these datasets. On the other hand, for \diabetes, we find that the number of entries in the factor tables are almost uniformly distributed, however, the structure of the elimination tree has large chain components, as reflected by its larger height relative to other similar-sized trees, which makes it rare for any chosen factor to be useful. % for queries.
% We also report 
In Table~\ref{table:runtimes} we report the average query-processing times when no materialization is used (i.e., \budget = 0) under the uniform-workload scheme. We observe that the running time increases with the number \qsize of free variables in a query. This is expected as free variables lead to larger factor tables during variable elimination.
\ReviewOnly{We omit the results per query size for the skewed-workload scheme and simply note that we observed similar trends for it, although with even higher cost savings. Results are provided in the extended version~\cite{aslay2020query}. 
% and report the overall comparison between the uniform- and skewed-workload schemes in Figure~\ref{fig:workloadMixed}.
} 
\FullOnly{
Figure~\ref{fig:biasr_per_r} reports the average cost savings  per query size for varying number of materialized factors in the skewed-workload scheme.  As in the case of uniform-workload scheme, we observe that a small number of materialized factors can achieve cost savings almost as well as in the case of materializing all the factors of the elimination tree in the skewed-workload scheme. We also observe that, while the relative performance of materialization, indicated by the numbers on the bars, does not differ significantly between the uniform- and skewed-workload schemes, the savings over the case of no materialization significantly improves in the skewed-workload scheme. }

We report the overall comparison between the uniform- and skewed-workload schemes in Figure~\ref{fig:workloadMixed}.
%% cigdem after ICDE submission edit: I am including here a more detailed interpretation of Figure~\ref{fig:workloadMixed} results of for the extended version 
\FullOnly{
We observe that the performance gains due to materialization in the skewed-workload scheme are significantly higher than in the uniform-workload scheme for all the datasets. 
This trend is especially visible for the datasets \munins, \andes, and \link, all of which contain only a small number of factors contributing to the majority of the computational cost as explained before: we observed that these large factors reside in middle layers of their elimination trees, hence, are less likely to be associated to free variables in the skewed-workload scheme.
% than its uniform counterpart. 
Since a random set of free variables in the skewed-workload scheme are more likely to be associated to the ancestors of the materialized factors, it is expected that the skewed-workload scheme can obtain higher benefit from materialization. 
This observation suggests that constructing an elimination order that is tuned to a given query workload can provide significant boost to the query-processing performance.
% when only factors that are the result of variable summation are materialized.  
%cigdem: should we add this sentence above to conclusions? 
}
For \mildew, average savings are slightly higher than $10\%$ under the uniform-workload scheme: 
when we drill down to savings specific to query size, as provided in Figure~\ref{fig:unif_per_r}, 
we see that for queries of size $1$, $2$, and $3$, the savings are around $70\%$, $35\%$, and $25\%$ respectively, 
and have a sharp de\-crease to under $10\%$ for queries of size $4$ and $5$, 
resulting to an average around $10\%$ over all queries. 
We remind that \mildew has only $35$ variables, which translate to a  small elimination tree, making it especially hard to find useful factors for large value of $\qsize$ under the uniform workload. 
Savings improve significantly to an average of $50\%$ in the skewed-workload scheme for \mildew. 
On the other hand, for \diabetes, average savings under both workload schemes remain around $10\%$, due to the elimination tree having large chain components, limiting the extent we can exploit materialization. 
For the rest of the datasets, \bnpathfinder, \muninm, and \muninb, we observe relatively high savings under both workload schemes, where average savings for $\budget = 20$ is $70\%$ for \bnpathfinder and $50\%$ for the other two Bayesian networks. 
Overall, we observe that it is possible to obtain up to $99\%$ savings 
with a small number of materialized factors in both schemes for each dataset.

\spara{Comparison with junction tree algorithms.} 
We begin by comparing the computational costs for inference, i.e., query answering. 
Results for the uniform scheme are shown in Figure~\ref{fig:comp_unif_per_r}, 
% Each plot corresponds to one dataset, 
% with the $x$-axis indicating the algorithms 
% (variable elimination with different levels of materialization is indicated by \qtm{$\budget$}) 
% and the $y$-axis indicating inference costs.
where variable elimination with different levels of materialization is indicated by \qtm{$\budget$}.
The reported costs are averages over the query workload per query size $\qsize$. 

Figure~\ref{fig:comp_unif_per_r} demonstrates that, for all the datasets, our proposed
variable elimination with different levels of materialization is competitive with the junction tree algorithms 
(\jt and \kanagal) for $\qsize > 1$, even with a very small materialization budget of $k=1$. 
When $\qsize = 1$, i.e., when there is only one query variable, \jt and \kanagal perform significantly better. 
This is expected, since the distribution of any single variable is readily available from the materialized junction tree.
% used by \jt and \kanagal.
%
%% Specifically, the junction tree contains materialized joint distributions of subsets of variables, such that the distribution of any single variable can be obtained from such joint distribution with a single summation operation.
%
However for $\qsize > 1$, \qtm{$\budget$} has comparable performance to \jt and \kanagal\ -- and, in fact, \qtm{$\budget$} significantly outperforms the other algorithms in half of the datasets, including the largest Bayesian networks \muninm and \muninb.
This is because in this case it is unlikely for the junction tree to have the joint distribution of $\qsize > 2$ variables readily available.
Specifically, it is unlikely for the junction tree to contain a materialized joint distribution of exactly the \qsize variables included in a given query.
In such cases, the junction tree algorithm \jt essentially performs variable elimination over the junction tree.
% (specifically: over a minimal subgraph of the junction tree that contains the query variables).
%%%%

\ReviewOnly{
We omit the results for the skewed-workload scheme as the overall trend is similar, with slightly better performance for \qtm{$\budget$}. They are made available in the extended version~\cite{aslay2020query}. Instead, we present an aggregate comparison of the algorithms in the uniform and skewed-workload schemes in Figure~\ref{fig:compMixed}.}

\FullOnly{
Figure~\ref{fig:comp_biasr_per_r} reports the comparison of the costs in the skewed-workload scheme. As in the case of uniform-workload scheme,  we observe that for all the datasets, our proposed
variable elimination with different levels of materialization is competitive with the junction tree algorithms (\jt and \kanagal) for $\qsize > 1$, even with a very small materialization budget of $k=1$. Figure~\ref{fig:compMixed} reports aggregate comparison of the algorithms in the uniform- and skewed-workload schemes.}

Note that \jt and \kanagal are quite sensitive to the query variables, as indicated by the high variability in the costs for answering different queries. 
This is expected, as their performance depends on how far apart the query variables are located on the junction tree. On the other hand, \qtm{$\budget$} is more robust under both workload schemes.
% , showing its capability to handle any kind of workload and query size using very small materialization budget.

Additionally, we report the computational costs for the off\-line materialization phase in Table~\ref{table:offlineStats}. 
In particular, we report the disk space occupied by the materialized structures, along with the running time required for materialization.
For our method, we report the results for \qtm{$n$}, i.e., for the maximum possible budget $\budget = n$, where all the factors are materialized. 
For \jt, the costs concern the pre\-computation of the junction tree (the ``calibration'' phase). 
For \kanagal, the costs concern both the junction tree and the additional index.

Table~\ref{table:offlineStats} demonstrates that \qtm{\budget} significantly outperforms \jt and \kanagal both in terms of precomputation time and materialization volume.
In fact, for \munins and \tpchverylarge, the junction tree algorithms did not terminate after a two-day-long execution (NA entries in Table~\ref{table:offlineStats}).
% In summary, we find that a modest amount of materialization for the variable elimination algorithm achieves comparable or better query efficiency for complex queries (two variables or more), at a much lower cost of precomputation and volume of materialization.
%
We conclude that, for settings that include large Bayesian networks and modest to large query sizes, a small level of materialization for variable elimination offers a significant advantage over junction tree based algorithms, as it provides efficient inference (Figures~\ref{fig:comp_unif_per_r}-\ref{fig:compMixed}), as well as faster and lighter pre\-computation (Table~\ref{table:offlineStats}). 

\eat{Additionally, we report statistics regarding the off\-line materialization phase. For each dataset, we report the results for maximum possible $\budget$ where all the factors are materialized. 
%\todo[MM]{The choice of $k = 7, 10$ seems a bit arbitrary. Can we report sizes among all factors (k = max)?}
In all the datasets except \munins and \link, selection and computation of the factors to materialize took less than 10 seconds while for \munins and \link it took 270 and 100 seconds respectively. The maximum materialized table size on disk was found to be 240 MB for \munins, 54 MB \link, and less than 2MB for the rest of the datasets. We conclude that by materializing a small number of tables with modest memory requirements, we can obtain significant performance gains in the query-processing time, reaching up to an average gain of $70\%$ in the uniform-workload scheme and $80\%$ in the skewed-workload scheme.}

\begin{table}[t]
\setlength\tabcolsep{3pt}
\fontsize{8}{9}\selectfont
\centering
% \begin{footnotesize}
\caption{\label{table:offlineStats} Materialization phase statistics.}
\begin{tabular}%{p{1.4cm}p{0.70cm}p{0.75cm}p{0.75cm}p{0.70cm}p{0.75cm}p{0.75cm}}
{lrrrrrrrr}
\toprule
 & \multicolumn{3}{c}{Disk Space (MB)} & \multicolumn{3}{c}{Time (seconds)}  \\
\cmidrule(lr{0.7em}){2-4}
\cmidrule(lr{0.7em}){5-7}
%Network & \qtm{n} & \jt & \kanagal & \qtm{n} & \jt & \kanagal  \\ \midrule
Network & \qtm{n} & \jt & \kanagal & \qtm{n} & \jt & \kanagal  \\ \midrule
\mildew & $1.7$ & $373$ & $1\,354$ & $5$  & $18\,360$ & $18\,360 $  \\
\bnpathfinder  & $<1$ & $17$ & $23$  & $<1$ & $302$ & $305$  \\
\munins  &  $317$ & NA & NA & $270$  & NA & NA  \\
\andes  &  $4.1$& $70$ & $78$ & $2$ & $3\,682$ & $3\,686 $  \\
\diabetes  &  $15$ & $945$ & $3\,286$ &  $2$ & $41\,228$  & $ 41\,247$  \\
\link   & $245$  & $3\,735$ &  $3\,824$ & $100$ & $98\,533$  & $98\,647 $  \\
\muninm    & $9$  & $480$ & $573$ & $8$ & $21\,348$  & $ 21\,635$  \\
\muninb  &  $14$ & $2\,866$ & $2\,972$ & $16$  & $110\,342$ & $110\,645$  \\
\revision{\tpchsmall} & \revision{$<1$} & \revision{$<1$} & \revision{$<1$} & \revision{0.01} & \revision{$0.306$} & \revision{$0.322$}  \\
\revision{\tpchmedium} & \revision{$<1$} & \revision{$<1$} & \revision{$<1$} & \revision{0.02} & \revision{$1.866$} & \revision{$1.882$}  \\
\revision{\tpchverylarge} & \revision{$<1$} & \revision{NA} & \revision{NA} & \revision{0.02} & \revision{NA} & \revision{NA}   \\
\revision{\tpchlarge} & \revision{$<1$} & \revision{$4.7$} & \revision{$6.9$} & \revision{0.02} & \revision{$106$} & \revision{$107$} \\
\bottomrule
\end{tabular}
% \end{footnotesize}
\end{table}

\smallskip
\ReviewOnly{
\noindent
\revision{\bf Robustness.} 
We experiment with settings where materialization is optimized for a training query workload \traindistr, but the test queries are drawn from a different workload \testdistr. 
Specifically, we set \traindistr to be either the uniform or the skewed workload; and \testdistr to consist of queries from the uniform workload in proportion $\lambda$ and from the skewed workload in proportion $(1-\lambda)$, with varying $\lambda\in[0,1]$. The results are shown for one Bayesian network (\mildew) in Figure~\ref{fig:robustness}. In both cases, performance decreases smoothly for increasingly uniform workload \testdistr (larger values of $\lambda$). This happens because in both cases the materialized factors that benefit the queries  are the ones that dominate the skewed workload, i.e., queries with variables near the root of the elimination~tree.}

\ReviewOnly{
\begin{figure}
\begin{center}
\begin{tabular}{cc}
    \includegraphics[width=.2\textwidth]{./cde_experiments/rev-robustness_plots/rev-robustness_mildew}&
     \hspace{-0mm} \includegraphics[width=.2\textwidth]{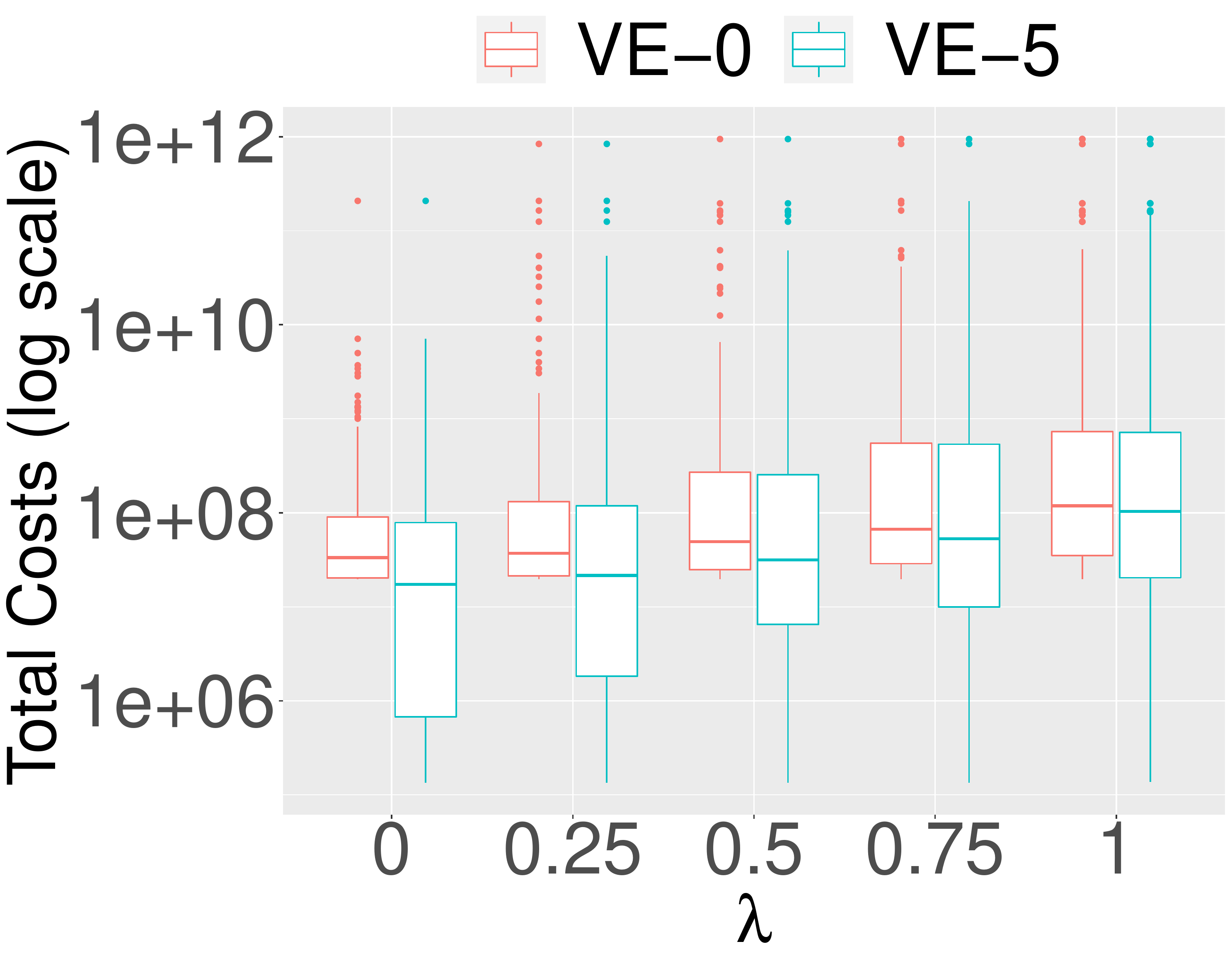}\\
  (a) \traindistr: Skewed & (b) \traindistr: Uniform 
\end{tabular}
\caption{\label{fig:robustness}\revision{Robustness with \mildew. We explore how performance of variable elimination (with and without materialization) changes when the query workload changes from skewed ($\lambda = 0$) to uniform ($\lambda = 1$).}
}
\end{center}
\end{figure}
}

\smallskip
\noindent
\FullOnly{{\bf Robustness.} 
We experiment with settings where materialization is optimized for a training query workload \traindistr, but the test queries are drawn from a different workload \testdistr. 
Specifically, we set \traindistr to be either ($i$) uniform or ($ii$) skewed workload; and \testdistr to consist of queries from the uniform workload in proportion $\lambda$ and from the skewed workload in proportion $(1-\lambda)$, with varying $\lambda\in[0,1]$. The results for both settings ($i$) and ($ii$) are shown in Figure~\ref{fig:robustness_lambda_uni} and \ref{fig:robustness_lambda_skewed}, respectively. In both cases, we see that the performance decreases smoothly for increasingly uniform workload \testdistr (larger values of $\lambda$). This happens because in both cases the materialized factors that benefit the queries are the ones that dominate the skewed workload, i.e., queries with variables near the root of the elimination~tree. Overall, the results suggest that our approach is fairly robust to changes in the query distribution
for both settings ($i$) and ($ii$). We see that the difference between the cost of \qtm{$0$} and  \qtm{$5$} does not deviate drastically across the five workloads corresponding to different values of $\lambda$. }

\FullOnly{Despite the robustness of the proposed method, it is certainly appropriate to recompute the optimal materialization whenever the query distribution changes significantly. 
Thus, implementing a query-distribution drift-detection mechanism, which automatically prompts an update in the materialization strategy, would be a valuable extension of our framework.}

\FullOnly{
\begin{figure*}[t]
	\begin{tabular}{cccc}
		\includegraphics[width=.235\textwidth]{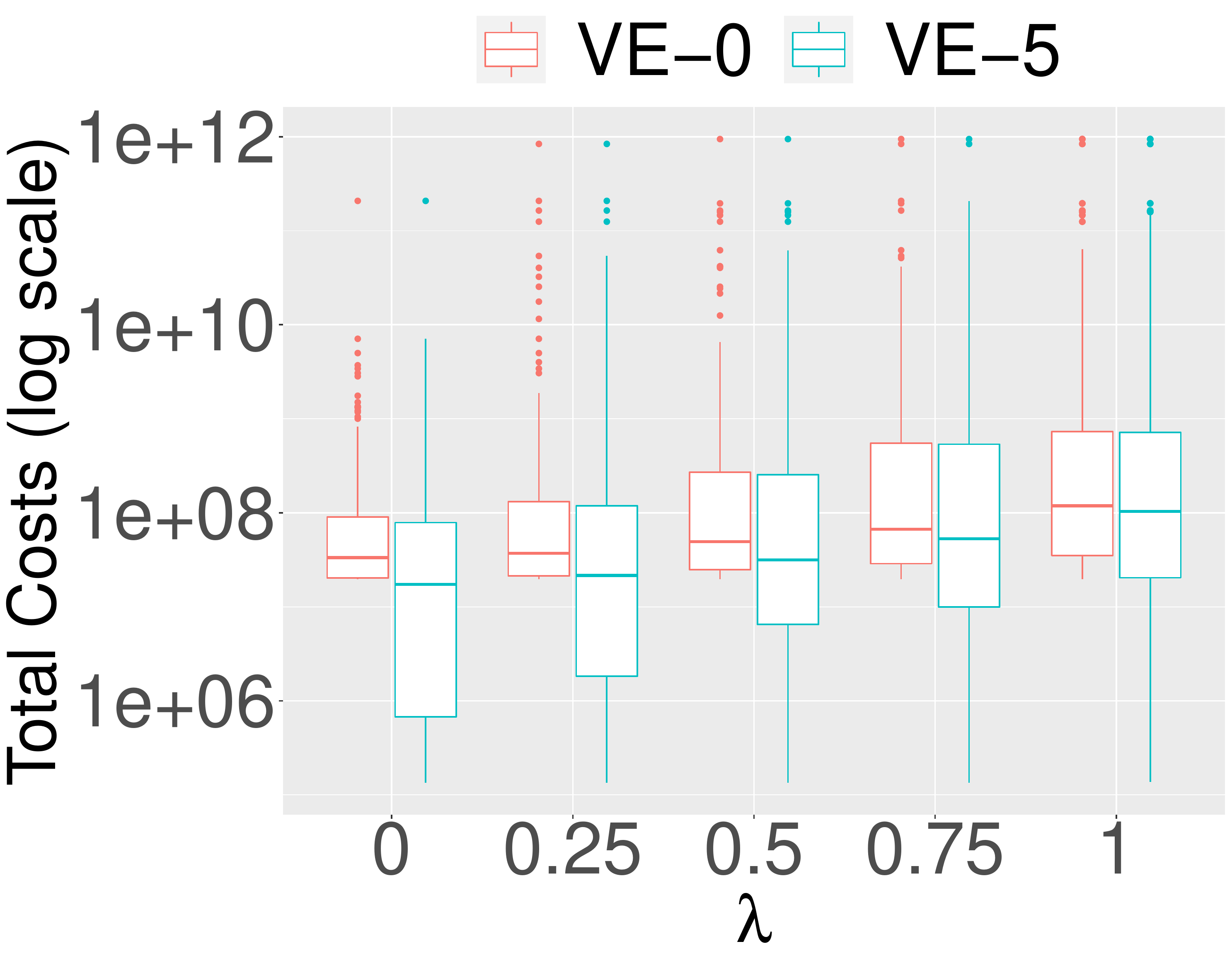}&
		\includegraphics[width=.235\textwidth]{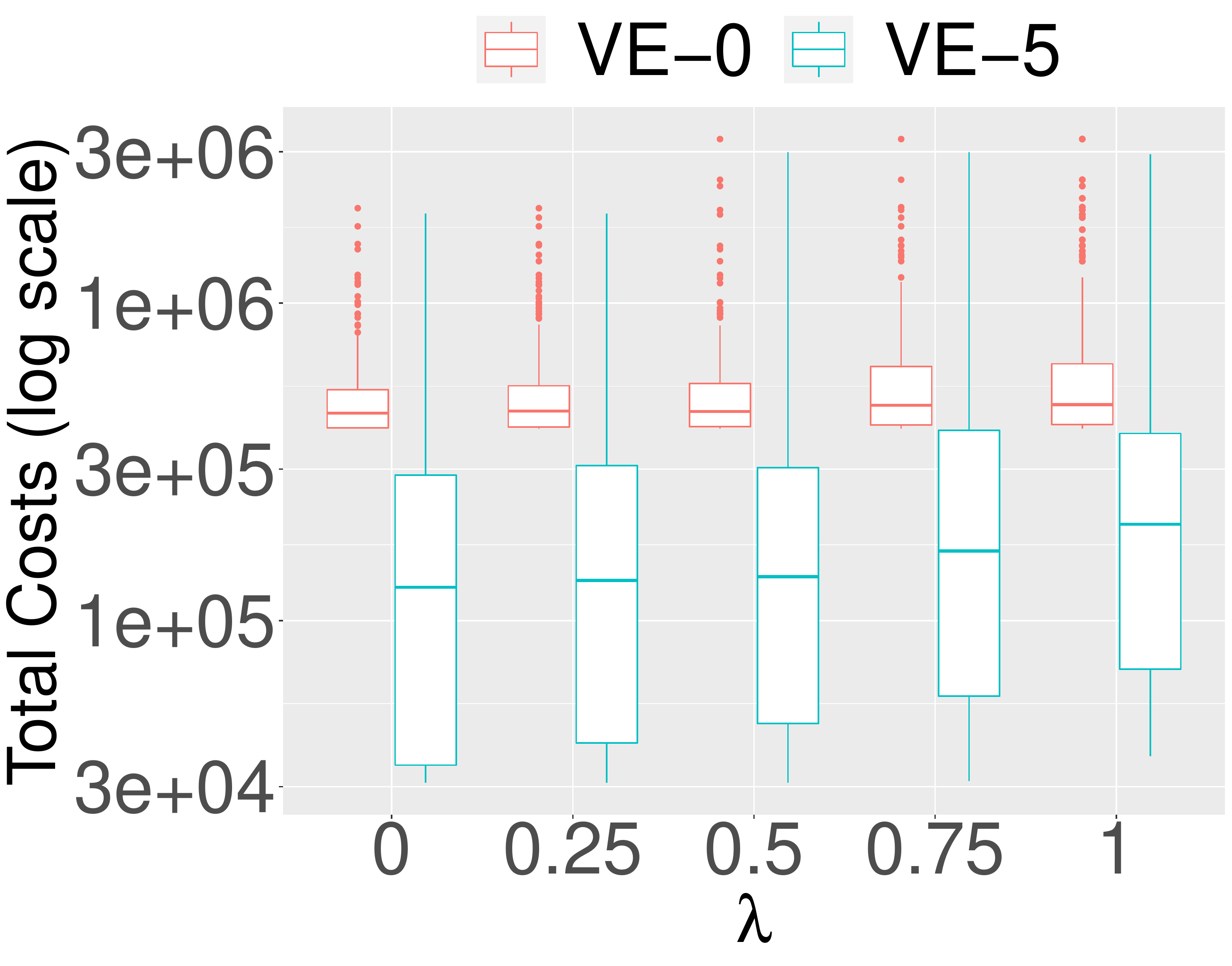}&
		\includegraphics[width=.235\textwidth]{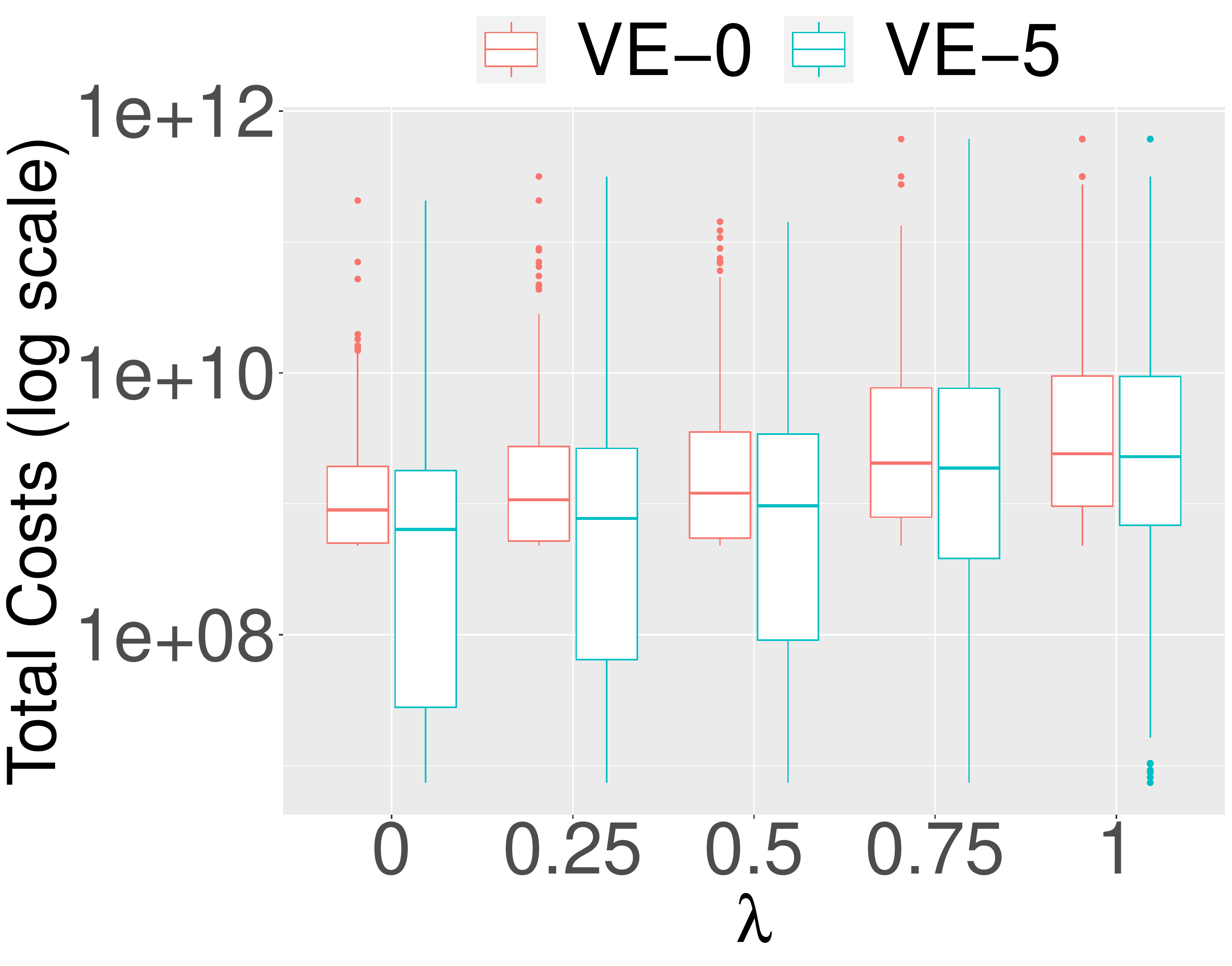}&
		\includegraphics[width=.235\textwidth]{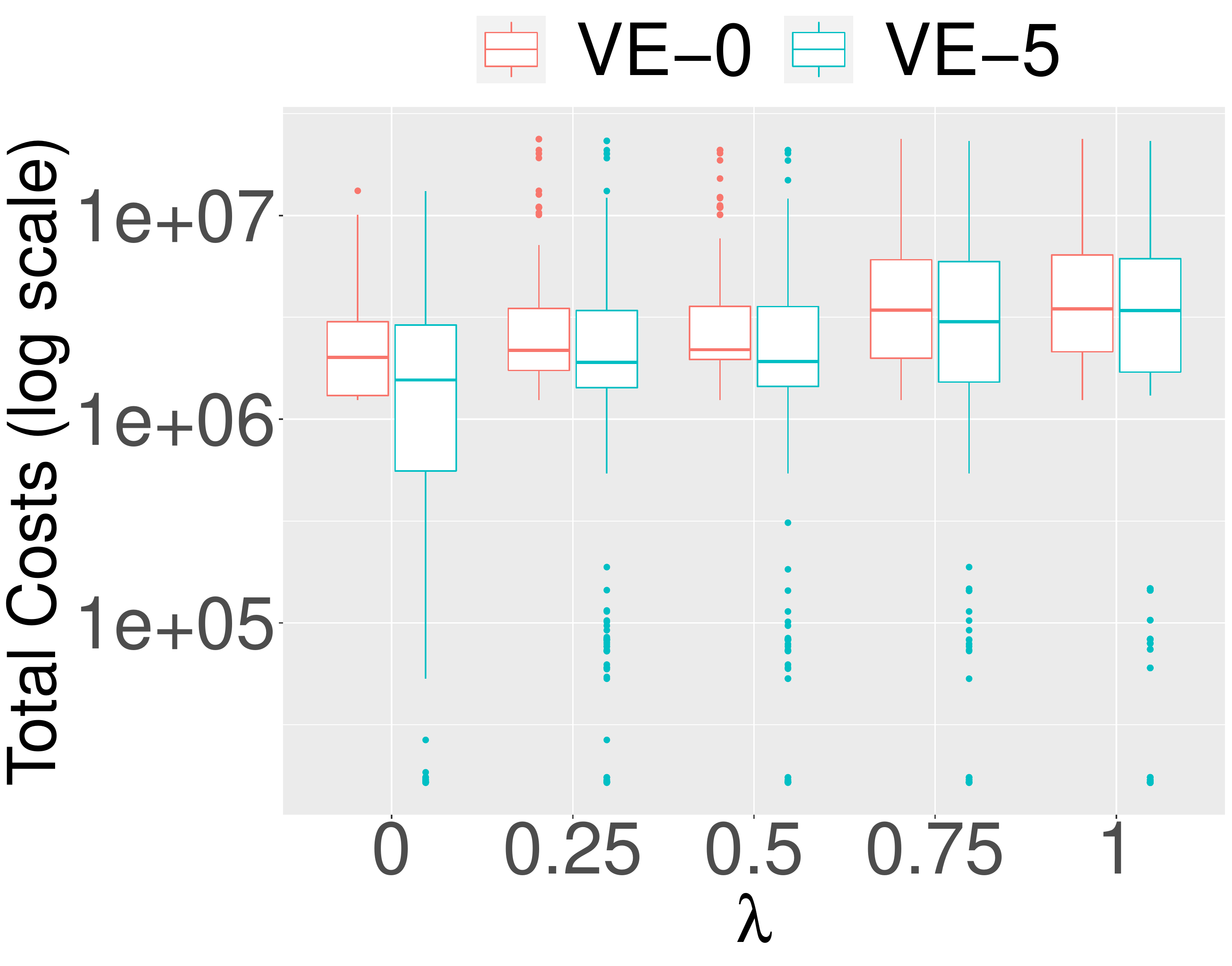}\\
		(a) \mildew  & (b) \bnpathfinder   & (c) \munins  & (d) \andes    \\
		\includegraphics[width=.235\textwidth]{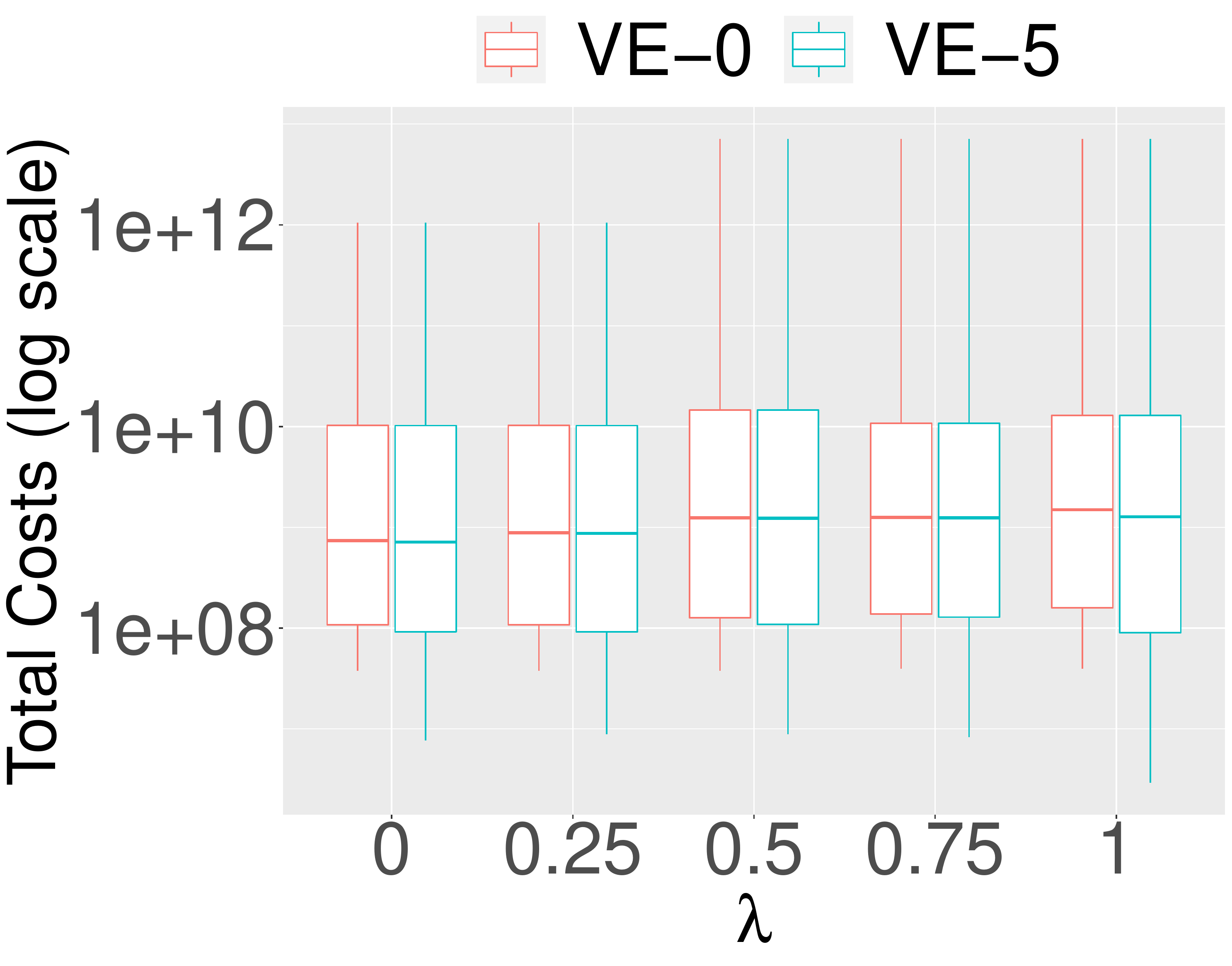} &
		\includegraphics[width=.235\textwidth]{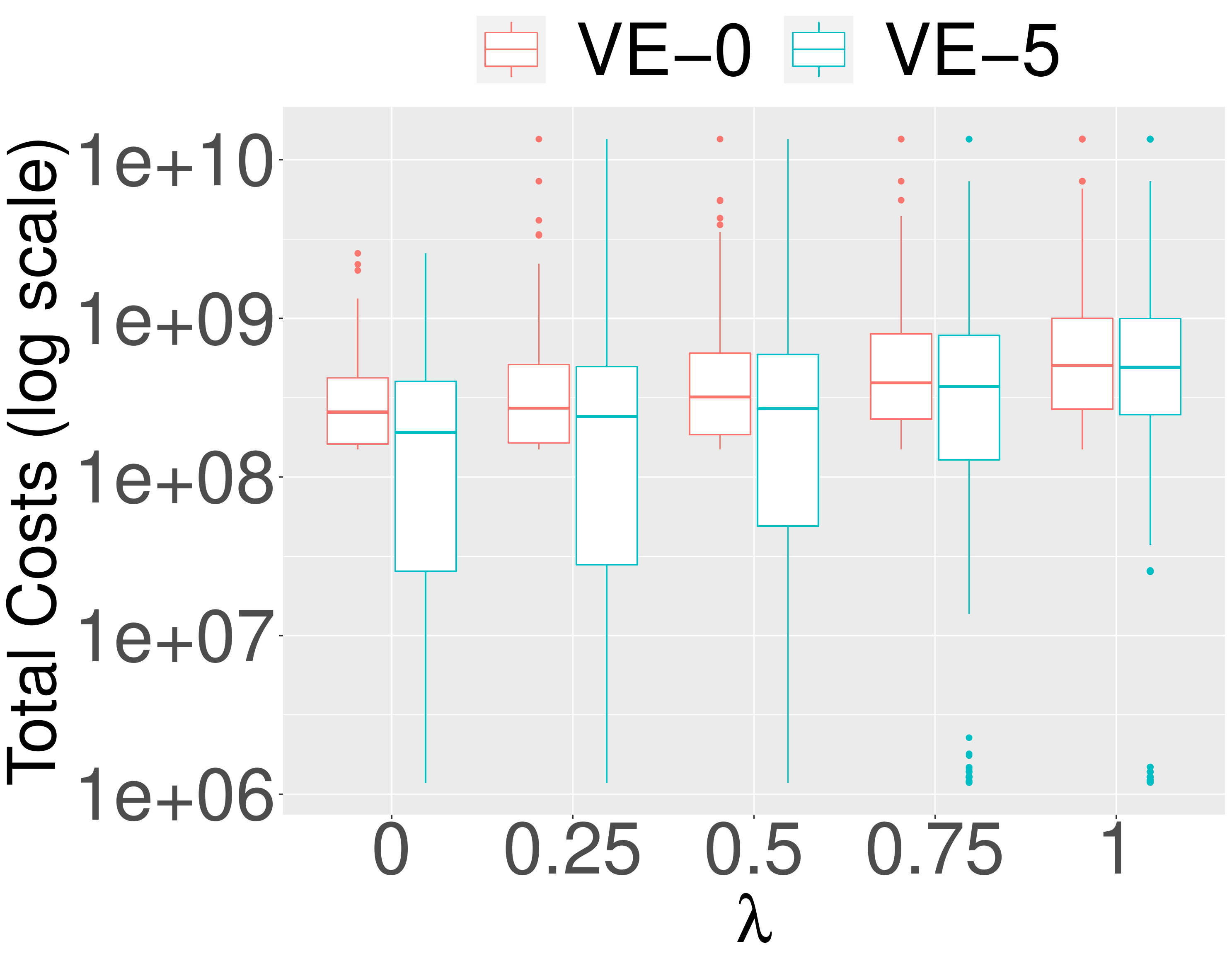} & 
		\includegraphics[width=.235\textwidth]{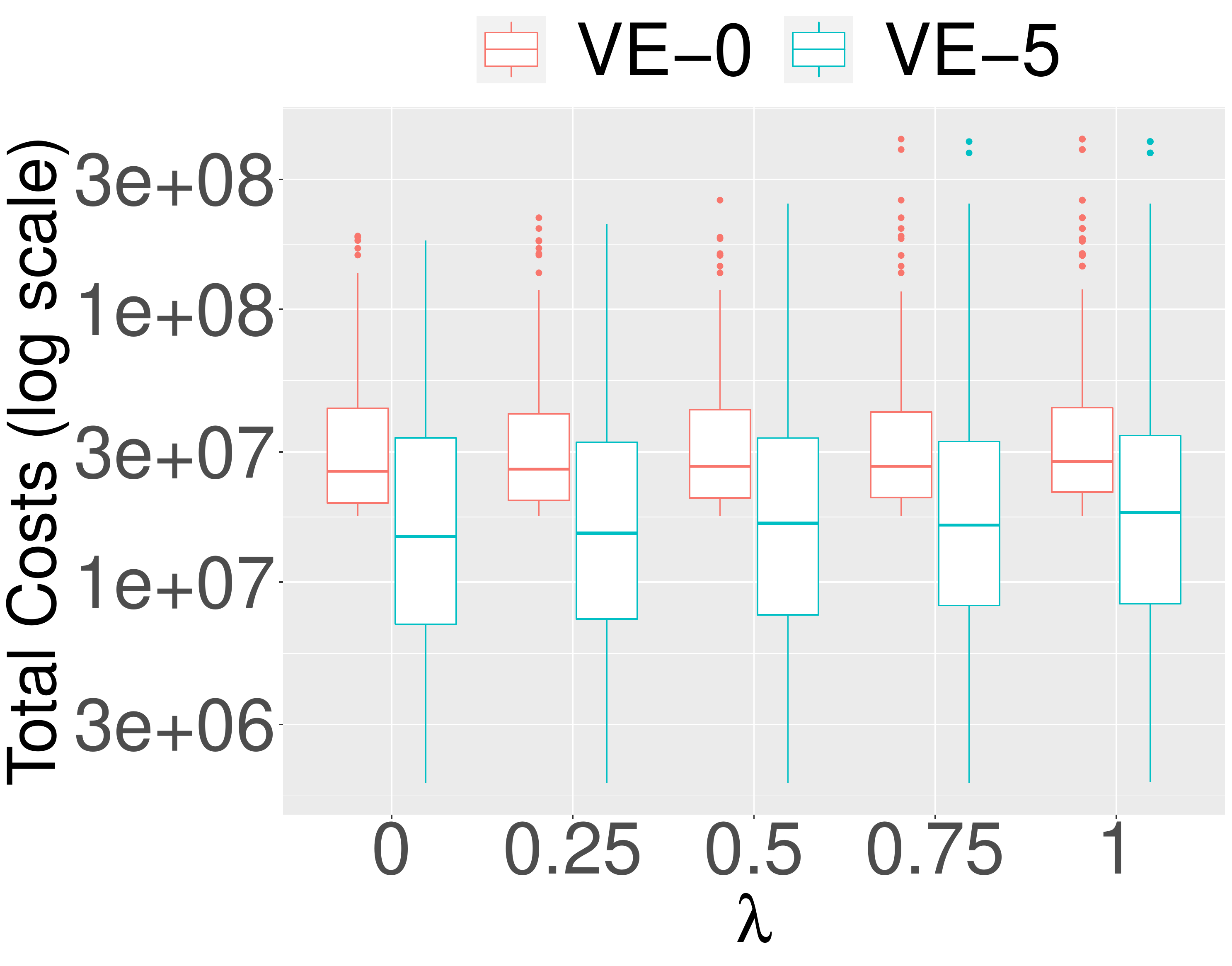} & 
		\includegraphics[width=.235\textwidth]{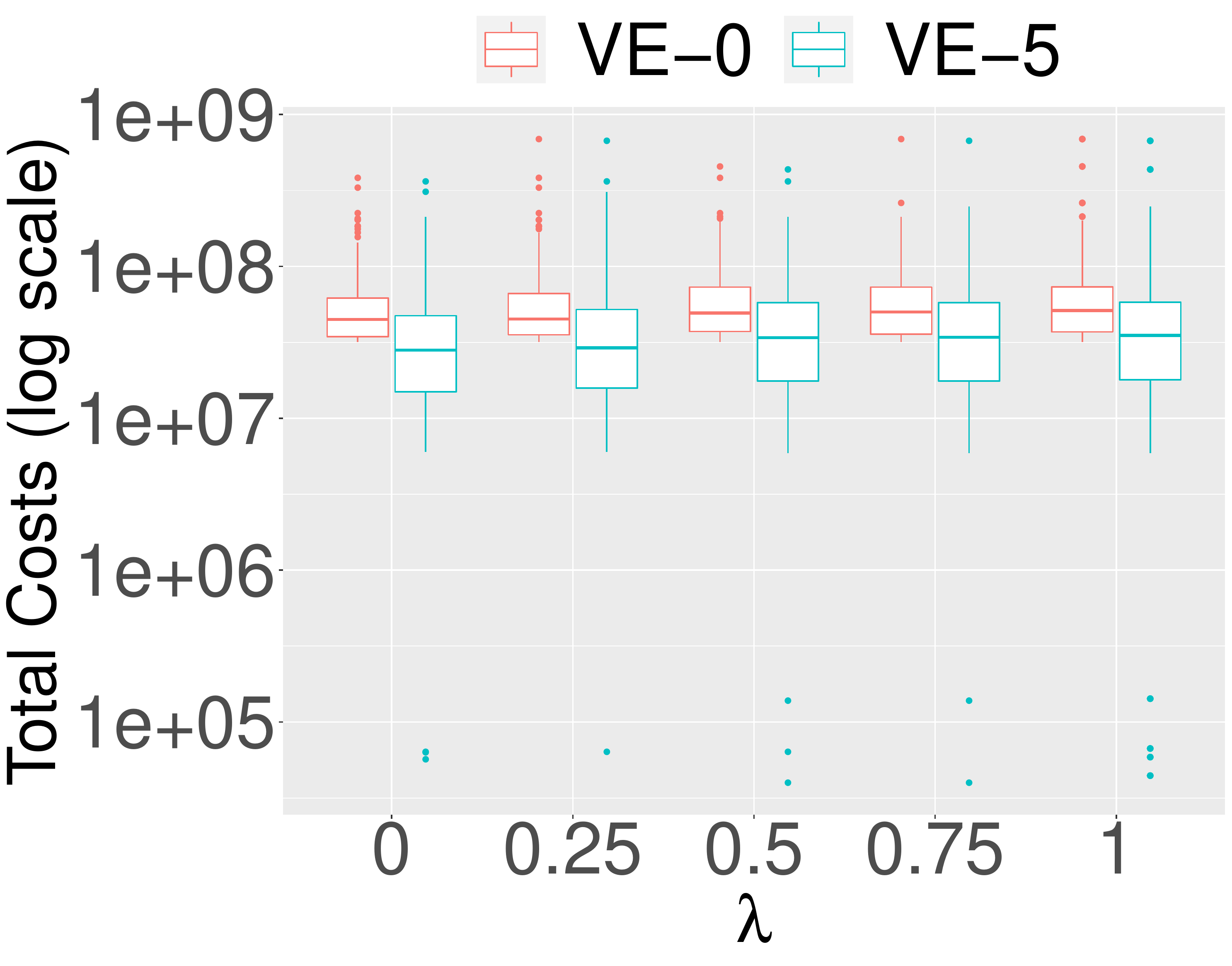}	\\
		(e) \diabetes   & (f) \link  & (g)  \muninm    & (h) \muninb   \\ 
		\includegraphics[width=.235\textwidth]{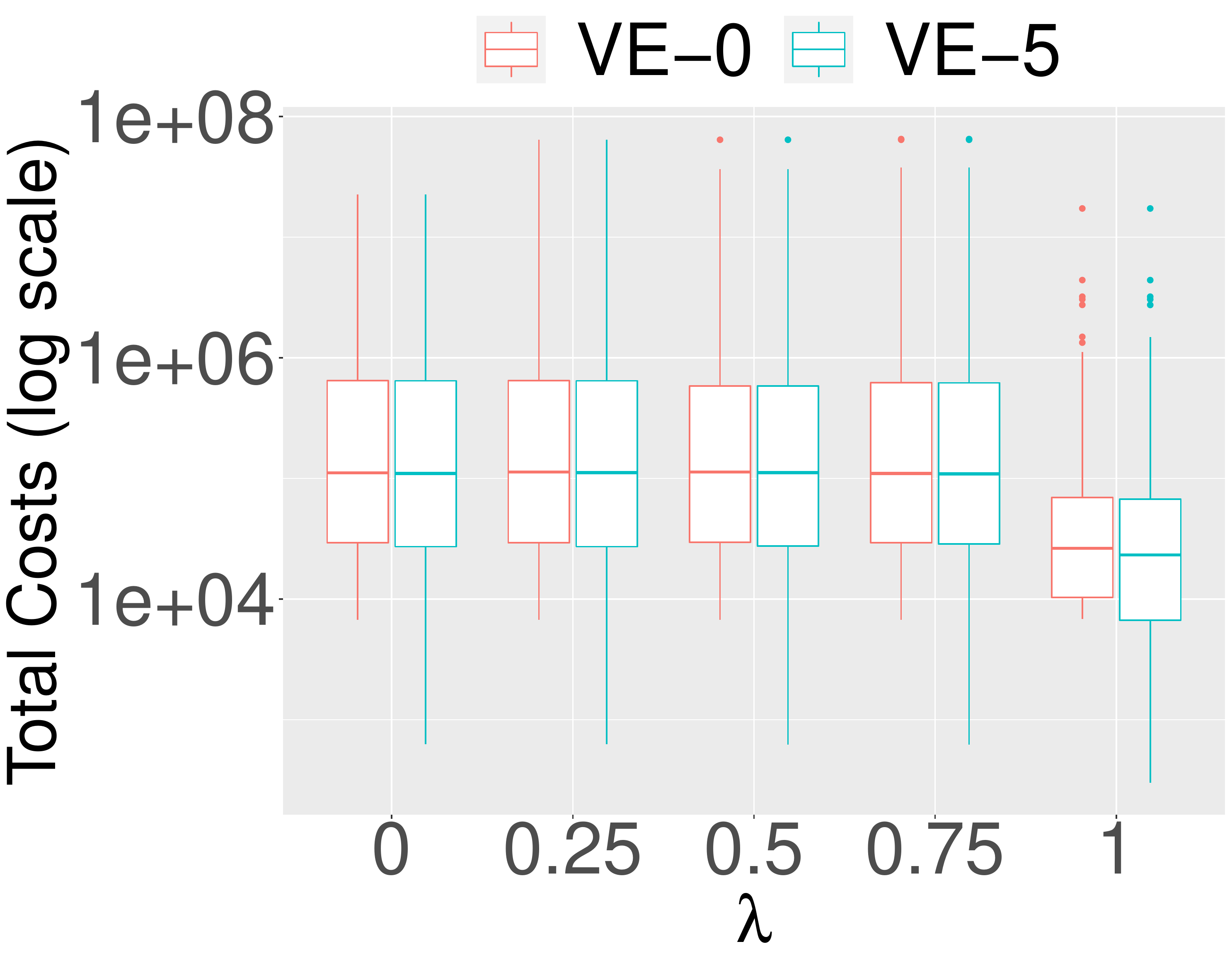} &
		\includegraphics[width=.235\textwidth]{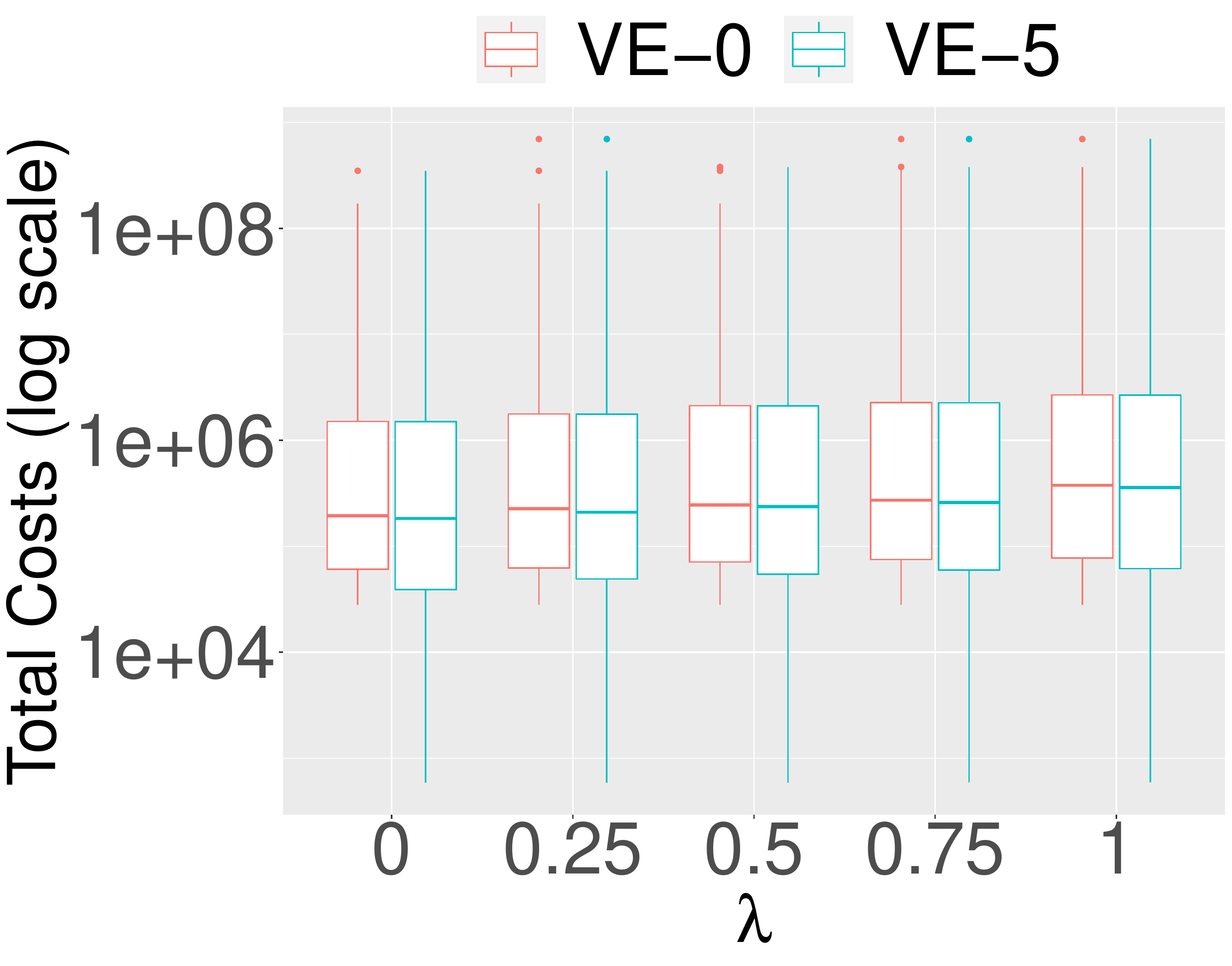} & 
		\includegraphics[width=.235\textwidth]{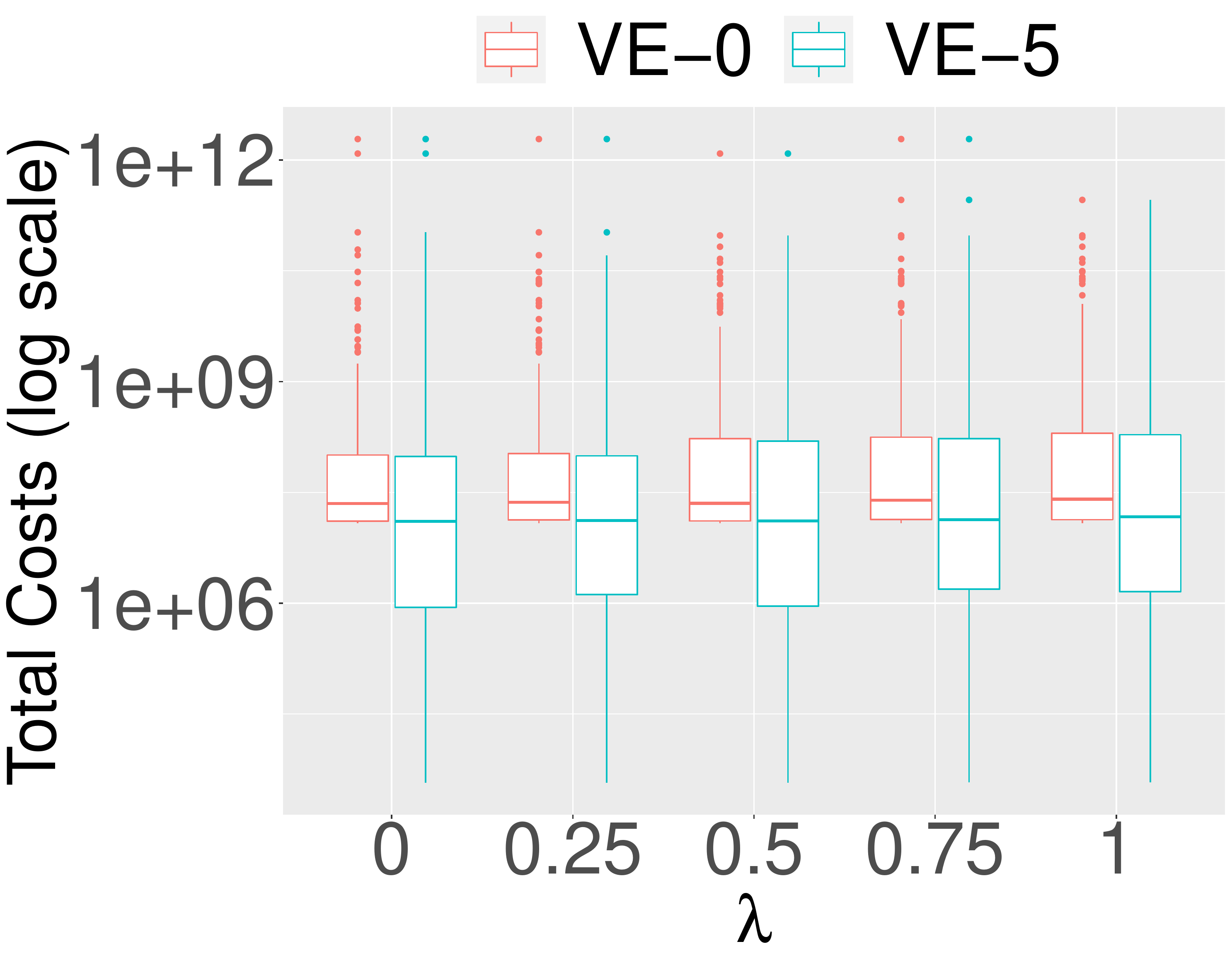} & 
		\includegraphics[width=.235\textwidth]{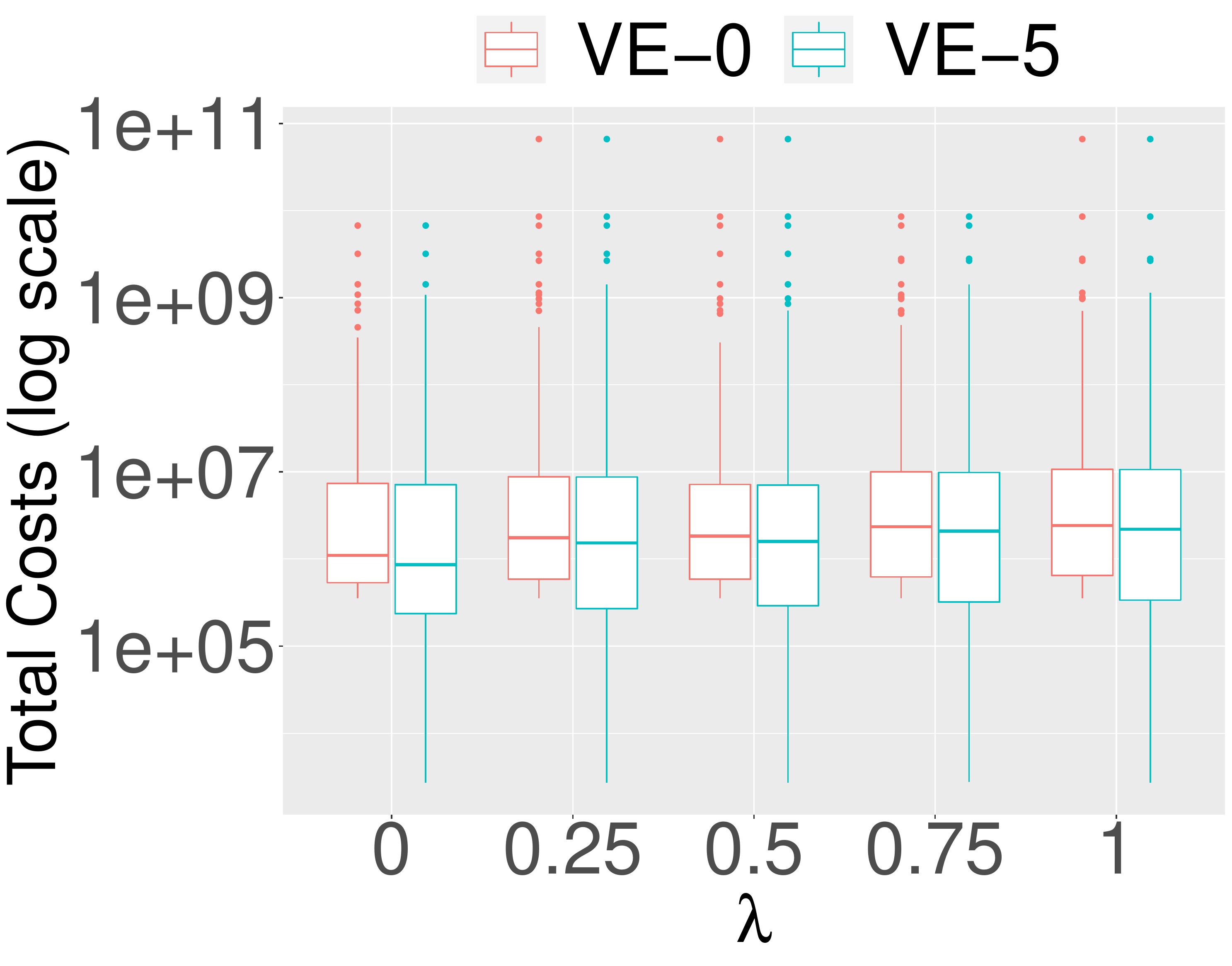}	\\
		(i) \tpchsmall   & (j) \tpchmedium  & (k)  \tpchverylarge    & (l) \tpchlarge   \\ 	
	\end{tabular}
	\caption{ {\color{black}   Robustness of materialization strategy with respect to drifts in query distribution when \traindistr is set to uniform workload. The value of $\lambda$ on the $x-$axis denotes the proportion of queries sampled from the uniform workload scheme, and $1-\lambda$ is the proportion of queries sampled from the skewed workload scheme. The $y-$axis shows on a logarithmic scale the total cost associated with \qtm{$0$} and  \qtm{$5$} in each workload.   }}
\label{fig:robustness_lambda_uni}
\end{figure*}
}

\FullOnly{
\begin{figure*}[t]
	\begin{tabular}{cccc}
		\includegraphics[width=.235\textwidth]{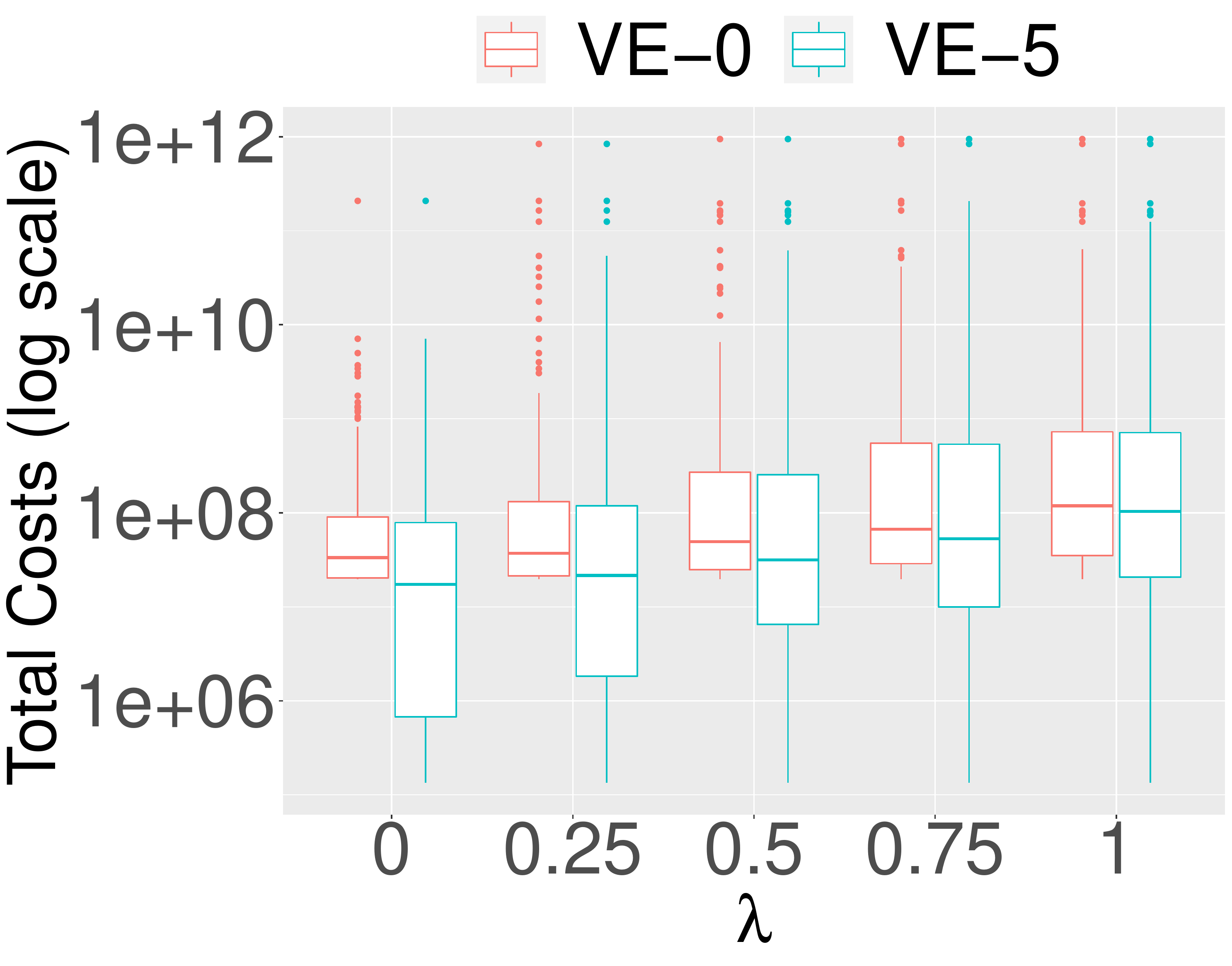}&
		\includegraphics[width=.235\textwidth]{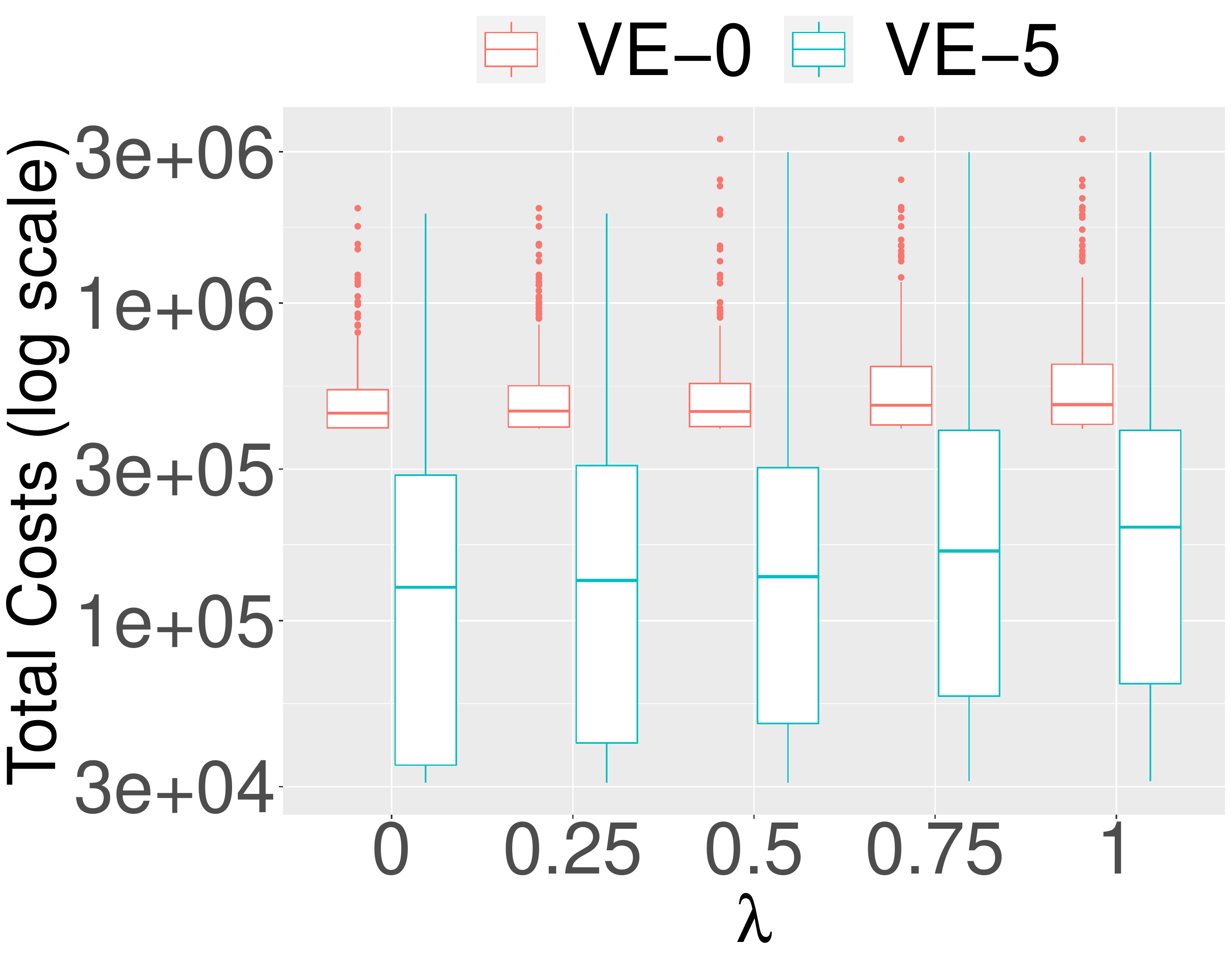}&
		\includegraphics[width=.235\textwidth]{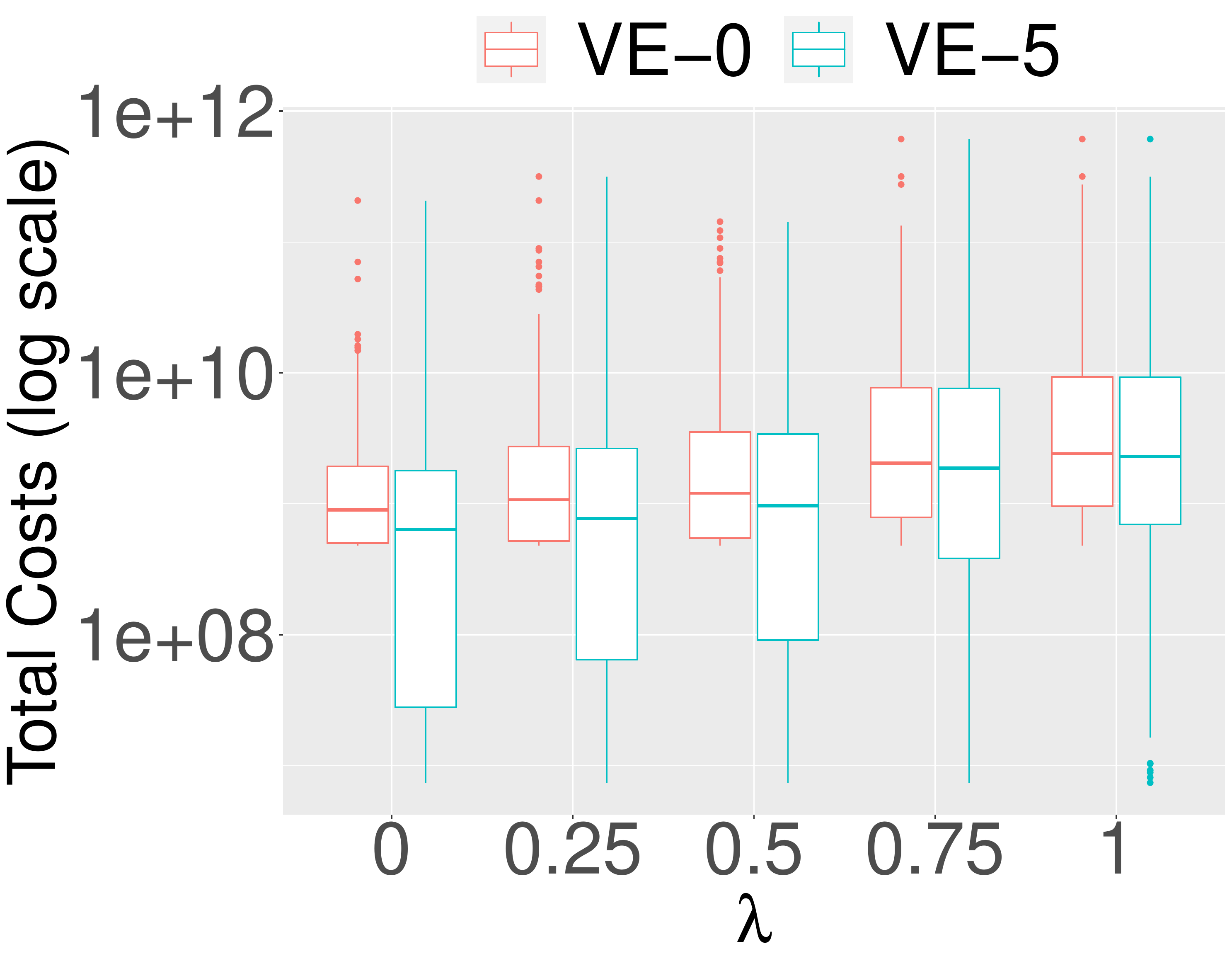}&
		\includegraphics[width=.235\textwidth]{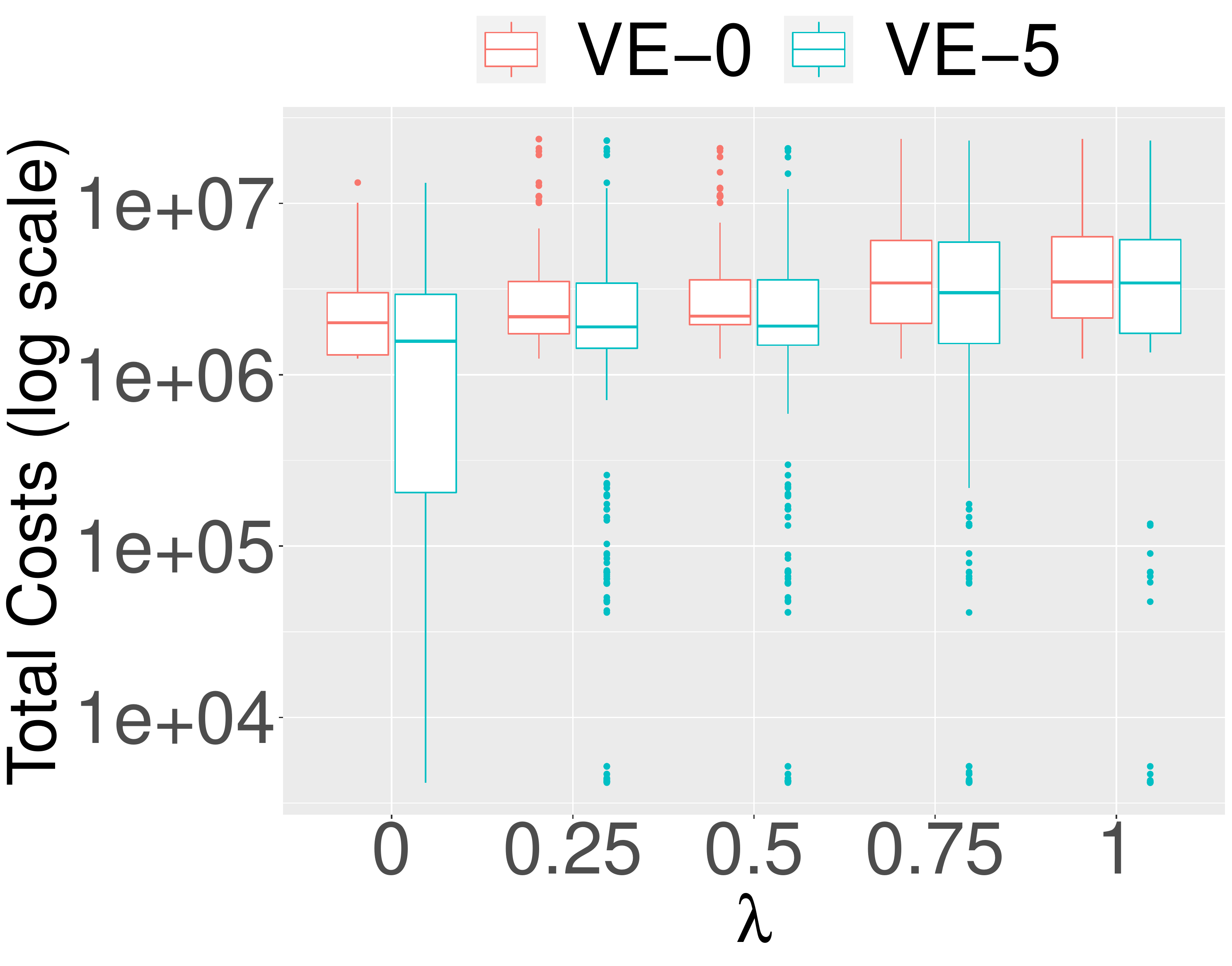}\\
		(a) \mildew  & (b) \bnpathfinder   & (c) \munins  & (d) \andes    \\
		\includegraphics[width=.235\textwidth]{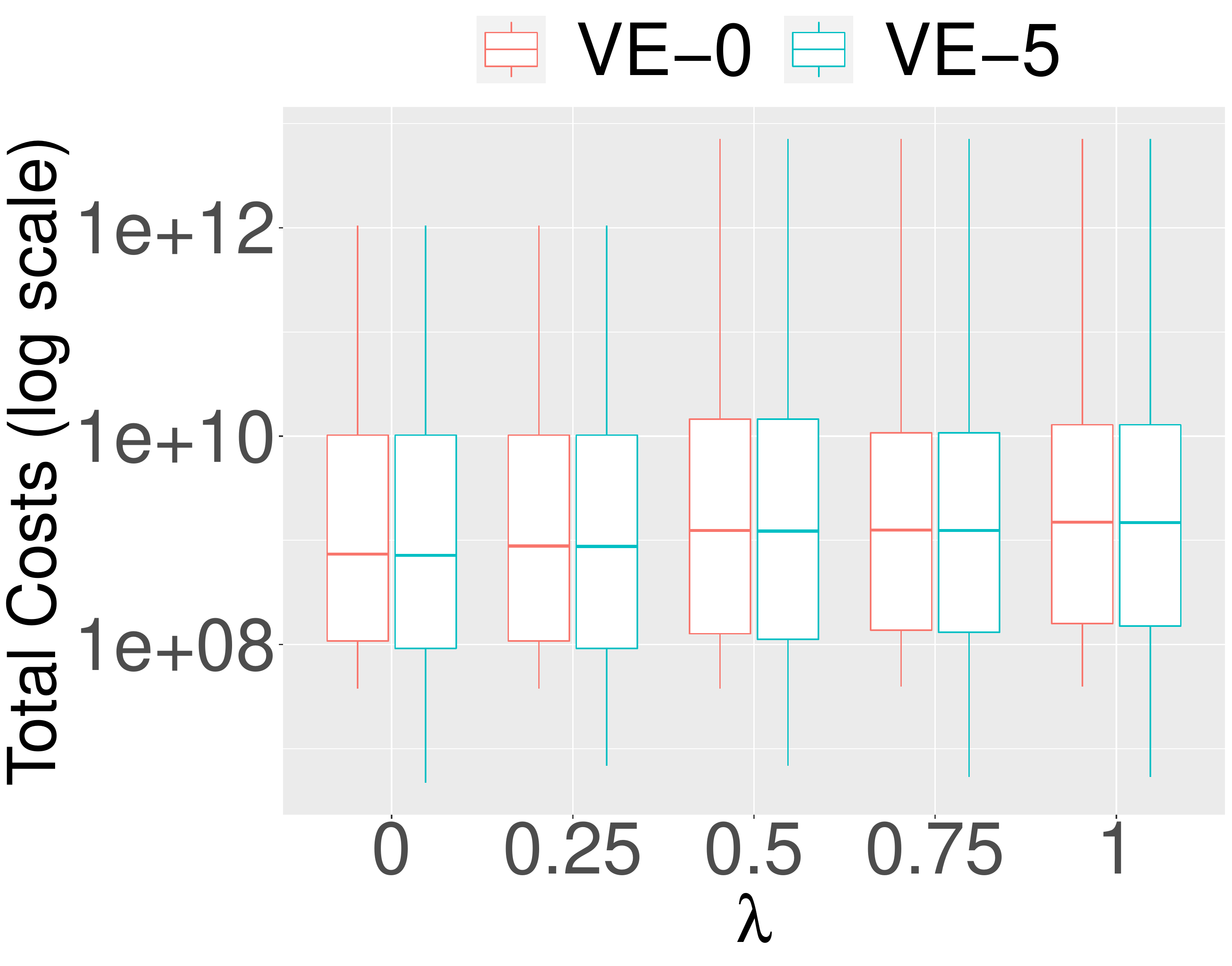} &
		\includegraphics[width=.235\textwidth]{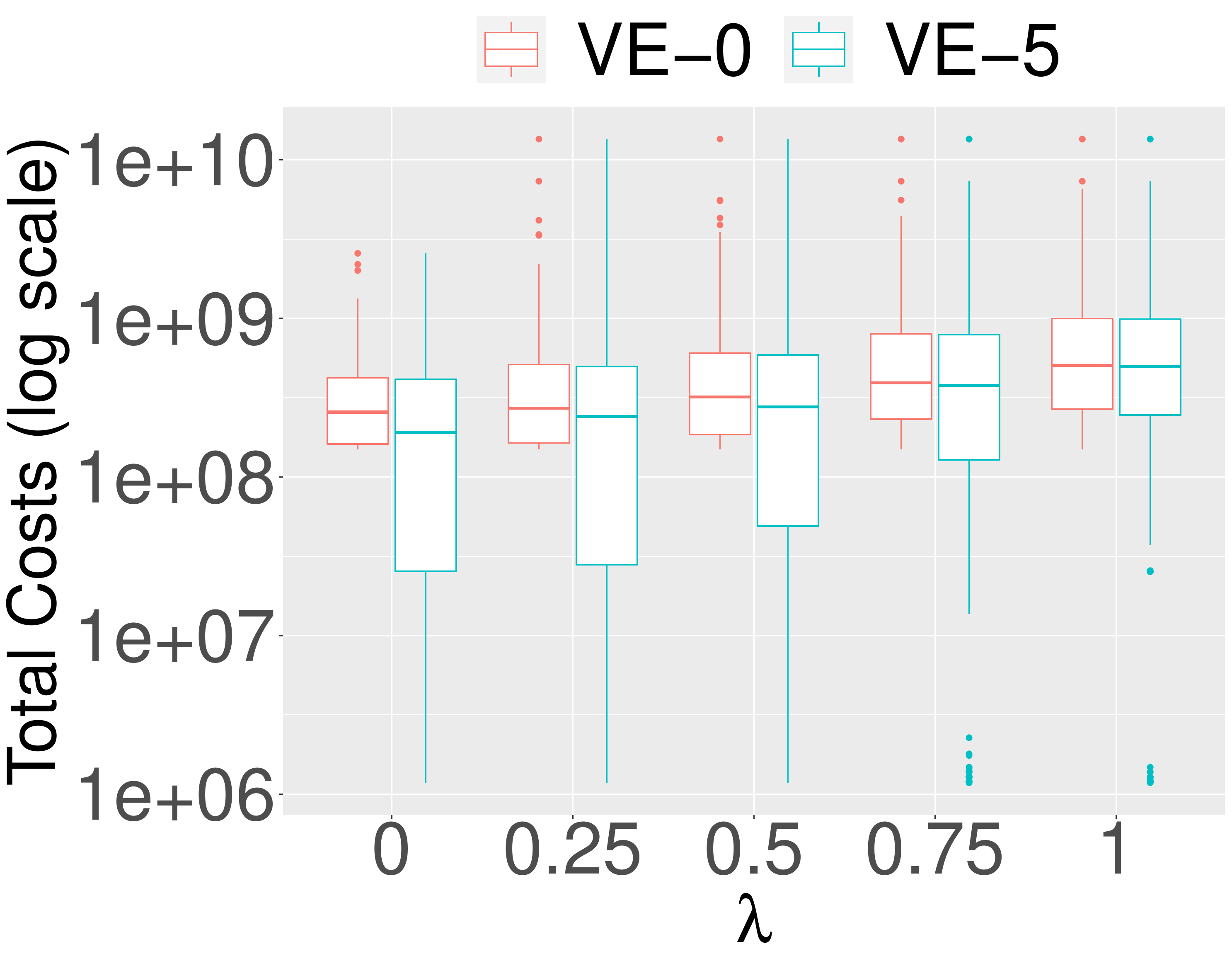} & 
		\includegraphics[width=.235\textwidth]{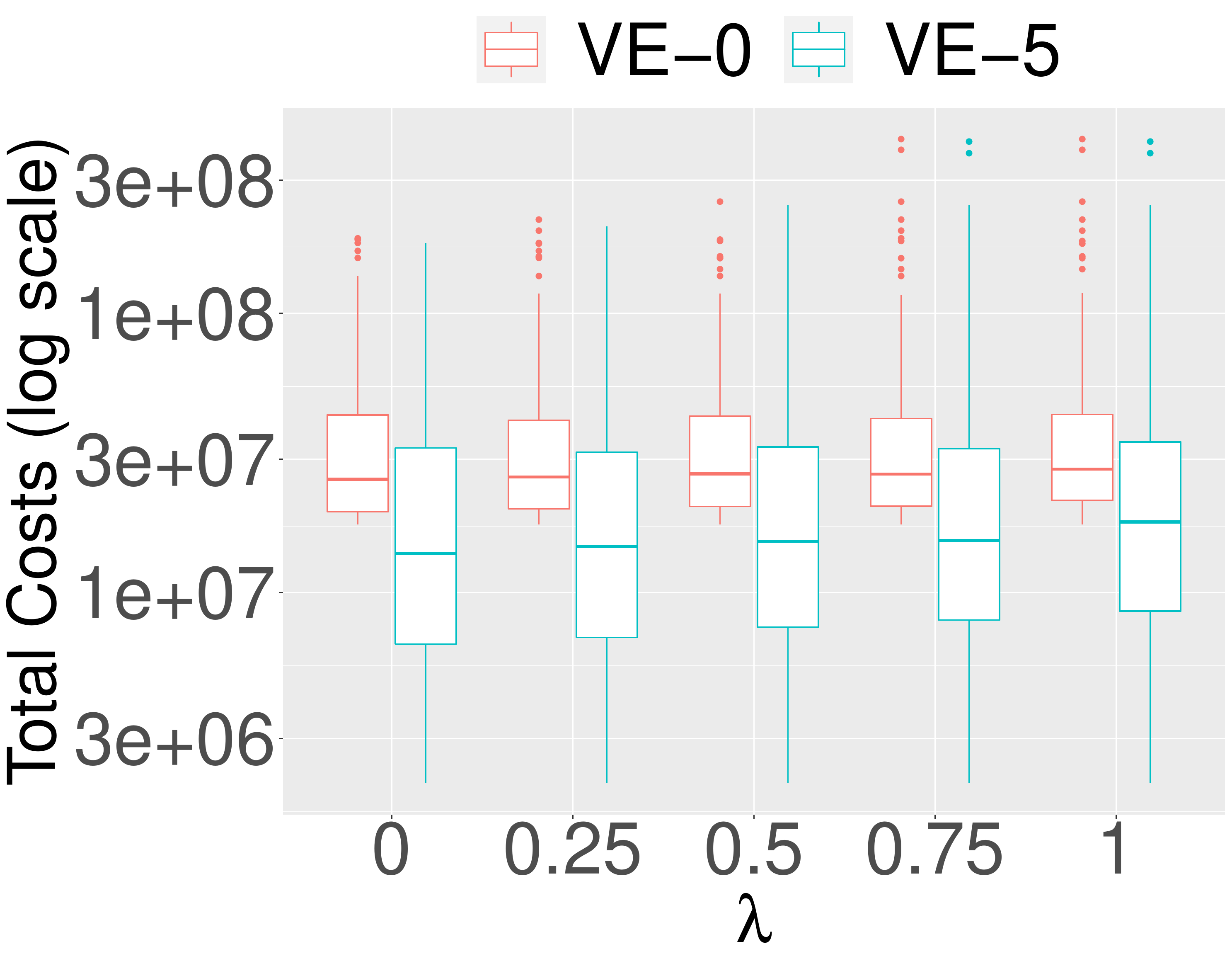} & 
		\includegraphics[width=.235\textwidth]{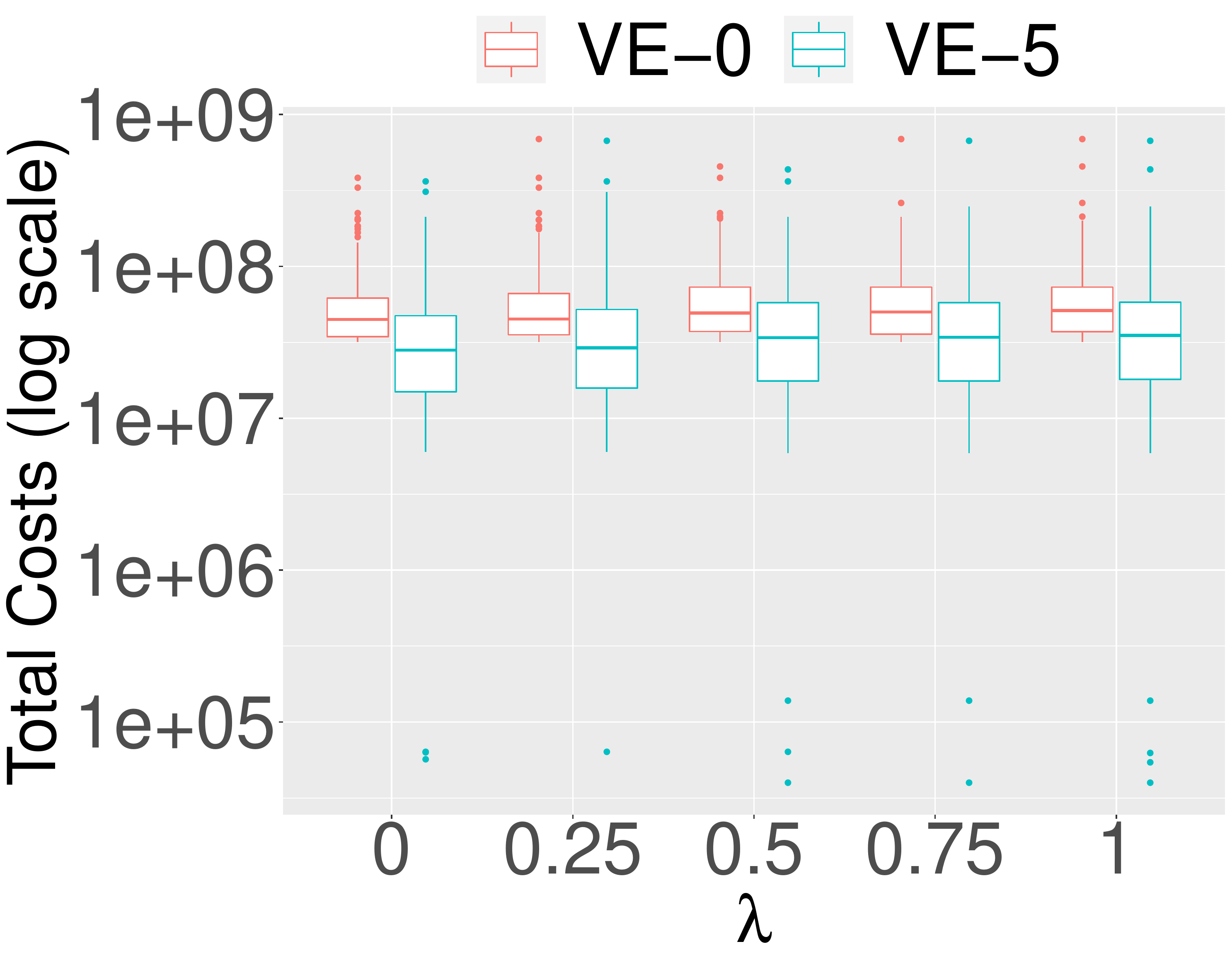}	\\
		(e) \diabetes   & (f) \link  & (g)  \muninm    & (h) \muninb   \\ 
		\includegraphics[width=.235\textwidth]{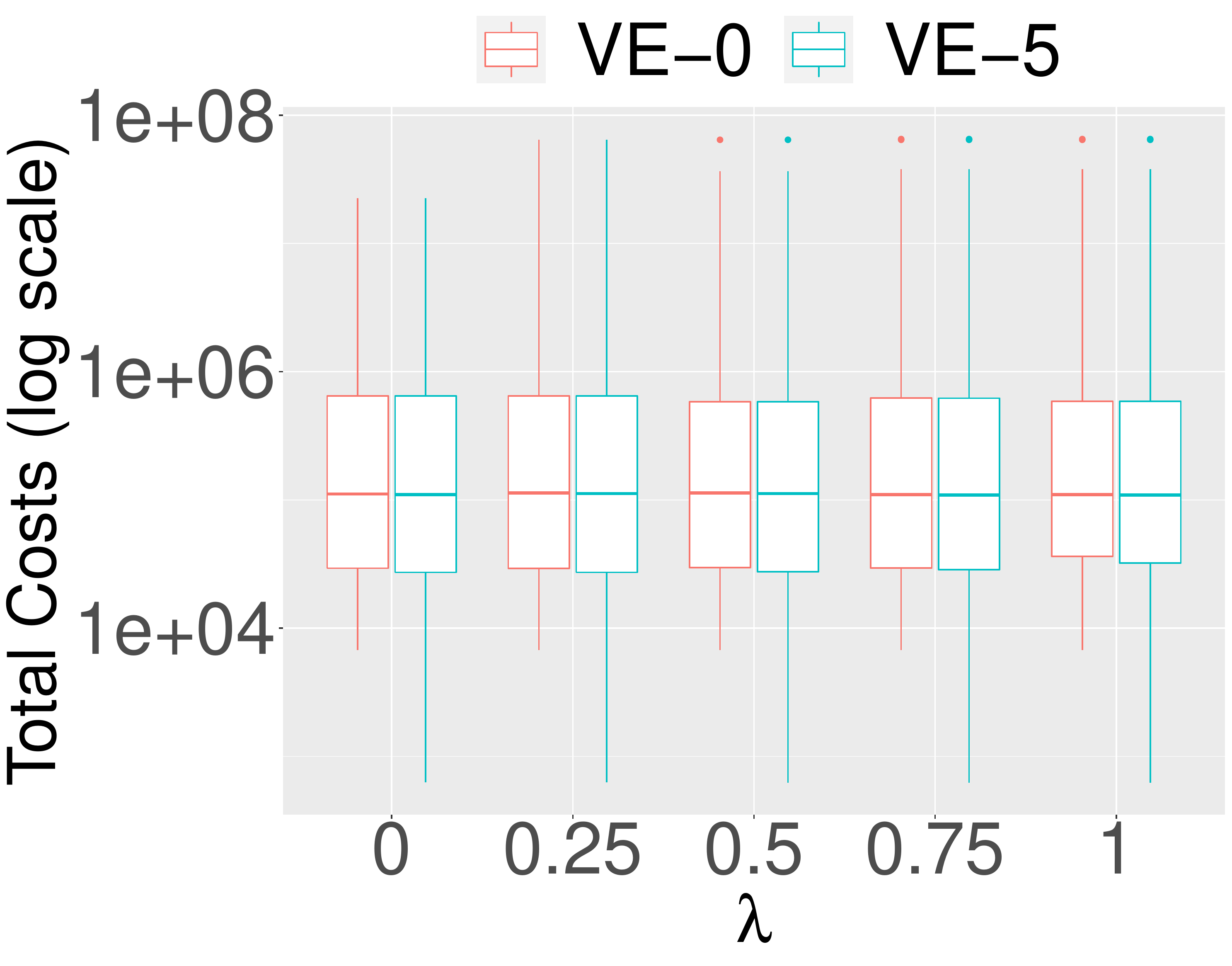} &
		\includegraphics[width=.235\textwidth]{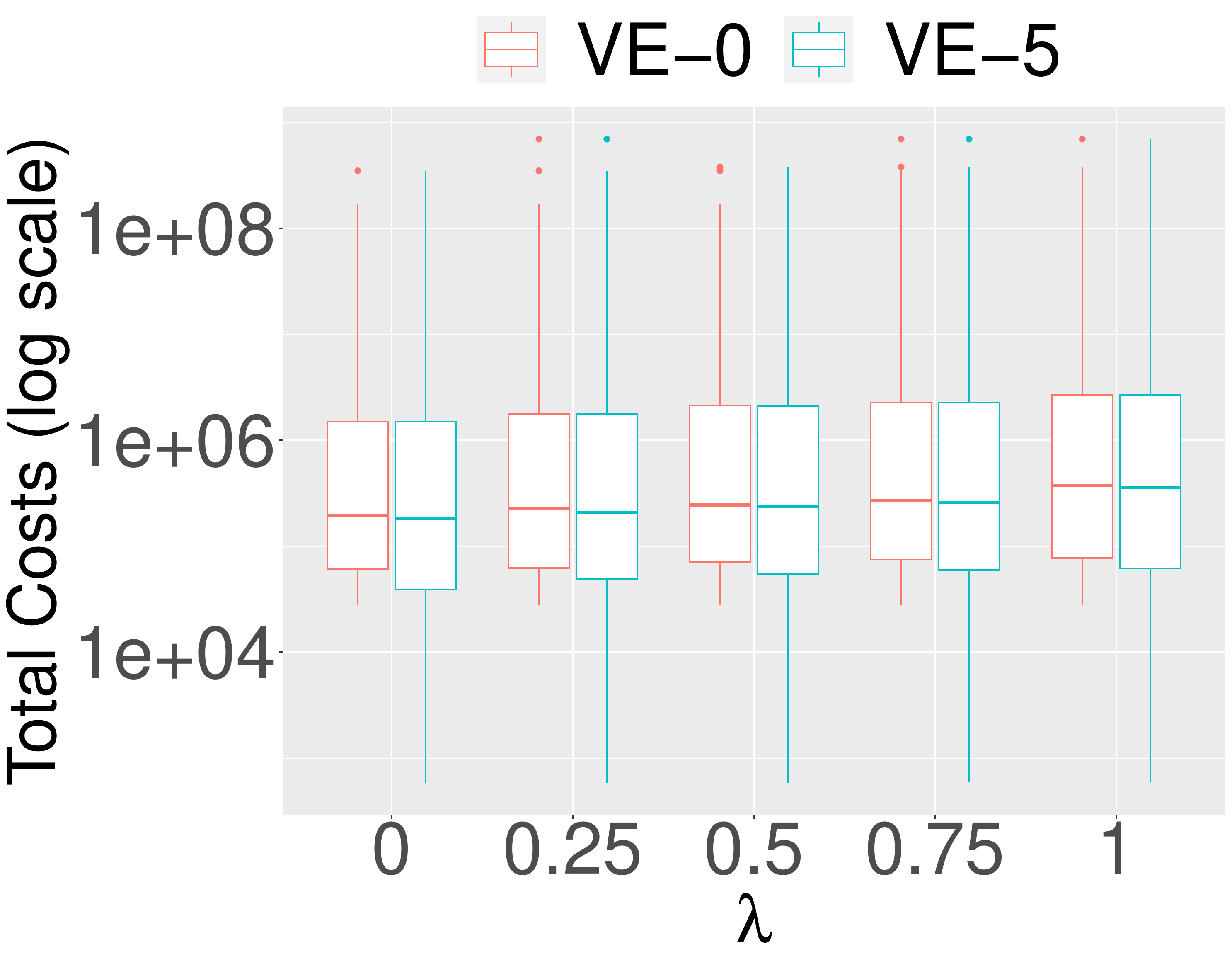} & 
		\includegraphics[width=.235\textwidth]{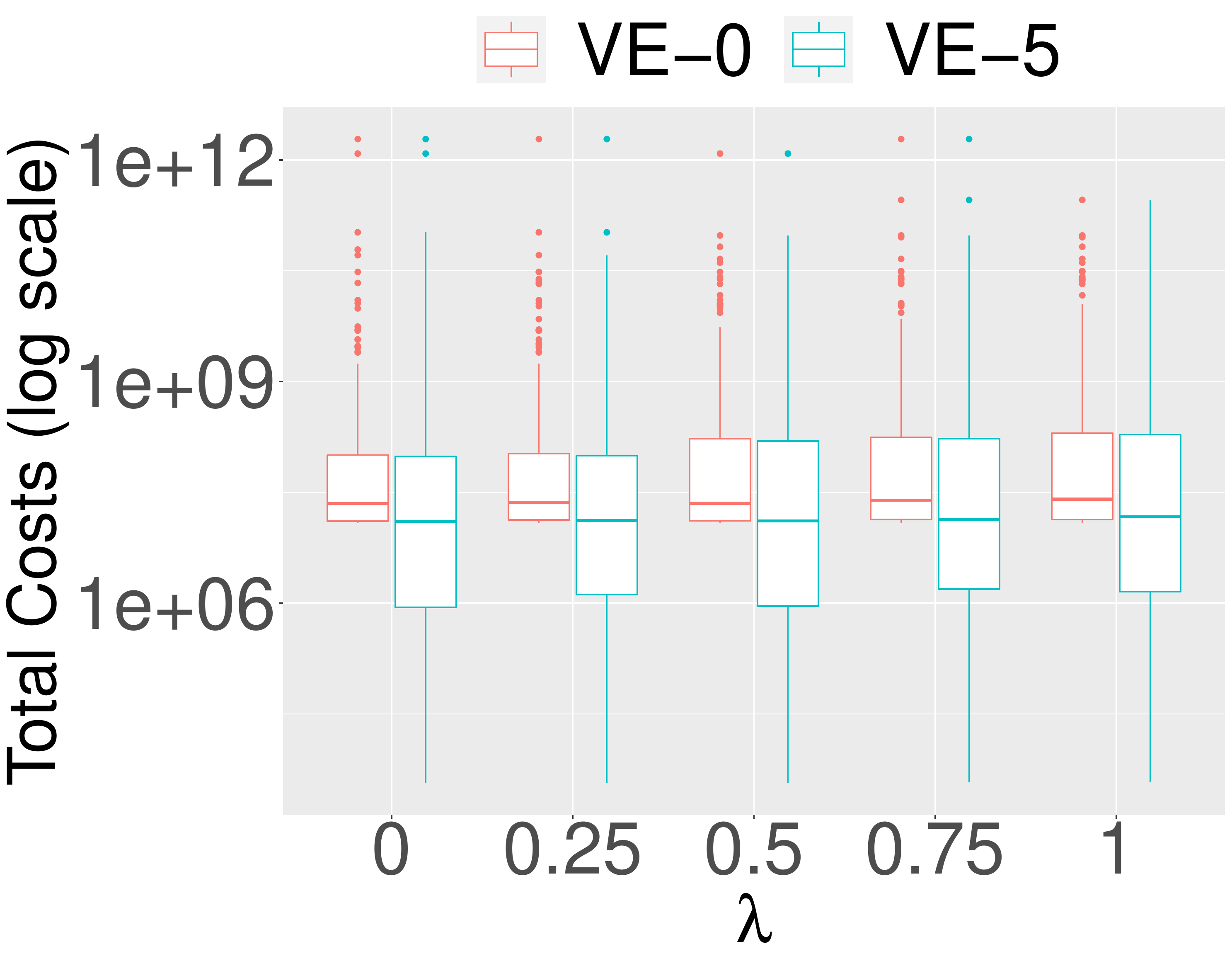} & 
		\includegraphics[width=.235\textwidth]{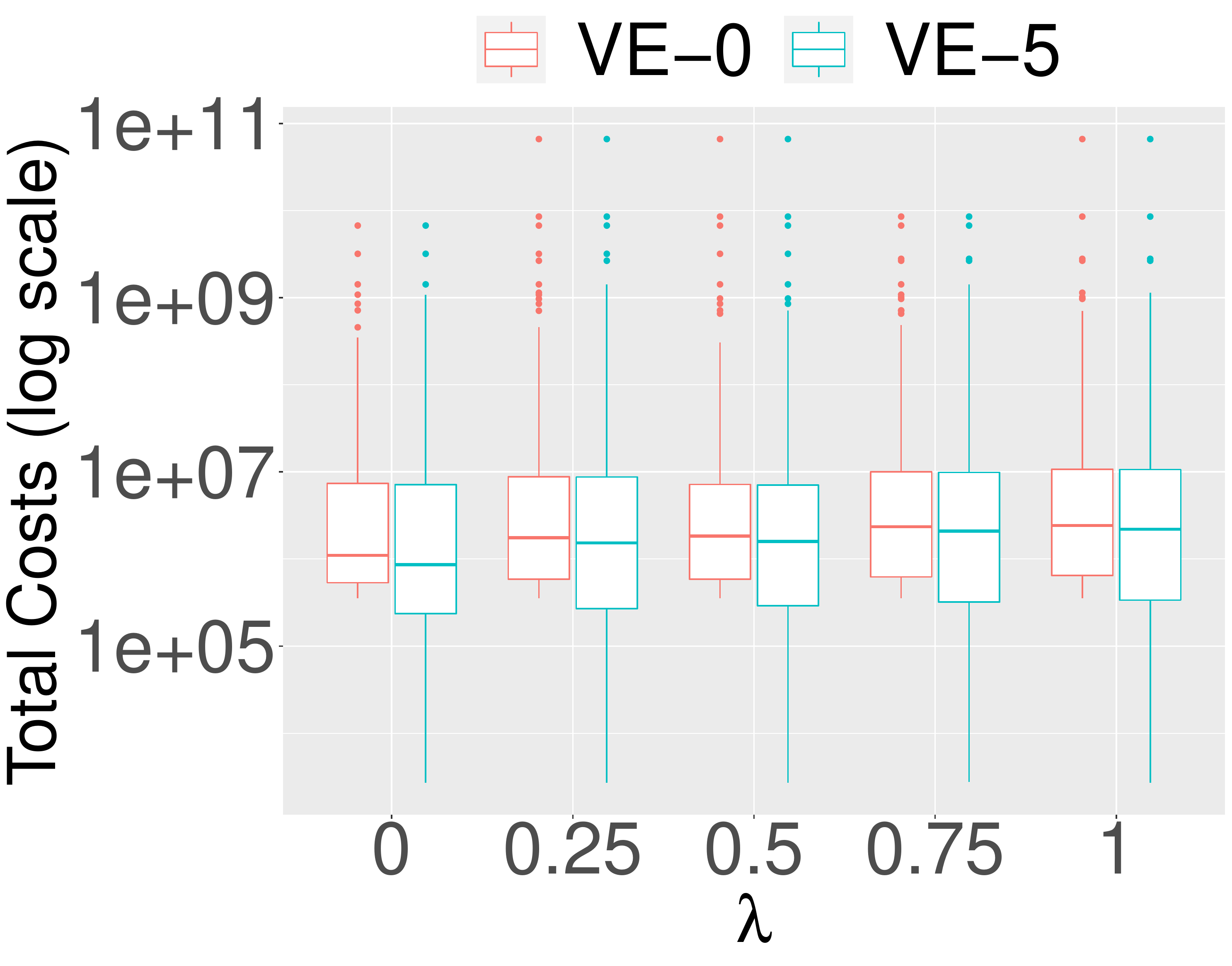}	\\
		(i) \tpchsmall   & (j) \tpchmedium  & (k)  \tpchverylarge    & (l) \tpchlarge   \\ 	
	\end{tabular}
	\caption{ {\color{black}  Robustness of materialization strategy with respect to drifts in query distribution under when \traindistr is set to skewed workload. The value of $\lambda$ on the $x-$axis denotes the proportion of queries sampled from the uniform workload scheme, and $1-\lambda$ is the proportion of queries sampled from the skewed workload scheme. The $y-$axis shows on a logarithmic scale the total cost associated with \qtm{$0$} and  \qtm{$5$} in each workload.   }}
\label{fig:robustness_lambda_skewed}
\end{figure*}
}

\FullOnly{
\smallskip
\noindent
{\bf Sensitivity to elimination order.} Finally, we demonstrate the impact of elimination order on the quality of materialization. We remind that the choice of elimination order is made heuristically based on factor sizes in the resulting elimination tree, i.e., before queries are actually executed.}
\FullOnly{Figure~\ref{fig:elim-order} shows the query-answering cost for different orders. Boxplots are shown for \qtm{$k$} with $k = 0$ (no materialization) and $k=5$. Our results suggest that the choice of elimination order has significant impact on the performance of the variable elimination algorithm, with or without materialization.}
\FullOnly{For example, \mw often provides the poorest performance, and in some cases results in factors with trillion parameters as given in Table~\ref{table:factorTableSizes}, thus is not shown in Figure~\ref{fig:elim-order}),  while in \tpch networks the same order works remarkably well. Figure~\ref{fig:elim-order} shows that our heuristic choice of order based on average parameter size works well: the chosen order is frequently the best performing or close to the best. This means that a good order can be identified and bad orders avoided before queries are actually executed.}
\FullOnly{Most importantly, notice that the \emph{relative} performance of our method is not significantly affected by such choice. For instance, in the \mildew dataset, \mf is the best order and the difference with the other orders is considerable, but the relative savings  do not vary considerably for different orders. A similar observation applies to other Bayesian networks.}

\FullOnly{
\begin{figure*}[t]
	\begin{tabular}{cccc}
		\includegraphics[width=.235\textwidth]{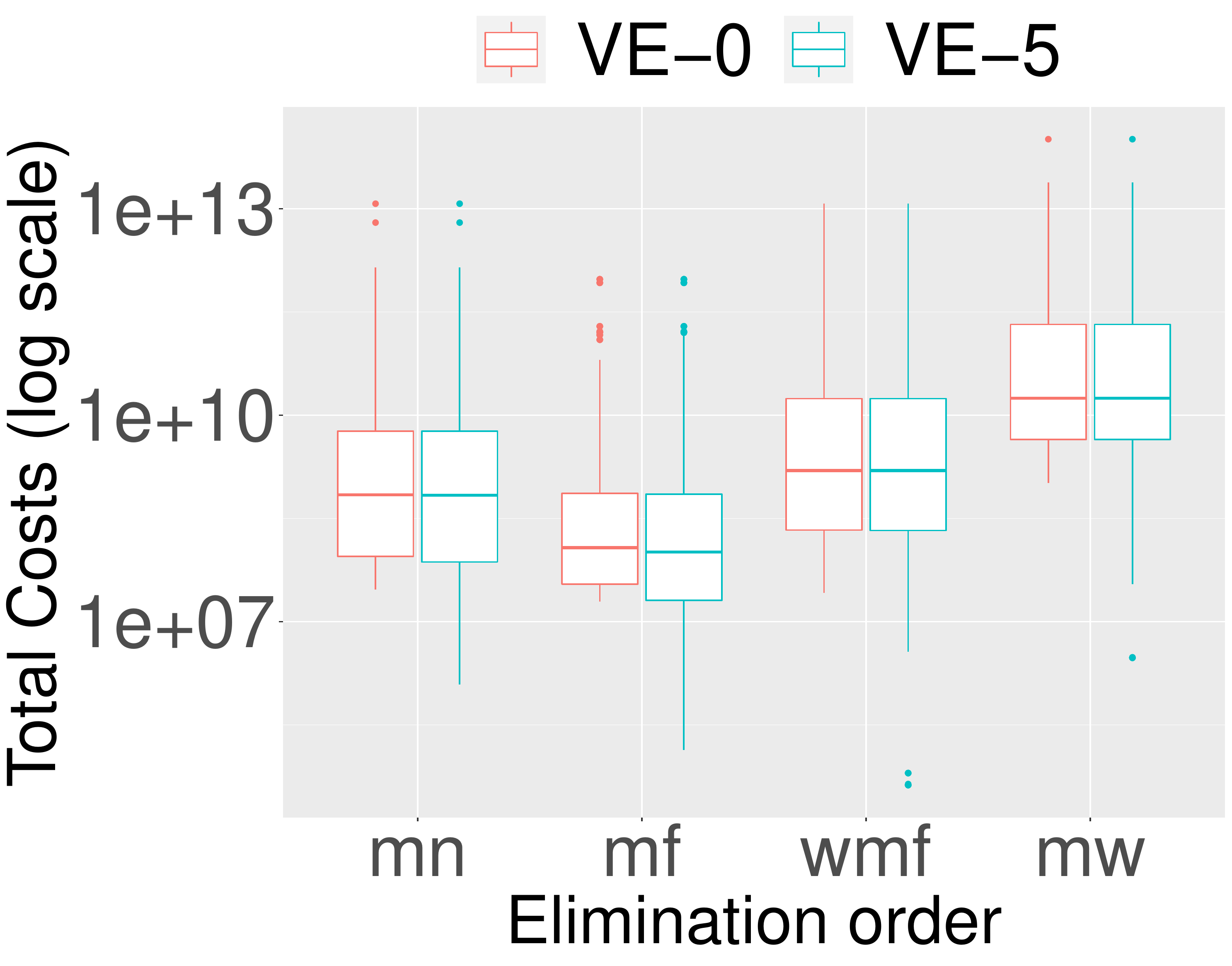}&
		\includegraphics[width=.235\textwidth]{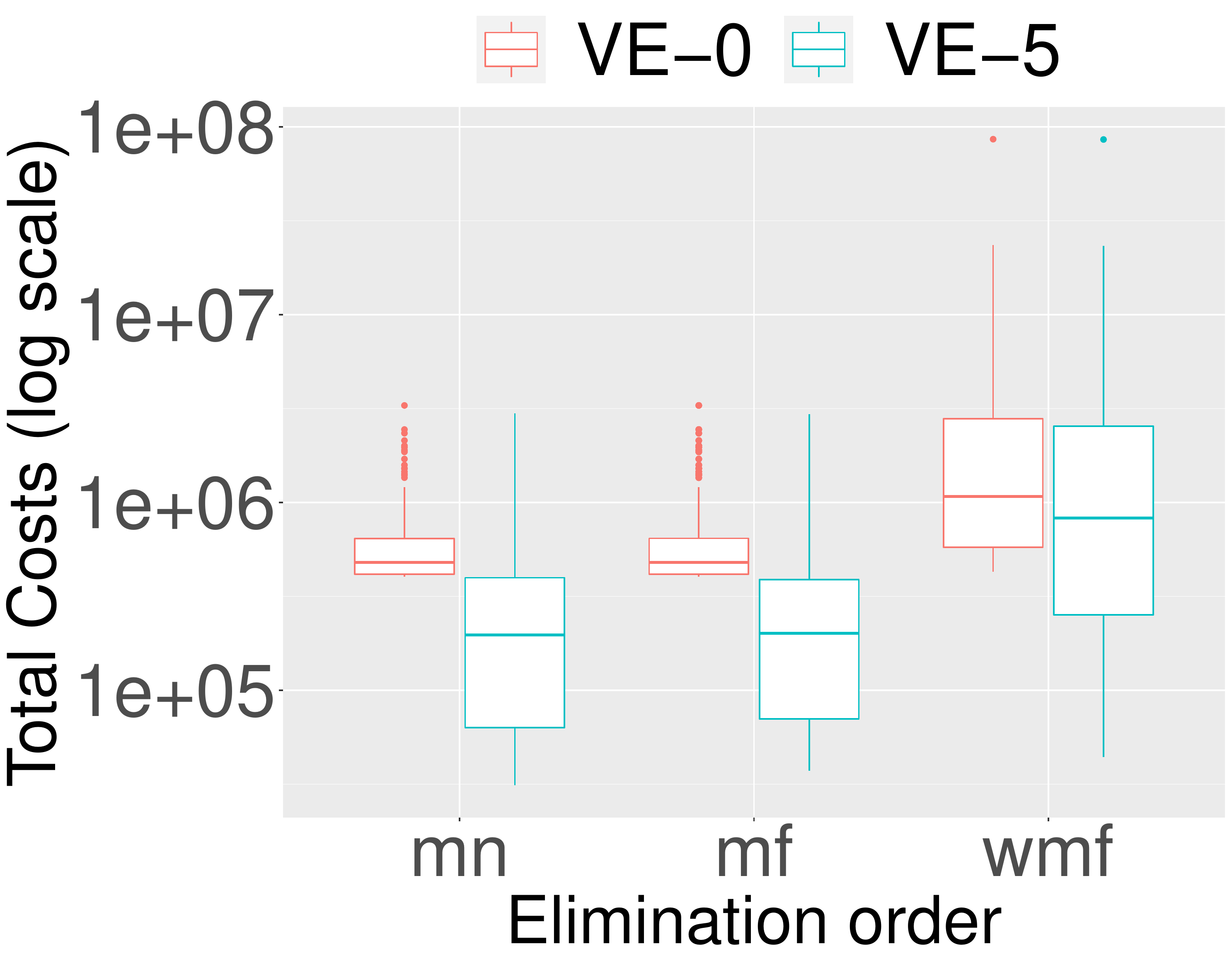}&
		\includegraphics[width=.235\textwidth]{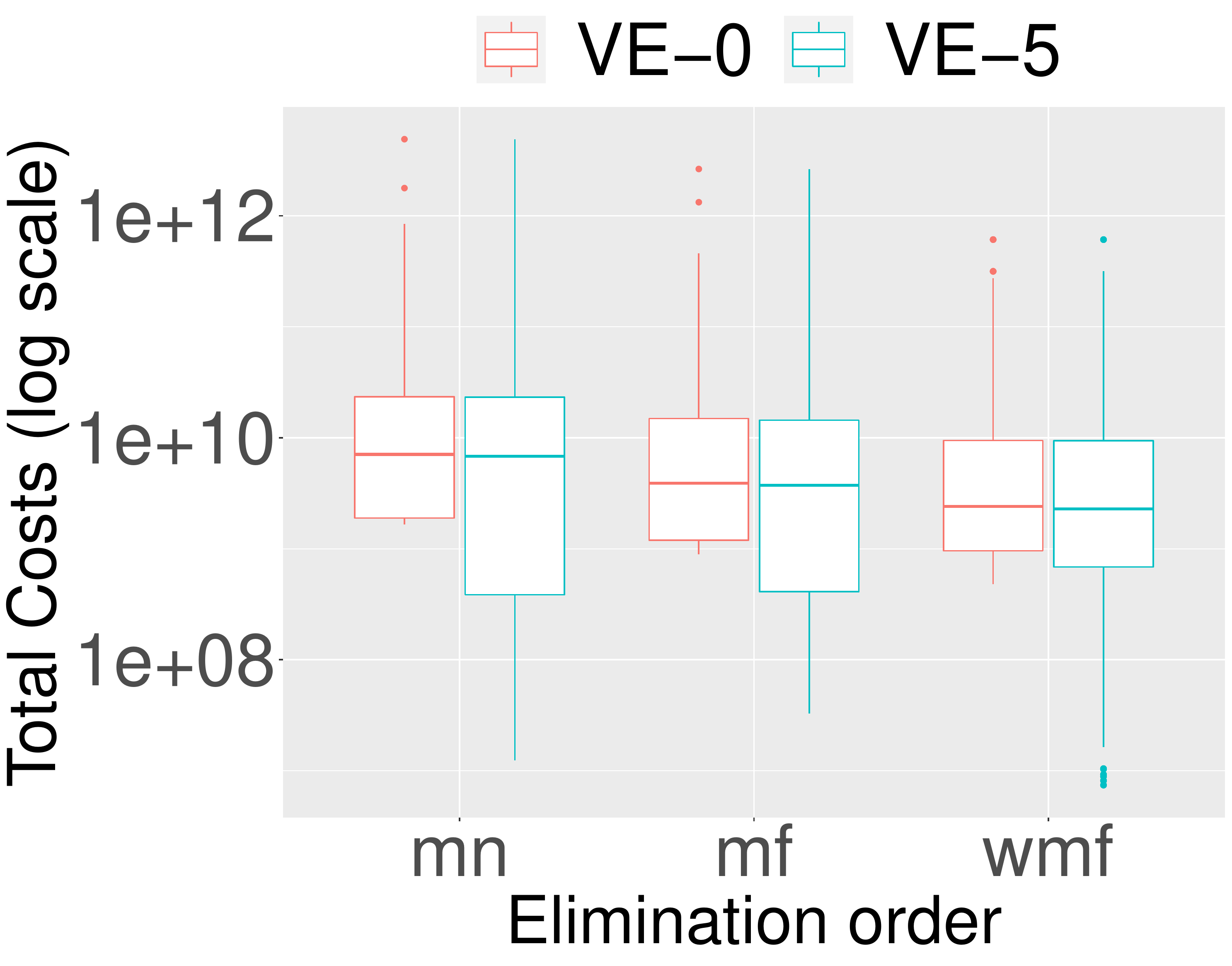}&
		\includegraphics[width=.235\textwidth]{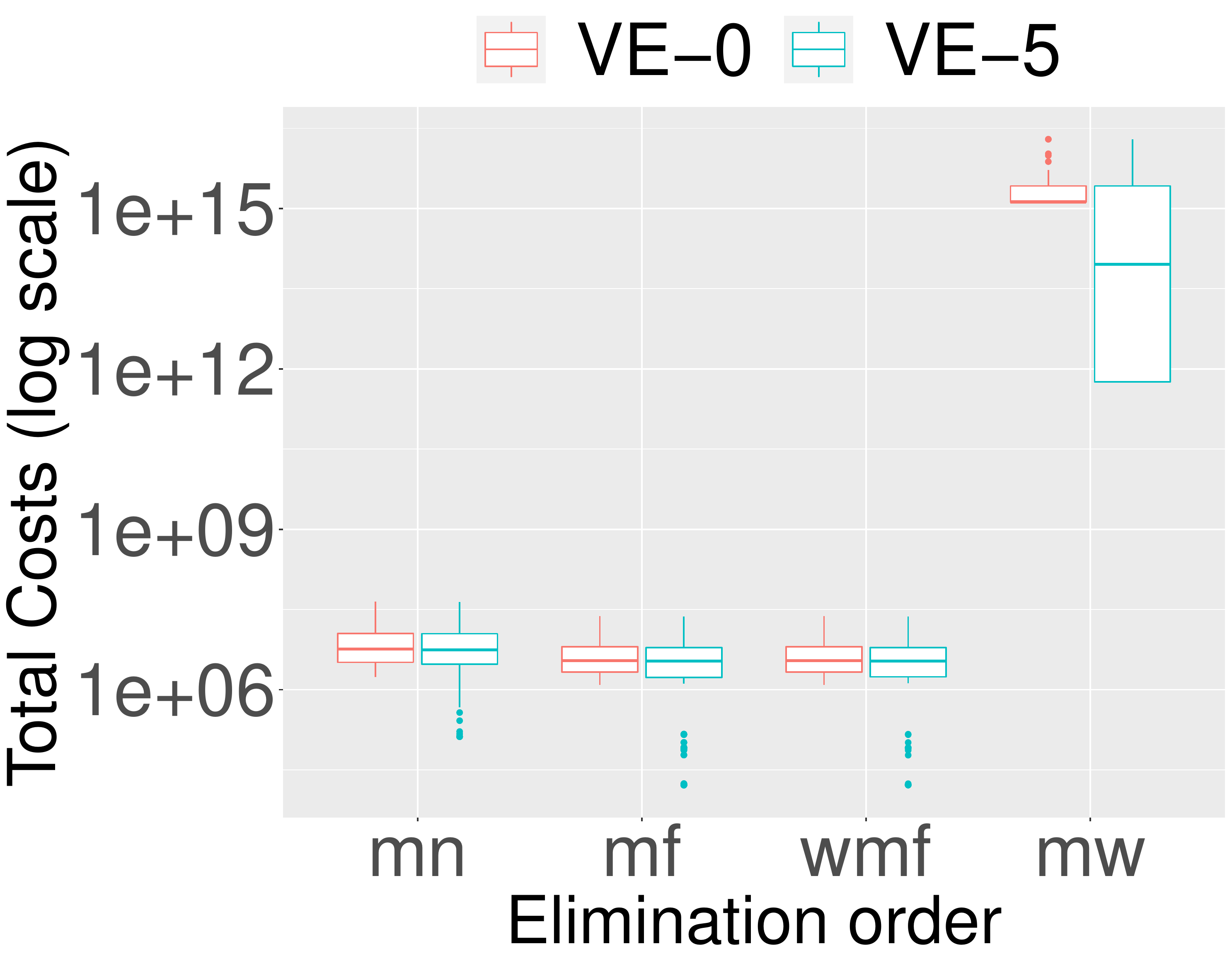}\\
		(a) \mildew (\mf)  & (b) \bnpathfinder (\mf)   & (c) \munins (\wmf)  & (d) \andes (\mf) \\
		\includegraphics[width=.235\textwidth]{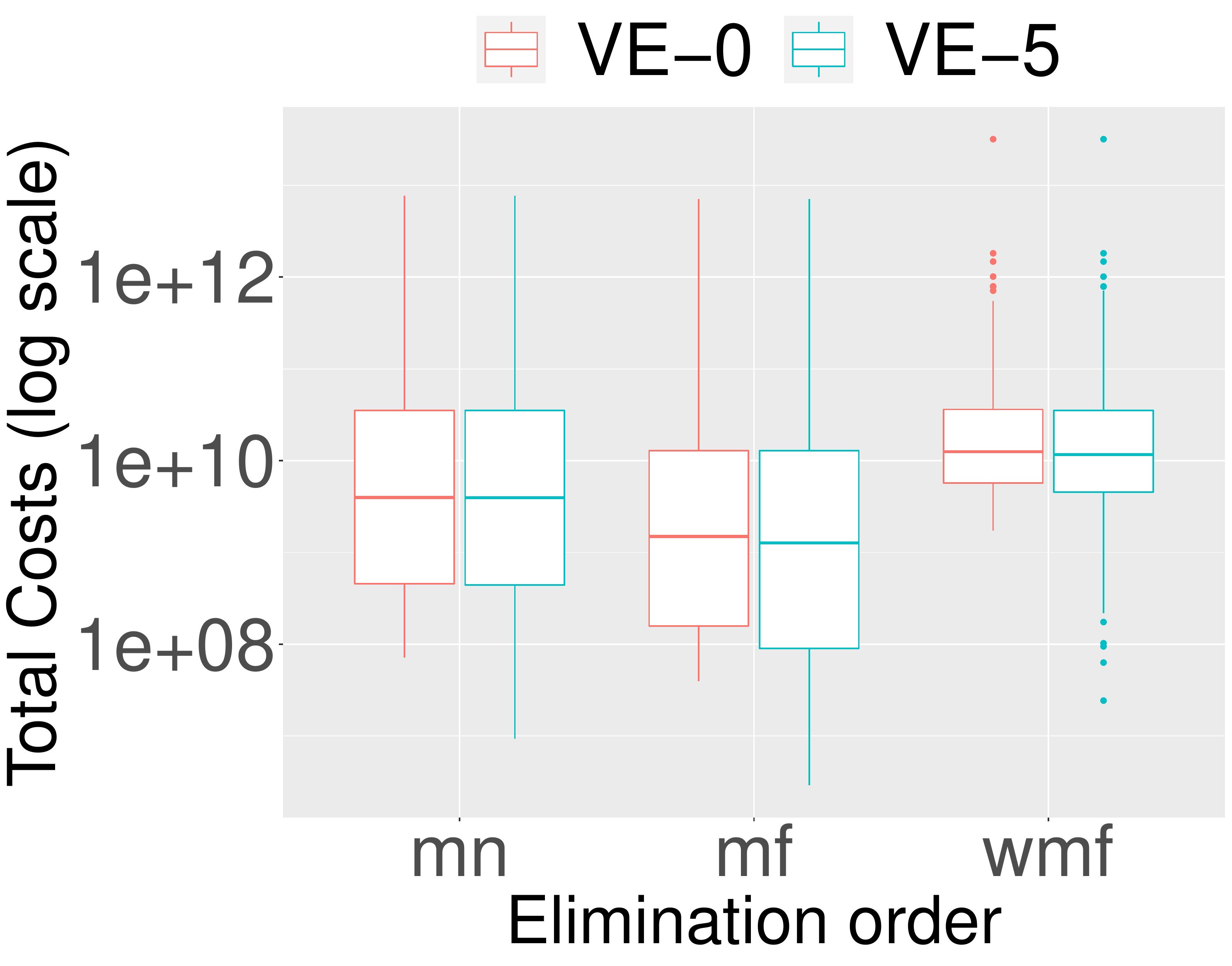} &
		\includegraphics[width=.235\textwidth]{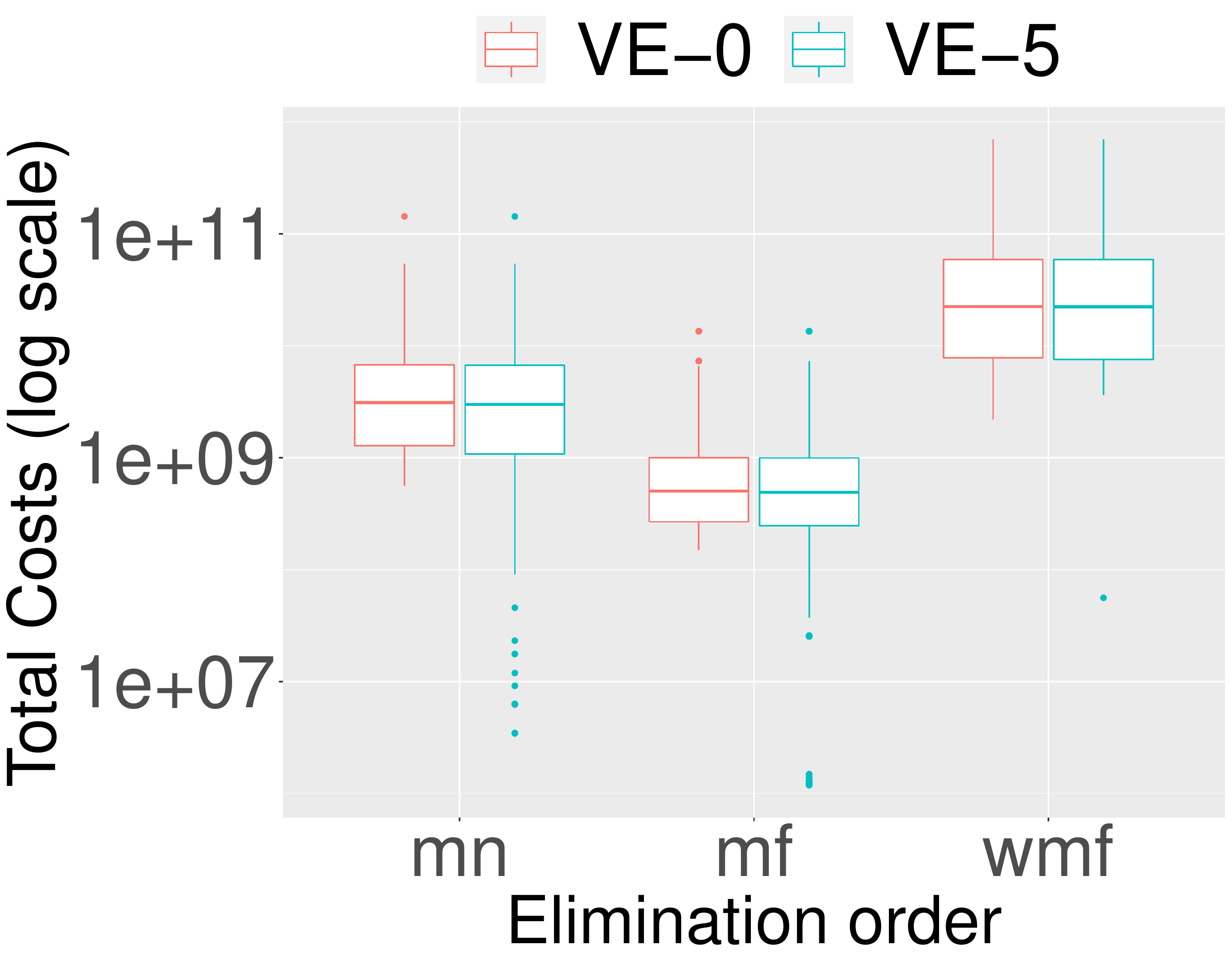} & 
		\includegraphics[width=.235\textwidth]{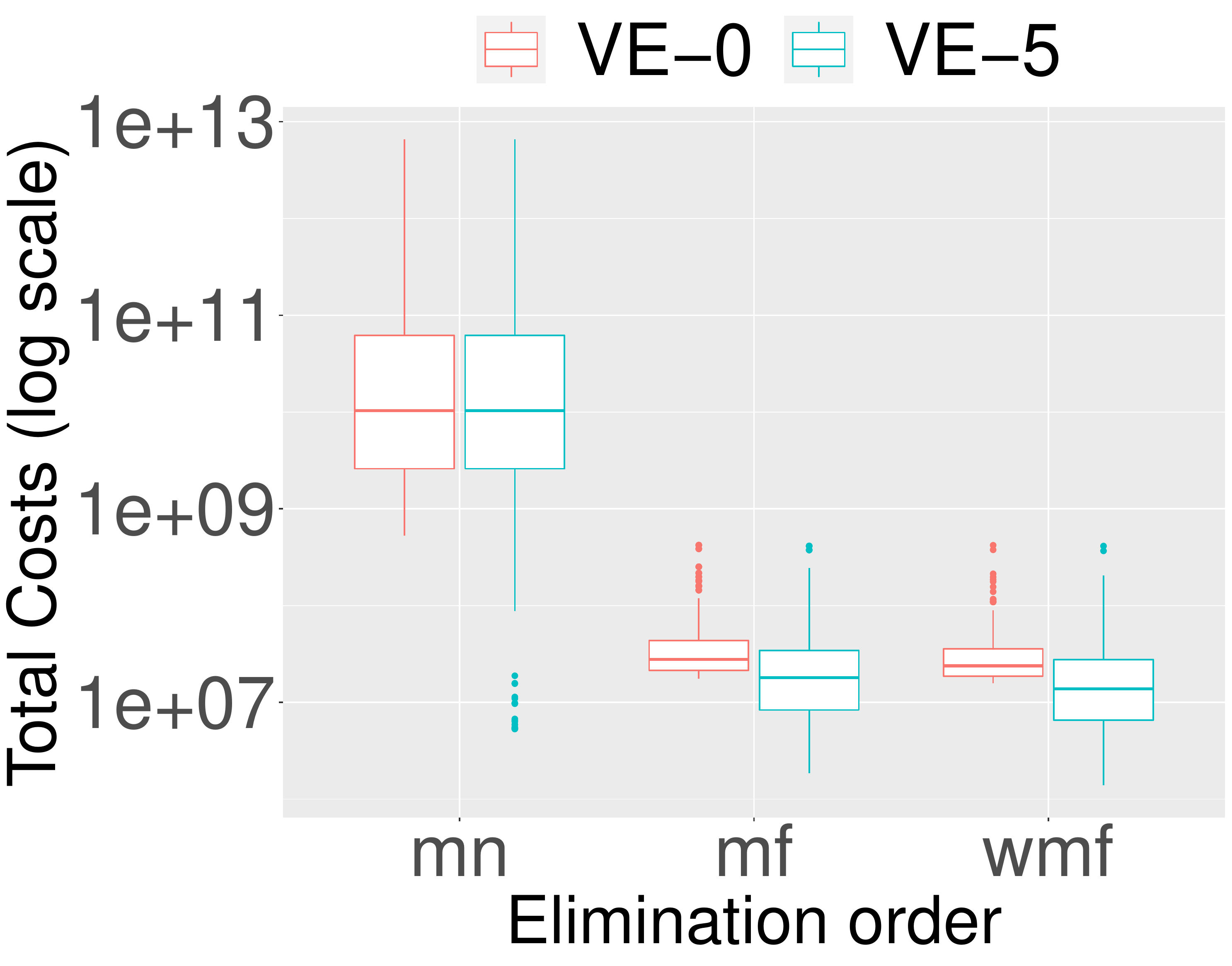} & 
		\includegraphics[width=.235\textwidth]{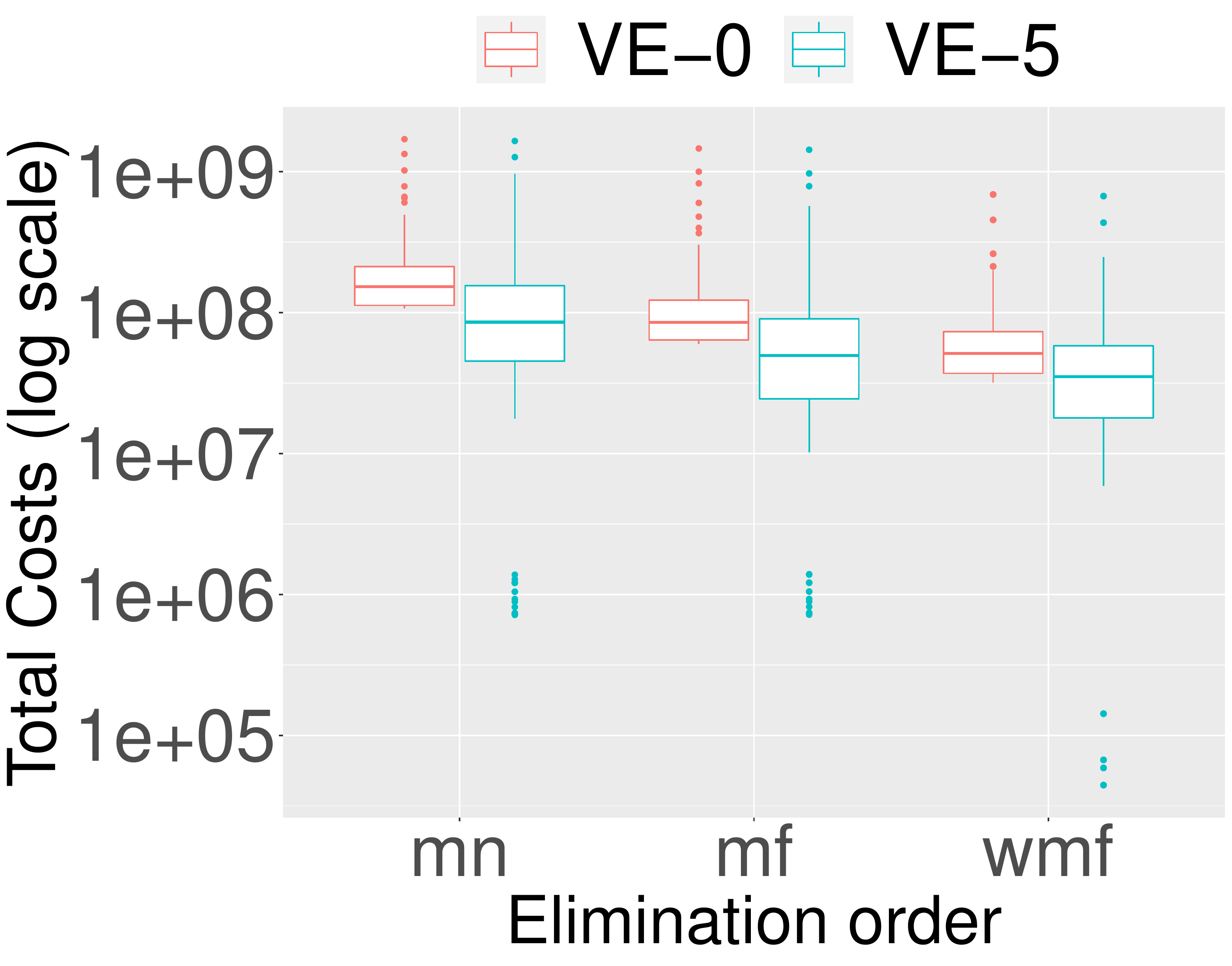}	\\
		(e) \diabetes (\mf)  & (f) \link (\mf)  & (g)  \muninm (\mf)    & (h) \muninb (\wmf)   \\ 
		\includegraphics[width=.235\textwidth]{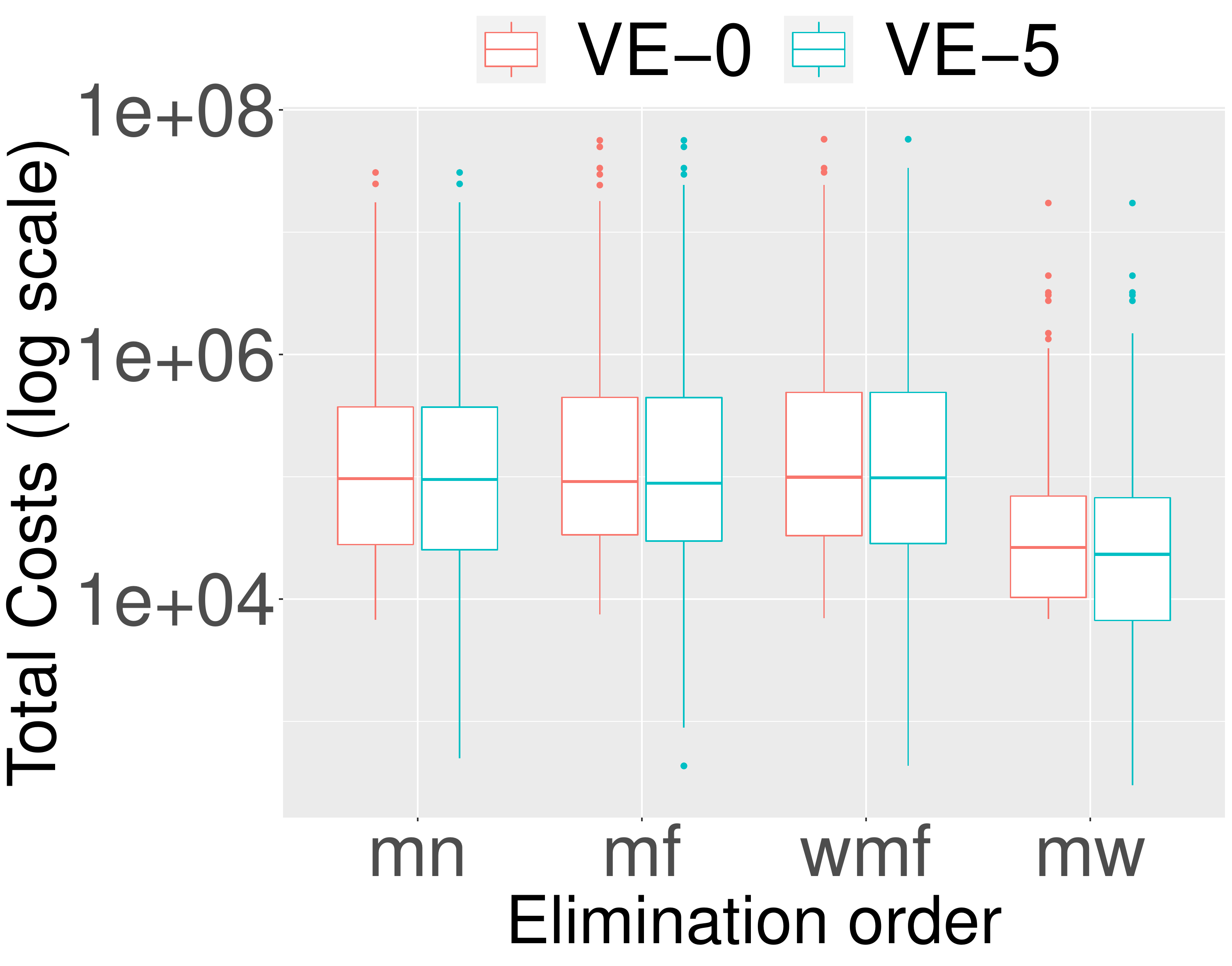} &
		\includegraphics[width=.235\textwidth]{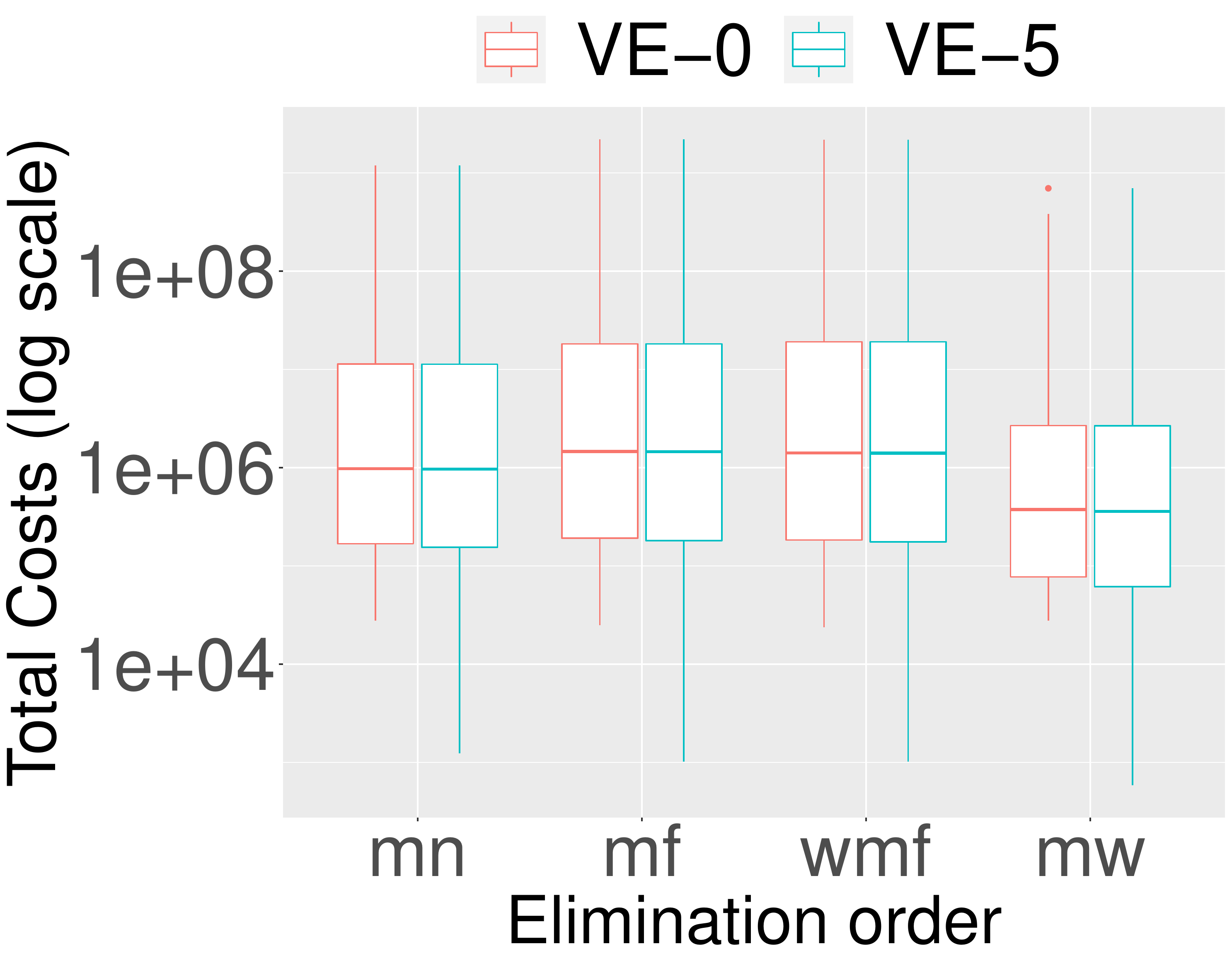} & 
		\includegraphics[width=.235\textwidth]{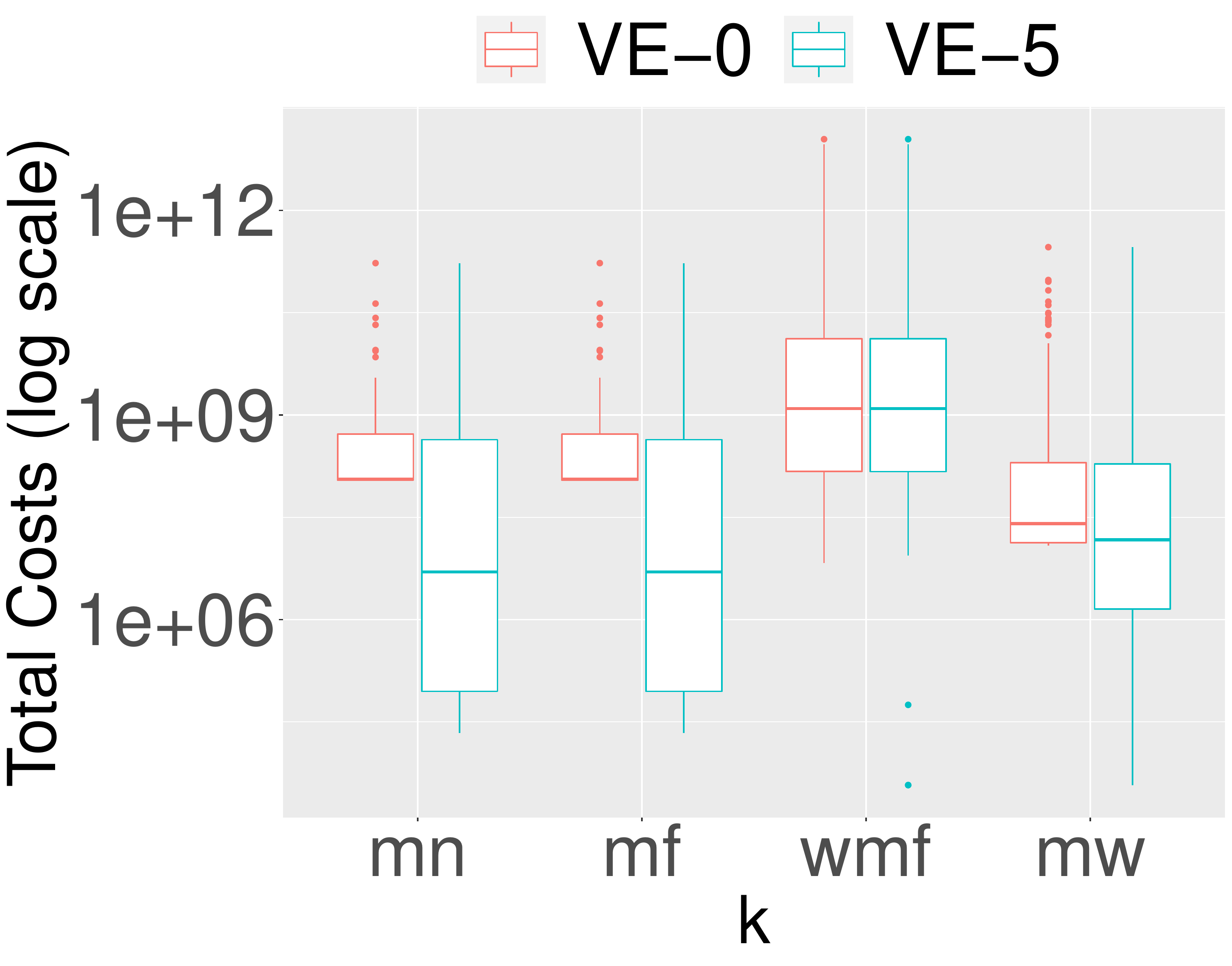} & 
		\includegraphics[width=.235\textwidth]{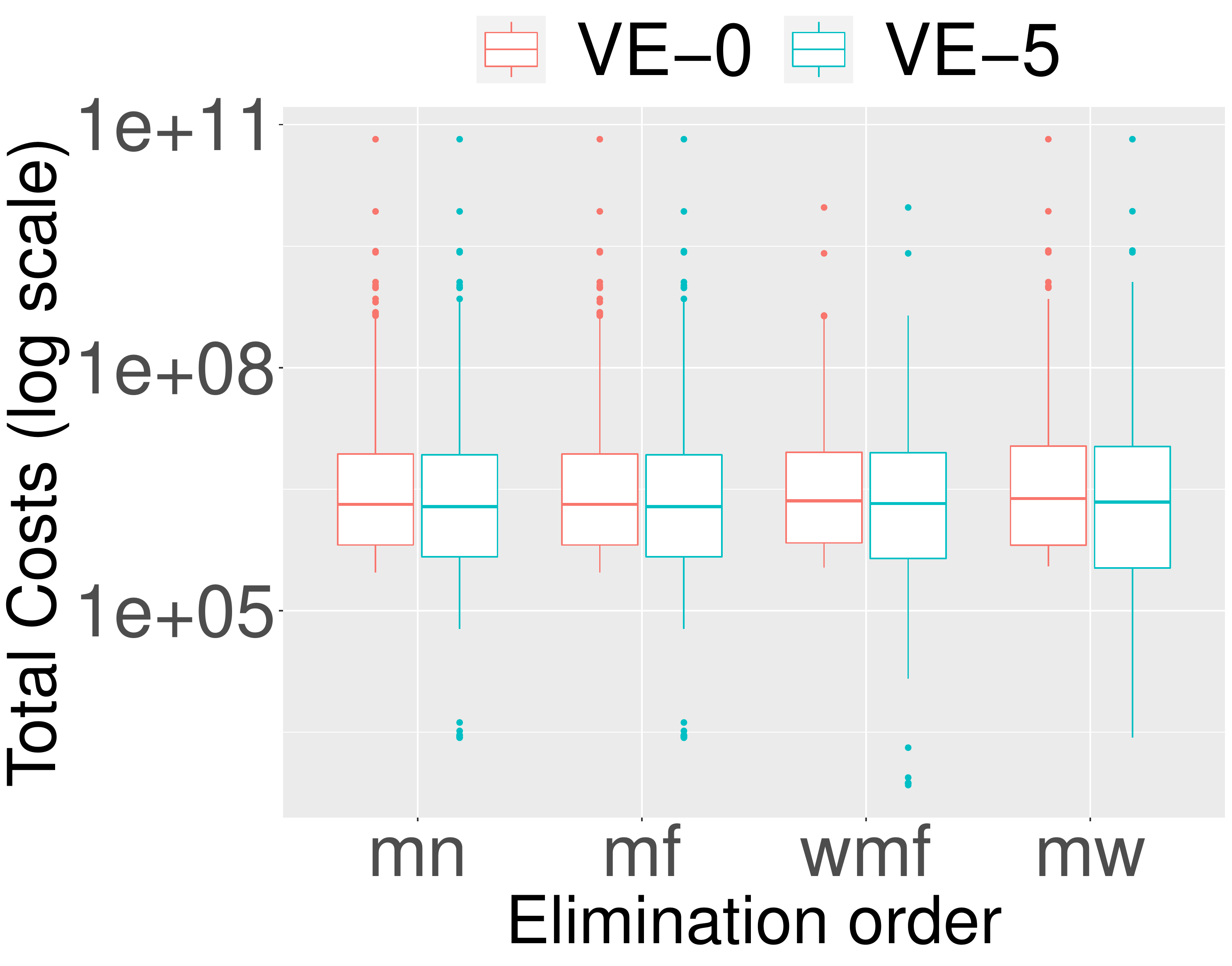}	\\
		(i) \tpchsmall ({\mw})  & (j) \tpchmedium ({\mw})  & (k)  \tpchverylarge ({\mw})    & (l) \tpchlarge ({\mw})   \\ 	
	\end{tabular}
	\caption{ {\color{black}  Impact of elimination order on materialization. Each pair of box-plots shows the total cost on the logarithmic scale associated with \qtm{$0$} and  \qtm{$5$} under a given elimination order. The chosen order is shown in parenthesis.}}
	\label{fig:elim-order}
\end{figure*}
}
\section{Conclusions}
\label{section:conclusion}

% State-of-the-art machine learning techniques allow us to build probabilistic models, such as Bayesian networks, and use them to perform approximate query processing and predictive querying over very large databases. To make such an approach even more scalable, it is crucial to consider what expensive model computations can be performed in a preprocessing step, so as to make query answering more efficient. 

In this paper, 
we studied the problem of materializing intermediate relational tables created during the evaluation of queries over a Bayesian network using variable elimination. We presented efficient algorithms to choose the optimal tables under a budget constraint for a given query workload and variable-elimination order. 
Our experiments show that appropriate materialization can offer significant inference speed-up and
competitive advantages over junction-tree algorithms. 

Our work contributes to the growing research on data-management issues in machine-learning systems~\cite{hasani2018efficient}. 
One direction for future work is to adapt our framework to other concise factor representations, such as arithmetic circuits~\cite{chavira2005compiling, chavira2007compiling} and sum-product networks~\cite{poon2011sum} or junction trees~\cite{kanagal2009indexing}. 
\revision{
Moreover, in the context of specific applications such as AQP, it is crucial to investigate the maintenance of the materialized factors in the presence of updates.
Finally, while the proposed method is shown to be quite robust with respect to changes of the query distribution, 
it will be valuable to implement a drift-detection mechanism, 
which automatically prompts an update in the materialization strategy. 
}

{% \small
\balance
\bibliographystyle{IEEEtran}
\bibliography{biblio-brief}
}

\cleardoublepage\makeatletter\makeatother
\FullOnly{
\begin{appendix}
\spara{Generality of Problem Setting.}

We formulated Problems~\ref{problem:qtm-space} and~\ref{problem:qtm} in terms of an elimination tree, without specifying any constraints on its structure. However, in the context where we're addressing Problem~\ref{problem:qtm}, elimination trees are defined by a Bayesian network and an elimination order. At this point, we should consider the question of whether the elimination trees that are input to Problem~\ref{problem:qtm-space} and~\ref{problem:qtm} are of general structure --- if they were not, then this would leave room for faster algorithms that take advantage of the special structure.

Lemma~\ref{lemma:generalStructure} states that the internal nodes of an elimination tree form a tree of general structure.
Note that we use the term ``internal subtree'' to refer to the subgraph of a tree that is indiced by internal (non-leaf) nodes.
For an elimination tree \tree, its internal subtree corresponds to the nodes that correspond to the variables of the Bayesian network.

\begin{lemma}
Consider any tree \tree and its internal subtree \coreTree. There is a Bayesian network and an elimination order that define an elimination tree \otherTree with $\otherCoreTree \simeq \coreTree$.
\label{lemma:generalStructure}
\end{lemma}

%\hide{
\begin{proof}
Assume that \coreTree is rooted at node \varomega with subtrees rooted at nodes $\nodeu_\vara$, $\nodeu_\varb$, \ldots, $\nodeu_\varz$.
Consider a Bayesian Network \bayesianNet that has one corresponding variable \varomega, \vara, \varb, \ldots, \varz for each node of \coreTree, and conditional probabilities in the opposite direction compared to the respective edges of \coreTree --- i.e., if \coreTree contains a directed edge $(\nodeu_\vara, \nodeu_\varomega)$, the Bayesian network \bayesianNet contains the edge $(\varomega, \vara)$ with associated conditional probability $\prob{\vara \mid \varomega}$.
Let \eliminationOrder be the elimination ordering of variables corresponding to the post-order traversal of \bayesianNet. 
Consider the directed edges (\varomega, \vara), (\varomega, \varb), (\varomega, \varc) etc, coming out of \varomega (see Figure~\ref{fig:tree_gen_simple}).
The elimination ordering we chose ensures that variables \vara, \varb, \varc, etc, are eliminated 
before \varomega. It also ensures that {\it the factor that results from the elimination of 
variable \vara is a function of \varomega} {\it but of none of the other variables \varb, \varc, etc, that depend on \varomega}, by construction.
The elimination of \varomega happens after that of each variable \vara, \varb, \varc, etc, it connects to, leading to an elimination tree \otherTree with an edge from each of those variables to \varomega\ --- and therefore to an internal subtree \otherCoreTree that is isomorphic of \coreTree.
\end{proof}
%}

\begin{figure}[t!]
\begin{center}
\includegraphics[width=0.90\columnwidth]{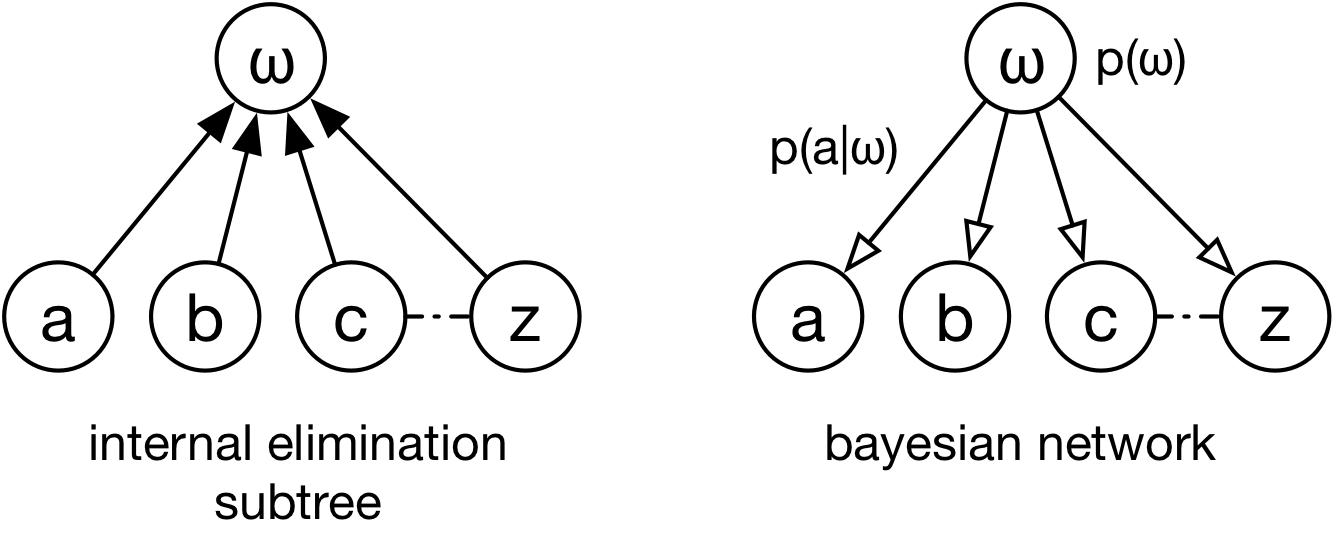}
\end{center}
\caption{The tree on the left represents \coreTree, the subtree of elimination tree \tree induced by its internal nodes. The Bayesian Network in the figure, constructed by reversing the direction of edges in \coreTree, along with its post-order traversal $(\vara, \varb, \ldots, \varz, \varomega)$ define \tree as the elimination tree.}
\label{fig:tree_gen_simple}
\end{figure}

We can also show, however, that if we limit ourselves to elimination orderings that follow a pre-order traversal (\preorderTraversal) of the Bayesian network, we end up with elimination trees of linear structure.
Note that the Bayesian network \bayesianNet we consider is generally a {\it directed acyclic graph} (\DAG) and not necessarily a tree.
It is therefore possible to have multiple roots in \bayesianNet --- i.e., multiple variables that 
do not depend on other variables.
% At the same time, however, for any given bayesian network \bayesianNet, we can construct an equivalent bayesian network $\otherBayesianNet$, by adding to \bayesianNet a single-valued dummy variable that is a parent to all the roots of \bayesianNet and therefore becomes the unique root in $\otherBayesianNet$.
For any variable \varomega, consider the set of variables \vara, \varb, \ldots, \varz that depend on \varomega.
A \preorderTraversal of \bayesianNet is any ordering of its variables such that \varomega precedes \vara, \varb, \ldots, \varz.
The breadth-first-search of \bayesianNet that starts from the roots of \bayesianNet is a \preorderTraversal.

% \begin{figure}[t!]
% \begin{center}
% \includegraphics[width=0.80\columnwidth]{img/linear_tree_simple.png}
% \end{center}
% \caption{\label{fig:linear_tree_simple}Consider the Bayesian Network in the figure below, along with its pre-order traversal $\eliminationOrder = (\varxi, \varomega, \varomega', \vara, \varb, \ldots, \varz)$. The two define the elimination tree shown in the figure, which includes a linear graph over the internal nodes.}
% \end{figure}

\begin{figure}[t!]
\begin{center}
\includegraphics[width=0.80\columnwidth]{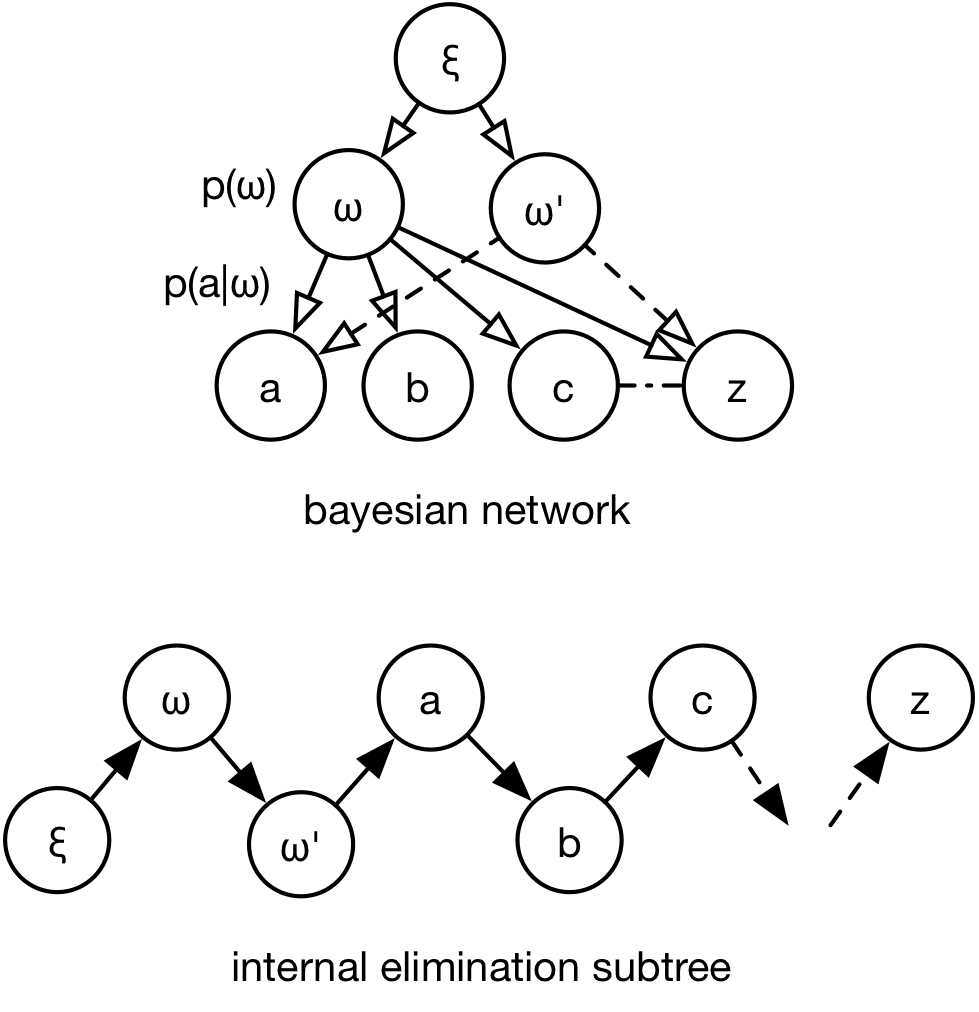}
\end{center}
\caption{Consider the Bayesian Network in the figure, along with its pre-order traversal $\eliminationOrder = (\varxi, \varomega, \varomega', \vara, \varb, \ldots, \varz)$. The two define the elimination tree shown in the figure, which includes a linear graph over the internal nodes.}
\label{fig:linear_tree_simple}
\end{figure}

\begin{lemma}
For a Bayesian Network \bayesianNet, let \tree be the elimination tree obtained for the elimination order defined by any pre-order traversal of \bayesianNet. The internal subgraph \coreTree of \tree is linear.
\label{lemma:linearStructure}
\end{lemma}

%\hide{
\begin{proof}
For any variable \varomega, let us consider the set of variables \{\vara, \varb, \ldots, \varz\} 
that depend on \varomega\ --- as well as ``co-dependents'' of \varomega, i.e.,
variables $\{\varomega'\}$ that depend on variables that overlap with the variables that \varomega depends on (see Figure~\ref{fig:linear_tree_simple}).
Notice that, in a pre-order traversal of \bayesianNet, these are the only variables that can follow \varomega in the traversal.
During elimination, node \varomega passes along a potential that involves all its dependants 
\{\vara, \varb, \ldots, \varz\} and co-dependents $\{\varomega'\}$ that have not been eliminated yet. 
As a result, {\it at every transition in the pre-order traversal, the factor that results from the 
latest elimination is a function of the next variable in the traversal, until the traversal of the last 
variable}. This leads to a directed edge in \tree from the current node in the elimination tree to the 
node that corresponds to the next variable in the traversal. These directed edges form the subgraph of 
internal nodes of \tree --- which is thus a linear graph that follows the pre-order traversal 
of \bayesianNet exactly.
\end{proof}
%}

\end{appendix}
}
\balance

\end{document}